%% file: Thesis_arxiv.tex
\documentclass[PhD]{iitmdiss}
\usepackage{times}
\usepackage{t1enc}
\usepackage{pdftricks}
\usepackage{mathrsfs}
\usepackage{amsmath,amssymb,float,amsthm}
\usepackage[verbose]{wrapfig}
\usepackage{balance}
\usepackage{url}
\usepackage{natbib}
\usepackage{adjustbox}
\usepackage{tikz}
\usepackage{array,multirow}
\usepackage{tabularx}
\usepackage{bm}
\usepackage{subfigure}
\usepackage{qcircuit}
\usepackage{blkarray,pifont}

\newtheorem{theorem}{Theorem}
\newtheorem{lemma}[theorem]{Lemma}

\newcommand{\id}{\bm{1}}
\newcommand{\tr}{\mathrm{Tr}}
\newcommand{\upi}{\mathrm{i}}
\newcommand{\upe}{\mathrm{e}}
\newcommand{\mi}{\mathrm{i}}

\newcommand{\cB}{\mathcal{B}}
\newcommand{\cC}{\mathcal{C}}

\newcommand{\cE}{\mathcal{E}}
\newcommand{\cF}{\mathcal{F}}
\newcommand{\cG}{\mathcal{G}}
\newcommand{\cH}{\mathcal{H}}

\newcommand{\cM}{\mathcal{M}}
\newcommand{\cN}{\mathcal{N}}

\newcommand{\cP}{\mathcal{P}}
\newcommand{\cQ}{\mathcal{Q}}
\newcommand{\cR}{\mathcal{R}}
\newcommand{\cS}{\mathcal{S}}
\newcommand{\cT}{\mathcal{T}}

\newcommand{\cV}{\mathcal{V}}
\newcommand{\cW}{\mathcal{W}}

\usepackage{hyperref,varioref}
\usepackage{amsmath} 
\usepackage{graphicx}
\usepackage{epstopdf}
\usepackage{fancyhdr}
\usepackage{ifpdf}
\ifpdf
\fi

\newcommand{\ket}[1]{\left|{#1}\right\rangle}

\newcommand{\be}{\begin{equation}}
\newcommand{\ee}{\end{equation}}
\newcommand{\bea}{\begin{eqnarray}}
\newcommand{\eea}{\end{eqnarray}}
\newcommand{\bdm}{\begin{displaymath}}
\newcommand{\edm}{\end{displaymath}}
\newcommand{\bit}{\begin{itemize}}
\newcommand{\eit}{\end{itemize}}
\newcommand{\ben}{\begin{enumerate}}
\newcommand{\een}{\end{enumerate}}

\catcode`@=11 \@addtoreset{equation}{chapter} \catcode`@=12

\newcommand{\nc}{\newcommand}
\nc{\beq}{\begin{equation}} \nc{\eeq}{\end{equation}}
\nc{\beqa}{\begin{eqnarray}} \nc{\eeqa}{\end{eqnarray}}
\nc{\lsim}{\begin{array}{c}\,\sim\vspace{-21pt}\\< \end{array}}
\nc{\gsim}{\begin{array}{c}\sim\vspace{-21pt}\\> \end{array}}
\def\Dslash{\not{\hbox{\kern-3pt $D$}}}

\usepackage{xcolor}
\usepackage{xparse,xcoffins}

\ExplSyntaxOn
\NewCoffin\imagecoffin
\NewCoffin\labelcoffin

\keys_define:nn { miguel/label }
 {
  label   .tl_set:N = \l_miguel_label_tl,
  labelbox .bool_set:N = \l_miguel_label_box_bool,
  labelbox .default:n = true,
  fontsize .tl_set:N = \l_miguel_label_size_tl,
  fontsize .initial:n = \normalsize,
  pos .choice:,
  pos/nw .code:n = \tl_set:Nn \l_miguel_label_pos_tl { left,up },
  pos/ne .code:n = \tl_set:Nn \l_miguel_label_pos_tl { right,up },
  pos/sw .code:n = \tl_set:Nn \l_miguel_label_pos_tl { left,down },
  pos/se .code:n = \tl_set:Nn \l_miguel_label_pos_tl { right,down },
  pos/n .code:n = \tl_set:Nn \l_miguel_label_pos_tl { hc,up },
  pos/w .code:n = \tl_set:Nn \l_miguel_label_pos_tl { left,vc },
  pos/s .code:n = \tl_set:Nn \l_miguel_label_pos_tl { hc,down },
  pos/e .code:n = \tl_set:Nn \l_miguel_label_pos_tl { right,vc },
  pos .initial:n = nw,
  unknown .code:n   = \clist_put_right:Nx \l_miguel_label_clist
                       { \l_keys_key_tl = \exp_not:n { #1 } }
 }
\clist_new:N \l_miguel_label_clist
\box_new:N \l_miguel_label_box
\box_new:N \l_miguel_label_image_box

\NewDocumentCommand{\xincludegraphics}{O{}m}
 {
  \group_begin:
  \tl_clear:N \l_miguel_label_tl
  \clist_clear:N \l_miguel_label_clist
  \keys_set:nn { miguel/label } { #1 }
  \tl_if_empty:NTF \l_miguel_label_tl
   {
    \miguel_includegraphics:Vn \l_miguel_label_clist { #2 }
   }
   {
    \SetHorizontalCoffin\imagecoffin
     {
      \miguel_includegraphics:Vn \l_miguel_label_clist { #2 }
     }
    \SetHorizontalCoffin\labelcoffin
     {
      \raisebox{\depth}
       {
        \bool_if:NTF \l_miguel_label_box_bool
         { \fcolorbox{white}{white}{\l_miguel_label_size_tl\l_miguel_label_tl} }
         { \l_miguel_label_size_tl\l_miguel_label_tl }
       }
     }
    \SetVerticalPole\imagecoffin{left}{3pt+\CoffinWidth\labelcoffin/2}
    \SetVerticalPole\imagecoffin{right}{\Width-3pt-\CoffinWidth\labelcoffin/2}
    \SetHorizontalPole\imagecoffin{up}{\Height-3pt-\CoffinHeight\labelcoffin/2}
    \SetHorizontalPole\imagecoffin{down}{3pt+\CoffinHeight\labelcoffin/2}
    \use:x{\JoinCoffins\imagecoffin[\l_miguel_label_pos_tl]\labelcoffin[vc,hc]} 
    \TypesetCoffin\imagecoffin
   }
   \group_end:
 }
\NewDocumentCommand{\setlabel}{m}
 {
  \keys_set:nn { miguel/label } { #1 }
 }

\cs_new_protected:Nn \miguel_includegraphics:nn
 {
  \includegraphics[#1]{#2}
 }
\cs_generate_variant:Nn \miguel_includegraphics:nn { V }

\ExplSyntaxOff
\input{title}

\begin{document}

\maketitle  
\fancyhf{}
\thispagestyle{empty}
\mbox{}
\setcounter{page}{0}
%
\input{cert}
%
\fancyhf{}
\mbox{}
\thispagestyle{empty}
\setcounter{page}{0}
%
\input{ack}
\fancyhf{}
\mbox{}
\thispagestyle{empty}
\mbox{}
\thispagestyle{empty}
\setcounter{page}{4}
\input{abs}
\fancyhf{}
\mbox{}
\mbox{}
\setcounter{page}{8}
\input{toc}

\mbox{}
\thispagestyle{empty}
\input{abbnot}


\pagenumbering{arabic}
\input{Chapter1_v3}

\input{Chapter2_v3} 
\input{Chapter3_v3}

\input{Chapter4_v2}

\input{Chapter5_v3}

\input{Chapter6}

\appendix
\input{appendixA_v2}
\input{appendixB_v2}
\input{appendixC_v3}

\begin{singlespace}
\bibliographystyle{utphys}
\bibliography{ref}
\end{singlespace}

\fancyhf{}
\thispagestyle{empty}
\thispagestyle{empty}
\mbox{}
\thispagestyle{empty}
\mbox{}
\thispagestyle{empty}

%
%
%
\fancyhf{}
\newpage
\mbox{}
\thispagestyle{empty}
\fancyhf{}
\mbox{}
\thispagestyle{empty}
%


\end{document}

%% file: title.tex
\title{Adaptive quantum codes: constructions, applications and fault tolerance}

\author{AKSHAYA J}

\date{OCTOBER 2020}
\department{PHYSICS}

%% file: cert.tex
\certificate
\setcounter{page}{-1}
\vspace*{0.5in}

\noindent This is to certify that the thesis entitled {\bf Adaptive quantum codes: constructions, applications and fault tolerance}, submitted by {\bf Akshaya J} to the Indian Institute of Technology Madras for the award of the degree of {\bf Doctor of Philosophy}, is a bonafide record of the research work done by her under our supervision. The contents of this thesis, in full or in part, have not been submitted to any other Institute or University for the award of any degree or diploma.

\vspace*{1.in}

\begin{minipage}{\textwidth}
\begin{minipage}{0.7\textwidth}
\end{minipage}%
\hfill
\begin{minipage}{0.29\textwidth}
\vspace*{1ex}
\vspace*{1ex}
\noindent {\bf Dr. Prabha Mandayam} \\ 
\noindent Research Guide \vspace{-1.2ex}\\ 
\noindent Associate Professor \vspace{-1.2ex} \\
\noindent Dept. of Physics  \vspace{-1.2ex}\\
\noindent IIT Madras \vspace{-1.2ex} \\

\vspace*{0.5in}

\end{minipage}
\end{minipage}

\vspace*{0.25in}
\noindent Place: Chennai, India\\
Date:  January 7, 2021

%% file: ack.tex
\acknowledgements
Hurray! It feels nice to finally see some light at the end of the tunnel. It has been a strenuous yet a wonderful experience striding into science. It is a great pleasure to thank everyone who has contributed to this ongoing journey in research.

First and foremost, I am extremely grateful to my advisor Dr.Prabha Mandayam who has been very encouraging and supportive right from the very beginning. She has inspired me with her knack in asking the right insightful questions while setting up a research problem and her perseverance in seeking a solution. I always admire her for her eloquence in expressing her thoughts, especially when it comes to explaining physics in a more comprehensible language. Her constructive criticism has certainly helped me in moulding myself as a student of physics. She has allowed me to grasp the subject at my pace and encouraged independent thinking, which helped me to identify my strengths and boundaries. I aspire to achieve her levels of perception, clarity and decisiveness some day. I am forever indebted to her for all that I have learnt from her.

I will always relish the moments of insightful discussions that I have had with Prof. Arul Lakshminarayan. He is a remarkable teacher and his systematic approach towards solving a research problem has always inspired me. I am very thankful for all his encouragement and kind words of advice. I am very thankful to my doctoral committee Chair, Prof. Lakshmi Bala who has always encouraged me. Her extraordinary teaching ability as well as her quick and witty responses have always awed me. I am very thankful to Prof. V. Balakrishnan, for his kind words of advice. I feel blessed to have had the opportunity to interact with him.

I am very thankful to Prof. Hui Khoon for all that I learnt from her through our collaborative projects. I am extremely thankful to Prof. M. V. Sathyanarayanan, Prof. Rajesh Narayanan, Prof. Suresh Govindarajan, Dr. Sunethra Ramanan, Dr. Vaibhav Madhok, Dr. Ashwin Joy, and Dr.Pradeep Sarvepalli for all the stimulating discussion sessions and their valuable feedback. I am indebted to Prof. K.P.N Murthy, who has always been a pillar of support and encouragement through my MSc days at Hyderabad. I thank our HoD, Prof. Sethupathy, for his suggestions and advice during my doctoral programme. I am very thankful to Prof. Sibasish Ghosh, Prof. Bhanu Pratap Das and Prof. M.V.N. Murthy for their keen interest in my progress and words of encouragement. I am thankful to IIT Madras for the financial support I received during my tenure as a research scholar and the computer facilities I could avail. I thank Tech. superintendent Mr.Rajan, for showing interest in my progress and motivating words. \\

I have thoroughly enjoyed the role of a teaching assistant. It enriched my learning experience through interactions and discussions with other students. I am glad that I have made a big 
group of friends here at IIT Madras. I am grateful to my friends- Shruti Dogra, Madhuparna, Vandana, Krishna Mohan, Vasumathy, Dipanwita, Roshna, Dileep and Sharmila for all their timely support, encouragement and discussions,
useful suggestions which kept me up through hard times. I will always cherish the interesting times I have had with them.  I thank my friends Aravinda, Abinash, N. Dileep, Suhail, Sreeram, Ipsita, Bhaswati, Anjala, Anant, Rohan, Kaushik, Manoj Gowda, Arindam, Sharat, Sashi, Malayaja for all the stimulating discussions and their help during needy hours. I thank My Duy Hoang Long for all the academic discussions that we have been having over the recent days.

I am sure words won't not be enough in expressing how indebted I feel towards my Amma and Appa, for their unconditional love, support, and encouragement. Their constructive criticisms have always helped me introspect myself and grow up to a better person. I am grateful to my immediate family- Paati, Banu, Viji Mami and Mama for their love and encouragement. I thank my parents-in-law for showing keen interest in my progress. Last but not the least, I am grateful to my husband Prasannaa for being an incredible source of strength and support over these years. I feel blessed with these people who have made this journey possible and a wonderful one.

%% file: abs.tex
\abstract

\noindent KEYWORDS: \hspace*{0.5em} \parbox[t]{4.4in}{Quantum error correction, Channel-adapted codes, Fault tolerance}

\vspace*{24pt}
 A major obstacle towards realizing a practical quantum computer is the `noise' that arises due to system-environment interactions.  While it is very well known that quantum error correction (QEC) provides a way to protect against errors that arise due to the noise affecting the system, a `perfect' quantum code requires atleast five physical qubits to observe a noticeable improvement over the no-QEC scenario.  However, in cases where the noise structure in the system is already known, it might be more useful to consider quantum codes that are adapted to specific noise models.  It is already known in the literature that such codes are resource efficient and perform on par with the standard codes. In this spirit, we address the following questions concerning such adaptive quantum codes. \\
(a) Construction:  Given a noise model, we propose a simple and fast numerical optimization algorithm to search for good quantum codes. Specifically, we search through the space of encoding unitaries by decomposing them using the Cartan form, into `local' and `nonlocal' parts. This allows us to explicitly search over the potential non-local parts thereby reducing our search parameters.  We also provide simple quantum circuits which can implement the optimal codes obtained via our search. \\
(b) Application: As a simple application of such noise-adapted codes, we propose an adaptive QEC protocol that allows transmission of quantum information from one site to the other over a $1$-d spin chain with high fidelity. We obtain explicit numerical and analytical results in the case of both ideal and disordered spin chains, provided the nearest neighbour interaction is governed by a total spin-preserving Hamiltonian. \\
(c) Fault-tolerance: Finally, we address the question of whether such noise-adapted QEC protocols can be made fault-tolerant. We show, by explicit construction, that it is possible to obtain fault-tolerant gadgets -- a universal gate set consisting of \textsc{ccz}, Hadamard gates and an error correction unit -- starting with a $[[4,1]]$ code, in systems where the elementary gates are assumed to be affected primarily by amplitude-damping noise. We obtain rigorous estimates of the critical error threshold, as $5.31\times 10^{-4}$ for the \textsc{cphase}-\textsc{exrec} unit and $2.8 \times 10^{-3}$ for the memory unit using our scheme.

\pagebreak

%% file: toc.tex

\begin{singlespace}
\tableofcontents
\thispagestyle{empty}
\mbox{}
\thispagestyle{empty}
\listoffigures
\addcontentsline{toc}{chapter}{LIST OF FIGURES}
\fancyhf{}
\mbox{}
\thispagestyle{empty}
\setcounter{page}{7}
\end{singlespace}

%% file: abbnot.tex
\setcounter{page}{18}
\abbreviations

\noindent 
\begin{tabbing}
xxxxxxxxxxx \= xxxxxxxxxxxxxxxxxxxxxxxxxxxxxxxxxxxxxxxxxxxxxxxx \kill
\textbf{CPTP}     \>Completely Positive and Trace Preserving \\
\textbf{QEC}     \> Quantum Error Correction \\
\textbf{CP}      \> Completely Positive \\
\textbf{AQEC}      \>Approximate Quantum Error Correction \\
\textbf{\textsc{cnot}}      \>Controlled-NOT \\
\textbf{\textsc{cz}}      \>Controlled-Phase \\
\textbf{\textsc{ccz}}      \>Controlled-Controlled-Phase \\
\textbf{\textsc{exrec}}      \>Extended Rectangle\\

\end{tabbing}
\fancyhf{}
\mbox{}
\thispagestyle{empty}
\pagebreak


\chapter*{\centerline{GLOSSARY OF SYMBOLS}}
\addcontentsline{toc}{chapter}{GLOSSARY OF SYMBOLS}

\begin{singlespace}
\begin{tabbing}
xxxxxxxxxxx \= xxxxxxxxxxxxxxxxxxxxxxxxxxxxxxxxxxxxxxxxxxxxxxxx \kill
$|\psi\rangle $               \> state vector \\ 
$\rho$                     \> density matrix \\
$\cH_S$                               \>  system Hilbert space  \\
$\cH_E$                               \>  environment Hilbert space  \\
$d$                        \> dimension of the Hilbert space \\ 

$\cB{(\cH_S)}$                               \> Set of bounded linear operators in the system Hilbert space  \\
$\cE$                               \>  arbitrary noise channel  \\
$E_i$                               \> Kraus operator of the noise channel $\cE$  \\
$\cH_S \otimes \cH_E$     \> Joint Hilbert space of the system and environment  \\
$\rho_{SE}$                               \> Joint density operator on the Hilbert space $\cH_S \otimes \cH_E$\\
$p, \gamma$                                     \>  noise parameter  \\ 
$\cC$                                     \>  codespace  \\ 
$\cW$                                     \>  encoding unitary map  \\ 
$\cR$                                     \>  recovery channel \\ 
$\cR_P$                                     \>  Petz recovery channel \\ 
$\cW^{-1}$                                     \>  decoding unitary map \\ 
$\{I,X,Y,Z\}$                \>single-qubit  Pauli operators \\ 
$P$                                     \>  projection onto the codespace \\ 
$\cP$                                     \> projection map \\ 
$\cN$                                     \>  normalization map \\ 
$F(.,.)$                        \> fidelity \\ 
$F_{\min}(.,.)$                        \>worst-case fidelity \\ 
$F^{2}_{\min}(.,.)$                        \>square of worst-case fidelity \\ 
$\eta(.,.)$                        \> fidelity-loss \\ 
$\Delta_{ij}$                        \>traceless matrices $\in$ codespace\\ 
$F(.,.)^{2}$                        \>sqaure of fidelity \\ 
$\sigma_X, \sigma_Y,\sigma_Z$                        \> alternate representation of single-qubit Pauli matrices\\ 
$SU(d)$                        \> special unitary group \\ 
$U$                        \>unitary operator \\ 
$\mathfrak{SU}(d)$                 \> Lie algebra \\ 
$\cE_{AD}$                        \> amplitude-damping channel \\ 
$\cE_{RAD}$                        \>rotated amplitude-damping channel \\ 
$\widetilde{\cE}_{AD}$                        \> amplitude-damping channel \\ 

$f_{r,s}^{N}(t)$                     \> transition amplitude with sender's site $s$, receiver's site $r$ on a $N$ spin chain after time $t$ \\ 
$\langle F^{2}_{\min}\rangle _{\delta}$                        \>disorder- averaged square of worst-case fidelity \\ 
$\delta$                        \> disorder strength\\ 
 $Loc$                        \> localization length\\
 $E_{1}$                        \>damping error\\
  $H$                        \>Hadamard\\
 $p_{\rm th}$                        \> noise threshold\\
 $\bar{X}$                        \> logical operator- XXII\\
  $\bar{Z}$                        \> logical operator- ZIZI\\
  $\beta_{ij}$                        \> Bell state\\
$\{s,t,u,v,h,g\}$                      \>syndrome bits\\
$R_{Z}$                                \>phase recovery\\
$S, S'$                                \>syndrome extraction units\\

\end{tabbing}
\end{singlespace}
\fancyhf{}
\mbox{}
\thispagestyle{empty}
\pagebreak
\clearpage

%% file: Chapter1_v3.tex

\chapter{Introduction} 

\label{Chapter1} 

\lhead{Chapter 1. \emph{Introduction}} 
A new paradigm for computing emerged in the early 1980s, hearlded by interesting results~\cite{deutsch, deutschjosza} that demonstrated that certain computational tasks could be performed much more efficiently using the laws of quantum mechanics. For example, the quantum algorithm devised by Shor~\cite{shor_factor} performs exponentially faster than any existing classical algorithm, whereas the quantum search algorithm offers a quadratic speedup over the best classical search algorithms~\cite{grover}. However, one of the major challenges in realizing a practical quantum computer is to tackle the \emph{noise} that arises out of inevitable system-environment interactions~\cite[Chapter 8]{nielsen}. A couple of remarkable results in the mid-nineties --- the invention of the  nine-qubit quantum error-correcting code~\cite{shor_qec} and the notion of quantum fault-tolerance ~\cite{shor_ft}, suggested the possibility of performing reliable quantum computation, even in the presence of noise. This laid the foundation for further developments in the field of quantum error correction (QEC) and fault tolerance~\cite{terhal}.

Quantum error correction (QEC) is a technique developed for protecting the quantum system against the effects of noise, by embedding the information to be protected in a subspace of the physical Hilbert space, called the \emph{quantum error-correcting code} or the \emph{codespace}~\cite{shor_qec, steane}. The process of encoding the information into a quantum code introduces redundancy, which makes it possible to detect and correct the errors that arise due to a noise process, via an appropriate recovery procedure. The challenges that a general theory of QEC needs to overcome are given below. 
\begin{itemize}
\item The \emph{no-cloning}~\cite{wootters} theorem prohibits the possibility of having quantum repetition codes similar to the classical codes. 
\item A continuum of errors may affect a quantum system, and identifying which error occurred could be difficult.
\item Measuring a quantum system can destroy the information that we wish to protect.
\end{itemize}

\noindent A general framework has been formulated~\cite{bennet,knill, ekert} providing the necessary and sufficient conditions for QEC, for any general noise model, overcoming  these challenges. Majority of the work on QEC centers around codes capable of removing the effects of \emph{arbitrary} errors on individual qubits perfectly. This is achieved by discretizing the errors in the Pauli operator basis, thereby offering protection against any unknown noise process. Such codes are called the \emph{perfect} codes, which strictly obey the so-called Knill-Laflamme conditions~\cite{nielsen}. The stabilizer codes~\cite{nielsen}, including the well-known Shor code~\cite{shor_qec}, Steane code~\cite{steane} and the five-qubit code~\cite{laflamme} fall in this category of codes. The shortest known perfect code protecting a qubit worth information is the standard five-qubit code~\cite{laflamme}. While quantum codes offer a certain degree of protection for noisy qubits, the theory of quantum fault tolerance deals with mitigating errors due to faulty quantum gates~\cite{preskill_FT}. The final step in realizing robust, universal quantum computation is therefore to identify fault-tolerant quantum gates and circuits, such that the quantum encoding and recovery processes work even when the elementary gates are noisy.

We are today in the so-called ``NISQ'' (noisy-intermediate scale quantum) era~\cite{preskill_nisq}, with access to quantum devices that have a small number of noisy qubits, that are not amenable to implement standard QEC and fault tolerance schemes. A standard QEC protocol assumes that the noise is unknown, and makes use of perfect codes which are resource intensive, to protect a single qubit against any arbitrary error. However, in practice, there are instances where a specific noise dominates the quantum system. Common examples are the case of an atom in a cavity undergoing a $T_1$ relaxation process due to the interaction with the incoming photons~\cite{haroche} and a superconducting qubit suffering from a dominant dephasing noise~\cite{dephasing}. In such a scenario, where we have prior knowledge about the noise afflicting the system, channel-adapted codes --- codes adapted to the noise in question --- are known to be more effective~\cite{leung, fletcher_codes, fletcher_rec, hui_prabha}. 

A four-qubit code adapted to the amplitude-dampingchannel protecting a single qubit worth information was constructed in~\cite{leung}. The four-qubit code was shown to satisfy a less restrictive, perturbed form of the perfect QEC conditions, yet performing comparably to the standard five-qubit code in protecting a single qubit of information. In subsequent works, a stabilizer-based approach of constructing codes~\cite{gottesman_stabilizer}, adapted to the amplitude-dampingchannel was provided in~\cite{fletcher_codes}. Furthermore, a numerical recovery map  was obtained~\cite{fletcher_rec} through semi-definite programming, adapted to the four-qubit code~\cite{leung}, recovering the information with high enough fidelity. At this juncture, it is important to point out that the Kraus operators that constitute the amplitude-dampingchannel are not scalable in terms of Pauli operators, thereby posing more challenge towards correcting them efficiently. 

In related work, \emph{approximate} QEC conditions were obtained~\cite{hui_prabha} by perturbing the Knill-Laflamme conditions and an analytical form of a universal, near-optimal recovery map, often referred to as Petz map~\cite{Petz} was obtained. It was further demonstrated~\cite{hui_prabha} that using the four-qubit code and Petz recovery one can achieve comparable protection against the amplitude-dampingchannel as the standard five-qubit code. All of these results suggest the need to develop ideas and techniques that go beyond the standard QEC formalism, that may lead to protocols that use less resources by taking advantage of the nature of the noise affecting the system.

Taking inspiration from these past results, we study three important aspects relating to channel-adapted codes, namely, construction, applications, and fault tolerance, in this thesis. In the subsequent sections, we give the background and motivation behind the work presented in this thesis.

\section{Construction of channel-adapted codes}
A QEC protocol is defined by the pair of encoding and recovery, for a given noise process. As mentioned earlier, while a vast majority of work centers around the \emph{perfect} QEC strategy requiring atleast five physical qubits to protect a single qubit, a couple of remarkable works demonstrated early on that the requirement of perfect QEC could be too restrictive~\cite{leung, fletcher_codes}. Rather, \emph{approximate} QEC strategies were proposed, and it was shown that a comparable protection against the amplitude-dampingnoise is possible, by chaarcterizing the \emph{worst-case fidelity} for the four-qubit code tailored to this noise~\cite{leung}. This construction was generalized to obtain a class of codes~\cite{fletcher_codes} adapted to the amplitude-dampingchannel, based on the stabilizer formalism. An optimal recovery map was constructed in Ref.~\cite{fletcher_rec} via semidefinite programming, adapted to the amplitude-dampingchannel, with optimality defined in terms of the \emph{entanglement fidelity}. A detailed note on the stabilizer based construction of the codes and optimal recovery maps adapted to the amplitude-dampingchannel can be found in~\cite{fletcherthesis}. Subsequently, convex-optimization techniques were used to obtain adaptive codes and adaptive recovery maps, using the entanglement fidelity~\cite{schum} as the figure of merit~\cite{kosut, yamamoto, reimpell}. 

Deviating from the numerical approach to approximate QEC problem in the past, a near-optimal analytical recovery map optimized for the \emph{average entanglement fidelity}, was constructed in~\cite{Barnum}.  Subsequenty, approximate quantum error correction conditions were obtained~\cite{beny} and optimal recovery maps based on the\emph{worst-case entanglement fidelity} were constructed. As mentioned earlier, more recently, a universal, near-optimal recovery map~\cite{Petz} was demonstrated for the \emph{worst-case fidelity}~\cite{hui_prabha} and an algebraic form of approximate QEC conditions that generalize the Knill-Laflamme conditions, was obtained. This led to a simple algorithm~\cite{hui_prabha} for finding approximate or the adaptive codes using the worst-case fidelity as the figure of merit.

Taking off from the approach outlined in~\cite{hui_prabha}, we study the problem of finding channel-adapted codes for an arbitrary noise process~\cite{ak_cartan}. Specifically, we study the optimization problem of AQEC, which involves finding the combination of the encoding and recovery that optimizes the chosen figure of merit. In our work, we focus on finding channel-adapted codes that minimize the worst-case fidelity for the storage of a single qubit of information. We reduce the original triple optimization to a single optimization by fixing the recovery as the universal, near-optimal map, namely the Petz recovery~\cite{Petz}. The use of the Petz recovery further permits the use of an analytical expression for the worst-case fidelity for codes encoding a single qubit~\cite{hui_prabha}, thereby requiring an optimization only over the encoding operations. 

A key aspect of our work is that we reduce the difficulty of this final numerical optimization over all possible encodings, by invoking the Cartan decomposition~\cite{khaneja_glaser} of the encoding operation. The Cartan decomposition splits up any $n$-qubit encoding unitary as an alternating product of single-qubit (local) and multi-qubit (nonlocal) unitaries. Under the well-motivated assumption that the noise process takes a tensor product structure over the $n$ qubits of the code, the Cartan decomposition allows us to explicitly search over the non-local encoding unitaries, thereby leading to the potential codespaces. Our algorithm uses the downhill-simplex technique to search over the space of the encodings. Altogether, these steps give a fast and easy algorithm for finding good channel-adapted codes for the worst-case fidelity, the preferred figure of merit for quantum computing tasks. The Cartan form also gives a nice structure to the optimal codes constructed using our numerical search procedure. This also leads to simple circuit implementations of these channel-adapted codes, which are amenable to implementation on the few-qubit quantum devices available today. 

\section{Applications: Quantum state transfer using channel adapted codes}
We next move on to demonstrate an important application of such channel-adapted QEC protocols in the context of quantum state transfer. Quantum state transfer entails transmission of an arbitrary quantum state from one spatial location to another. Spin chains are a natural medium for quantum state transfer over short distances, with the dynamics of the transfer being governed by the Hamiltonian describing the spin-spin interactions along the chains. Starting with the original proposal by Bose~\cite{bose} for state transfer via a $1$-d Heisenberg chain, several protocols were developed for {\it perfect} as well as {\it pretty good} quantum state transfer via spin chains. 

Perfect state transfer protocols typically involve engineering the coupling strengths between the spins in such a way as to ensure perfect fidelity between the state of the sender's spin and that of the receiver's spin~\cite{christandl,christandl2005perfect,albanesemirror,karbach,di}. Alternately, there are proposals to use multiple spin chains in parallel, and apply appropriate encoding and decoding operations at the sender and receiver's spins so as to transmit the state perfectly~\cite{conclusive,perfect,efficient}. Experimentally, perfect state transfer protocols have been implemented in various architectures including nuclear spins~\cite{bochkin} and photonic lattices using coupled waveguides~\cite{perez2013,chapman}. 

Relaxing the constraint of perfect state transfer, protocols for pretty good transfer aim to identify optimal schemes for transmitting information with high fidelity across permanently coupled spin chains~\cite{godsil2012,godsil}. One approach was to encode the information as a Gaussian wave packet in multiple spins at the sender's end~\cite{osborne,hasel}. Moving away from ideal spin chains, quantum state transfer has also been studied over disordered chains, both with random couplings and as well as random external fields~\cite{perfect, ashhab, chiara}.  

In our work, we study the problem of pretty good state transfer from a quantum channel point of view~\cite{ak_statetransfer}. It is known~\cite{bose} that state transfer over an ideal $XXX$ chain (also called the Heisenberg chain) can be realized as the action of an amplitude-dampingchannel~\cite{nielsen} on the encoded state. Naturally, this leads to the question of whether  quantum error correction (QEC) can improve the fidelity of quantum state transfer. While QEC-based protocols that achieve pretty good transfer were developed for noisy $XX$~\cite{kay, kay2018perfect} and Heisenberg spin chains~\cite{allcock} in the past, using perfect QEC strategies, we study the role of adaptive QEC in achieving pretty good transfer over a class of $1$-d spin systems, which preserve the total spin. This includes both the $XX$ as well as the Heisenberg chains, and more generally, the $XXZ$ chains. Motivated by \cite{leung, fletcher_codes, hui_prabha}, we propose a protocol to transmit the information efficiently by encoding it using channel-adapted codes, across multiple, identical and parallel spin chains. This is in contrast to the protocols in~\cite{kay, kay2018perfect} which use perfect QEC codes and encode into multiple spins on a single chain. Using the worst-case fidelity between the states of the sender and receiver's spins as the figure of merit, we demonstrate that pretty good state transfer may be achieved over a class of spin-preserving Hamiltonians using a channel adapted code and the Petz recovery map.

Finally, we present explicit results for the fidelity of state transfer obtained using our QEC scheme, for disordered $XXX$ chains. The presence of disorder in a $1$-d spin chain is known to lead to the phenomenon of localization~\cite{anderson}. This naturally leads to the question of whether it is possible to achieve state transfer over disordered spin chains~\cite{allcock}. In our work~\cite{ak_statetransfer}, we study the distribution of the transition amplitude, both numerically as well as analytically, for a disordered $XXX$ chain, with random coupling strengths drawn from a uniform distribution. We modify the QEC protocol suitably so as to ensure pretty good transfer when the disorder strength is small. As the disorder strength increases, our analysis points to a threshold beyond which QEC does not help in improving the fidelity of state transfer.

\section{Fault-tolerant quantum computation using the four-qubit code}

Finally, we study an important open question in the context of channel-adapted codes, namely, whether such codes can be used to develop fault-tolerant schemes for quantum computation. A typical quantum algorithm involves several hundred gate operations, so that even a single error anywhere could lead to an incorrect outcome. Therefore, a practical large-scale quantum computer must be tailored \emph{fault tolerantly}, so as to perform reliably even in the presence of noise.  A fault-tolerant quantum computer is protected from noise not just while storing the information, but at every stage of processing, thereby preventing the errors from spreading catastrophically within the circuit~\cite{gottesman_intro, gottesman_nature}. 

A fault-tolerant simulation of an ideal computation proceeds by performing encoded operations on the logical qubits of a quantum code. An encoded operation is implemented by a composite object called a \emph{gadget}, made of elementary noisy physical operations such as single-qubit and two-qubit gates which constitute the finite universal gate set~\cite{aliferis}. In a fault-tolerant computation, errors are removed from the system via periodic syndrome checks performed on ancillary qubits, before they accumulate to a point where the damage becomes irreparable. The central idea behind any such fault-tolerance scheme is the threshold theorem, which states that a scalable quantum computation is possible provided the physical error rate per gate or per time step is below a certain critical value called the \emph{threshold}~\cite{gottesman_intro}. 
The following assumptions~\cite{aliferis_thesis} are necessary for obtaining such a fault tolerance threshold:
\begin{itemize}
\item Faults affecting multiple qubits simultaneously are assumed to be suppressed in amplitude.
\item An inexhaustible supply of fresh ancilla qubits, or the ability to refresh and reuse the ancilla qubits an indefinite number of times.
\item The ability to apply quantum gates in parallel on a disjoint set of qubits, so as to not let the errors spread uncontrollably.
\end{itemize}

The final step enroute to building a robust quantum computer is to construct a scalable fault-tolerant architecture with high noise threshold, and minimum resource overheads. Shor's original proposal~\cite{shor_ft} shows how to build a polynomial size quantum circuit, which can tolerate an error rate per gate that decays polylogarithmically with the size of the computation. Currently, two families of quantum error correcting codes, namely, concatenated stabilizer codes and topological codes are considered leading candidates for building a fault-tolerant architecture, specifically for their lower error rate per logical qubit~\cite{campbell}. Concatenated codes are obtained by combining two codes of smaller distance, where the logical qubits of one code are encoded into the logical qubits of the other code recursively, to various levels of hierarchy. Although this increases the resources exponentially, the logical error rate falls off as a double exponential, in the number of levels that one concatenates into~\cite{nielsen}. Topological codes are codes whose physical qubits are laid on a lattice, allowing only the nearest neighbour interactions. Such codes are known to have larger distances (correct for more number of errors) and higher noise thresholds~\cite{lidar_brun}. For instance, a family of topological codes, called the surface codes~\cite{bombin}, have atleast $d^{2}$ physical qubits and achieve a distance $d$~\cite{campbell}. 

Once the quantum code architecture is decided, the next step is to be able to perform the finite universal set of gates in a fault-tolerant way on the logical qubits of the code. The finite universal gate set includes the clifford gates (Hadamard, $S$, \textsc{cnot}) and the $T$ gate~\cite{nielsen}. Some of the stabilizer codes allow for Clifford logical gates which are transversal, that is, the logical gate can be implemented as an $n$-fold tensor product of the elementary physical gates. For example, the logical Hadamard gate on the Steane seven-qubit code is implemented as $H^{\otimes 7}$, where $H$ is the single qubit Hadamard gate. Similarly, the logical $S$ gate and the logical \textsc{cnot} are also transversal in the case of the Steane code~\cite{campbell}. However, the logical $T$ gate is not transversal and is realised by a much more complicated construction. It has been shown ~\cite{eastin_restrictions,chen_restrictions} that a nontrivial code cannot realise a set of universal gates fault-tolerantly, in a tranversal manner, emphasizing the need for alternate constructions. One approach is to use the magic state distillation procedure~\cite{bravyi} to perform non-trivial gates such as the logical $T$ gate.

\begin{table}
\hspace*{-1cm}\begin{tabular}{|c|c|c|}
\hline
Fault-tolerant Scheme& Noise model& Threshold\\
\hline
\hline
General&Simple stochastic noise&$\sim 10^{-6}$~\cite{shor_ft}\\
\hline
Polynomial codes (CSS codes)&general noise model&$10^{-6}$~\cite{aharonov}\\
\hline
Concatenated $7$-qubit code&Adversarial stochastic noise&$2.7 \times 10^{-5}$~\cite{aliferis}\\
\hline
Error detecting concatenated codes ($C_4/ C_6$)&depolarizing noise&$1\%$~\cite{knill_nature}\\
\hline
Concatenated $7$-qubit code&Hamiltonian noise&$10^{-8}$~\cite{ng}\\
\hline
Concatenated repetition code& biased dephasing noise & $0.5 \%$~\cite{aliferis_biasedExp}, $0.24\%$~\cite{gourlay} \\
\hline
Surface codes&depolarizing&$0.75\%$~\cite{raus_FT} , $0.5\% -  1\%$~\cite{stephens}\\
\hline
$[[4,2,2]]$ concatenated toric code& depolarizing&$0.41\%$~\cite{criger}\\
\hline
\end{tabular}
\caption{\label{tab:table1} Noise thresholds for various fault-tolerant schemes.}
\end{table}

%


The current status of quantum fault tolerance is  summarized in Tab.~\ref{tab:table1}, where we present various fault-tolerant schemes and the thresholds obtained against a given noise model. We see from Tab.~\ref{tab:table1}, that the fault tolerance threshold obtained for a given QEC code or a class of codes, depends crucially on the error model under consideration~\cite{knill_nature}. Most of the threshold estimates obtained so far have assumed either symmetric depolarizing noise or erasure noise. However, when the structure of the noise associated with a given qubit architecture is already known, a fault tolerance scheme tailored to the specific noise process could be advantageous, both in terms of less resource requirements as also in improving the threshold. This is borne out by the fault tolerance prescription developed for \emph{biased} noise, in systems where the dephasing noise is known to be dominant ~\cite{aliferis_biased, FT_assymetry}. This construction was used to obtain a universal scheme for pulsed operations on flux qubits~\cite{aliferis_biasedExp}, taking advantage of the preferentially high degree of dephasing noise over bit-flip noise, in the \textsc{cphase} gate, leading to a numerical threshold estimate of $.0.5\%$ for the error rate per gate operation, as pointed out in Tab.~\ref{tab:table1}.

Motivated by this example, we address the question of whether it may be possible to develop fault-tolerant schemes for approximate, channel-adapted codes. As mentioned earlier, QEC codes adapted to specific noise models are often known to offer a similar level of protection as the general purpose codes, while using fewer qubits to encode the information~\cite{leung, fletcherthesis, hui_prabha}.  More recently, it has been demonstrated~\cite{gottesman_2016} that a family of fault-tolerant circuits can be constructed for the $[[4,2,2]]$ code, using which a criterion to test experimental fault tolerance may also be derived. The  $[[4,2,2]]$ code is known to detect arbitrary single qubit erasure errors at unknown locations~\cite{vaidman} and correct any arbitrary single qubit errors at known locations~\cite{grassl}. The usefulness of the $[[4,2,2]]$ error-detecting code in improving the fidelity of fault-tolerant gates over the fidelity of physical unencoded qubits has been studied in~\cite{flammia}. 

In our work, we demonstrate a universal scheme for fault-tolerant quantum computation in a scenario where the qubits are susceptible to amplitude-damping noise~\cite{ak_ft}. Physically, this corresponds to a $T_{1}$ relaxation process and is known to be a dominant source of noise in superconducting qubit systems. Specifically, we use the four-qubit code~\cite{leung} and develop a syndrome unit, recovery unit and a universal set of fault-tolerant gadgets, assuming \emph{local}, stochastic amplitude-damping noise. Finally, we analytically estimate the error threshold for the logical \textsc{cz} gadget as well as the memory unit.

\section{Thesis outline}

We have given a comprehensive overview of the motivations and background for this thesis work in the sections above. The rest of the thesis is organized as follows. In Chapter~\ref{Chapter2} we introduce the necessary mathematical preliminaries describing standard as well as approximate QEC. In Chapter~\ref{Chapter3}, we decribe our work on the construction of channel- adapted quantum codes using Cartan decomposition. In Chapter~\ref{Chapter4} we elaborate on the application aspect, where we study quantum state transmission using adaptive codes. In Chapter~\ref{Chapter5}, we describe in detail, how we may achieve fault tolerance using the four-qubit code. Finally, in Chapter~\ref{Chapter6}, we conclude with a  summary and future outlook of the work presented in this thesis.

%% file: Chapter2_v3.tex

\chapter{Preliminaries: Channel-adapted Quantum Error Correction} 

\label{Chapter2} 

\lhead{Chapter 2 \emph{Chap:2}} 
\section{Dynamics of open quantum systems}
\hspace{0.4cm} In this chapter, we motivate the need for QEC and formally define the concepts of \emph{perfect} and \emph{approximate} QEC. It is known that the dynamics of a closed quantum system is described by a unitary transformation as depicted in Fig.~\ref{fig:sys_env}. However, a quantum system is never fully isolated due to unavoidable interaction with the environment. This unwanted interaction can give rise to \emph{noise} that affects the quantum system. For example, an electron under the action of an unwanted stray magnetic field tends to precess along the direction of the external field, leading to a loss of information about the original state that it started with. 

The evolution of an \emph{open} quantum system is described by a \emph{quantum channel}~\cite[Chapter 8]{nielsen} as shown in Fig.~\ref{fig:sys_env}. This is mathematically described by a \emph{completely-positive and trace-preserving} (CPTP) map  $\cE$, acting on the system state $\rho$ $\in$ $\cB(\cH_S)$, as, \[ \cE(\rho) = \sum_{i}E_{i}\rho E_{i}^{\dagger}, \] 
where $\cH_S$ refers to the Hilbert space of the system and $\cB(\cH_S)$ refers to the set of bounded linear operators on $\cH_S$. The operators $\{E_{i}\}$ that constitute $\cE$ are known as Kraus operators, satisfying the \emph{trace-preserving} (TP) condition $\sum_{i}E_{i}^{\dagger}E_{i}=I $. The TP condition ensures that $\rm tr(\cE(\rho))=1$, $\forall \ \rho$. The map $\cE$ is \emph{completely positive} (CP) if and only if, (i) $\cE(\rho) > 0$, $\forall \ \rho$, implying that the resulting operator after the action of the map is also a positive definite matrix, (ii) $\cE\otimes \mathbb{I}$ is a positive map on $\cB(\cH_{S} \otimes \cH_{E})$, for an extended Hilbert space $\cH_{S} \otimes \cH_{E}$, where $\cH_S$ and $\cH_E$ refer to the system and environment Hilbert space respectively. This is an important physical requirement to ensure that $(\cE\otimes \mathbb{I}) \ (\rho_{SE})$ is a valid density operator on the Hilbert space $\cH_{S} \otimes \cH_{E}$,  with the joint density matrix $\rho_{SE}$ $\in$ $\cH_{S} \otimes \cH_{E}$,  such that $\cE$ acts only on $H_{S}$ and $\mathbb{I}$ is the identity operator on $\cH_{E}$. It was shown~\cite{choi,kraus1971,kraus} that a map $\cE$ is completely positive, if and only if, there exists a set of operators $\{E_{i}\}$ (the so-called Kraus operators), such that the action of the map $\cE$ on the density matrix $\rho$ is given by, $\cE(\rho)$=$\sum_{i}E_{i}\rho E_{i}^\dagger$. This form of the map is also referred to as the \emph{Kraus representation} of $\cE$. Note that the Kraus representation of a channel $\cE$ is not unique, since two different sets of Kraus operators $\{E_{i}\}$ and $\{F_{i}\}$ related by a unitary transformation can describe the same channel~\cite[Theorem 8.2]{nielsen}.
  
\begin{figure}[H]
\centering
\includegraphics[scale=0.8]{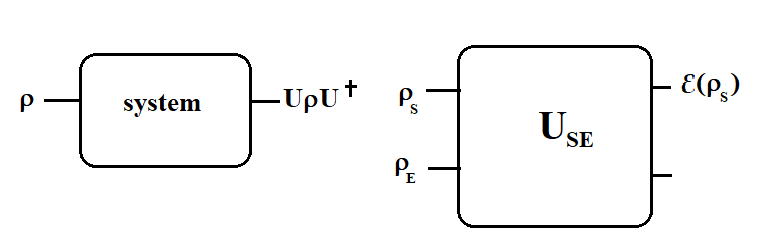}
\caption{ Closed system evolution depicting a system with state $\rho$ evolving under the unitary $U$ (left); Open system evolution (right), where $\rho_S$ is the initial state of system, $\rho_E$ is the initial state of the environment, $U_{SE}$ is the system-environment interaction, $\cE(\rho_S)$ is the evolved state of the system.}
 \label{fig:sys_env}
\end{figure}

\noindent We present a few examples of single-qubit quantum channels below. \\
{\bf (1) Bit-flip and phase-flip channels}: A bit-flip channel is constructed analogous to a classical binary symmetric channel, flipping a qubit from the state $|0\rangle$  to $|1\rangle $ and vice-versa, with a probability $p$. It leaves the state unchanged with a probability $1-p$. The action of the bit-flip channel is depicted in Fig.~\ref{fig:bitflip} below.
\begin{figure}[H]
\centering
\includegraphics[scale=1]{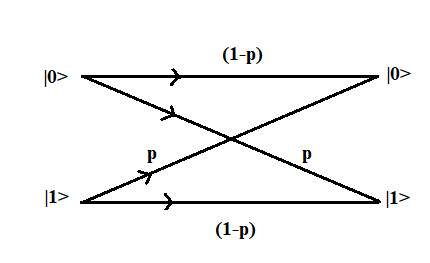}
\caption{Bit-flip channel}
\label{fig:bitflip}
\end{figure}
\noindent The Kraus operators of this channel in the basis $\{|0\rangle, |1\rangle\}$ are given by, 
\begin{equation}\label{eq:bitflipnoise}
E_{0}  = \sqrt{1-p} \ {I}= \sqrt{1-p} \left( \begin{array}{cc}
1 & 0 \\
0 &1
\end{array} \right), \; E_{1} = \sqrt{p}X =\sqrt{p}\left( \begin{array}{cc}
0 & 1\\
1 & 0
\end{array} \right).
\end{equation} where ${I}$ refers to the identity map and $X$ refers to the Pauli X operator.
The effect of the phase-flip channel is similar to the bit-flip channel, but with the Pauli Z operator replacing the Pauli X operator. The Kraus operators for the phase-flip channel in the basis $\{|0\rangle, |1\rangle\}$ are given as,  \begin{equation}
E_{0}  = \sqrt{1-p} \ {I}= \sqrt{1-p} \left( \begin{array}{cc}
1 & 0 \\
0 &1
\end{array} \right), \; E_{1} = \sqrt{p}Z =\sqrt{p}\left( \begin{array}{cc}
1 & 0\\
0 & -1
\end{array} \right).
\end{equation}

{\bf (2) Depolarizing channel}: The \emph{depolarizing} channel leaves a state completely mixed (depolarized) with a probability $p$ and leaves the state unchanged with a probability $1-p$. The effect of  a depolarizing channel on a density matrix $\rho$ can be written as $\cE(\rho)$= $(1-\frac{3p}{4})\rho + \frac{p}{4} (X\rho X +Y\rho Y+ Z\rho Z)$, where $X,Y,Z$ are the single-qubit Pauli operators. The Kraus operators in this case are given by, \[\{\sqrt{(1-3p/4)}{I},\sqrt{p}X/2,\sqrt{p}Y/2,\sqrt{p}Z/2\}.\]

\noindent {\bf (3) Amplitude-damping channel}: This is an important physical channel, describing the energy dissipation process in a two-level system. For instance, the state of an atom spontaneously decaying to the ground state by losing a photon, is well described by the action of the amplitude-damping channel (see chapters 7, 8 of ~\cite{nielsen}). The Kraus operators of the amplitude-damping channel $\cE_{AD}$ described in the $\{|0\rangle, |1\rangle \}$ basis are given as,
 \begin{equation}\label{eq:ampdamp}
E_{0}  =  \frac{(1+\sqrt{1-p}){I}}{2} +  \frac{(1-\sqrt{1-p}) {Z}}{2}=  \left( \begin{array}{cc}
1 & 0 \\
0 &\sqrt{1-p}
\end{array} \right), \nonumber
\end{equation}
\begin{equation}
 E_{1} =  \frac{\sqrt{p}(X+ i Y)}{2}=\sqrt{p}\left( \begin{array}{cc}
0 & 1\\
0 & 0
\end{array} \right).
\end{equation}
The operator $E_{0}$ leaves $|0\rangle$ unchanged, but reduces the amplitude of  $|1\rangle$, with a probability $1-p$. The operator $E_{1}$ damps $|1\rangle$ to $|0\rangle$ with a probability $p$. An interesting point to note here is that no linear combination of $E_{0}$ and $E_{1}$ can give an element proportional to a single-qubit Pauli operator, which in turn makes it challenging to correct for the noise arising from such a channel.

\section{Quantum error correction (QEC)}

\hspace{0.4cm} QEC provides a framework to protect quantum systems against noise resulting from system-environment interactions, as described in the last section. The basic idea behind QEC is to \emph{encode} the information that we wish to protect, namely a $d_0$-dimensional Hilbert space $\cH_0$, in a $d_0$-dimensional subspace $\cC$ of a larger Hilbert space $\cH$. The \emph{encoding} is an isometry $\cW$ whose action is described as the map $\cW:\cB(\cH_0)\rightarrow \cB(\cC)\subseteq\cB(\cH)$, where the encoded space is often $\cH$= $\cH_{0}^{\otimes n}$. The action of noise on the encoded space is described by a CPTP map $\cE: \cB(\cC) \rightarrow \cB(\cH)$. One then applies a recovery $\cR: \cB(\cH) \rightarrow \cB(\cC)$, a CPTP map, that reverses the effects of noise and maps the state back into the codespace. Finally, the decoding unitary $\cW^{-1}$ is applied, bringing the corrected information back to the original physical Hilbert space that we started with, given by, $\cW^{-1}: \cB(\cC) \rightarrow \cB(\cH_0)$. A schematic of a typical QEC protocol is shown in Fig.~\ref{fig:schematic} below. \\
\begin{figure}[H]
\centering
\includegraphics[scale=.5]{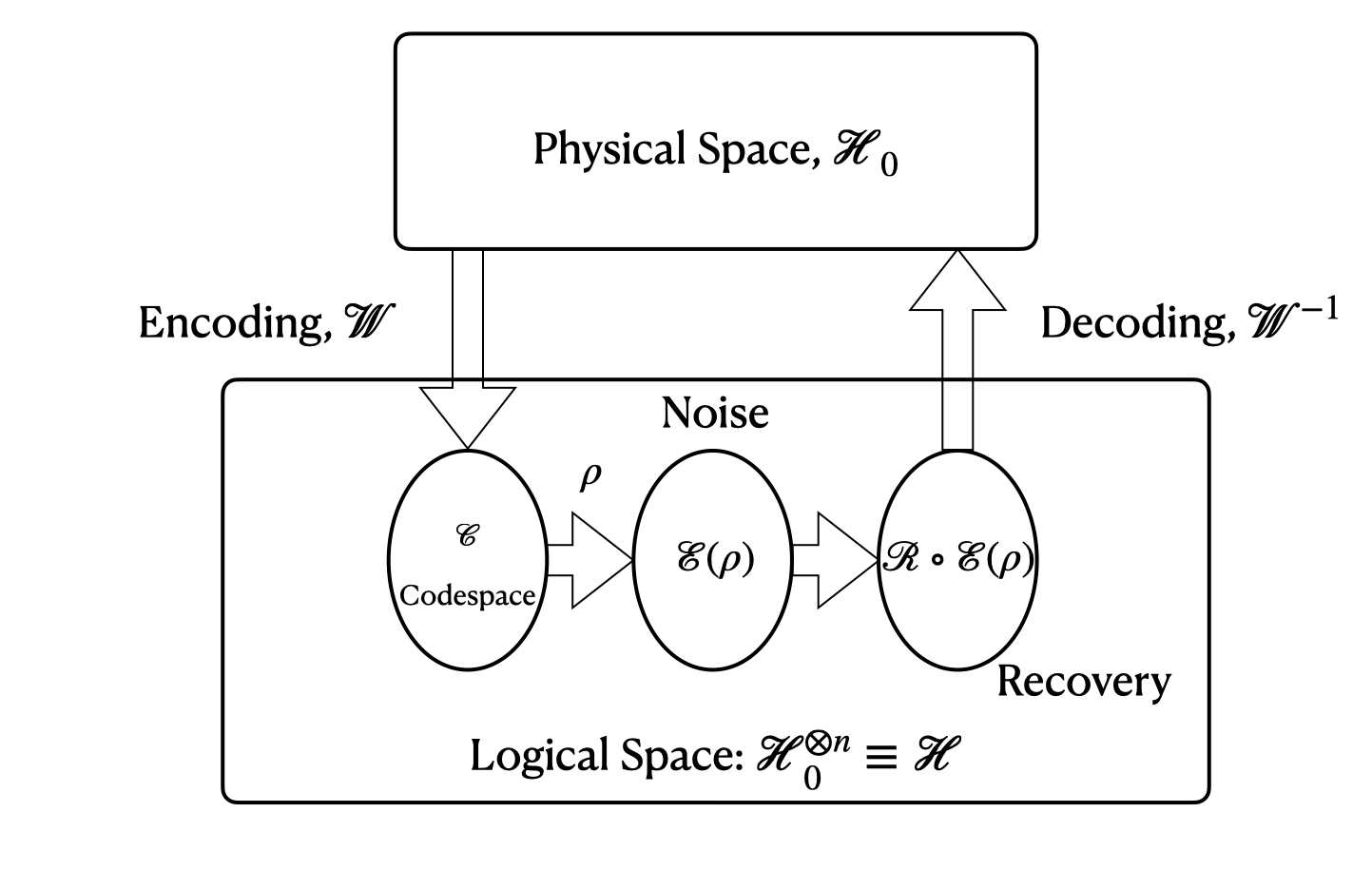}
\caption{Schematic of QEC}
\label{fig:schematic}
\end{figure}
We move on to formally elaborate on the perfect and approximate QEC strategies in subsequent sections.
\subsection{Perfect QEC}
\hspace{.4cm}A QEC protocol is described by an encoding $\cW$ and a recovery $\cR$, for a given noise process $\cE$. This allows one to perfectly correct for the complete CPTP noise channel or correct upto $t$ errors or $t$ Kraus operators that constitute the channel $\cE$. Such a strategy is deemed as \emph{perfect} QEC. Formally, a quantum code $\cC$ is said to perfectly protect against the noise channel $\cE$ composed of the Kraus operators $\{ E_{i}\}$, if there exists a recovery map $\cR$ such that $\cR\circ\cE (\rho) \propto \rho$, for all $\rho$ $\in$ $\cC$. Algebraically, the conditions for perfect QEC~\cite{ekert,knill,bennet} are given below.
\begin{theorem}[Conditions for perfect QEC]
A necessary and sufficient condition for the existence of a recovery operation $\cR$ correcting a set of Kraus operators $\{E_{i}\}$ on $\cC$~\cite[Theorem 10.1]{nielsen}, is given by,
\begin{equation}\label{eq:perfectqec}
 PE_{i}^{\dagger}E_{j}P = \alpha_{ij}P, \; \forall i, j, 
\end{equation}
where $P$ is the projection onto $\cC$, and $\alpha_{ij}$ are complex elements of a Hermitian matrix $\alpha$.
\end{theorem}
These conditions allow us to identify codes for a given noise channel $\cE$, such that the noise becomes perfectly correctible via an appropriate CPTP recovery channel $\cR$.  We refer to \cite{nielsen} for a formal proof of this condition. Rather, in what follows, we will describe what this condition implies about the nature of the noise and its action on the codespace.

Note that we can always find a unitary matrix $u$, such that the matrix $\alpha$ in Eq.~\ref{eq:perfectqec} can be diagonalized as $d = u^{\dagger} \alpha u$, where $d$ is a diagonal matrix. Now, defining the operators, $F_{k}$= $\sum_{i} u_{ik} E_{i}$, where $u_{ij}$ are the elements of unitary $u$, we can rewrite the perfect QEC conditions in Eq.~\ref{eq:perfectqec} as,
\begin{equation}\label{eq:diagqec}
PF_{k}^{\dagger}F_{l}P=\delta_{kl}d_{kk}P.
\end{equation}
Here, $d_{kk}$ are the elements of a diagonal matrix $d$, such that $d_{kk} > 0$, $\forall \ k$, since the left side of Eq.~\ref{eq:diagqec} is a positive semi-definite matrix when $k=l$. Using polar decomposition, and following Eq.~\ref{eq:diagqec}, we can write,
\begin{equation}\label{eq:polar}
F_{k}P = U_{k}\sqrt{PF_{k}^{\dagger}F_{k}P}= \sqrt{d_{kk}}U_{k}P.
\end{equation}
Thus, the action of the Kraus operator $F_{k}$ on the codespace $\cC$ is to rotate the codespace into a different subspace, dictated by $U_{k}$, such that the new subspace is given by $U_{k}PU_{k}^{\dagger}$. It is now easy to notice that the overlap between the two rotated subspaces, $U_{k}PU_{k}^{\dagger}U_{l}PU_{l}^{\dagger}$ =0, when $k\neq l$, from Eq.~\ref{eq:polar}. Thus, the effect of the operators $\{F_{k}\}$ on the codespace is to scramble the information into mutually orthogonal subspaces as depicted in Fig.~\ref{fig:perfectqec}. This allows us to construct the recovery $\cR_{\rm perf}$, whose Kraus operators $\{R_{k}\}$ are $\{PU_{k}^{\dagger}\}$, satisfying, \[ (\cR_{\rm perf} \circ \cE)  (\rho) =\sum_{k}d_{kk} P, \ \forall \rho \in \cC\] where $\sum_{k}d_{kk}$ is the trace of $\cE(\rho)$, which is $1$ only when $\cE$ is TP.

\begin{figure}[H]
\centering
\includegraphics[scale=.4]{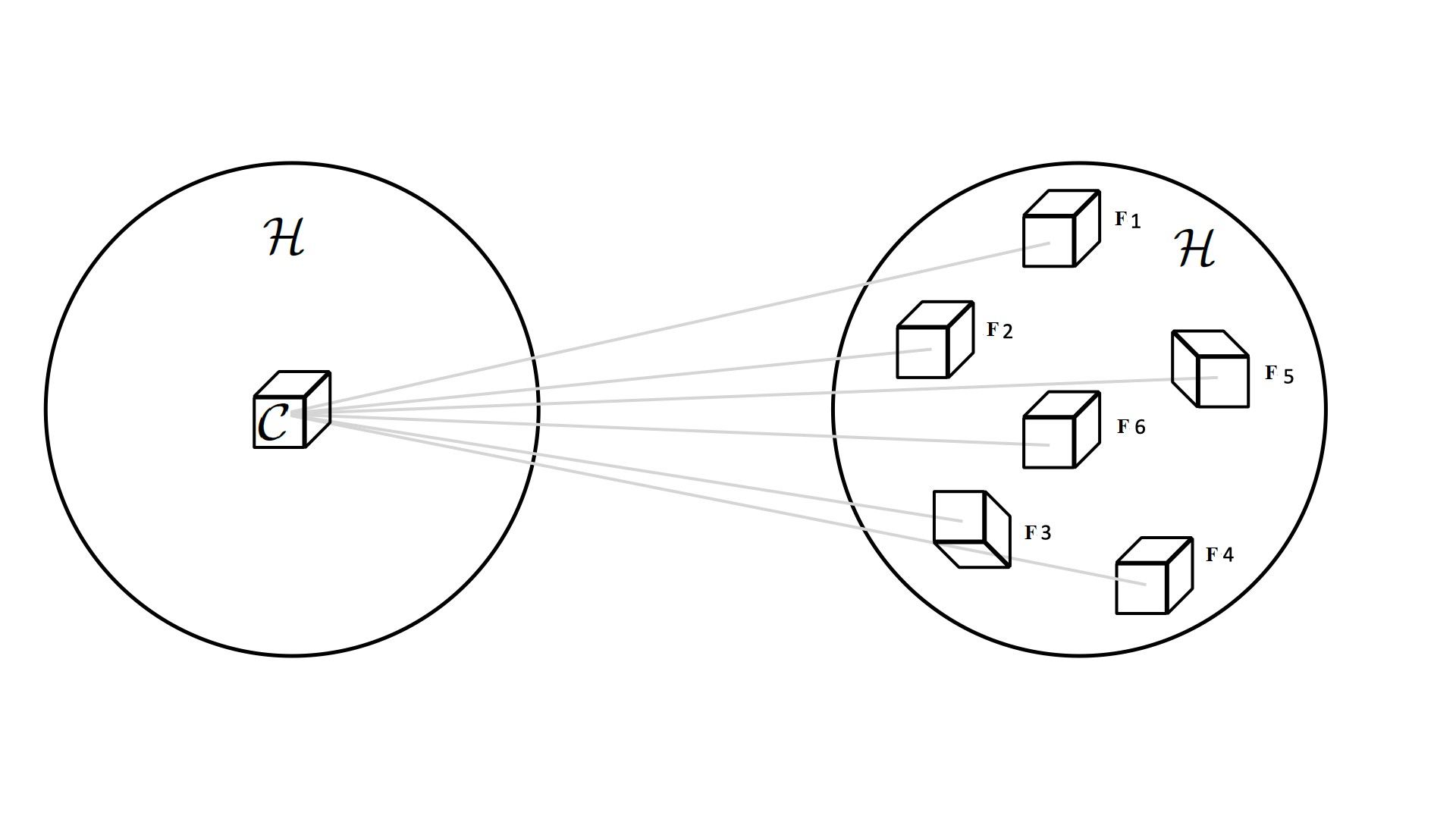}
\caption{Perfect QEC: $\cB(\cC) \rightarrow \cB(\cH)$~\cite{nielsen}}
\label{fig:perfectqec}
\end{figure}

The set of errors $\cE$ $\sim$ $\{E_{i}\}$ that satisfy Eq.~\ref{eq:perfectqec} are said to be \emph{correctable}. Any linear combination of $\{E_{i}\}$ is also correctable since Eq.~\ref{eq:perfectqec} is linear. Hence, it suffices to check the perfect QEC conditions in Eq.~\ref{eq:perfectqec} for the ''Pauli errors'', namely $\{{I},X,Y,Z\}$, since any single-qubit channel can always be expanded in terms of this Pauli error basis. 
Based on its error-correction properties, a quantum code is characterized by three parameters $n$, $k$, $d$. An $[[n,k,d]]$ quantum code is one where $k$ qubits that we wish to protect are encoded into $n$ physical qubits. $d$ refers to the \emph{distance} of a code, given by $2t+1$, where $t$ represents the number of correctable errors.

The shortest length perfect quantum code is the five-qubit code protecting a single qubit worth of information, denoted as, $[[5,1,3]]$~\cite{laflamme}.  It can correct for any arbitrary error occuring on a single qubit, such that $t=1$ and distance $d=3$, by discretizing errors in the Pauli basis. It saturates the so called quantum Hamming bound~\cite{nielsen}(Chapter 8). The other examples of perfect codes include Steane's $[[7,1,3]]$ code~\cite{steane} and Shor's $[[9,1,3]]$ code~\cite{shor_qec}. There are many known constructions of perfect codes in the literature~\cite{nielsen, lidar, gottesman_stabilizer,terhal}. 

The QEC schemes discussed above, that involve encoding, followed by detecting errors, and finally use a recovery operation to correct for the errors, fall under the ambit of \emph{active} QEC. An alternate approach adopted in fixing the errors is the \emph{passive} QEC strategy~\cite[Chapter 3]{lidar}, where the information to be protected is stored in subspaces or subsytems of a larger Hilbert space that remain unaffected by the noise, called \emph{Decoherence-free subspaces} (DFS)~\cite{lidar_paper,zanardi} and \emph{Noiseless subsystems} (NS)~\cite{shabani,holbrook}, respectively. The existence of such DFS/NS is based on underlying symmetries in the noise structure.  A unified algebraic framework called \emph{operator QEC}~\cite{kribs} lays down the conditions for both active and passive QEC techniques. Besides these approaches, the technique of dynamical decoupling~\cite[Chapter 4]{lidar} involves applying periodic control pulses to nullify the system environment interactions.
In this thesis we focus on active QEC schemes.

\subsection{Stabilizer codes}

Stabilizer codes form an important class of quantum error correcting codes, that are derived from the classical linear codes. Their construction is based on the group-theoretic structure of the Pauli matrices. In what follows, we will briefly review the \emph{stabilizer formalism} underlying the construction of such codes~\cite[Chapter~10]{nielsen}. Recall that the single-qubit Pauli operators $\{I,X,Y,Z\}$, form a group under multiplication. The Pauli-operator basis for $n$ qubits, which comprises $n$-fold tensor products of the single-qubit Pauli operators, along with the multiplicative factors $\{\pm 1,\pm i\}$ also form a group $\mathscr{G}_{n}$, represented as,
\[\mathscr{G}_{n}= \{\pm I^{\otimes n}, \pm i I^{\otimes n}, \pm X I^{\otimes n-1}, \pm i X I^{\otimes n-1}, \ldots, \pm i Z^{\otimes n}\} .\]
A stabilizer group $\mathscr{S}$ is defined as an abelian subgroup of $\mathscr{G}_{n}$, with elements that commute with each other, excluding the element $-I$.

\emph{Stabilizer code~\cite{gottesman_stabilizer}}: A stabilizer code is defined as the vectorspace $\cV$ which is left invariant after the action of the abelian subgroup $\mathscr{S}$ of $\mathscr{G}_{n}$ in the following way.
\begin{equation}\label{eq:stabilizer}
A |v\rangle = |v\rangle,  A \in \mathscr{S}, \{ |v\rangle \} \in \mathscr{V} ,
\end{equation}
where $A$ is any element of $\mathscr{S}$ and the set of vectors $|v\rangle$ which span $\cV $ are the $+1$ eigenvalue eigenstates, common to all the elements in $\mathscr{S}$. This set of vectors $\{|v\rangle\}$ is said to be stabilized by the elements of $\mathscr{S}$ and is said to span the stabilizer code.

Any Pauli operator within the group $\mathscr{G}_n$ either commutes or anticommutes with the elements of $\mathscr{S}$. A detectable error $B$ on the codespace is the one that anticommutes with atleast one element of $\mathscr{S}$, such that $B A =-A B, A \in\mathscr{S}$. Following Eq.~\ref{eq:stabilizer}, the action of the error $B$ on $\cV$ is then given by,
\begin{equation}
A(B|v\rangle) =- B( A|v\rangle) = -B|v\rangle .
\end{equation}
Thus the error operator $B$ acting on $\cV$ maps the information to be protected from the $+1$ eigenspace to the  $-1$ eigenspace of $A$, thus making it possible to detect and correct the action of the error. 

The well known $3$-qubit bit-flip code is an example of a stabilizer code (as well as a perfect code), constructed from a classical repetition code. Obtained as the span of $\{|000\rangle,|111\rangle\}$, this code is stabilized by the abelian subgroup $\{I^{\otimes 3}, ZIZ, ZZI, IIZ\}$ of the three-qubit Pauli group. A minimal set of operators called the \emph{generators} are enough to obtain all the elements of $\mathscr{S}$. The generators of the three-qubit bit-flip code are given by $\langle ZZI, IZZ\rangle$.

Consider the case of single-qubit bit-flip noise shown in Eq.~\ref{eq:bitflipnoise}. Observe that a single-qubit bit-flip error on any of the three qubits anticommutes with atleast one of the stabilizer generators, thereby mapping the codespace to the $-1$ eigenvalue eigenspace of those stabilizer elements that anticommute with the single-qubit bit-flip error. For example, a bit-flip error on the first qubit, denoted by $XII$, anticommutes with $ZZI$ and $ZIZ$. The information about which qubit is affected due to the error is identified by performing a \emph{syndrome measurement}. Such a measurement procedure typically involves measuring the generators of the stabilizer group $\mathscr{S}$ without disturbing the quantum state of the system, by adding ancillary systems. The measurement outcomes are referred to as \emph{syndrome bits} and they uniquely identify the correctable errors. Finally, after the syndrome extraction procedure, one applies a recovery operation to correct for the error. We show the syndrome bits obtained against each single-qubit bit-flip error, and the corresponding recovery operator in the table below.
 \begin{figure}[H]
 \hspace{-1cm} \includegraphics[scale=.9]{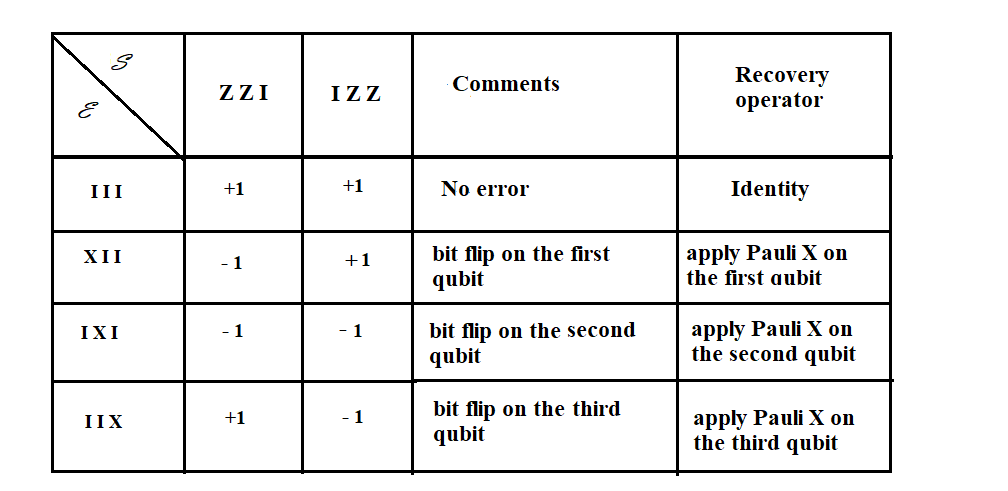}
 \caption*{Syndrome bits obtained by measuring the generators of $\mathscr{S}$ for the no error and single-qubit bit-flip errors $\in$ $\cE$, along with the recovery operators.}
   \label{tab:bitflip}
  \end{figure}
In the table above, the values $\pm 1$ correspond to obtaining $+1$ and $-1$ eigenstates of the corresponding generators of $\mathscr{S}$, during the syndrome measurement shown in Fig.~\ref{fig:syndrome}. The encoding circuit for the three-qubit bit-flip code is shown in Fig.~\ref{fig:enc} below.
 \begin{figure}[H]
 \centering
 \includegraphics[scale=.6]{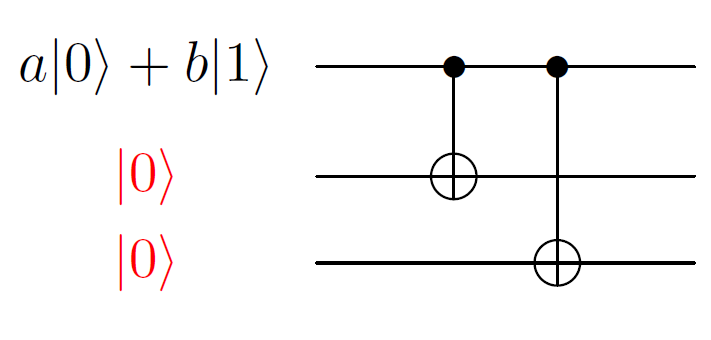}
 \caption{Encoding circuit for three-qubit bit-flip code }
 \label{fig:enc}
 \end{figure}
In Fig.~\ref{fig:enc}, the single qubit information to be protected from the bit-flip noise is given by $a|0\rangle+b|1\rangle$, where $a$ and $b$ satisfy $\mid a\mid^2 +\mid b\mid^2=1$. The state after encoding is obtained as $a|000\rangle +b|111\rangle$. The syndrome measurement unit extracting the syndrome bits for the bit-flip noise is shown in Fig.~\ref{fig:syndrome} below. This unit uses two ancilla qubits indicated in red, such that the first ancilla measures $ZZI$ and the second ancilla qubit measures $IZZ$ on the three-qubit encoded state.
 \begin{figure}[H]
 \centering
 \includegraphics[scale=.5]{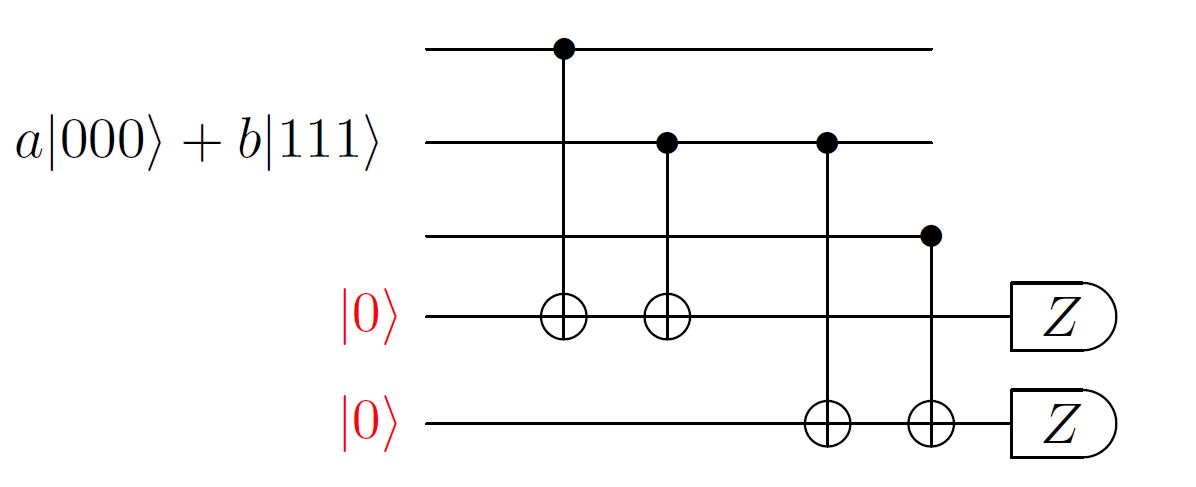}
  \caption{Syndrome measurement unit measuring $ZZI$ and $IZZ$}
  \label{fig:syndrome}
  \end{figure}
We see from this example of the $3$-qubit bit-flip code that the stabilizer formalism serves as a powerful tool in constructing quantum codes, detecting and correcting errors. We refer to~ \cite[Chapters~2,6]{lidar} for a complete overview of the stabilizer formalism.

\section{Approximate QEC}

\begin{figure}[H]
\centering
\includegraphics[scale=.4]{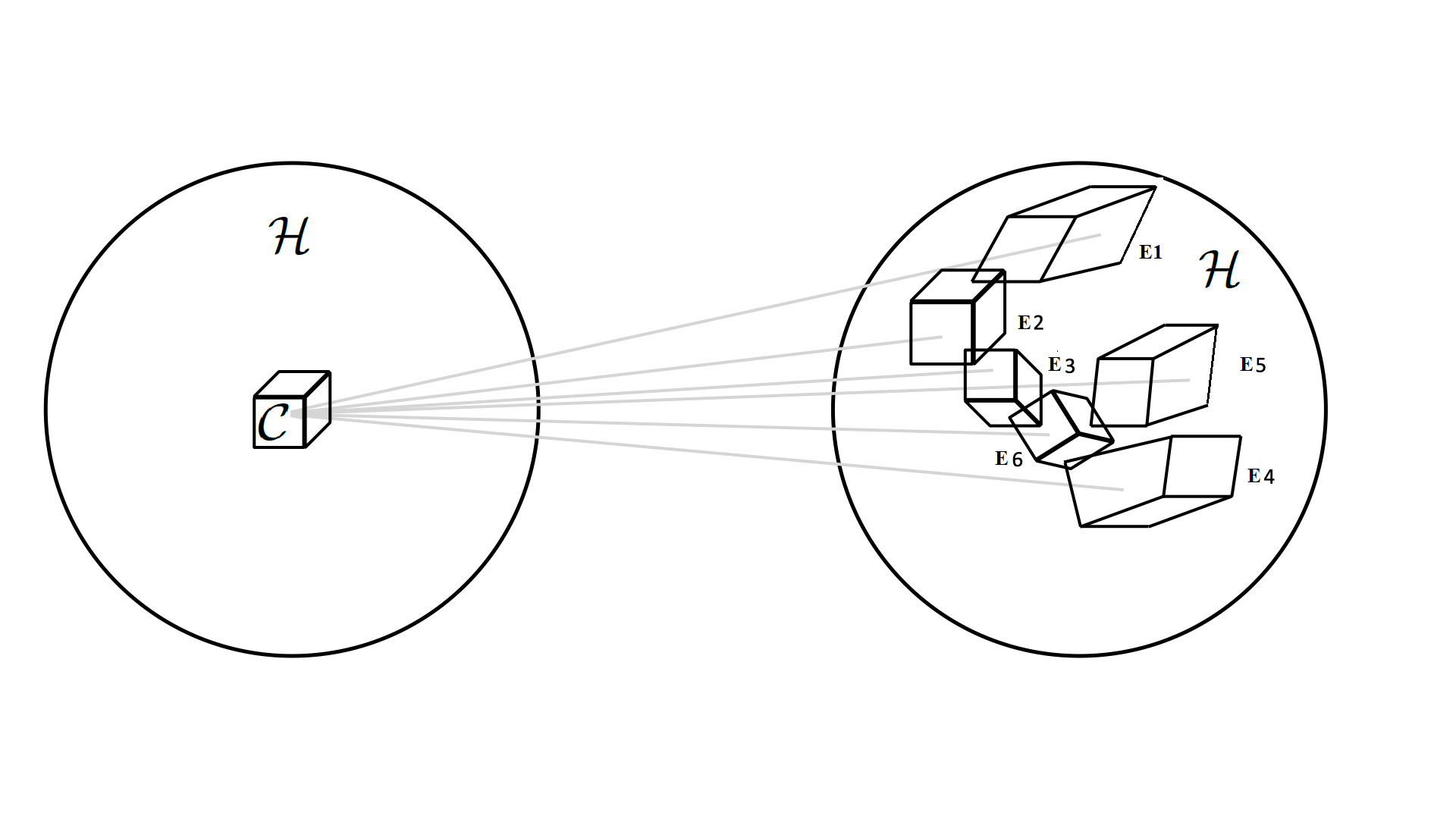}
\caption{Approximate QEC: $\cB(\cC) \rightarrow \cB(\cH)$~\cite{nielsen}}
\label{fig:aqec}
\end{figure}

A perfect QEC protocol demands that a quantum code get mapped to mutually orthogonal subspaces, as captured by the QEC conditions in Eq.~\ref{eq:diagqec}, under the action of errors. In practice, such a constraint might be too stringent and one can expect a codespace to get mapped to overlapping subspaces under the action of errors, as shown in Fig.~\ref{fig:aqec}. The idea of \emph{approximate} QEC owes its origins to a four-qubit code that was constructed to protect a single qubit of information against amplitude-damping noise~\cite{leung}. This code was shown to satisfy \emph{approximate QEC} conditions, which were obtained as a perturbed form of the perfect QEC conditions described in Eq.~\eqref{eq:perfectqec}. As we discuss below, such an approximate quantum code with just four qubits performs comparable to the perfect $5$-qubit code, in terms of the fidelity function. 

Formally, a quantum code $\cC$ is said to be an approximate code for the noise channel $\cE$, if there exists a recovery map $\cR$ such that $\cR\circ\cE (\rho)$ is close, in terms of some well defined distance measure, to the initial state $\rho$. The four-qubit approximate code, denoted as $[[4,1]]$ is obtained as the span of~\cite{leung},
\begin{equation}\label{eq:4qubit}
 |0_{L}\rangle = \tfrac{1}{\sqrt 2}(|0000\rangle +|1111\rangle), \quad
 |1_{L}\rangle = \tfrac{1}{\sqrt 2}(|1100\rangle +|0011\rangle),
 \end{equation}
where $|0_{L}\rangle$ and $|1_L\rangle$ are the logical qubits encoding a single qubit of information.
The encoding circuit for the $[[4,1]]$  code can be obtained as shown in Fig.~\ref{fig:enc_four} below.
\begin{figure}[H]
\centering
\includegraphics[scale=1.1]{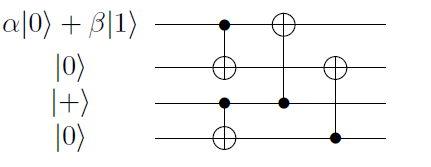}
\caption{ Encoding circuit for $[[4,1]]$ code.}
 \label{fig:enc_four}
\end{figure}

 It is easy to verify that the $[[4,1]]$ code satisfies the perfect QEC conditions in Eq.~\ref{eq:perfectqec} upto a deviation of $O(p^{2})$, for amplitude-damping noise acting on each of the four qubits in an identical and independent manner. This four-qubit channel is represented as $\cE_{AD}^{\otimes 4}$, where $\cE_{\rm AD}$ is the single-qubit amplitude-damping noise described in Eq.~\ref{eq:ampdamp}. Note that the representation of the four-qubit code, $[[4,1]]$ does not have a distance parameter, since the notion of distance is not well defined for an approximate code.

The $[[4,1]]$ code maybe considered the first known example of a \emph{channel-adapted code}, since it is designed to correct specifically for amplitude-damping noise, rather than to correct for generic Pauli errors. Furthermore, it is also an approximate code since it satisfies a perturbed form of the Knill-Laflamme condition. Subsequently, there was a lot of interest in constructing channel-adapted codes as well as channel-adapted recovery schemes. A generalization of the approximate code for the amplitude-damping channel, based on stabilizer formalism was soon developed in Ref.~\cite{fletcher_codes}. A channel-adapted recovery map that corrects for the four-qubit code given in Eq.~\ref{eq:4qubit} affected by the amplitude-damping channel, was also constructed via a convex-optimization procedure~\cite{fletcher_rec}. A few other works in the past~\cite{kosut,yamamoto,reimpell} adopted similar numerical optimization strategies to look for channel-adapted codes or channel-adapted recoveries, or both, using the \emph{average entanglement fidelity} as the figure of merit. Some of the past works also focussed on developing the approximate QEC (AQEC) conditions from an information-theoretic perspective~\cite{beny, beny2009,klesse, schumacher}.  

On the analytical front, a near-optimal map for reversing the dynamics was identified, using average entanglement fidelity as the figure of merit in ~\cite{Barnum}. More recently, simple algebraic AQEC conditions were established in terms of the \emph{worst-case fidelity}, based on the construction of a universal, recovery map, which was shown to be near-optimal for any encoding~\cite{hui_prabha}. This approach forms the basis of the methodology used in this thesis to construct and characterise approximate codes. In the following sections, we outline this algebraic approach to approximate QEC
 as well as the recovery map presented in Ref.~\cite{hui_prabha} in some detail, since we make use of these in subsequent chapters.
 
\subsection{Petz Recovery map $\cR_{P}$}\label{sec:Petzmap}
We now define and discuss the universal near-optimal recovery map $\cR_{P}$~\cite{hui_prabha}, often referred to as the Petz map in the literature. For a given quantum channel $\cE$ with Kraus operators $\{E_{i}\}$ and a code $\cC$ with projector $P$, the Petz map is defined as,
 \begin{equation}\label{eq:Petzmap}
\cR_{P}(\cdot) \equiv \sum_{i=1}^{N} PE_{i}^{\dagger} \cE(P)^{-1/2}(\cdot) \cE(P)^{-1/2}E_{i}P, 
\end{equation}
with Kraus operators $\{R_{i}\equiv P E_{i}^{\dagger}\cE(P)^{-1/2}\}_{i=1}^{N}$. Noe that the inverse of $\cE(P)=\sum_{i}E_{i}PE_{i}^{\dagger}$ is taken on its support. The recovery $\cR_{P}$ can be seen as the composition of three CP maps, $\cP\circ \cE^{\dagger}\circ \cN$, where $\cP=P(.)P$ is the projection onto the codespace, $\cE^{\dagger}$ describes the adjoint of the noise channel $\cE$ with the Kraus operators $\{E^{\dagger}_{i}\}$, and $\cN$ is a normalization map given by, $\cE(P)^{-1/2}(.)\cE(P)^{-1/2}$. The normalization map $\cN$ ensures that the map $\cR_{P}$ is trace-preserving (TP). 

The recovery map $\cR_P$ in Eq.~\ref{eq:Petzmap} was originally introduced in the work of Petz in an information-theoretic context~\cite{Petz}, and is hence known in the literature as the Petz map. This map was shown to saturate Ulhmann's theorem on the monotonicity of relative entropy~\cite{petz_entropy}. The structure of the Petz map $\cR_{P}$ is independent of the choice of the noise channel $\cE$, it satisfies the TP condition below.
\begin{equation}
\sum_{i} (P E_{i}^{\dagger}\cE(P)^{-1/2})^{\dagger} P E_{i}^{\dagger}\cE(P)^{-1/2}= P_{\cE},
\end{equation}
where $P_{\cE}$ is the projector onto the support of $\cE(P)$. We note that the map $\cR_{P} \circ \cE$ is unital for any noise channel $\cE$, satisfying $\cR_{P} \circ \cE (P)= P$.

\subsection{The worst-case fidelity function}\label{sec:worst_case}
The performance of a QEC protocol can be studied using the \emph{fidelity} function~\cite[Chapter 9]{nielsen}, which is a measure of distance between quantum states. The fidelity between any two quantum states $\rho$ and $\sigma$ is defined as, 
\begin{equation}\label{eq:fid}
F(\rho,\sigma)=\rm tr(\sqrt{\rho^{1/2}\sigma\rho^{1/2}})
\end{equation}
$F(\rho,\sigma)$ takes values between $0$ and $1$. Observe that for $\rho=\sigma$, $F=1$, and when $\rho$ and $\sigma$ have orthogonal support, $F=0$. For a pure state $\rho = |\psi\rangle\langle \psi|$, and for $\sigma\equiv \cM(\rho)$ where $\cM$ is a CPTP map, the fidelity in Eq.~\ref{eq:fid} is obtained as,
\begin{equation}
F(|\psi\rangle,\cM) =\sqrt{ \langle \psi|\,\cM( |\psi\rangle\langle\psi|) \,|\psi\rangle}.
\end{equation}

Specifically, the {\it worst-case fidelity}~\cite[Chapter 9]{nielsen} is a useful figure of merit in quantifying the performance of a QEC protocol, since it assures that a certain minimum fidelity is achieved by all the states on a given Hilbert space. The \emph{worst-case fidelity} for a given pair of encoding and recovery map $(\cW,\cR)$, for a noise $\cE$, is defined as,
\begin{equation}
F_{\min}(\cW, \cR;\cE) \equiv \min _{|\psi\rangle \in \cH_{0}} F(|\psi\rangle,\cW^{-1}\circ \cR \circ \cE \circ \cW). \label{eq:worstcase_fidelity}
\end{equation}
Equivalently, it is also useful to describe the performance of a QEC scheme in terms of the \emph{fidelity-loss} function, which is defined as
\begin{equation}
\eta(\cW,\cR;\cE)\equiv 1-F_{\min}(\cW,\cR;\cE).
\end{equation}
Note that the fidelity is a quantity that is jointly concave in its arguments, which means that for any probability distribution $\{p_i\}$ and density matrices $\rho_i$, $\sigma_i$, 
\begin{eqnarray}
F(\sum_i p_i \rho_i, \sum_i p_i \sigma_i) \geq \sum_i p_i F(\rho_i, \sigma_i) 
\end{eqnarray}
Suppose $\rho= \sum_i p_i \rho_i$ = $\sum_i p_i |\psi_i\rangle \langle \psi_i|$, and $\sigma_i$ = $\cE(\rho_i)$, then
\begin{equation}
F(\rho,\cE(\rho))=F(\sum_i p_i \rho_i, \sum_i p_i \sigma_i) \geq \sum_i p_i F(\rho_i, \sigma_i) 
\end{equation}
For any pure state $|\psi\rangle_i$, it thus follows that,
\begin{eqnarray}\label{eq:lowerbound}
F(\rho,\cE(\rho)) &&\geq \sum_i p_i F(\rho_i,\cE(\rho_i)) \\
&&\geq F(|\psi\rangle_i, \cE(|\psi\rangle_i))
\end{eqnarray}
Thus, the lower bound on fidelity is achieved by a pure state as illustrated in Eq.~\ref{eq:lowerbound}. Hence, it is enough to minimize the fidelity function over the state space of just pure states to obtain the worst-case fidelity. 

\subsection{AQEC conditions and near-optimality of the Petz map}
We now state and explain the AQEC conditions based on the Petz map in Eq.~\ref{eq:Petzmap}. Although the Petz recovery may not achieve the largest worst-case fidelity for given channel $\cE$ and code $\cC$, it was shown to perform close to optimal in \cite{hui_prabha}. Specifically, this near-optimality is captured by the bounds~\cite[Corollary 4]{hui_prabha}),
\begin{equation}\label{eq:optimalrange}
\eta_{\cR_\mathrm{op}}\leq \eta_P\leq \eta_{\cR_\mathrm{op}}{\left[(d+1)+O(\eta_{\cR_\mathrm{op}})\right]},
\end{equation}
where $\eta_{\cR_\mathrm{op}}$ and $\eta_P$ are the fidelity losses if we used the optimal and Petz recoveries, respectively, for a given encoding $\cW$ and noise $\cE$.  An interesting point to note is that the code dimension $d$ also appears in Eq.~\ref{eq:optimalrange}. Thus, for smaller values of $d$, the Petz recovery map is a good approximation of the optimal map. 

Based on the near-optimality of the Petz map, one can obtain conditions for approximate QEC, as shown in~\cite[Theorem 6]{hui_prabha}.
\begin{theorem}[Conditions for AQEC]
Consider a noise channel $\cE$ with Kraus operators $\{E_{i}\}$, acting on a codespace $\cC$ with projector $P$, such that,
\begin{equation}\label{eq:aqec}
PE_{i}^{\dagger}\cE(P)^{-1/2}E_{j}P= \beta_{ij}P +\Delta_{ij},
\end{equation}
where the scalars $\beta_{ij} \in \mathbb{C}$ are obtained as $\rm tr(PE_{i}^{\dagger}\cE(P)^{-1/2}E_{j}P)/d$, $d$ being the dimension of the codespace $\cC$, and $\Delta_{ij} \in \cB(\cC)$ are traceless matrices. Then, there exists a recovery map with  fidelity loss $\eta$ given by,
\begin{equation}
\eta =\underset{|\psi\rangle \in \cC}{\textrm{max}} [\langle \psi| \Delta{ij}^{\dagger} \Delta_{ij}|\psi\rangle - |\langle \psi| \Delta_{ij}|\psi\rangle|^{2}],
\end{equation}
such that, for any $\epsilon \in [0,1]$, the code $\cC$ achives a worst-case fidelity $\epsilon$ for the noise $\cE$ (a) if $\eta \leq \epsilon$, and (b) only if $\eta \leq f(\epsilon; d)$, where $f(\epsilon; d)$= $d+1+ O(\eta)$.
\end{theorem}

In the case of a \emph{perfect} QEC code, the $\Delta_{ij}$ matrices in Eq.~\ref{eq:aqec} vanish and the matrix of scalars $\beta=\sqrt{\alpha}$ given in Eq.~\ref{eq:perfectqec}, leading us back to the perfect QEC conditions in  Eq.~\ref{eq:perfectqec}. This shows that the Petz map is indeed the optimal map for perfect quantum codes.

In contrast to a perfect code, an approximate code with a non-zero  $\Delta_{ij}$ in Eq.~\ref{eq:aqec}, gets mapped to overlapping, and hence indistinguishable subspaces under the action of errors as illustrated in Fig.~\ref{fig:aqec}. Hence, recovering the information could be even more challenging and one cannot construct unitaries to recover as in the case of perfect QEC. A detailed treatment of this approximate QEC approach and Petz recovery that we outlined in this section above can be found in~\cite{hui_prabha, prabha_thesis}.

\subsection{Recent developments in AQEC}

We conclude this chapter with a brief review of some of the recent developments in AQEC. We first note that the AQEC formalism elaborated in the sections above, based on the Petz recovery, has been extended to study the case of subsystem codes in~\cite{prabha_aqec}. The performance of AQEC  schemes has been studied for the case of generalized amplitude-damping channel, using a recovery similar to the Knill-Laflamme perfect recovery, in~\cite{cafaro}. This scheme uses the entanglement fidelity to study the efficiency of the protocol.

More recently, the work due to~\cite{elizabeth} demonstrates how approximate codes arise naturally in translation-invariant many body systems. This work also  draws connections between quantum chaotic systems with ETH (Eigenstate Thermalization Hypothesis) and approximate quantum codes. More generally, the interplay between continuous symmetries and approximate quantum codes has been analyzed in great detail, in \cite{preskill_symmetry}, which also demonstrated the construction of quantum codes covariant with respect to a general group.

The Petz map and its variants have provided a fertile area of study, leading to several interesting directions of work. The Petz construction has been used in \cite{todd} to show how orbital angular momentum properties of the photons can be protected from decoherence caused by atmospheric turbulence. An explicit construction of the Petz map has been provided for the case of guassian channels in \cite{wilde}. More recently,~\cite{petz_implementation} demonstrated a quantum algorithm to implement the Petz map using the quantum singular value transformation. This algorithm provides a procedure to perform the pretty-good measurement which is a special case of Petz recovery, allowing for near-optimal state discrimination.

%% file: Chapter3_v3.tex

\chapter{Constructing adaptive codes using the Cartan form} 

\label{Chapter3} 

\lhead{Chapter 3. \emph{Chap:3}} 

\section{\label{sec:intro}Introduction}

Interest in building quantum computing devices has grown steadily, with rapid progress in the last few years thanks to the fresh injection of industry support. Current quantum computing devices, like the ones being built by IBM, Google, and Rigetti, comprise only a few (at best, tens of) qubits, and are quite noisy. We are right now in the ``NISQ era" \cite{preskill2018quantum}, a term referring to the near-to-intermediate-term situation where physical devices are too noisy and too small to implement regular quantum error correction (QEC) and fault tolerance schemes to deal with the noise in the device. Hence, there is a strong need to find better QEC and fault tolerance schemes with lower resource overheads, essential for the eventual implementation of robust and scalable quantum computing devices. 
 Generally, QEC~\cite{shor,  steane,  calderblank,  gottesman, knill} (see also a recent review~\cite{terhal}) tries to store the information to be protected in a special part of the quantum state space with the property that errors due the noise can be identified and their effects removed through a recovery procedure.
 
Much of the existing work on QEC centers around codes capable of removing the effects of \emph{arbitrary} errors on individual qubits, powerful enough to deal with general, even unknown, noise. {The stabilizer codes~\cite{nielsen}, including the well-known Steane code~\cite{steane} and Shor code~\cite{shor}, fall in this category of codes which can correct for single-qubit errors.} 

One cannot find codes capable of correcting an arbitrary error on any single qubit unless one uses at least five physical qubits to encode a one qubit of information \cite{laflamme}. However, when there is a reasonable level of characterization of the noise afflicting the qubits, \emph{channel-adapted codes}~\cite{fletcherthesis}--- codes tailor-made to deal with the specific noise channel encountered in the physical device--- become of interest. Such codes can be expected, and are known (see, for example, the $4$-qubit code for amplitude-damping noise discovered in Ref.~\cite{leung}), to be less demanding in resources.

The most general formulation of channel-adapted codes requires full knowledge of the noise. We make the following experimentally well-motivated assumption~\cite{tannu} about the structure of the noise. We assume that the noise takes a tensor-product structure and propose a numerical algorithm to find good codespaces--- regions of the state space resilient to the noise--- among states with a nonlocal structure, using a Cartan decomposition of the encoding operation.

The question of finding channel-adapted codes can be formulated as an optimization problem~\cite{yamam, reimpell, fletcher_codes,kosut,wang}, one of finding the combination of code and recovery that optimizes a chosen figure of merit for the given noise channel. When the figure of merit is the \emph{average} entanglement fidelity~\cite{schum}, one only has a double optimization over encoding (of a given block length) and recovery. This problem is known to be tractable via convex optimization techniques~\cite{reimpell, fletcher_rec, kosut}. 

In our work, we focus on finding channel-adapted codes that minimize the \emph{worst-case fidelity} for the storage of a single qubit of information. We argue that we can, in practice, reduce the original triple optimization to a single optimization by making use of the Petz recovery ~\cite{Petz}, shown to be near-optimal~\cite{hui_prabha} for any noise channel and choice of code. The use of the Petz recovery leads to an analytical expression for the worst-case fidelity for codes encoding a single logical qubit. Furthermore, a key aspect of this work is that we can reduce the difficulty of this remaining numerical optimization over all possible encodings by employing a Cartan decomposition of the encoding operation, motivated by the noise-locality--code-nonlocality dichotomy. We vary only over the nonlocal pieces of the decomposition, thereby reducing the dimension of the search space. Altogether, these steps give a fast and easy algorithm for finding good channel-adapted codes for the worst-case fidelity, the preferred figure of merit for quantum computing tasks.

\section{The optimization problem}\label{sec:prelim}
We begin by describing the various ingredients in our approach to the problem of finding channel-adapted codes for the worst-case fidelity measure. We explain how to reduce the problem to a single optimization, over the encoding operation.

\subsection{Basic formulation}
Consider a physical quantum information processing system of dimension $d$, with Hilbert space $\cH$. The noise acting on the system can be described by a \emph{quantum channel}, i.e., a completely positive (CP) and trace-preserving (TP) linear map, denoted by $\cE$. $\cE$ acts on $\cB(\cH)$, the set of linear operators on $\cH$, $\cE:\cB(\cH)\rightarrow \cB(\cH)$. Its action can be written as $\cE(\,\cdot\,)=\sum_{i=1}^NE_i(\,\cdot\,)E_i^\dagger$, for a set of (non-unique) Kraus operators $\{E_i\}_{i=1}^N$, a structure that assures the CP nature of the map. The Kraus operators further satisfy $\sum_{i=1}^NE_i^\dagger E_i=\id$, for the TP property.

To protect the quantum information from damage by the noise, QEC proposes to store the information---assumed to be a $d_0$-dimensional Hilbert space $\cH_0$ of states---in a $d_0$($\leq d$)-dimensional subspace $\cC$ of $\cH$, the Hilbert space of the physical system. We refer to $\cC$ as the \emph{codespace}. The encoding operation $\cW$---a unitary operation, and hence invertible---is a one-to-one mapping of states from $\cH_0$ to $\cC$, $\cW:\cB(\cH_0)\rightarrow \cB(\cC)\subseteq\cB(\cH)$.  The action of the noise $\cE$ on the encoded state, the output of $\cW$, can then  be regarded as $\cE: \cB(\cC) \rightarrow \cB(\cH)$. After the action of the noise, the QEC protocol applies a suitable recovery map $\cR$, a CPTP map $\cR: \cB(\cH) \rightarrow \cB(\cC)$ that restores the state into the codespace, and in the process removing (hopefully most of) the errors due to the noise. If we want, we can then decode the physical state back into the quantum informational state of $\cH_0$ by applying the decoding operation, $\cW^{-1}$.

Traditionally, the QEC protocol, specified by the pair $(\cW,\cR)$ for given $\cH_0$ and $\cH$, is chosen to satisfy (at least approximately) what are known as the QEC conditions \cite{knill,hui_prabha}, for successful removal of the errors caused by the noise. Here, it is more straightforward to think directly in terms of an optimization problem. For that, we first quantify the performance of a code $\cC$ (or, equivalently, $\cW$) with recovery $\cR$ for the noise process $\cE$ by a measure that compares the output state of the QEC protocol $(\cR\circ\cE)(\rho)$ to the input state $\rho\in\cB(\cC)$. We then characterize the performance of a given pair $(\cW,\cR)$ for noise $\cE$, using the {\it worst-case fidelity} which is the \emph{square} of the quantity defined in   Eq.~\ref{eq:worstcase_fidelity}, as shown below.
\begin{equation}
F^2_{\min}(\cW, \cR;\cE) \equiv \min _{|\psi\rangle \in \cH_{0}} F^{2}(|\psi\rangle,\cW^{-1}\circ \cR \circ \cE \circ \cW). \label{eq:worstcase_fid}
\end{equation}
Recall from Sec.~\ref{sec:worst_case} that it is enough to minimize $F(.,.)$ and by extension $F^{2}(.,.)$ also, over pure states. This minimization over $|\psi\rangle\in\cH_0$ usually has to be done numerically unless one has special properties that simplify the problem (as we will see below). Alternatively, one can make use of the \emph{fidelity loss} quantity defined as, 
\begin{equation}
\eta(\cW,\cR;\cE)\equiv 1-F_{\min}^2(\cW,\cR;\cE).
\end{equation}

We can now state the basic formulation of the optimization problem for channel-adapted codes: For given noise $\cE$, and the available dimension $d$ of the physical system, the best code is given by the solution to the following optimization over encoding operations $\cW$ and recovery maps $\cR$.
\begin{align}\label{eq:optimize}
&\quad~\underset{\cW}{\textrm{argmax}}\,\underset{\cR}{\textrm{argmax}} ~F^2_{\min}(\cW,\cR;\cE)\\
&=\underset{\cW}{\textrm{argmin}}\,\underset{\cR}{\textrm{argmin}} ~\eta(\cW,\cR;\cE)\nonumber\\
&=\underset{\cW}{\textrm{argmax}}\,\underset{\cR}{\textrm{argmax}} \min_{|\psi\rangle\in \cH_0} F^2(|\psi\rangle,\cW^{-1}\circ \cR \circ \cE \circ \cW).\nonumber
\end{align}
This is the triple optimization, over the encoding $\cW$, the recovery $\cR$, and the input state $|\psi\rangle$, mentioned in the introduction.

We note that, in principle, one could also add an optimization over the dimension $d$ of the physical state space used to encode $\cH_0$. For the current situation of independent noise on the physical system, one expects better fidelity with a larger number of physical qubits, as this will gives better ``de-localization" of the information. However, in the current NISQ era, the number of physical qubits available for encoding the information will largely come from practical constraints. We thus take $d$ to be fixed, and find the best $(\cW,\cR)$ for that given $d$.

We first reduce this triple optimization problem to a double optimization over $\cW$ and $|\psi\rangle$, by choosing the Petz recovery $\cR_{P}$. Recall from Sec.~\ref{sec:Petzmap} that the Petz map is defined as~\cite{Petz, hui_prabha},
\begin{equation}\label{eq:Petzmap}
\cR_{P}(\cdot) \equiv \sum_{i=1}^{N} PE_{i}^{\dagger} \cE(P)^{-1/2}(\cdot) \cE(P)^{-1/2}E_{i}P, 
\end{equation}
where $\{R_{i}\equiv P E_{i}^{\dagger}\cE(P)^{-1/2}\}_{i=1}^{N}$ constitute the Kraus operators of $\cR_{P}$. Here, $P$ is the projector onto the codespace $\cC$, and the inverse of $\cE(P)$ is taken on its support. Having fixed the recovery, our optimization problem now reduces to a double optimization of the form,
\begin{equation}\label{eq:minfidelity}
  \underset{\cW}{\textrm{argmax}} \ \underset{|\psi\rangle\in \cH_{0}}{\text{min}} F^{2}(|\psi\rangle,\cW^{-1}\circ \cR_{P} \circ \cE \circ \cW).
\end{equation}
We denote the fidelity loss for an encoding $\cW$ as $\eta_\cW\equiv \eta(\cW,\cR_P;\cE)$; the optimal encoding $\cW_{\mathrm{op}}$ is then the one that attains $\eta_{\mathrm{op}}\equiv \min_\cW\eta_\cW$.

\subsection{ Fidelity loss for qubit codes}

In this work we search for codes which preserve a qubit worth information. It turns out that the optimization problem of Eq.~\eqref{eq:minfidelity} can be further simplified, since the worst-case fidelity $\min_{|\psi\rangle\in \cH_0} F^{2}(|\psi\rangle,\cW^{-1}\circ \cR_{P} \circ \cE \circ \cW)$ for encoding $\cW$, or equivalently, the fidelity loss function $\eta_\cW$, has a simple form for the case of qubit codes (i.e., $d_0=2$) with the Petz recovery. Specifically, $\eta_{\cW}$ can be easily computed via eigenanalysis~\cite{hui_prabha}. We recall the steps here, for completeness.

We encode a qubit $\cH_0$ into a two-dimensional codespace $\cC$. For an orthonomal basis $\{|v_1\rangle,|v_2\rangle\}$ on $\cC$, the Pauli basis (orthogonal but not normalized) for operators on $\cC$, denoted as $\{\sigma_\alpha\}_{\alpha=0,x,y,z}$, can be defined in the usual way as,
\begin{align}\label{eq:Pauli}
\sigma_0 =& |v_1\rangle\langle v_1|+|v_2\rangle\langle v_2|=P \equiv I,\\
\sigma_x =& |v_1\rangle\langle v_2|+|v_2\rangle\langle v_1|, \nonumber \\
\sigma_y =& -i(|v_1\rangle\langle v_2|-|v_2\rangle\langle v_1|), \nonumber\\
\textrm{and}\quad\sigma_z =& |v_1\rangle\langle v_1|-|v_2\rangle\langle v_2|.
\end{align}
Codestates $\rho\in\cB(\cC)$ can then be described using the Bloch representation,
\begin{equation}
\rho = \dfrac{1}{2}(I + \mathbf{s} . \bm{\sigma}),
\end{equation} 
where $\mathbf{s}=(s_x,s_y,s_z)$ is a real $3$-dimensional vector --- the Bloch vector for $\rho$ --- with Euclidean length $|\mathbf{s}|\leq 1$, and $\bm{\sigma}=(\sigma_x,\sigma_y,\sigma_z)$.

Consider the channel $\cM:\cB(\cC)\rightarrow \cB(\cC)$ constructed by composing the noise followed by the Petz recovery, $\cM\equiv\cR_P\circ\cE\circ\cP$, acting on the codespace. $\cP(\cdot)\equiv P(\cdot)P$ is the map that enforces the pre-condition that we start in the codespace. $\cM$ is both trace-preserving [$\cM^\dagger (P)=P$] and unital [$\cM(P)=P$], and its action can be expressed in the Pauli operator basis defined above, as the following matrix.
 \begin{equation}\label{eq:tmatrix}
M  = {\left(\begin{array}{c|c}
1 & 0~~0~~0\\
\hline
0 &  \\
0 & T \\
0 & 
\end{array}\right)},
\end{equation}
Note that $M$ is a real matrix with real matrix entries $M_{\alpha\beta}\equiv \frac{1}{2}\tr\{\sigma_\alpha\cM(\sigma_\beta)\}$, with $T$ denoting a $3\times 3$ matrix of the non-zero $\alpha,\beta=x,y,z$ entries. The action of $\cM$ on an input state $\rho\in\cB(\cC)$ can then be expressed in terms of the action on the Bloch vector as $\mathbf{s}\mapsto \mathbf{s}'\equiv T\mathbf{s}$.
The fidelity loss $\eta_\cW$ (for a given encoding $\cW$ that defines the $\cC$ subspace) is then, by straightforward algebra,
\begin{equation}\label{eq:fidelityloss}
\eta_\cW=\max_{\mathbf{s},|\mathbf{s}|=1}\tfrac{1}{2}(1-\mathbf{s}^\mathrm{T}T_\mathrm{sym}\mathbf{s})=\tfrac{1}{2}{\left[1-t_{\min}(\cW)\right]},
\end{equation}
where $T_\mathrm{sym}\equiv \frac{1}{2}(T+T^\mathrm{T})$, and $t_{\min}(\cW)$ is the smallest eigenvalue of $T_\mathrm{sym}$. Here, the superscript $\mathrm{T}$ denotes the transpose operation.

\subsection{A single numerical optimization}
In this way, we have reduced the minimization needed to compute the worst-case fidelity in Eq.~\ref{eq:minfidelity} to a simple diagonalization of a $3\times 3$ matrix and finding its smallest eigenvalue. Thus we have only a single optimization left to do, to find the best channel-adapted code, namely,
\begin{equation}
  \underset{\cW}{\textrm{argmin}}~\frac{1}{2}{\left[1-t_{\min}(\cW)\right]}.
\end{equation}

The optimization over the encoding $\cW$ has to be done numerically. We parameterize the search space as follows. Every codespace $\cC$ is specified by $d_0$ orthogonal pure states in $\cH$, forming a basis for $\cC$. Varying over the codespace can then be thought of as starting with a fixed basis with $d_0$ elements, and then applying a rotation of the basis, via a unitary operator $U$ acting on the full $d$-dimensional Hilbert space of the physical system. Choosing different codespaces $\cC$ then corresponds to choosing different unitary operators $U$. The search space is then the set of all $d$-dimensional unitary operators, each element of which is specified by $d^2$ real parameters.

As noted above, $t_{\min}(\cW)$ has to be computed numerically for each $\cW$. This means that we do not have a closed-form expression for the gradient of our objective function, so that standard optimization methods that require a formula for the gradient do not work. This is easily solved, however, by going to methods that estimate the gradient numerically via gradient-descent algorithms. A well-known approach, the one that we used here, is the Nelder-Mead search technique (also known as the downhill simplex method)~\cite{nelderpaper, NumericalRecipes}, which has been explained in Sec.~\ref{sec:nmsearch} of Appendix~\ref{AppendixA}.

\section{Simplifying the search: The Cartan decomposition}\label{sec:search}

As stated earlier, the optimization over $\cW$ involves a $d^2$-dimensional search. For $n$-qubit physical systems, the typical experimental scenario, where $d=2^n$, the search space dimension grows exponentially with $n$. It would hence be useful to further reduce the complexity of the search by considering a restricted search over the set of encoding unitaries $\cW$. For that, we recall our focus, as motivated in the introduction, on noise channels with a tensor-product structure over the $n$ qubits. This local structure in the noise suggests  the use of codes with a nonlocal nature. To separate the nonlocal pieces of the unitary search space from the local pieces, we make use of the Cartan decomposition, as described in the following paragraphs.

Our search space, originally comprising elements of the unitary group $U(2^{n})$ for an $n$-qubit code, can be restricted to elements of the special unitary group $SU(2^{n})$ without loss of generality. We then use the Cartan decomposition originally proposed in~\cite{khaneja_glaser}, whereby any $n$-qubit unitary is realised as a product of single-qubit (local) and multi-qubit (nonlocal) unitaries. The specific paramterization we use is due to~\cite{cartan}, where the standard Pauli basis is employed to decompose an arbitrary element of  $SU(2^{n})$ in terms of its local and nonlocal parts in an iterative fashion.

\subsection*{Cartan form of $SU(2^{n})$}\label{sec:cartan}

Recall that the special unitary group $SU(d)$ --- the group of $d \times d$ complex matrices with determinant one --- forms a real Lie group of dimension $d^{2}-1$. Let $\mathfrak{SU}(d)$ denote the corresponding Lie algebra, the algebra of traceless anti-Hermitian $d\times d$ complex matrices with the Lie bracket $-\upi[\,\cdot\,,\,\cdot\,]$, that is, $(-\upi)$ times the commutator.

The central idea behind the Cartan form is the fact that any element of $SU(2^{m})$ can be represented, up to local unitaries, using elements of two Abelian subalgebras $\mathfrak{h_{m}}$ and $\mathfrak{f_{m}}$ ($m=2, 3, \ldots , n$) of $\mathfrak{SU}(2^{m})$. This was shown in~\cite{khaneja_glaser} via an iterative decomposition of the form $U = U' H U''$, where $H$ is generated alternately from elements of $\mathfrak{h_{m}}$ and $\mathfrak{f_{m}}$, while $U'$ and $U''$ belong to the subgroup of $SU(2^{m})$ generated by a subalgebra orthogonal to $\mathfrak{h_{m}}$ and $\mathfrak{f_{m}}$. The exact structure of the decomposition depends on the choice of an appropriate basis for $\mathfrak{SU}(2^{m})$ that can be obtained recursively for $m=2,3,\ldots,n$. 

For example, for $n=2$, one can use twofold tensor products of the single-qubit Pauli operators ($I,X,Y,Z$) as basis elements for the Lie algebra $\mathfrak{SU}(4)$. Using this basis to partition $\mathfrak{SU}(4)$ into orthogonal subspaces leads to an identification of the Abelian subalgebras $\mathfrak{h}_{2} = {\rm span}\{XX,YY, ZZ\}$ and $\mathfrak{f}_{2} = \{0\}$. This leads to the well known Cartan form for any $U \in SU(4)$~\cite{Zhang2003, Rezakhani2004},
\begin{equation}
 U = (U_{1}\otimes U_{2})\,\upe^{-\upi  (c_{1}XX + c_{2}YY + c_{3}ZZ)} (U_{3}\otimes U_{4}), \label{eq:su2}
 \end{equation}
where $U_{1}, U_{2}, U_{3}$, and $U_{4} \in SU(2)$ are local, single-qubit unitaries, and $c_{1}, c_{2}$, and $c_{3}$ are scalar parameters characterizing the nonlocal operators.
 
Following this intuition from $SU(2)$, it was shown that a basis comprising $n$-fold tensor products of the single-qubit Pauli basis can be obtained for any $\mathfrak{SU}(2^{n})$ by an iterative process~\cite{cartan}, which partitions $\mathfrak{SU}(2^{n})$ into Abelian subalgebras $\mathfrak{h}_{n}$ and $\mathfrak{f}_{n}$. The Cartan decomposition of any $n$-qubit unitary operator can then be obtained as follows. \\[1ex]
{\textbf{Cartan decomposition}~\cite{cartan}\textbf{.~}} \textit{Any $U \in SU(2^{n})$, for $n > 2$ can be decomposed as,
\begin{equation}\label{eq:cartan}
G = K^{(1)} F^{(1)} K^{(2)} J K^{(3)} F^{(2)} K^{(4)}.
\end{equation}
Here, $K^{(i)}$ are product operators from $SU(2^{n-1}) \otimes SU(2)$, $F^{(j)}\equiv\exp(-\upi f^{(j)})$ and $J\equiv \exp{(-\upi h)}$ are unitary operators \emph{nonlocal} on the entire $n$-qubit space, with $h \in \mathfrak{h}_{n}$ and $f^{(j)} \in \mathfrak{f}_{n}$. The decomposition can be applied recursively, to further decompose each $SU(2^m)$ operator in $K^{(i)}$ in the same form as in Eq.~\eqref{eq:cartan}, for all $m=2,3,\ldots,n-1$.}\\

As in the $n=2$ case, the Cartan decomposition for $n >2$ again separates out the local and nonlocal degrees of freedom  in an iterative fashion. However, for $n > 2$, a second Cartan decomposition is required in order to identify the factors that are nonlocal on the entire $n$-qubit space. This stems from the fact that a pair of nontrivial Abelian subalgebras $\mathfrak{h}_{n}, \mathfrak{f}_{n}$ maybe identified for any $n >2$, and this leads to a two{-}step decomposition. First, using the generators of the subalgebra $\mathfrak{h}_{n}$, we obtain the unitary $J\in SU(2^{n})$, as well as $U', U'' \in SU(2^{n})$, such that $G = U' J U''$ for any $G \in SU(2^{n})$, $n>2$. Further Cartan decompositions of $U'$ and $ U''$ using the generators of $\mathfrak{f}_{n}$ gives the form $ G= K^{(1)} F^{(1)} K^{(2)} J K^{(3)} F^{(2)} K^{(4)}$, where the operators $K^{(i)} \in SU(2^{n-1})\otimes SU(2)$ are no longer nonlocal on the entire $n$-qubit space. Starting with the bases for $\mathfrak{h}_{2}, \mathfrak{f}_{2}$ identified above,~\cite{cartan} provides a simple recursive prescription to identify the bases for the subalgebras $\mathfrak{h}_{n}, \mathfrak{f}_{n}$, for any $n>2$. 

To illustrate how the above prescription can be used to obtain a nice parameterization of the encoding unitaries for QEC, we explicitly write down the Cartan form for $n=3$ and $4$. Any element of $SU(2^{3})$ can be constructed using the formalism in Eq.~\ref{eq:cartan} as,
\begin{align}\label{eq:su8}
U &= K^{(1)} F^{(1)} K^{(2)} J K^{(3)} F^{(2)} K^{(4)},\\
F^{(i)} &= \upe^{-\upi (c^{(i)}_{1}XXZ + c^{(i)}_{2}YYZ + c^{(i)}_{3}ZZZ)}, \nonumber \\ 
\textrm{and }\quad J &= \upe^{-\upi(a_{1}XXX+a_{2}YYX+  a_{3}ZZX+ a_{4}IIX)}.\nonumber
\end{align}
As stated above, $K^{(j)} \in SU(4) \otimes SU(2)$, and each element in $SU(4)$ can be obtained similarly from Eq.~\ref{eq:su2}. Recall that the standard description of any unitary in $SU(2^{3})$ requires $63$ real parameters, whereas the recursive Cartan decomposition described in Eq.~\ref{eq:su8} requires a total of $82$ real parameters. However, the key advantage of using the Cartan parameterization is that the nonlocal factors of any unitary in $SU(2^{3})$ are easily described in terms of $22$ real parameters, namely, the set of ten real parameters $\{a_{1}, a_{2}, a_{3}, a_{4}, c^{(i)}_{1}, c^{(i)}_{2}, c^{(i)}_{3}\}$, along with the three real parameters for each of the four $SU(4)$ factors. 

Similarly, we note that any element $ U \in SU(2^{4})$ can be decomposed as,
\begin{align}\label{eq:su16}
 U &= K^{(1)} F^{(1)} K^{(2)} {J} K^{(3)} F^{(2)} K^{(4)}, \nonumber \\ 
 F^{(i)} &=\exp\bigl(-{\upi}(c^{(i)}_{1} XXIZ+c^{(i)}_{2}YYIZ\nonumber\\
 &~\quad\qquad + c^{(i)}_{3} ZZIZ + c^{(i)}_{4} IIXZ + c^{(i)}_{5} XXXZ\nonumber\\
 &~\quad\qquad + c^{(i)}_{6} YYXZ + c^{(i)}_{7} ZZXZ) \bigr),\quad i=1,2, \nonumber \\ 
\textrm{and } J &= \exp\bigl(-{\upi} (a_{1} IIIX+a_{2} XXIX + a_{3} YYIX  \nonumber \\  
 &~\quad\qquad + a_{4} ZZIX+ a_{5} IIXX+ a_{6}XXXX+ \nonumber \\  
 &~\quad\qquad + a_{7} YYXX + a_{8} ZZXX) \bigr), 
\end{align}
where $ K^{(j)}$ $\in$ $SU(8)\otimes SU(2)${, $j=1$, 2, 3, and 4}. Such a decomposition requires a total of $362$ real parameters including the set $\{a_{1}, a_{2}, a_{3},\ldots, a_{8}, c^{(i)}_{1}, c^{(i)}_{2},\ldots c^{(i)}_{7}\}$ which parameterizes the fully nonlocal factors.
 
We simplify our numerical search for optimal encodings by fixing the local components in the Cartan form and searching only over the nonlocal degrees of freedom. This greatly reduces the dimension of our search space {and allows us to search much faster, compared to an unstructured search}. This restriction only to nonlocal degrees of freedom, as we will see from our examples in Sec.~\ref{sec:results}, does not lead to substantial loss in fidelity. Furthermore, choosing the local unitaries in the decomposition appropriately allows us to construct encoding unitaries with simple structures that permit easy circuit implementations of the encoding procedure. 

An example of such a \emph{structured} encoding would be to set all the $K^{(i)}$s in Eq.~\ref{eq:cartan} to be the identity operator, thus reducing the form of the encoding unitary to $U =F^{(1)} J F^{(2)}$. This implies the following form for the encoding unitary in the case of $SU(2^3)$,
\begin{equation}
U=
   \begin{bmatrix}
\ast  & \ast&   0 &   0 &   0 &   0 &   \ast & \ast &\\
  \ast & \ast &   0 &   0 &   0 &   0 &  \ast & \ast &\\
   0 &   0 & \ast & \ast & \ast & \ast &   0 &   0 &\\
     0 &   0 &  \ast & \ast & \ast & \ast &   0 &   0 &\\
   0 &   0 & \ast & \ast & \ast &\ast & 0 &   0 &\\
   0 &   0 &  \ast  & \ast & \ast & \ast &   0 &   0 &\\
  \ast & \ast &   0 &   0 &   0 &   0 &  \ast & \ast &\\
   \ast & \ast &   0 &   0 &   0 &   0 &  \ast & \ast &
  \end{bmatrix} , \label{eq:U_8}
  \end{equation}
where $\ast$ refers to some non-zero complex number.  Such a structured encoding $U$ with only non-local Cartan factors is easy to implement using only single{-}qubit gates and the \textsc{cnot} gate, as explained in Appendix~\ref{AppendixA}, Sec.~\ref{sec:circuit}.

\section{Examples}\label{sec:results}

We now demonstrate the efficacy of our numerical approach to finding good codes through some examples. We consider $n$-qubit noise channels of the form $\cE=\cE_1\otimes\cE_2\otimes\ldots\otimes \cE_n$, where $\cE_i$ is a single-qubit channel on the $i$-th qubit. The first few examples are for the case where all the $\cE_i$s are the same channel, corresponding to the common experimental situation where all the qubits see the same environment and hence undergo the same noise dynamics. We focus on three examples, namely, the amplitude-damping channel, the rotated amplitude-damping channel, and an arbitrary (randomly chosen, no special structure) single-qubit channel. 

In each example, we use the form of the encoding unitary in Eq.~\ref{eq:cartan} to perform both {\it unstructured} as well as {\it structured} search over the space of all encoding unitaries.
In an unstructured search, we retain the general form of the encoding unitary in Eq.~\ref{eq:cartan}, using \emph{all} parameters, local and nonlocal, in our search. For the structured search, we adopt two different approaches. In the case of a structured search with \emph{trivial} local unitaries, we set all the local ($SU(2)$) unitaries in the Cartan decomposition in Eq.~\ref{eq:cartan} equal to the identity, and search only over the nonlocal parameters in Eq.~\eqref{eq:cartan}. For example, while searching over the $4$-qubit space, we retain the nontrivial $3$-qubit nonlocal pieces as well as the $2$-qubit pieces, but set all the single-qubit unitaries to identity. Alternately, we also implement structured search with \emph{nontrivial} local{ ($SU(2)$)} unitaries, where the choice of the local unitaries in the Cartan decomposition is guided by the structure of the channel. This kind of structured search with nontrivial local unitaries will be particularly relevant for the example of the rotated amplitude-damping channel discussed in Sec.~\ref{sec:rotated_AD}.

\subsection{Amplitude-damping channel}\label{sec:amp_damp}

\begin{figure}
\includegraphics[trim=8mm 12mm 15mm 5mm, clip, width=\columnwidth]{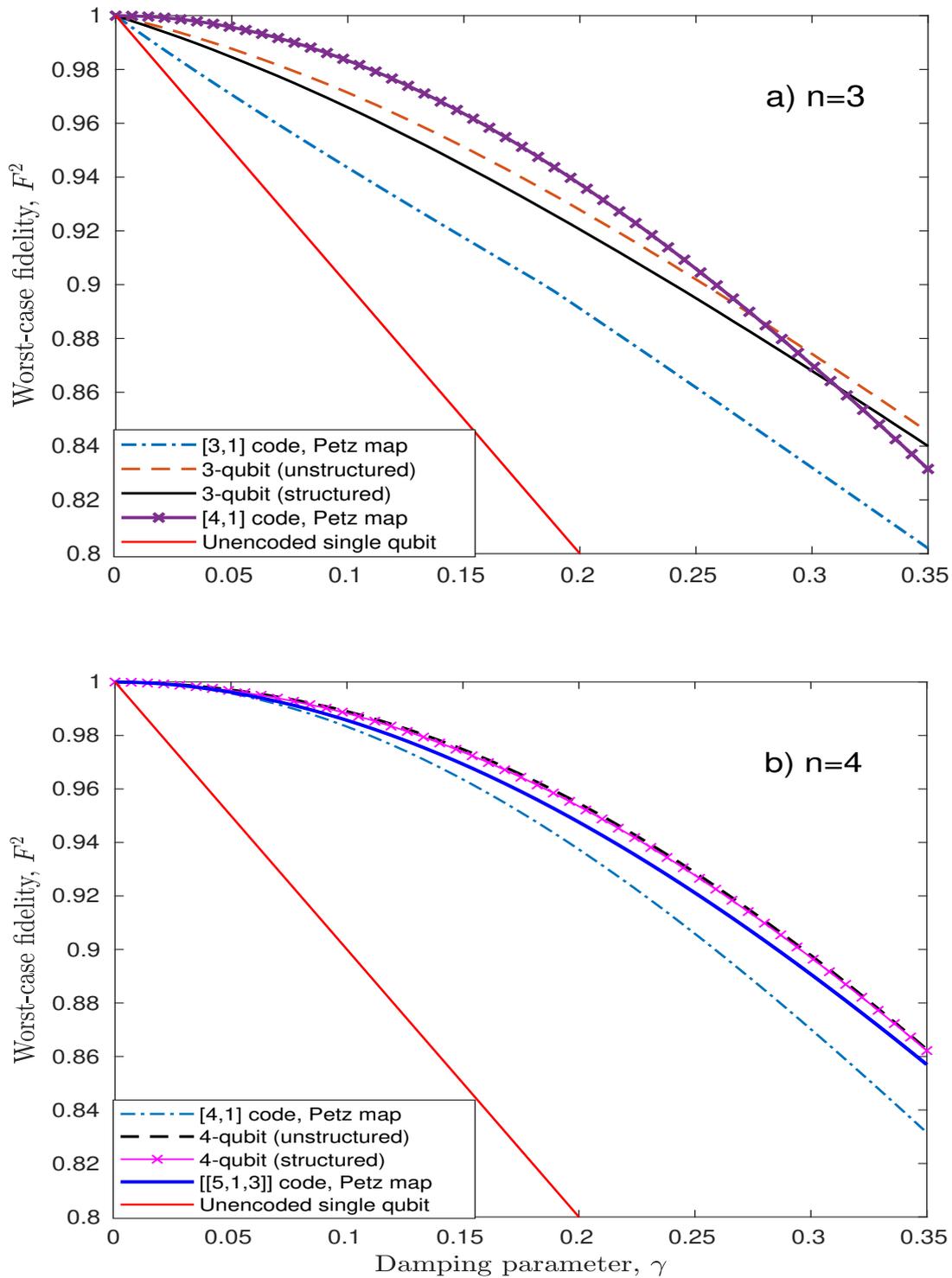}
\caption{\label{fig:3-4qubit}Performance of $n$-qubit codes for amplitude{-}damping, using the Petz recovery [Eq.~\eqref{eq:Petzmap}] with structured and unstructued encodings, for (a) $n=3$, and (b) $n=4$.} 
\end{figure}

Recall from Eq.~\ref{eq:ampdamp} that the single-qubit amplitude-damping channel $\cE_\mathrm{AD}$ is described by a pair of Kraus operators in the computational basis $\{|0\rangle,|1\rangle\}$, as,
{
\begin{align}\label{eq:amplitudedamping}
E_{0}  &= |0\rangle\langle 0|+ \sqrt{1-\gamma}\,|1\rangle\langle 1|,\nonumber\\
\textrm{and}\quad E_1 &= \sqrt{\gamma}\,|0\rangle\langle 1|,
\end{align}
}
where $E_1$ flips the $|1\rangle$ state to the $|0\rangle$ state, imitating a ``decay" to the $|0\rangle$ state; the deviation of $E_0$ from the identity ensures the trace-preserving nature of the channel. We perform the numerical search outlined in Sec.~\ref{sec:search} and obtain optimal encodings for $\cE=(\cE_\mathrm{AD})^{\otimes n}$, when $\gamma\ll 1$, for $n=3$ and $4$. We compare the performance of the codes we find with the various known codes in Fig.~\ref{fig:3-4qubit}. In particular, we compare with the $[3,1]$ approximate code~\cite{langshor}, given by the span of the states,
\begin{equation}\label{eq:3qubit}
 |0_{L}\rangle = \tfrac{1}{\sqrt 2}(|000\rangle +|111\rangle), \quad
 |1_{L}\rangle = \tfrac{1}{\sqrt 2}(|100\rangle +|011\rangle);
 \end{equation}
and the $[4,1]$ approximate code~\cite{leung}, which is the span of,
\begin{equation}\label{eq:4qubit}
 |0_{L}\rangle = \tfrac{1}{\sqrt 2}(|0000\rangle +|1111\rangle), \quad
 |1_{L}\rangle = \tfrac{1}{\sqrt 2}(|1100\rangle +|0011\rangle).
 \end{equation}
Fig.~\ref{fig:3-4qubit} shows that the numerically obtained codes via structured (with trivial local unitaries) and unstructured search outperform the known {\it approximate} codes of the same length described in Eq.~\eqref{eq:3qubit} and Eq.~\eqref{eq:4qubit} respectively. We observe that the performance of the $4$-qubit optimal codes is even better than the standard $[[5,1,3]]$ code, as seen in Fig.~\ref{fig:3-4qubit}(b). In both figures, we have also plotted the worst-case fidelity for a single unprotected qubit under the noise channel. The fidelity of the unencoded qubit falls off linearly with the noise parameter, thus demonstrating the advantage of using the $4$-qubit  codes found using our procedure. 

The codewords for the optimal $3$,$4$-qubit codes found in our search are presented in Appendix~\ref{AppendixA}, Sec.~\ref{sec:numericalcodes}. We also provide the encoding circuit corresponding to the optimal, structured $3$-qubit code, as an example of how the codes that emerge out of the structured search admit simple encoding circuits. Finally, we note that our numerical search procedure is indeed fast. The unstructured search for a specific value of damping parameter $\gamma$ takes a few hundred seconds on a standard desktop computer, while each structured search takes only a few milliseconds on the same computer.

\subsection{Rotated amplitude-damping channel}\label{sec:rotated_AD}

\begin{figure}
\includegraphics[trim=5mm 1cm 7cm 3mm, clip, width=18cm]{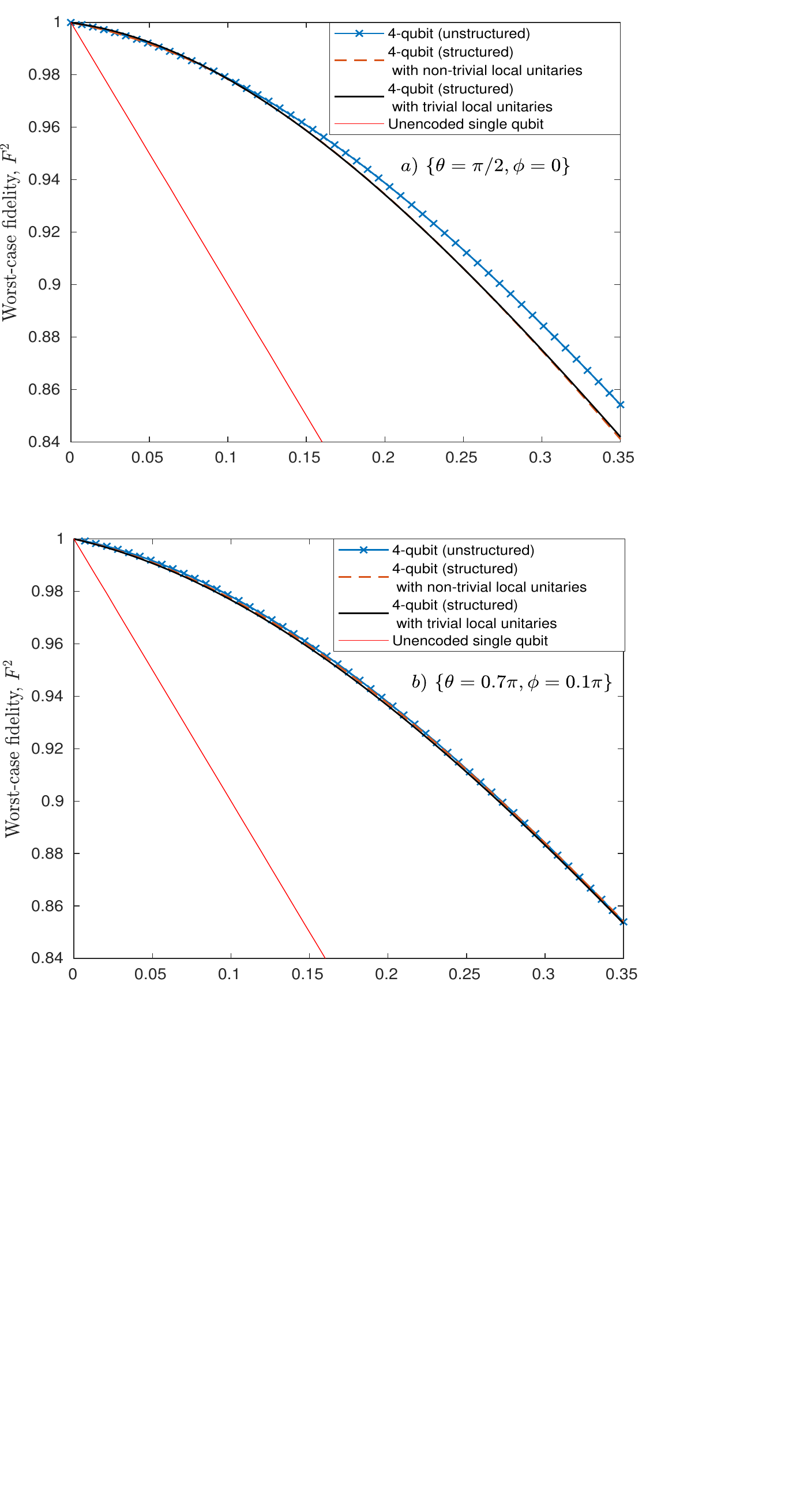}
\end{figure}
\begin{figure}
\includegraphics[trim=3mm 18cm 7cm 3mm,clip, width=20cm]{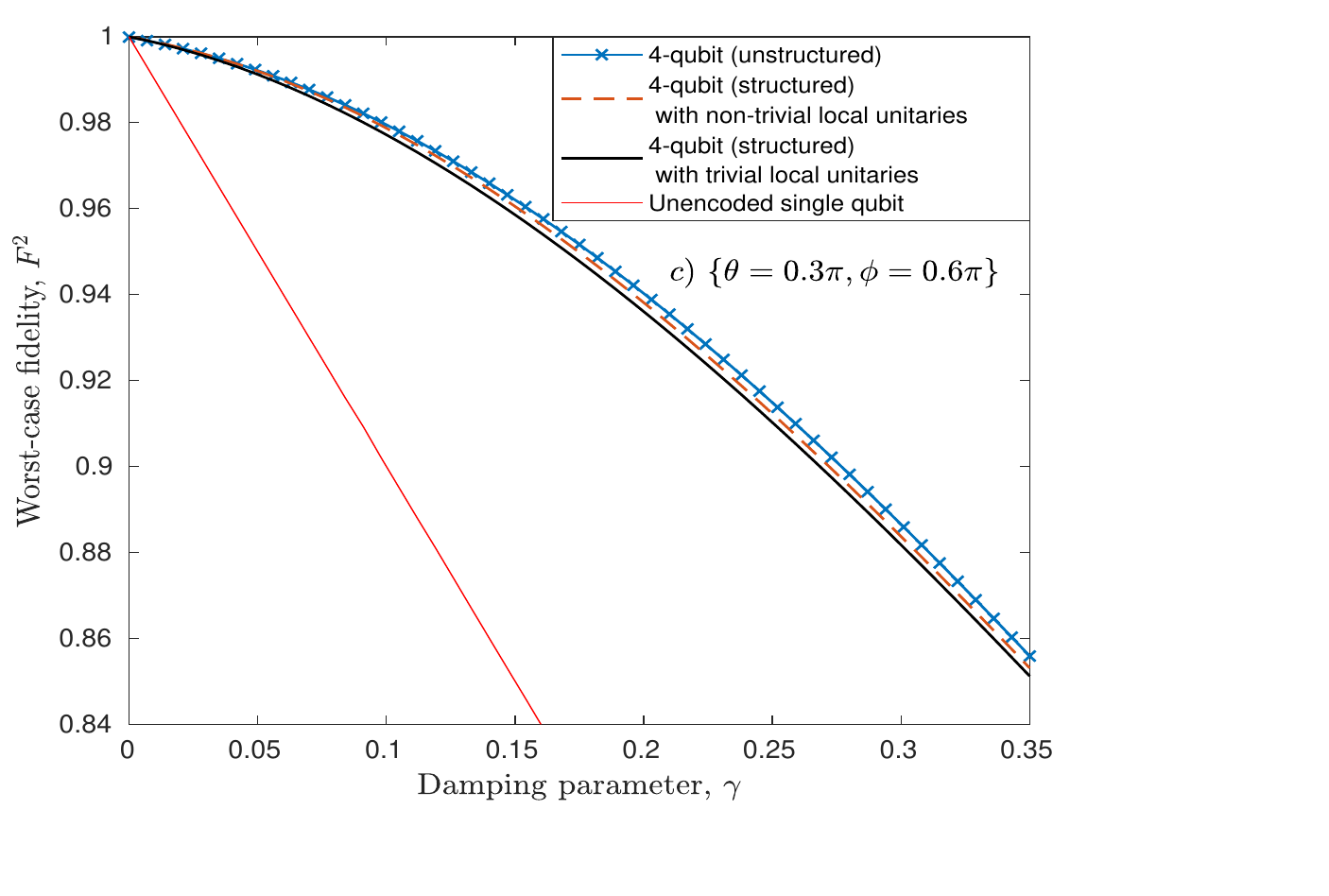}
\caption{Approximate $4$-qubit codes for damping along different directions in the Bloch sphere: (a) the $x$ direction (spherical coordinates $\{\theta,\phi \}$= $\{\pi/2,0\}$); (b) along the direction $\{\theta,\phi\}=\{0.7\pi,0.1\pi\}$; and (c) along the direction $\{\theta,\phi\}= \{0.3\pi,0.6\pi\}$.}
\label{fig2}
\end{figure} 

As our Cartan decomposition uses the Pauli basis, it is important to test if our numerical search is robust against noise not aligned along the axes used to define the Pauli basis. We therefore consider amplitude-damping channels where the damping is no longer in the $Z$ basis. Specifically, we consider the single-qubit \emph{rotated} amplitude-damping channel $\cE_\mathrm{RAD}$ described by the Kraus operators,
\begin{align} \label{eq:arbitrary}
E'_{0} &= |v\rangle\langle v| +\sqrt{1-\gamma} \, \vert v^{\perp}\rangle \langle v^{\perp}\vert, \nonumber \\
\textrm{and }\quad E'_{1} &= \sqrt{\gamma} \, \vert v \rangle \langle v^{\perp} \vert ,
\end{align}
where $\{|v\rangle,|v^\perp\rangle\}$ is a pair of orthonormal vectors on the Bloch sphere. Such a pair of vectors can be parameterized with respect to the $\{|0\rangle, |1\rangle\}$ basis using spherical coordinates,
\begin{eqnarray}\label{eq:parametrise}
|v\rangle &=& \cos(\theta/2) |0\rangle  + \upe^{\upi \phi} \sin (\theta/2) |1\rangle, \nonumber \\
\textrm{and }\quad |v^{\perp}\rangle &=& - \upe^{-\upi \phi} \sin(\theta/2) |0 \rangle + \cos(\theta/2) |1\rangle ,
\end{eqnarray}
with $\theta$ $\in$ $[0,\pi]$, $\phi$ $\in$ $[0,2\pi]$. The values of $\{\theta,\phi\}$ thus determine the damping direction. 

We present numerical search results for the amplitude-damping channel aligned along three different directions in Fig.~\ref{fig2}. In all three examples, the structured search with nontrivial local unitaries was implemented by fixing the local unitaries as $U \equiv (|v\rangle \langle 0|+|v^{\perp}\rangle \langle 1|)$ $\in$ $SU(2)$. For example, when the damping noise is aligned along the $x$-direction on the Bloch sphere, the basis $\{|v\rangle, |v^{\perp}\rangle\}$ is the eigenbasis of the Pauli $X$ operator and the local unitaries are fixed to be the Hadamard gate, which rotates the $\{|0\rangle, |1\rangle\}$ basis to the $\{|+\rangle, |-\rangle\}$ basis. 

Fig.~\ref{fig2}(a) shows the performance of different codes when the damping is with respect to the $X$ eigenstates, whereas Figs.~\ref{fig2}(b) and (c) present the results for choices of damping direction $|v\rangle$ not aligned with one of the standard Pauli axes. In all three cases, we observe that the codes obtained using the unstructured search offer only slightly better fidelity than the codes obtained using the structured searches. Furthermore, the codes obtained using nontrivial local unitaries are often distinct from, and offer better fidelity compared to the codes obtained using trivial local unitaries in the search. Once again, our search procedure is efficient, with the structured and unstructured searches taking between tens to hundreds of seconds on a standard desktop computer. As in the earlier case, we have also compared the performance of the $4$-qubit optimal codes with the fidelity of the single unprotected qubit.

\subsection{Random local noise}

As a third example of the usefulness of our numerical search procedure, we search for good codes for $\cE^{\otimes n}$, where $\cE$ is a randomly chosen single-qubit channel. A random qubit channel $\Phi$ is generated using a Haar-random unitary acting on the system qubit and a single-qubit ancilla initialized to the state $|0\rangle$; the unitary acts jointly on the qubit and the ancilla, after which the ancilla is traced out, yielding a single-qubit channel. We then admix $\Phi$ with the identity channel to give a family of qubit noise channels $\cE$, for different $\alpha\in[0,1]$, of the form,
\begin{equation}\label{eq:weaknoise}
\cE(\cdot)=(1-\alpha)(\cdot) + \alpha \Phi(\cdot).
\end{equation}
$\alpha$ parametrizes the noise strength: for small values of $\alpha$, $\cE$ describes weak noise, which is the practically relevant case.

\begin{figure}[H]
\includegraphics[trim=5mm 9cm 12mm 5mm, clip, height=10cm, width=\columnwidth]{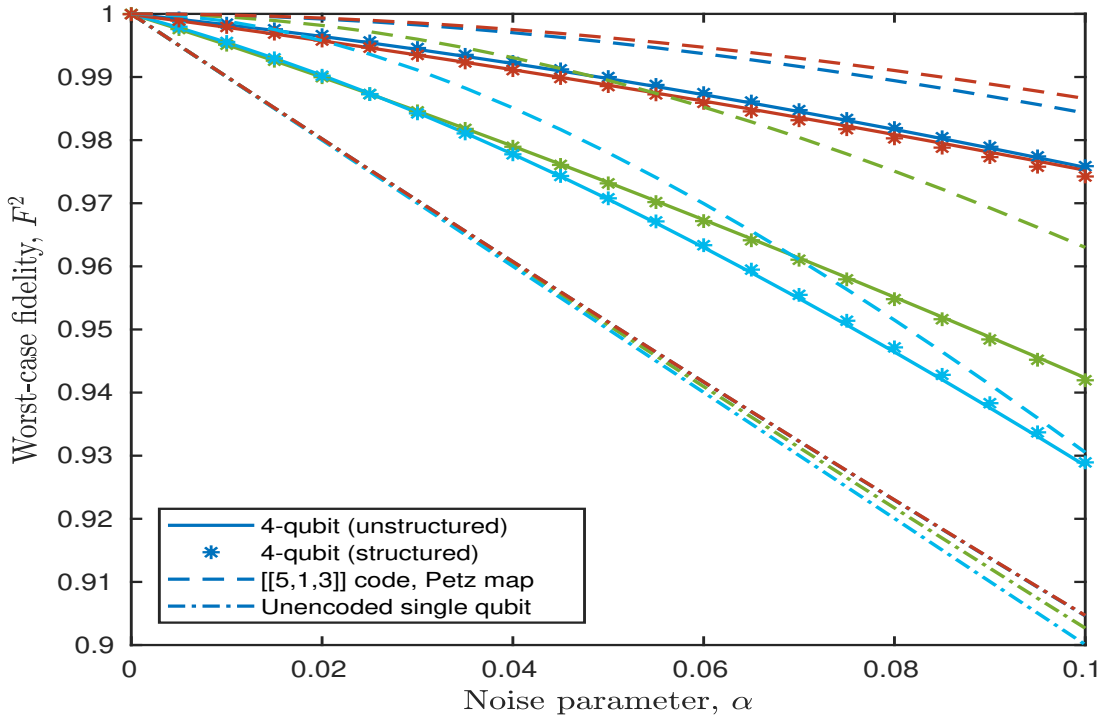}
\caption{Performance of $4$-qubit numerical code and the $[[5,1,3]]$ code for random local noise.}
\label{fig:random}
\end{figure}

Using our numerical search procedure we now obtain optimal $4$-qubit codes for the class of random local noise channels described by Eq.~\eqref{eq:weaknoise} in the weak noise regime, $\alpha \in [0,0.1]$. For each choice of random $\Phi$, and hence $\cE$, we use our numerical AQEC approach to identify good $4$-qubit codes for the 4-qubit channel $\cE^{\otimes 4}$.

Fig.~\ref{fig:random} shows the fidelities obtained for the optimal codes for four random choices of $\Phi$. We compare the performance of the four-qubit codes obtained via structured and unstructured searches with the performance of the $[[5,1,3]]$ code, for each random channel. The recovery procedure used in each case is the corresponding Petz recovery. 

In Fig.~\ref{fig:random}, we observe that the best $4$-qubit codes --- structured or unstructured --- have fidelities linear in the noise strength $\alpha$, suggesting that the $4$-qubit Hilbert space might not be sufficient to distinguish amongst the no-error case and the eight single-qubit errors arising from the weak noise $\cE$. This issue is clearly resolved when we use the $[[5,1,3]]$ code and the corresponding Petz recovery, since the fidelity is now quadratic to leading order in $\alpha$. Finally, we note that the $4$-qubit codes do yield a better worst-case fidelity than the single unencoded qubit under the action of the noise channel.

\section{Conclusions}\label{sec:conc1}

We have described a numerical search algorithm to find good quantum codes, using the worst-case fidelity as the figure of merit. By choosing the recovery map as the Petz recovery, we have reduced the general problem from a triple optimization over the encoding, recovery, and input states, to a single optimization over the encoding map only. Furthermore, the use of the Cartan decomposition, motivated by the typical scenario of independent per-qubit noise, allowed for a reduction of the search space to structured encodings, with performance comparable with the more expensive unstructured ones, as illustrated by our examples.           

The ability to identify channel-adapted codes that involve fewer qubits than the stabilizer codes targeting arbitrary noise, might suggest that the corresponding encoding/decoding circuits might also be smaller in size. In fact, the well-known $4$-qubit channel-adapted code due to Leung et al~\cite{leung} does have a much simpler encoding circuit~\cite{fletcherthesis} than the $5$-qubit stabilizer code, and the encoding circuit is made up of only Clifford gates. Furthermore, in Chapter~\ref{Chapter4} we demonstrate the usefulness of our shorter-length optimal numerical codes using a specific application.

{In our search for the quantum codes we have assumed that the encoding/decoding gates are perfect. This is of course the standard assumption in any discussion on QEC codes in the literature. However, the gates used to do the error correction are indeed the same ones as those used to do computation, so this assumption has to be relaxed, taking us into the domain of fault-tolerant quantum computing~\cite{gottesman_FT}. Our work deals with the first step of finding the optimal code for a given noise process. Extending this framework to fault tolerance is the next step for future work. In this context, our work provides an easy platform to explore and discover a large number of good candidate codes, which can then be individually examined to find the one that can be most easily implemented in a fault-tolerant manner. For the first time, we make an attempt to build a fault-tolerant computation scheme using the channel-adapted $4$-qubit code~\cite{leung} in Chapter~\ref{Chapter5}.}

%% file: Chapter4_v2.tex

\chapter{Quantum State Transfer using Adaptive Quantum Codes} 

\label{Chapter4} 

\lhead{Chapter 4. \emph{Chap:4}} 

\section{Introduction}
Quantum communication entails transmission of an arbitrary quantum state from one spatial location to another. Spin chains are a natural medium for quantum state transfer over short distances, with the dynamics of the transfer being governed by the Hamiltonian describing the spin-spin interactions along the chains. Starting with the original proposal by Bose~\cite{bose} for state transfer via a $1$-d Heisenberg chain, several protocols have been developed for {\it perfect} as well as {\it pretty good} quantum state transfer via spin chains. 

Perfect state transfer protocols typically involve engineering the coupling strengths between the spins in such a way as to ensure perfect fidelity between the state of sender's spin and that of the receiver's spin~\cite{christandl,christandl2005perfect,albanesemirror,karbach,di}. Alternately, there have been proposals to use multiple spin chains in parallel, and apply appropriate encoding and decoding operations at the sender and receiver's spins so as to transmit the state perfectly~\cite{conclusive,perfect,efficient}. Experimentally, perfect state transfer protocols have been implemented in various architectures including nuclear spins~\cite{bochkin} and photonic lattices using coupled waveguides~\cite{perez2013,chapman}. 

Relaxing the constraint of perfect state transfer, protocols for pretty good transfer aim to identify optimal schemes for transmitting information with high fidelity across permanently coupled spin chains~\cite{godsil2012,godsil}. One approach is for example to encode the information as a Gaussian wave packet in multiple spins at the sender's end~\cite{osborne,hasel}. Moving away from ideal spin chains, quantum state transfer has also been studied over disordered chains, both with random couplings and as well as random external fields~\cite{perfect, ashhab, chiara}.  

Here, we study the problem of pretty good state transfer from a quantum channel point of view. It is known~\cite{bose} that state transfer over an ideal $XXX$ chain (also called the Heisenberg chain) can be realized as the action of an amplitude-damping channel~\cite[Chapter 8]{nielsen} on the encoded state. Naturally, this leads to the question of whether quantum error correction (QEC) can improve the fidelity of quantum state transfer. Such an improvement has been reported by making use of the $5$-qubit code for pretty good transfer over the Heisenberg spin chain~\cite{allcock}. On a related note, QEC-based protocols have also been developed for pretty good state transfer over noisy $XX$~\cite{kay, kay2018perfect} chains. 

In our work we study the role of adaptive QEC in achieving pretty good transfer over a class of $1$-d spin systems which preserve the total spin. This includes both the $XX$ as well as the Heisenberg chains, and more generally, the $XXZ$ chain. We use an {\it approximate} QEC (AQEC) code, which has been shown to achieve the same level of fidelity as perfect QEC codes for certain noise channels while making use of fewer physical resources~\cite{leung, fletcher_codes, hui_prabha, cafaro}. Our protocol involves the use of multiple identical spin chains in parallel, with the information encoded in an entangled state across the chains. This is in contrast to the protocols in~\cite{kay, kay2018perfect} which use perfect QEC codes and encode into multiple spins on a single chain. Using the worst-case fidelity between the states of the sender and receiver's spins as the figure of merit, we demonstrate that pretty good state transfer maybe achieved over a class of spin-preserving Hamiltonians using an approximate code and a channel-adapted recovery map.

We present analytical and numerical results for the fidelity of state transfer obtained using our QEC scheme, for ideal as well as disordered $XXX$ chains. The presence of disorder in a $1$-d spin chain is known to lead to the phenomenon of localization~\cite{anderson}. Here, we analyze the distribution of the transition amplitude for a disordered $XXX$ chain, with random coupling strengths which are drawn from a uniform distribution. We modify the QEC protocol suitably so as to ensure pretty good transfer when the disorder strength is small. As the disorder strength increases, our analysis points to a threshold beyond which QEC does not help in improving the fidelity of state transfer. 

\section{Preliminaries}\label{sec:prelim}

We consider a general $1$-d spin chain with nearest neighbour interactions described by the Hamiltonian,
\begin{eqnarray} \label{eq:H_gen}
\cH &=& -\sum_{k} J_{k}\left(\sigma^{k}_{x}\sigma^{k+1}_{x}+\sigma^{k}_{y}\sigma^{k+1}_{y}\right) - \sum_{k}\tilde{J}_{k}\sigma_{z}^{k}\sigma^{k+1}_{z} \nonumber \\
&&  + \sum_{k}B_{k}\sigma_{k}^{z},
 \end{eqnarray}
 where,  $\{J_{k}\}>0$ and $\{\tilde{J}_{k}\}>0$ are site-dependent exchange couplings of a ferromagnetic spin chain, $\{B_{k}\}$ denote the magnetic field strengths at each site, and, $(\sigma^{k}_{x},\sigma^{k}_{y},\sigma^{k}_{z})$ are the Pauli operators at the $k^{\rm th}$ site. The spin sites are numbered as $j = 1,2, \ldots ,N$. We assume that the sender's site is the $s^{\rm th}$ spin and receiver's site is the $r^{\rm th}$ spin. 


We denote the ground state of the spin as $|\textbf{0}\rangle = |000\ldots 0\rangle $. Since we are interested in transmitting a qubit worth of information along the chain, we will work within the subspace spanned by the set of single particle excited states $|\textbf{j}\rangle$, with $|\textbf{j}\rangle$ denoting the state with the $j^{\rm th}$ spin alone flipped to $|1\rangle$. The Hamiltonian in Eq.~\eqref{eq:H_gen} preserves the total number of excitations, that is, $\left[ \cH,\sum_{i=1}^{N} \sigma _{z}^{i} \right] = 0 $ and hence the resulting dynamics is restricted to the $(N+1)$-dimensional subspace spanned by the single particle excited states and the ground state. 

The sender encodes an arbitrary quantum state $|\psi_{\rm{in}}\rangle = a|0\rangle + b|1\rangle $ at the  $s^{ th}$ site, with the coefficients $a$ and $b$ parameterized using a pair of angles $(\theta,\phi)$ as $a=\cos(\frac{\theta}{2})$, $b = e^{-i \phi} \sin(\frac{\theta}{2})$. The initial state of the spin chain is thus given by,
\begin{equation}
|\Psi(0)\rangle = a |\textbf{0}\rangle + b|\textbf{s}\rangle,  
\end{equation}
where $|\textbf{s}\rangle$ is the state of the spin chain with only the $s^{ th}$ spin is flipped to $|1\rangle$ and all other spins set to $|0\rangle$. Under the action of the Hamiltonian $\cH$ described in Eq.~\eqref{eq:H_gen}, after time $t$, the spin chain evolves to the state (here, and in what follows, we set $\hbar = 1$),
\begin{eqnarray}
|\Psi(t)\rangle &=& e^{-i\cH t}|\Psi(0)\rangle,\nonumber \\
 &=& a |\textbf{0}\rangle + b \sum_{j=1}^{N}\langle\textbf{j}\vert e^{-i\cH t}\vert \textbf{s}\rangle|\textbf{j}\rangle. \nonumber
\end{eqnarray}
Following~\cite{bose}, the state of the receiver's spin at the $r^{ th}$ site after time $t$, denoted as $\rho_{\rm out}(t)$, is obtained by tracing out all the other spins from the state of the full spin chain $\rho(t)= \vert\Psi(t)\rangle\langle\Psi(t)\vert$, as follows.
\begin{eqnarray}
&& \rho_{\rm out}(t) = \tr_{1,2,\ldots,r-1,r+1,N-1}\left[\rho(t)\right] \nonumber \\
 &=& \left[ |a|^{2} +|b|^{2}\left(1-|f_{r,s}^N (t)|^{2}\right) \right] |0\rangle\langle 0| + ab^{*} (f_{s,r}^{N}(t))^{*}|0\rangle\langle 1| \nonumber\\
&+&  ba^{*}f_{r,s}^{N}(t)|1\rangle\langle 0| + |b|^{2}|f_{r,s}^N(t)|^{2}|1\rangle\langle 1|, \label{eq:rho_out}
\end{eqnarray}
where, 
\begin{equation}
 f_{r,s}^N(t) = \langle \textbf{r} |e^{(-i \cH t)}|\textbf{s} \rangle \label{eq:trans_amp0}
 \end{equation}
is the {\it transition amplitude}, which gives the probability amplitude for the excitation to transition from the $s^{\rm th}$ site to $r^{\rm th}$ site. The function $f_{r,s}^{N}(t)$ satisfies,
\begin{eqnarray}
\sum_{r=1}^{N}|f_{r,s}^{N}(t)|^{2} &=& 1, \, \forall \; s = 1,2,\ldots, N . \nonumber \\
\sum_{k=1}^{N}f^{N}_{r,k}(t)(f^{N}_{k,s}(t))^{*} &=& \delta_{rs} , \, \forall \; k = 1,2,\ldots, N . \label{eq:trans_amp}
\end{eqnarray}
where $\delta_{rs}$ is the delta function with $\delta_{rs} = 1 $ for $r = s$ and $\delta_{rs} = 0$ for $r\neq s$.

As in the case of the Heisenberg chain~\cite{bose}, we obtain the reduced state in  Eq.~\eqref{eq:rho_out} at receiver's end as the action of a quantum channel $\widetilde{\cE}_{AD}$ on the input state. Specifically, 
\begin{equation} 
\rho_{\rm{out}}(t) = \widetilde{\cE}_{AD}(\rho_{\rm in}) = \sum_{k}E_{k}\rho_{\rm in}E_{k}^{\dagger}, 
\end{equation}
where $E_{0}$ and $E_{1}$ are the Kraus operators that describe the action of the channel. It is easy to see that the operators $E_{0}, E_{1}$ have the following form when written in the $\{|0\rangle, |1\rangle\}$ basis.
\begin{equation}
E_{0}  = \left( \begin{array}{cc}
1 & 0 \\
0 & f_{r,s}^{N}(t)
\end{array} \right), \; E_{1} = \left( \begin{array}{cc}
0 & \sqrt{1-|f_{r,s}^{N}(t)|^{2}} \\
0 & 0
\end{array} \right). \label{eq:Kraus_ideal}
\end{equation}
The Kraus operators in Eq.\eqref{eq:Kraus_ideal} lead to a channel that has the same structure as the amplitude-damping channel described in Eq.~\ref{eq:ampdamp}, but is more general since the parameter $f_{r,s}^{N}(t)$ characterizing the noise in the channel is complex. 

The standard amplitude-damping channel (Eq.~\ref{eq:ampdamp}) parameterized by a {\it real} noise parameter $p$, 
is the quantum channel induced in the original state transfer protocol in~\cite{bose}, where the Hamiltonian considered is a Heisenberg chain in the presence of an external field of the form $\vec{B} = B \hat{z}$, 
\begin{equation}
\tilde{\cH} =  - \frac{J}{2}\sum_{\langle i,j \rangle}\vec{\sigma}^{i}\cdot\vec{\sigma}^{j} - B\sum_{i}\sigma_{z}.  \label{eq:Heisenberg_B}
\end{equation}
By choosing the intensity of the $\vec{B}$-field appropriately, it is possible to adjust the phase of the complex amplitude $f_{r,s}^{N}(t)$ to be a multiple of $2\pi$ and hence replace $f_{r,s}^{N}(t)$ by $\vert f_{r,s}^{N}(t)\vert$, thus obtaining the amplitude-damping channel described in Eq.~\eqref{eq:ampdamp}. 

While much of the past work on state transfer has focused on the Heisenberg Hamiltonian in Eq.~\eqref{eq:Heisenberg_B}, here, we will focus on the more general Hamiltonian in Eq.~\eqref{eq:H_gen}. We study the problem of transmitting an arbitrary quantum state from the $s^{\rm th}$ site to the $r^{\rm th}$ site of an $N$-spin chain. We quantify the performance of the protocol in terms of the fidelity between the final state $\rho_{\rm out} \equiv \widetilde{\cE}_{AD}(|\psi_{\rm in}\rangle\langle\psi_{\rm in}|)$ and the input state $|\psi_{\rm in}\rangle$. Specifically, we use the square of the {\it worst-case} fidelity defined in Eq.~\ref{eq:worstcase_fidelity},
\begin{equation}
 F^{2}_{\rm min} (\widetilde{\cE}_{AD})  = \min_{a,b} \, \langle \psi_{\rm in} \vert \rho_{\rm out} \vert \psi_{\rm in} \rangle, \nonumber
\end{equation}
where the minimization is over all possible input states $a|0\rangle + b |1\rangle$. We say that pretty good state transfer is achieved when the worst-case fidelity $F^{2}_{\rm min}(\widetilde{\cE}_{AD}) \geq 1 -\epsilon$, for some $\epsilon > 0$. 

Let $|f^{N}_{r,s}(t)|$ and $\Theta$ refer to the amplitude and phase respectively, of the noise parameter $f^{N}_{r,s}(t) = e^{i\Theta}|f^{N}_{r,s}(t)|$ of the general quantum channel in Eq.~\eqref{eq:Kraus_ideal}. For such a channel, the worst-case fidelity depends on both the amplitude $|f^{N}_{r,s}(t)|$ as well as the phase $\Theta$. However, following the original protocol in~\cite{bose}, if we choose the magnetic fields $\{B_{k}\}$ so as to ensure that $\Theta$ is a multiple of $2\pi$, we can show that,
\begin{equation}
F^{2}_{\rm min} (\widetilde{\cE}_{AD}) = |f_{r,s}^{N}(t)|^{2}. \label{eq:fmin_noQEC}
\end{equation}
In what follows, we examine how the worst-case fidelity may be improved using techniques from quantum error correction. In particular, by obtaining a functional relationship between the worst-case fidelity and the transition amplitude using an adaptive QEC procedure, we show how the fidelity can be improved by an order of the magnitude in the noise parameter.


\section{State Transfer protocol based on adaptive QEC}\label{sec:ideal}

Given a specific form of the spin-conserving Hamiltonian in Eq.~\eqref{eq:H_gen}, it is possible to estimate $|f^{N}_{r,s}(t)|$ and $\Theta$ for a specific choice of sites $s,r$ and $t$ by making repeated measurements on the spin chain~\cite{perfect}. Knowing $\Theta$, we may apply a phase gate of the form,
\begin{equation}
 U_{\Theta} = \left( \begin{array}{cc}
1 & 0 \\
0 & e^{-i\Theta} 
\end{array} \right),  \label{eq:theta-gate}
\end{equation}
to change the encoding basis to $\{|0\rangle, e^{-i\Theta}|1\rangle\}$. In this rotated basis, the channel in Eq.~\eqref{eq:Kraus_ideal} is identical to the amplitude-damping channel described in Eq.~\eqref{eq:ampdamp}. At the level of the Hamiltonian, this is the same as choosing the field strengths $\{B_{k}\}$ so as to make the phase $\Theta$ trivial. Indeed, by making an appropriate choice of magnetic fields, it is always possible to transfom the spin-preserving Hamiltonian in Eq.~\eqref{eq:H_gen} into an $XXX$ interaction as in Eq.~\eqref{eq:Heisenberg_B}  (see~\cite{kayreview}) and hence map the underlying noise channel to an amplitude-damping channel. 

One na\"ive approach to improving the fidelity of state transfer is to therefore first apply the $U_\Theta$-gate and then use any of the well known QEC protocols which correct for amplitude-damping noise~\cite{leung, hui_prabha, fletcher_codes, AD_reliable2017}.  However, such an approach fails in the presence of disorder. When we consider a disordered $1$-d spin chain wherein either the  couplings $\{J_{k}, \tilde{J}_{k}\}$ or the fields $\{B_{k}\}$  in Eq.~\eqref{eq:H_gen} maybe random, the underlying noise channel is stochastic. The two real parameters $|f^{N}_{r,s}(t)|$ and $\Theta$ characterizing the noise in the channel vary with each disorder realization, and hence an encoding procedure that relies on knowledge of a specific realization of $\Theta$ is not useful. Moreover, implementing a phase gate as in Eq.~\eqref{eq:theta-gate} based on the disorder-averaged value of $\Theta$ does not help -- such a phase gate will no longer cancel out the arbitrary (random) phase in Eq.~\eqref{eq:Kraus_ideal} and we do not obtain an amplitude-damping channel in the rotated basis after the action of the phase gate.  

We would therefore like to tackle the problem of correcting for the more general noise channel in Eq.~\eqref{eq:Kraus_ideal}. Taking inspiration from the structural similarity to the amplitude-damping channel, we propose a QEC protocol using an {\it approximate} $4$-qubit code~\cite{leung} along with the channel-adapted near-optimal recovery proposed in~\cite{hui_prabha}. The $[[4,1]]$ code has already been described in Eq.~\ref{eq:4qubit}. We recall here that the $4$-qubit codespace $\cC$ is realised as the span of the following pair of orthogonal states,
\begin{eqnarray}\label{eq:4qubit_4}
|0_{L}\rangle &=& \frac{1}{\sqrt{2}}\left( \, |0000\rangle + |1111\rangle \, \right),\nonumber \\ 
|1_{L}\rangle &=& \frac{1}{\sqrt{2}}\left( \, |1100\rangle + |0011\rangle \, \right) .
\end{eqnarray}
This code was shown to be {\it approximately} correctable for amplitude-damping noise, both in terms of worst-case fidelity~\cite{leung} as well as entanglement fidelity~\cite{cafaro}. The code is approximate in the sense it does not satisfy the conditions for perfect quantum error correction~\cite{nielsen}, for any single-qubit error. 

We now use the Petz recovery mentioned in Sec.~\ref{sec:Petzmap} (Chapter~\ref{Chapter2}), known to achieve better worst-case fidelity for the case of amplitude-damping channel and $4$-qubit code compared to the standard QEC procedure~\cite{hui_prabha}. Recall from Sec.~\ref{sec:Petzmap} that the Petz map $R_P$ can be described in terms of the Kraus operators of the noise $\widetilde{\cE}_{AD}$ and the projector $P$ onto the codespace, as follows,
\begin{equation}\label{eq:Petz}
\cR_{P}(.) =\sum_{i}P E_{i}^{\dagger}\widetilde{\cE}_{AD}(P)^{-1/2}(.)\widetilde{\cE}_{AD}(P)^{-1/2}E_{i}P ,
\end{equation}
where the inverse of $\widetilde{\cE}_{AD}({P})$ is taken on its support. 
\begin{figure}[H]
\flushleft \hspace*{-1cm}\includegraphics[scale=0.8]{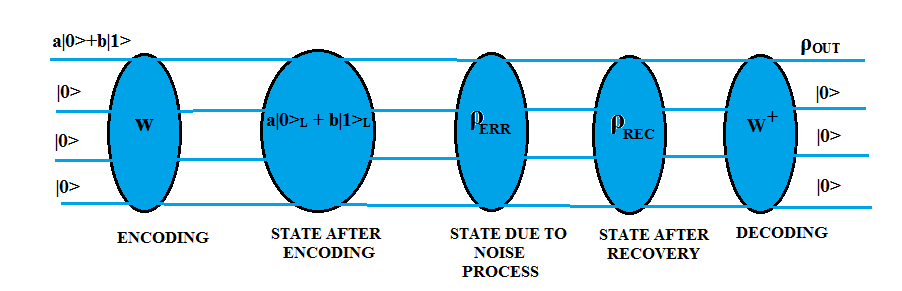}
\caption{$4$-qubit QEC on spin chains}
\label{fig:QEC_schematic}
\end{figure}

The quantum state transfer protocol with QEC is implemented using a set of $4$ unmodulated, identical, spin chains. Fig.~\ref{fig:QEC_schematic} depicts a schematic of our protocol. The initial, encoded state $|\psi_{\rm enc}\rangle$ is now an entangled state across the four chains, involving only a single spin (the $s^{\rm th}$ site)in each of the chains.
\begin{equation}\label{eq:psi_enc}
\vert\psi_{\rm enc}\rangle = a|0\rangle_{L} +b |1 \rangle_{L} .
\end{equation}
Once the initial state is prepared, the four chains are allowed to evolve in an uncoupled fashion, according to the Hamiltonian in Eq.~\eqref{eq:H_gen}. After time $t$, the state at the receiver's site is a joint state of the $r^{\rm th}$ site of the four chains, and is described by action of the map ${\widetilde{\cE}_{AD}}^{\otimes 4}$ with the time-dependent noise parameter $f_{r,s}^{N}(t)$. Thus,
\[ \rho_{\rm err} = \small{{\widetilde{\cE}_{AD}}^{\otimes 4}(\rho_{\rm enc})} = \sum_{i} E^{(4)}_{i} \rho_{\rm enc} \left( E^{(4)}_{i}\right)^{\dagger}, \]
where $E^{(4)}_{i}$ are the Kraus operators of the $4$-qubit noise channel realized as four-fold tensor products of the operators $E_{0}$ and $E_{1}$ in Eq.~\eqref{eq:Kraus_ideal}. After evolving the chains for time $t$, the recovery map $\cR_{P}^{(4)}$ is applied at the receiver's site of the four spin chains. The final state at the receiver's end after the QEC protocol is obtained as,
\begin{equation} 
\rho_{\rm rec} = \sum_{i,j} R^{(4)}_{j} E^{(4)}_{i} \rho_{\rm{enc}} \left(E^{(4)}_{i}\right)^{\dagger} \left(R^{(4)}_{j}\right)^{\dagger}, \nonumber
\end{equation}
with the Kraus operators $R^{(4)}_{i}$ given by, 
\begin{equation}
R^{(4)}_{i} = P\small{\left(E^{(4)}_{i}\right)^{\dagger}}\widetilde{\cE}_{AD}^{\otimes 4}(P)^{-1/2} . \label{eq:4qubit_petz}
\end{equation}
As usual, $P\equiv |0_{L}\rangle\langle 0_{L}| + |1_{L}\rangle\langle 1_{L}|$ is the projector onto the $4$-qubit space described in Eq.~\eqref{eq:4qubit_4}. The fidelity of the $4$-chain quantum state transfer protocol is then given by,
\[F^{2}_{\rm min} \left( \cR_{P}^{(4)}\circ\widetilde{\cE}_{AD}^{\otimes 4}, \cC \right) \equiv \min_{a,b}\langle \psi_{\rm enc}\vert \rho_{\rm rec}\vert\psi_{\rm enc}\rangle.,\]
where the minimization is over all states in the codespace $\cC$. As before, pretty good transfer is achieved when the worst-case fidelity is high, that is, $F^{2}_{\rm min}\left( \cR_{P}^{(4)}\circ\widetilde{\cE}_{AD}^{\otimes 4}, \cC \right) \geq 1 - \epsilon$, for $\epsilon > 0$. We now present a bound on the fidelity of state transfer using our adaptive QEC protocol, in terms of the transition amplitude $f^{N}_{r,s,}(t)$.

\begin{theorem}\label{thm:aqec_fid}
The fidelity of quantum state transfer from site $s$ to site $r$ under a spin-conserving Hamiltonian as in Eq.~\eqref{eq:H_gen}, using the $4$-qubit code $\cC$ and adaptive recovery $\cR_{P}^{(4)}$ at time $t$, is given by,
\begin{equation}
F_{\rm min}^{2} \left( \cR_{P}^{(4)}\circ\widetilde{\cE}_{AD}^{\otimes 4}, \cC \right) \approx 1- \frac{7p^{2}}{4}+ O(p^{3}), \label{eq:4qubit_fid}
\end{equation}
where $p= 1-|f_{r,s}^{N}(t)|^{2}$. 
\end{theorem}
\begin{proof}
We first rewrite the Kraus operators given in Eq.~\eqref{eq:Kraus_ideal}, as, 
\begin{eqnarray}
 E_{0} &=& |0 \rangle \langle 0| +|f_{r,s}^{N}(t)|e^{i \Theta} |1\rangle \langle 1| \nonumber \\ \nonumber
 E_{1} &=& |0\rangle \langle 1| \sqrt{1-|f_{r,s}^{N}(t)|^{2}} ,
\end{eqnarray}
where, $|f_{r,s}^{N}(t)|$ and $\Theta$ are the absolute value and phase of the complex-valued transition amplitude $f_{r,s}^{N}(t)$. The state after the $4$-qubit recovery map is then given by,
\[ \rho _{\rm rec} = \left(\small{\cR_{P}^{(4)}\circ\widetilde{\cE}_{AD}^{\otimes 4}}\right) \left(\rho_{\rm enc}\right).\]
The composite map $\left(\cR_{P}^{(4)}\circ\widetilde{\cE}_{AD}^{\otimes 4}\right)$ comprising noise and recovery has Kraus operators of the form,
\begin{equation}
 \small{P\left(E^{(4)}_{j}\right)^{\dagger}}\widetilde{\cE}_{AD}^{\otimes 4}(P)^{-1/2}E^{(4)}_{i}P. \label{eq:kraus_composite}
 \end{equation}
The key step in obtaining the desired fidelity is to show that the Kraus operators of the composite map written above are independent of $\Theta$. First, we write out $\small{\widetilde{\cE}_{AD}^{\otimes 4}(P)^{-1/2}}$ in the (standard) computational basis of the $4$-qubit space. 
\begin{eqnarray}
&& \small{\widetilde{\cE}_{AD}^{\otimes 4}(P)^{-1/2}} = \sum_{i=1}^{16}\cG_{i}|i\rangle\langle i| + e^{-4 i \Theta } \cG_{17}|0000\rangle\langle 1111| \nonumber \\
 &+& e^{i 4\Theta} \cG_{17}|1111\rangle\langle0000| + \cG_{18}(|1100\rangle \langle0011| +|0011\rangle \langle 1100|), \nonumber
\end{eqnarray}
where $\{\cG_{i}\}$ are polynomial functions of the transition amplitude $|f_{r,s}^{N}(t)|$. The $\Theta$-dependence in this pseudo-inverse operator occurs only in the span of $\{|0000\rangle, |1111\rangle\}$. Since $\small{\widetilde{\cE}_{AD}^{\otimes 4}(P)^{-1/2}}$ is sandwiched between the Kraus operators of the $4$-qubit channel and their adjoints, we also write down the Kraus operators $\{E^{(4)}_{i}\}$ in the computational basis. Then, an explicit computation reveals that the $\Theta$-dependence gets conjugated out for each of the Kraus operators in Eq.~\eqref{eq:kraus_composite}.  We refer to Appendix~\ref{sec:E(P)} for the details of this calculation. Furthermore, we also prove a general theorem that states that $\cR_P^{(4)} \circ \widetilde{\cE}_{AD}^{\otimes 4}$ is independent of $\Theta$ for any choice of the $\cC$ $\in$ $\ \cH_2^{\otimes 4}$ in Appendix~\ref{sec:theta_cancellation}.


\noindent Hence the final state after noise and recovery $\rho_{\rm rec}$ can be expressed as a linear sum of terms that are independent of $\Theta$. Since the parameter $\Theta$ is effectively suppressed, the fidelity after using $4$-qubit code and the universal recovery in Eq.~\eqref{eq:Kraus_ideal}, is purely a function of $p = 1-|f_{r,s}^{N}(t)|^{2}$.  

\noindent The fidelity corresponding to the initial state $|\psi_{\rm enc}\rangle = a|0_{L}\rangle + b|1_{L}\rangle$ can thus be obtained as,
\begin{eqnarray} \label{eq:4qubitmin}
&& F^{2} ( \cR_{P}^{( 4)}\circ\widetilde{\cE}_{AD}^{\otimes 4},\cC)  \nonumber \\
&=& 1 - p^{2}\left((|a|^{2} - |b|^{2})^{2} - ((ba^{*})^{2}+(ab^{*})^2)+ 5 |a|^{2}|b|^{2}\right) \nonumber \\
&& + \, O(p^{3}) ,
\end{eqnarray}
where $O(p^{3})$ refers to terms of order $p^{3}$ and higher. Parameterizing $a$ and $b$ as $a=\cos{\frac{\theta}{2}}$, $b=e^{-i \phi}\sin{\frac{\theta}{2}}$, the fidelity attains its minimum value at $\{\theta,\phi\}= \{\frac{(2n+1)\pi}{2},\frac{(2n+1)\pi}{2}\}$ ($n=1,2,\ldots$), so that the worst-case fidelity over the $4$-qubit code $\cC$ is given by,
\[
F_{\rm min}^{2} \left( \cR_{P}^{(4)}\circ\widetilde{\cE}_{AD}^{\otimes 4}, \cC \right) \approx 1- \frac{7p^{2}}{4}+ O(p^{3}).
\] 
\end{proof}

Our result shows that using the adaptive recovery in conjunction with the approximate code leads to a fidelity that is independent of the phase $\Theta$ of the complex noise parameter $f^{N}_{r,s}(t)$. Thus, to optimize the fidelity of state transfer between the $s^{\rm th}$ and $r^{\rm th}$ site of a chain of $N$ spins evolving according to the Hamiltonian in Eq.~\eqref{eq:H_gen}, we simply need to find the time $t$ at which $\vert f^{N}_{r,s}(t)\vert^{2}$ is maximized. Recall that the worst-case fidelity without QEC (using the single chain protocol) is linear in the parameter $p$, as observed in  Eq.~\eqref{eq:fmin_noQEC}. Thus we see an $O(p)$ improvement in fidelity with QEC, as expected. 

Furthermore, our estimate of the worst-case fidelity implies that so long as the noise strength $p$ is such that $1 - (7/4)p^{2} > 1- p$, the adaptive QEC protocol achieves better fidelity than the single chain protocol without QEC. This constraints the noise strength $p$ to satisfy $0<p< (4/7) \approx 0.57$. This in turn implies a threshold for the transition amplitude, namely, $|f_{r,s}^{N}(t)|^{2} > 0.43$, below which our adaptive QEC protocol will not offer any improvement in the fidelity of state transfer. 


\section{Results for the $1$-d Heisenberg Chain} \label{sec:Heisenberg}

As a simple example to illustrate the performance of the adaptive QEC protocol, we now consider a special case of the Hamiltonian in Eq.~\eqref{eq:H_gen}, namely, an $N$-length, ideal Heisenberg chain, with $J_{k} = \tilde{J}_{k} =  J/2 (J>0)$ and $B_{k} =0$, for all $k$. This is also often referred to as the $XXX$-chain in the literature. Setting $J=1$ without loss of generality, we present numerical results on the fidelity of state transfer from the first ($s=1$) to the $N^{\rm th}$ ($r=N$) site.
\begin{figure}[H]
\centering
\includegraphics[width=1\textwidth]{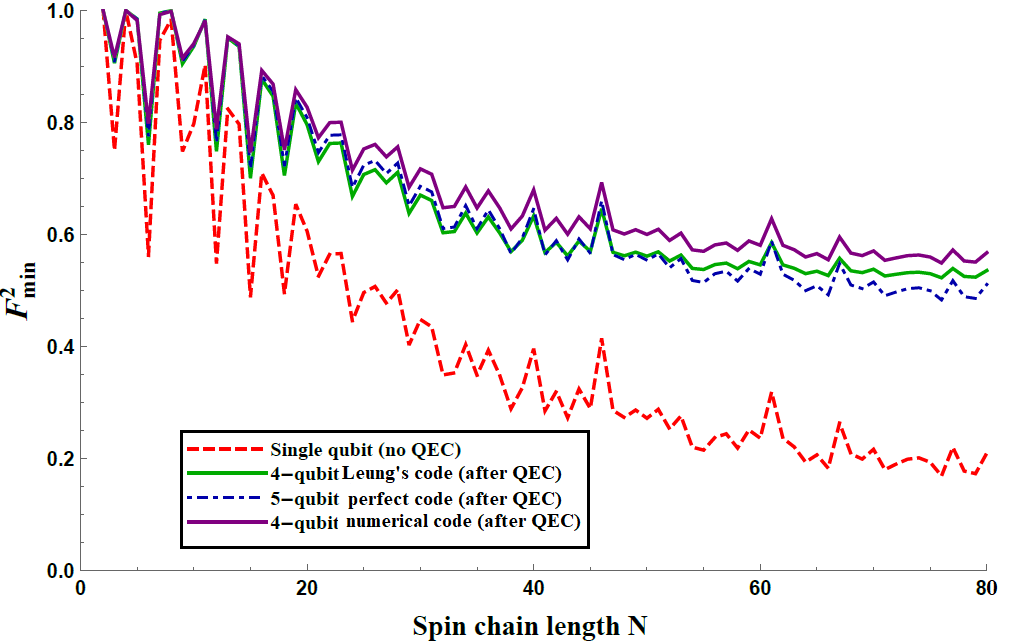}
\caption{Worst-case fidelity as a function of chain length $N$.}
\label{fig:fmin_QEC}
\end{figure}
Fig.~\ref{fig:fmin_QEC} compares the performance of state transfer protocols with and without QEC. In particular, it compares the performance of our $4$-chain state transfer protocol using the $4$-qubit code in \cite{leung} as well as a $4$-qubit code given in Appendix~\ref{sec:numerical_code}, obtained via our numerical search procedure described in Chapter~\ref{Chapter3}, with the single-chain (no QEC) protocol~\cite{bose} and the $5$-chain protocol proposed in~\cite{allcock}. For each $N$, we plot the fidelity of state transfer from the $1^{\rm st}$ site to the $N^{\rm th}$ site on a $N$-length spin chain, after a time $t^{*}$ chosen such that $|f_{N,1}^{N}(t)|$ is maximum at $t=t^{*}$, for $0<t<4000/J$.

From the plot we see that the QEC-based protocols achieve pretty good state transfer over longer distances than the single chain protocol. Furthermore, using {\it approximate} QEC it is possible to achieve as high as fidelity as with the standard $5$-qubit code, using fewer spin chains. Specifically, in the regime of small noise parameter $p$, it is easy to compute that the worst-case fidelity obtained using the $5$-qubit code is,
\begin{equation}
F^{2}_{\min} \approx 1-\frac{15p^{2}}{8}+ O(p^{3}) . \label{eq:5qubit_fid}
\end{equation}
Correspondingly, a $5$-chain protocol performs better than the single chain protocol when $0< p < (8/15) \approx 0.53$, implying that the transition amplitude should satisfy $|f_{r,s}^{N}(t)|^{2} > 0.47$, which is a higher threshold than that required by our adaptive QEC protocol.

For the ideal Heisenberg chain, it was recently shown that~\cite{godsil}, there always exists a time $t$ at which $|f_{1,N}^{N}(t)|^{2} > 1 - \epsilon$ if and only if the length of the chain is a power of $2$, that is, $N = 2^{m}$. In other words, pretty good state transfer is always possible between the ends of a Heisenberg spin chain whose length $N$ is of the form $N = 2^{m} (m>1)$.  We may therefore consider improving the performance of our QEC-based protocol by repeating the error correction procedure every $2^{m}$ sites. Specifically, we can achieve pretty good state transfer over a chain of arbitrary length $L$, by stitching together smaller chains whose lengths are of the form $N =2^{m}$. At every stage of the repeated QEC protocol, there are exactly $2^{m}$ interacting spins and the rest of the spin-spin-interactions are turned off. 
 \begin{figure}[H] 
 \centering
 \includegraphics[width=1\textwidth]{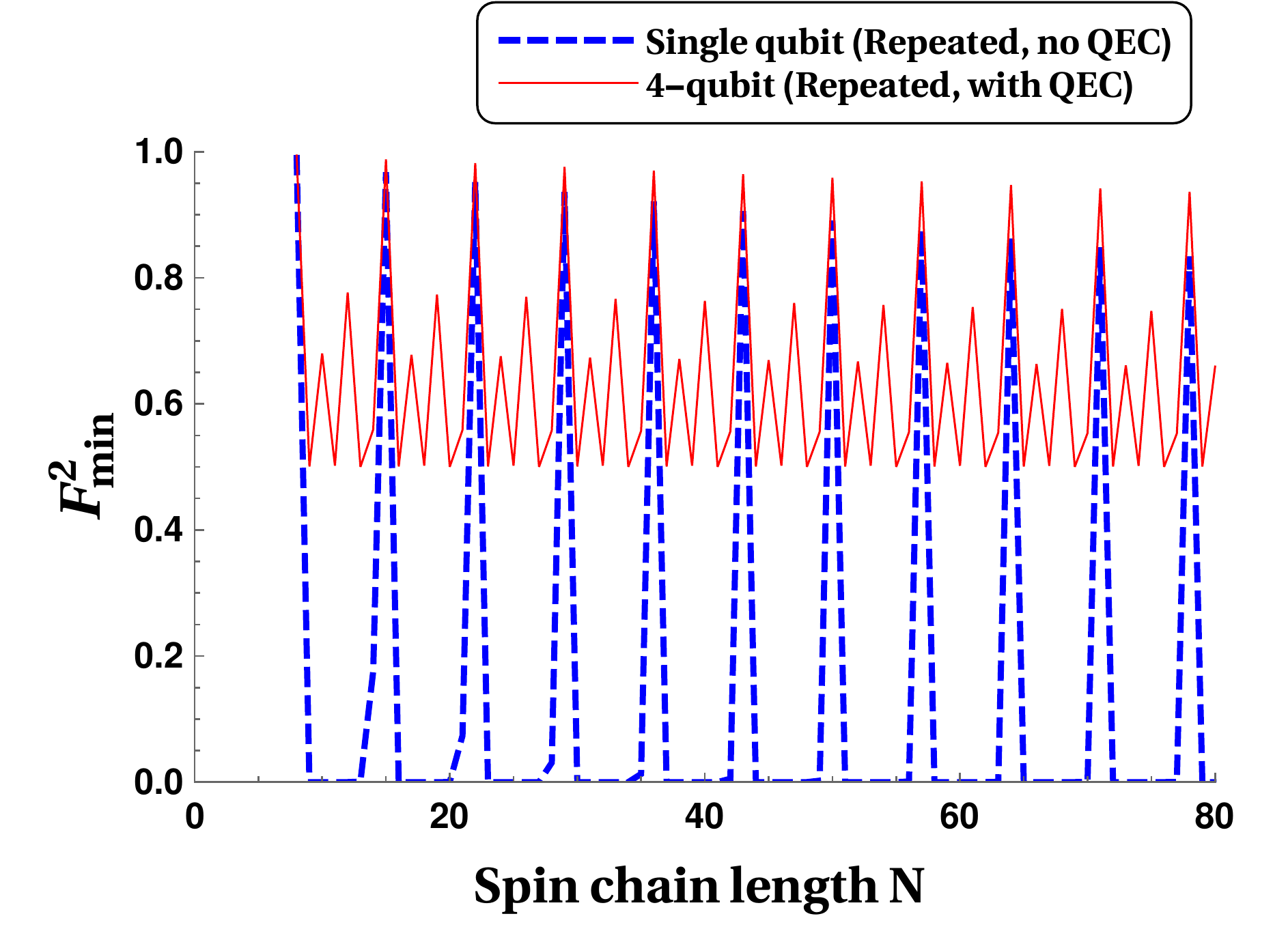}
 \caption{Worst-case fidelity using repeated QEC}
 \label{fig:repeated_QEC}
 \end{figure}
Fig.~\ref{fig:repeated_QEC} shows an example of the resulting improvement in fidelity when the QEC protocol is repeated every $8$ sites. For comparison, we plot the worst-case fidelity obtained by stitching together a sequence of length-$8$ chains, without QEC. The repeated QEC protocol proceeds as follows. We first implement our QEC protocol for an $8$-spin chain, evolving for time $t^{*}$ at which $|f_{8,1}^8 (t)|$ maximizes. We repeat this procedure some $k$ times, where $k$ is the largest integer such that $8k < N$ and finally perform QEC for the remaining $N-7k$ sites for the same waiting time $t^{*}$. Such a repeated QEC protocol indeed enables pretty good transfer for much  longer lengths, as seen in the plot.  

More generally, if $F^{2}_{\rm min} \approx 1 - \alpha p^{2}$ is the fidelity of the single-shot QEC protocol, repeating the procedure $k$ times gives us a fidelity of $F^{2}_{\min} = 1-(p_{\rm{new}})$ , with,
\begin{eqnarray}
p_{\rm{new}}= (1-(1-\alpha p^{2})^{k}) \nonumber
\end{eqnarray}
where $p_{\rm{new}}$ is the noise parameter obtained after repeating QEC $k$ times.

\section{Quantum state transfer over a disordered Heisenberg chain}\label{sec:disorder}

Moving away from an ideal spin chain with a fixed, uniform coupling between successive spins, we now study state transfer over a disordered $XXX$ chain, where the spin-spin couplings are randomly drawn from some distribution. It is well known that the presence of disorder in a $1$-d spin chain leads to the phenomenon of localization~\cite{anderson} of information close to one end of the chain. It is therefore a challenging task to identify protocols which achieve perfect or pretty good transfer over disordered spin chains, overcoming the effects of localization. 

Past work on disordered chains has primarily focused on the $XX$ chain. Starting with a modulated chain that admits perfect state transfer, both random magnetic field and random couplings have been studied~\cite{chiara}. Alternately, an unmodulated chain with random couplings at all except the sender and receiver sites has also been studied~\cite{ashhab}. 

When viewed in the quantum channel picture, the presence of disorder becomes an additional source of noise. The role of QEC in overcoming the effects of disorder has been studied both for the $XX$~\cite{kay} as well as the Heisenberg chains~\cite{allcock}. The QEC protocol for a noisy $XX$ chain with random couplings involves encoding into multiple spins on a single chain using modified CSS codes~\cite{kay}. The QEC protocol in~\cite{allcock} encodes into multiple identical, uncoupled chains using the standard $5$-qubit code, while also requiring access to multiple spins at the sender and receiver ends of each of the chains. Furthermore, the protocol based on the $5$-qubit code involves choosing an encoding based on the phase $\Theta$ of the transition amplitude (as explained in Sec.~\ref{sec:ideal}), which in turn is specific to the disorder realization. This makes the QEC procedure hard to implement in a practical sense. 

Here, we show how the channel-adapted QEC procedure described in Sec.~\ref{sec:ideal} can be used to achieve pretty good state transfer over an $XXX$ chain with random couplings. As before, we quantify the performance of the state transfer protocol in terms of the fidelity between the initial and final states. When the underlying quantum channel is stochastic, as in the case of a disordered chain, we use the {\it disorder-averaged} worst-case fidelity $\langle F^{2}_{\rm min}\rangle_{\delta}$, to characterize the performance of the state transfer protocol. We say that pretty good state transfer is achieved by a certain choice of code $\cC$ and recovery $\cR_{P}$ when the corresponding disorder-averaged fidelity $\langle F^{2}_{\rm min}\rangle_{\delta} \geq 1 -\epsilon$, for some $\epsilon > 0$. 

We consider a disordered Heisenberg chain with couplings $J_{k}$= $\frac{\overline{J}}{2}(1+\Delta_{k})$, where $\Delta_{k}$ are independent, identically distributed random variables drawn from a uniform distribution between $ \left[ -\delta,\delta \right ] $ and $\overline{J}$ is the mean value of the coupling strength, which we may set to $1$, without loss of generality.  Note that such a Hamiltonian conserves the total spin and hence falls within the universality class discussed in Sec.~\ref{sec:prelim}. 

Consider a state transfer protocol, where the sender wishes to transmit the state $|\psi_{\rm in}\rangle = a|0\rangle + b|1\rangle$ from the $s^{\rm th}$ site to the $r^{th}$ site via the natural dynamics of the chain. As before, the final state at the receiver's site, tracing out the other spins can be realized as the action of a quantum channel $\widetilde{\cE}_{AD}$,
\[ \rho_{\rm out} = \widetilde{\cE}_{AD}(\rho_{\rm in}) = \sum_{k}E_{k}\rho_{\rm in} E_{k}^{\dagger}, \]
with the same Kraus operators $\{E_{0}, E_{1}\}$ as in Eq.~\eqref{eq:Kraus_ideal}. The key difference however is in the nature of the noise parameter $p \equiv 1 - |f_{r,s}^{N}(t,\{\Delta_{k} \})|^{2}$: in the case of the {\it disordered} chain, the transition amplitude $f_{r,s}^{N}(t, \{\Delta_{k}\} ) $ between site $s$ and $r$ for a chain of length $N$ allowed to evolve for a time $t$, is a random variable whose value depends on the specific realization of the disorder variables $\{\Delta_{k}\}$. The distribution of $f_{r,s}^{N}(t, \{\Delta_{k}\} )$ for given set of $r,s,N,t$ values depends on the distribution over which the disorder variables $\{\Delta_{k}\}$ are sampled. To illustrate our point, we specifically consider below the case where the coupling strengths $\{\Delta_{k}\}$ are independently sampled from a uniform distribution.

\subsection{Transition amplitude in the presence of disorder}\label{sec:transAmp_disorder}

The Heisenberg Hamiltonian $\mathcal{H}$ with static disorder in the coupling strengths, has the form,
\begin{equation}\label{eq:H_dis}
\mathcal{H}_{\rm dis} = -\sum_{k}\frac{\overline{J}(1+\Delta_{k})}{2}(\sigma^{k}_{x}\sigma^{k+1}_{x}+\sigma^{k}_{y}\sigma^{k+1}_{y}+ \sigma_{z}^{k}\sigma^{k+1}_{z}).
\end{equation}
Here, the effect of disorder is introduced via the i.i.d. random variables $\{\Delta_{i}\}$ which take values over a uniform distribution between $\left [-\delta,\delta \right ]$. The quantity $\delta$ is called the disorder strength, and  $\overline{J}$ is the mean value of the coupling coefficient. We may view the disordered Hamiltonian as a sum of the form $\mathcal{H}_{\rm dis} =\mathcal{H}_{o}+ \mathcal{H}_{\delta}$, where $\mathcal{H}_{o}$ denotes the ideal $XXX$ Hamiltonian studied in the previous section and $\mathcal{H}_{\delta}$ is given by,
\[ \mathcal{H}_{\delta} = -\frac{\overline{J}}{2} \sum_{k}\Delta_{k} \overrightarrow{\sigma^{k}} \cdot \overrightarrow{\sigma^{k+1}}. \]
$\mathcal{H}_{\delta}$ captures the effect of disorder in the spin chain and can be treated as a perturbation of the Hamiltonian $\cH_{0}$. Since $[\cH_{0},\cH_{\delta}] \neq 0$, the transition amplitude maybe evaluated using the so-called time-ordered expansion, also referred to as the Dyson-series~\cite{dyson}. 

Specifically, the transition amplitude between the $r^{\rm th}$ and $s^{\rm th}$ site for the disordered Hamiltonian $\mathcal{H}_{\rm dis}$ in Eq.~\eqref{eq:H_dis} is given by (setting $\hbar = 1$),
\begin{eqnarray}
&& f^{N}_{r,s}(t, \{\Delta_{k}\} \,) \nonumber \\
&=& \langle \textbf{r} | e^{- i (\mathcal{H}_{o}+ \mathcal{H}_{\delta})t} |\textbf{s}\rangle \nonumber \\ 
&=& \langle \textbf{r}| e^{-i \mathcal{H}_{o} t}\mathcal{ T}\left[\exp{\left(-i\int_{0}^{t}e^{i\mathcal{ H}_{o} t'}\,\mathcal{H_{\delta}}\,e^{-i \mathcal{H}_{o} t'}dt'\right)}\right]| \textbf{s}\rangle \nonumber \\  
&=& f^{N}_{r,s}(t) - i\sum_{k=1}^{N} f^{N}_{r,k}(t) \int_{0}^{t}\langle \textbf{k}|e^{i \mathcal{H}_{o} t'} \mathcal{H}_{\delta}e^{-i\mathcal{ H}_{o} t'}|\textbf{s}\rangle dt' \nonumber \\
&& + \;  O(H_{\delta}^{2}), \nonumber
\end{eqnarray}
where $\cT$ is the time-ordering operator which has been expanded to first order in the perturbation in the final equation. As before, $f^{N}_{r,k}(t)$ denotes the transition amplitude between the $r^{\rm th}$ and $k^{\rm th}$ sites in the case of an ideal chain of length $N$, without disorder.  

Thus, using the time-ordered expansion, the transition amplitude in the presence of disorder can be evaluated as a perturbation around the zero-disorder value $f^{N}_{r,s}(t)$, of the form,
\begin{equation}
 f^{N}_{r,s}(t,\{\Delta_{k}\}) = f^{N}_{r,s}(t)+ \sum_{i=1}^{N-1} c^{N}_{i}(t) \Delta_{i} + \sum_{i,j=1}^{N-1}d^{N}_{ij} \Delta_{i}\Delta_{j} + \ldots . \label{eq:transAmp_final}
 \end{equation}
The explicit forms of the complex coefficients $c^{N}_{i}(t)$ are given in Eq.~\eqref{eq:c-coeff} in Appendix~\ref{sec:transAmp_dist}. A similar approach was used in~\cite{chiara} to study deviations from perfect state transfer due to the presence of disorder in an $XX$ chain.

Using the form of the transition amplitude stated in Eq.~\eqref{eq:transAmp_final}, we obtain the distribution of real part of the transition amplitude $x \equiv \texttt{Re}[f^{N}_{r,s}(t,\{\Delta_{k}\})]\,$, up to first order in the perturbation $H_{\delta}$, as,
\begin{equation}
 \cP^{\delta, N,t}(x) \propto \sum_{j=1}^{2^{N-1}} (-1)^{u_{j}}(q_{j})^{N-2}\,{\rm Sign}[q_{j}], 
 \end{equation}
where  $u_{j} \in [0,1]$ and the $Sign$ function is defined as
\begin{equation}
 {\rm Sign }(x-a) = \left\lbrace \begin{array}{cc}
   -1 , & x< a , \\ 
   0 , & x=a , \\ 
   1, & x>a . \end{array} \right. \nonumber
\end{equation}
The functions $q_{j} \left(x,\texttt{Re}[f^{N}_{r,s}(t)], \{\texttt{Re}[c^{N}_{i}(t)]\} \right)$ are linear combinations of the form,
\begin{equation}
 q_{j} \equiv x - \texttt{Re}[f^{N}_{r,s}(t)] + \delta\sum_{i=1}^{N-1} (-1)^{r_{i}^{j}}\texttt{Re}[c^{N}_{i}(t)]  , 
 \end{equation}
where $ r_{i}^{j} \in [0,1] \, \forall \, i=1,\ldots, N-1$ and $\texttt{Re}[c^{N}_{i}(t)]$ denote the real part of the coefficients in Eq.~\ref{eq:transAmp_final}. Since there are $N-1$ such coefficients for a spin chain of length $N$, the sum over $i$ ranges from $1$ to $N-1$. There are $2^{N-1}$ distinct linear combinations of the form $q_{j}$, corresponding to the $2^{N-1}$ distinct $(N-1)$-bit binary strings parameterized by $r^{j}$, so that the sum over $j$ runs from $1$ to $2^{N-1}$. The form of the distribution is identical for the imaginary part $\texttt{Im}[f^{N}_{r,s}(t,\{\Delta_{k}\})]$, with the real parts of $\{c^{N}_{i}(t)\}$ and $f^{N}_{r,s}(t)$ replaced by their imaginary parts. We refer to Eqs.~\eqref{eq:dist_real2},~\eqref{eq:dist_im2} in Appendix~\ref{sec:transAmp_dist} for a detailed description of the distributions of the real and imaginary parts of the disordered transition amplitude.

\begin{figure} 
\centering
\includegraphics[width=1\textwidth]{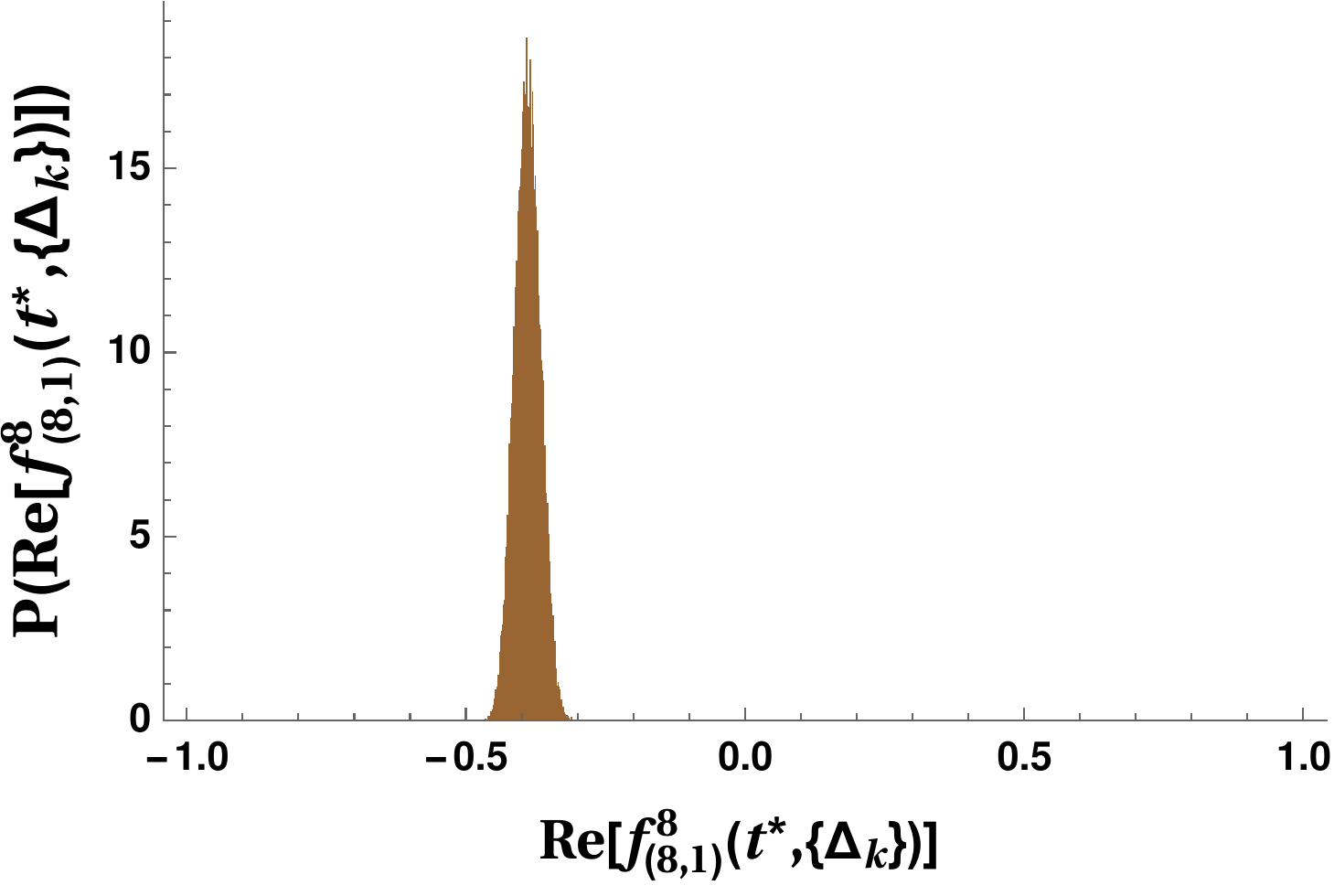}

\caption{Distribution of $\texttt{Re}[f_{8,1}^{8}(t^{*}, \{\Delta_{k}\})]$ 
for different disorder realizations, drawn from a uniform distribution with disorder strength $\delta = 0.001$.}
\label{fig:f_dist1}
\end{figure}

The key salient feature we observe from calculating the distribution functions above is that the limiting distribution in the case of no disorder ($\delta \rightarrow 0$), is indeed a delta distribution peaked around $f^{N}_{r,s}(t)$. Furthermore, in Appendix~\ref{sec:transAmp_dist} we also explicitly evaluate the mean and variance of $f^{N}_{r,s}(t,\{\Delta_{k}\})$ and show that the mean is equal to the zero-disorder value of $f^{N}_{r,s}(t)$, up to $O(\delta^{2})$ (see Eq.~\ref{eq:trans_ampAvg}). The variance goes as $O(\delta^{2})$, as shown in Eq.~\eqref{eq:variance}, making it vanishingly small in the limit of small $\delta$. This observation leads us to propose a modified QEC protocol for state transfer over disordered $XXX$ chains, using an adaptive recovery $\cR_{P}^{\rm avg}$ based on the {\it disorder-averaged} transition amplitude $\langle f_{r,s}^{N}(t,\{\Delta_{k}\})\rangle_{\delta}$.

\begin{figure}  
\centering
\includegraphics[width=1\textwidth]{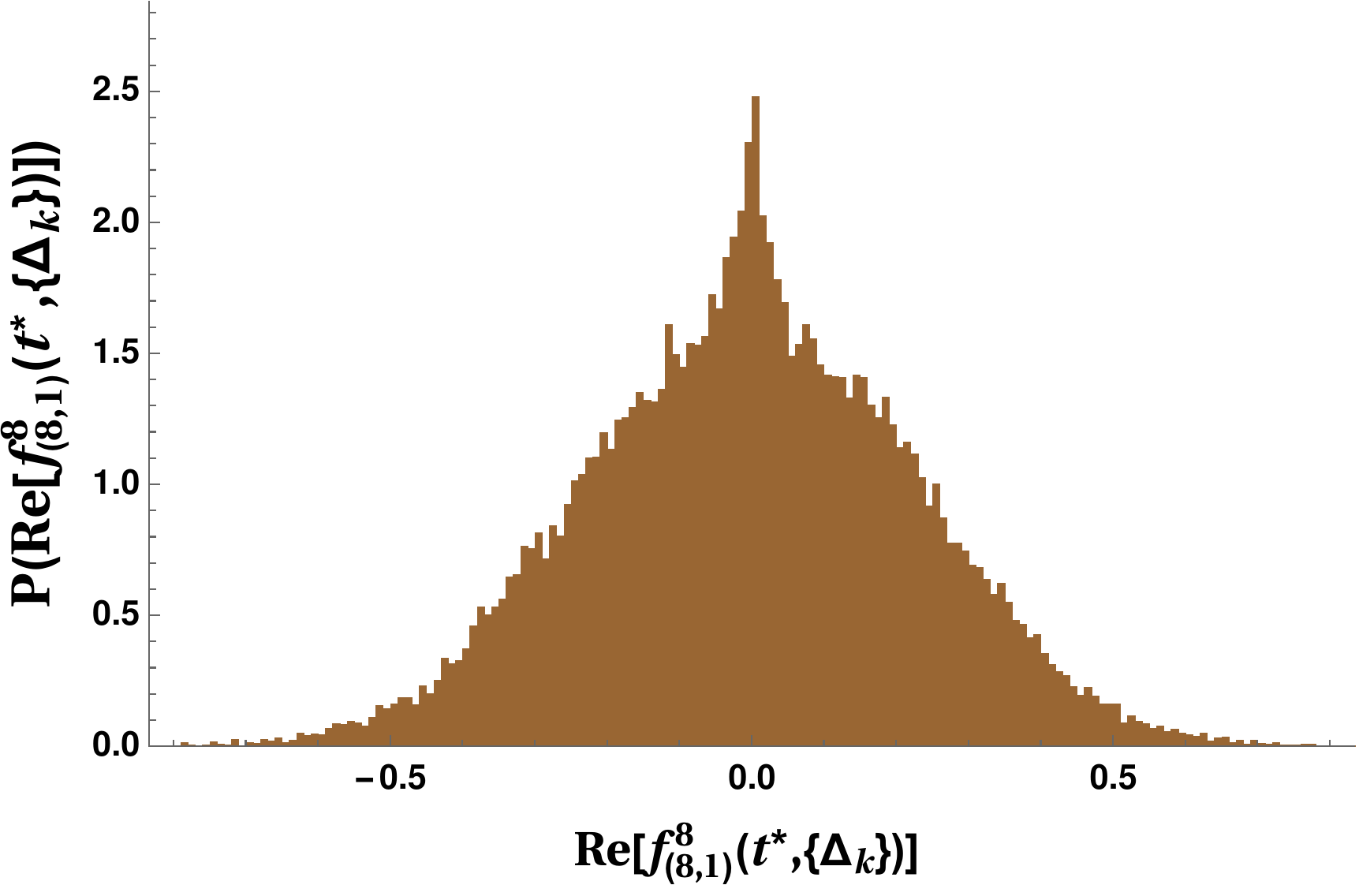}

\caption{Distribution of $\texttt{Re}[f_{8,1}^{8}(t^{*}, \{\Delta_{k}\})]$ 
for different disorder realizations, with disorder strength $\delta = 1$.}
\label{fig:f_dist2}
\end{figure}

The analysis presented thus far holds for any pair of sites $(s,r)$ on a spin chain of length $N$. As an example, we consider the specific case of an $8$-length chain, with $s=1$ and $r=8$. We plot the distribution of the real part of the transition amplitude at some fixed time $t^{*}$, for disorder strengths $\delta = 0.001$ and $\delta=1$, in Figs.~\ref{fig:f_dist1} and~\ref{fig:f_dist2} respectively. We see that when the disorder strength is small enough, the transition amplitude is indeed distributed like a delta function peaked around the zero-disorder value. For large values of $\delta$, the distribution spreads out quite a bit and its mean also shifts closer to zero, giving rise to a very small transition amplitude. The corresponding figures for $\texttt{Im}[f_{8,1}^{8}(t^{*}, \{\Delta_{k}\})]$ are presented in Appendix~\ref{sec:transAmp_dist}. 

\subsection{Adaptive QEC for a $1$-d disordered chain}\label{sec:aqec_diorder}
To summarize, the quantum channel for state transfer in the presence of disorder has the same structure as that of the ideal chain, but with a stochastic noise parameter $p \equiv 1 - |f^{N}_{r,s}(t,\{ \Delta_{k} \})|^{2}$, since the transition amplitude $f_{r,s}^{N}(t,\{\Delta_{k}\})$ is now a random variable whose value depends on the random couplings $\{\Delta_{k}\}$. However, as discussed in Sec.~\ref{sec:transAmp_disorder}, for small enough disorder strengths, $f_{r,s}^{N}(t,\{\Delta_{k}\})$ is peaked sharply around its mean value, and we may consider the disorder-averaged amplitude $\langle f_{r,s}^{N}(t,\{\Delta_{k}\})\rangle_{\delta}$ as a good estimate of the noise. 

We therefore propose an adaptive QEC procedure for a disordered $XXX$ chain involving the $4$-qubit code in Eq.~\eqref{eq:4qubit_4} and a recovery map $\cR_{P}^{\rm avg}$ with the same structure as that used in the case of the ideal chain, described in Eq.~\eqref{eq:4qubit_petz}. However, unlike the ideal case, the value of the channel parameter used in the recovery is different from the one in actual noise channel :  the recovery map uses the disorder-averaged amplitude $\langle f_{r,s}^{N}(t,\{\Delta_{k}\})\rangle_{\delta}$, and is therefore independent of the specific disorder realization, whereas the noise channel has the parameter $ f_{r,s}^{N}(t,\{\Delta_{k}\})$ which changes with every realization.

\begin{figure} [h!]
\centering
\includegraphics[width=1\textwidth]{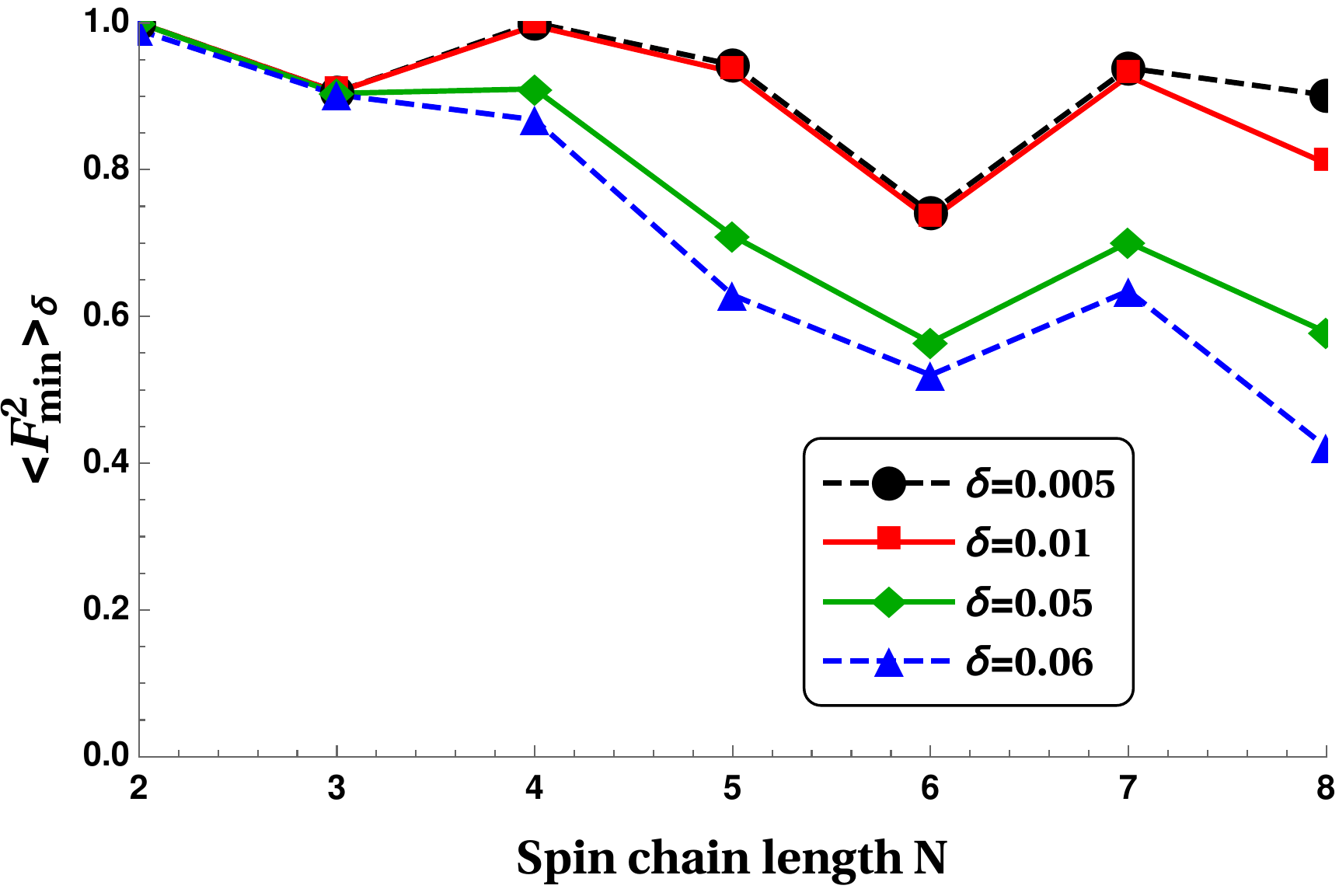}
\caption{Disorder-averaged worst-case fidelity $\langle F^{2}_{\rm min}\rangle_{\delta}$ obtained using the adaptive recovery $\cR_{P}^{\rm avg}$.}
\label{fig:disavg_fmin_8}
\end{figure}
To illustrate the performance of this modified recovery map, we present numerical results for quantum state transfer from the first site ($s=1$) to the $8^{\rm th}$ site ($r=8$)on an $8$-spin chain. Fig.~\ref{fig:disavg_fmin_8} shows the disorder-averaged worst-case fidelity $\langle F^{2}_{\rm min}\rangle_{\delta}$ obtained using the $4$-qubit code and the adaptive recovery $\cR_{P}^{\rm avg}$, for an $8$-spin chain. For disorder strengths $\delta \leq 0.01$,  our adaptive QEC protocol achieves pretty good transfer, with fidelity-loss $\epsilon < 0.2$. Beyond $\delta \geq 0.06$, we notice that $\langle F_{\min}^{2} \rangle_{\delta} < 0.5$ since the effects of localization are too strong to be counteracted by QEC. 

This is further borne out by our detailed analysis of the distribution of the transition amplitude in the presence of disorder (see Appendix~\ref{sec:transAmp_dist}). In particular, our expressions for the mean and standard deviation of the transition amplitude indicate that until $\delta \leq 0.01$, the disorder-averaged value $\langle f_{r,s}^{N}(t,\{\Delta_{k}\})\rangle_{\delta}$ is close to the value of the transition amplitude in the ideal (zero-disorder) case, and the standard deviation is insignificant compared to the mean. However, as the disorder strength increases further, the disorder-averaged value $\langle f_{r,s}^{N}(t,\{\Delta_{k}\})\rangle_{\delta}$ starts dropping and the standard deviation becomes comparable to the average value. Thus the effective noise parameter of the underlying quantum channel becomes too strong for the QEC procedure to be effective. 

The fact that $\delta = 0.06$ is a threshold of sorts can be seen more directly by studying the variation of the disorder-averaged transition amplitude with disorder strength. Previous studies on localization in disordered chains have used such a quantity, namely  $\langle |f_{n,1}^{N}(t,\{\Delta_{k}\})|^{2} \rangle_{\delta}$, as an indicator of the extent of localization~\cite{ashhab,chiara}.

\begin{figure} [H]
\centering
\includegraphics[width=1\textwidth]{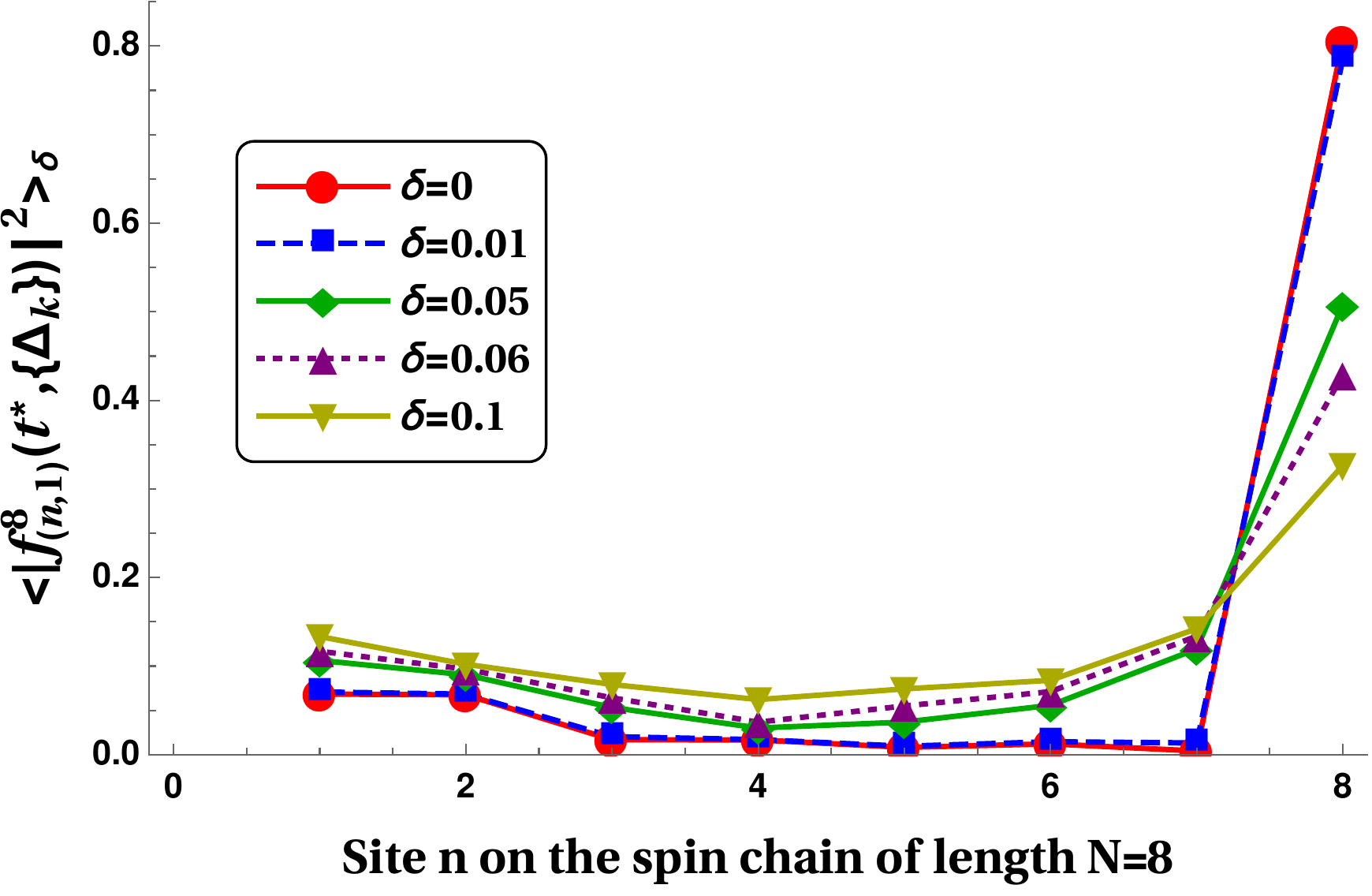}
\caption{$\langle |f_{n,1}^{8}(t)|^{2}\rangle_\delta$ for an $8$-spin chain  as a function of the site $n$.}
\label{fig:transAmp_dis8}
\end{figure}

In Fig.~\ref{fig:transAmp_dis8}, we plot the disorder-averaged transition amplitude $\langle |f_{n,1}^{N}(t,\{\Delta_{k}\})|^{2} \rangle_{\delta}$ for a fixed time $t^{*}$ and different disorder strengths $\delta$, as a function of the receiver site $n$, for the Heisenberg chain in Eq.~\eqref{eq:H_dis}. Empirically, we see that this plot follows an exponential distribution. The curves take the form $e^{-(\alpha n + \beta)/{\rm Loc}}$, where $\alpha, \beta$ are functions of disorder strength $\delta$ and ${\rm Loc}$ is the localization length, i.e. the length at which $\langle |f_{n,1}^{N}(t,\{\Delta_{k}\})|^2\rangle_{\delta}$ falls to $(1/e)$ of its maximum value. This would mean that in Fig.~\ref{fig:transAmp_dis8} for each site $n$, the disorder averaged transition probability $\langle |f_{n,1}^{8}(t,\{\Delta_{k}\})|^2\rangle_{\delta}$ gets suppressed exponentially depending on the strength of the disorder $\delta$. Therefore, with increase in disorder strength $\delta$, the localization effects become more pronounced. 

Specifically, when the disorder strength crosses $\delta=0.06$, the square of the transition amplitude between the ends of the $8$-spin chain falls below $0.43$ on average. However, we know from the fidelity estimate in Theorem~\ref{thm:aqec_fid} that the adaptive QEC protocol improves fidelity if only if $|f_{N,1}^{N}(t)|^{2} > 0.43$. Thus, for end to end state transfer on an $8$-length disordered Heisenberg chain, $\delta=0.06$ is indeed a threshold beyond which the adaptive QEC protocol cannot help in improving fidelity. Since our analysis of the distribution of the transition amplitude presented in Sec.~\ref{sec:transAmp_disorder} as well as the fidelity expression in Theorem~\ref{thm:aqec_fid} hold for any $s,r,N$ we can always identify such a threshold for a specific set of values. 


\section{Conclusions}\label{sec:concl2}

We develop a pretty good state transfer protocol based on adaptive quantum error correction (QEC), for a universal class of Hamiltonians which preserve the total spin excitations on a linear spin chain. Based on the structure of the underlying quantum channel, we choose an approximate code and near-optimal, adaptive recovery map, to solve for the fidelity of state transfer explicitly. For the specific case of the ideal Heisenberg chain, our protocol performs as efficiently as perfect-QEC-based protocols. Using repeated QEC on the chain, we are able to achieve high enough fidelity over longer distances for an ideal spin chain. 

In the case of disordered spin chains the underlying quantum channel is stochastic. For the case of a disordered $1$-d Heisenberg chain, we study the distribution of the transition amplitude, which in turn is directly related to the stochastic noise parameter of the noise channel. By suitably adapting the recovery procedure, we demonstrate pretty good transfer on average, for low disorder strengths. 

It is an interesting question as to whether such channel-adapted QEC techniques maybe used to achieve pretty good state transfer for other universal classes, such as the transverse-field Ising model and the $XYZ$-chain. It is also an open problem to obtain an efficient circuit implementation of the adaptive recovery map discussed here.

%% file: Chapter5_v3.tex
\chapter{ Achieving fault tolerance against amplitude-damping noise} 

\label{Chapter5} 

\lhead{Chapter 5. \emph{Chap:5}} 

\title{Achieving fault tolerance against amplitude damping noise}


\section{Introduction}
A practical quantum computer is prone to noise, both due to the fragile nature of quantum states but also due to imperfect gate operations. This necessitates the theory of quantum fault tolerance, which shows that reliable and scaleable quantum computing is possible even with a noisy quantum memory and noisy quatum gates, provided the noise rates are below a certain threshold~\cite{nielsen, gottesman97}. A fault-tolerant quantum computation proceeds by first encoding the physical qubits into a quantum error correcting (QEC) code. Subsequently, encoded gate operations are performed on the encoded qubits in a way that does not let the errors spread catastrophically. Errors are removed from the system via periodic syndrome checks performed on ancillary qubits, before they accumulate to a point where the damage becomes irreparable. The fundmental basis for any quantum fault tolerance scheme is the threshold theorem, which states that scalable quantum computation is possible provided the physical error rate per gate or per time step is below a certain critical value called the \emph{threshold}~\cite{gottesman_FT}. 

Starting with Shor's original proposal~\cite{shor_ft}, 
the early works on quantum fault tolerance focussed on developing fault-tolerant schemes using general pupose QEC codes, specifically CSS codes such as polynomial codes~\cite{aharonov} and the Steane $[[7,1,3]]$ code~\cite{preskill_FT}. Subsequently, fault tolerance thresholds have been estimated incorporating concatenation~\cite{steane, aliferis}, magic-state distillation~\cite{magic_2005}, as well as teleportation-based schemes~\cite{knill_FT}. Typical threshold estimates of the allowed error rates range between $10^{-6}$ and $10^{-3}$, as seen from Table~\ref{tab:table1}, which contains a summary of various fault tolerance schemes and the threshold values achieved for specific noise models. More recently, surface codes have emerged as a very promising candiate for achieving fault-tolerance~\cite{raus_FT} with threshold estimates ranging between $0.1\%-10\%$, although the resource overhead requirements are often as high as $10^{4}$-$10^{7}$ physical qubits per logical qubit. We refer to~\cite{terhal} for a recent overview of fault-tolerant schemes using surface codes and colour codes in different dimensions.

The fault tolerance threshold obtained for a given QEC code or a class of codes depends crucially on the error model under consideration~\cite{knill_nature}. Most of the threshold estimates obtained so far have assumed either symmetric depolarizing noise or erasure noise. However, when the structure of the noise associated with a given qubit architecture is already known, a fault tolerance scheme tailored to the specific noise process could be advantageous, both in terms of less resource requirements as also in improving the threshold. This is borne out by the fault tolerance prescription developed for \emph{biased} noise, in systems where the dephasing noise is known to be dominant ~\cite{aliferis_biased, FT_assymetry}. This prescription was used to obtain a universal scheme for pulsed operations on flux qubits~\cite{aliferis_biasedExp}, taking advantage of the high degree of dephasing noise in the \textsc{cphase} gate, leading to a numerical threshold estimate of $0.5\%$ for the error rate per gate operation.

Motivated by this example, we address the question of whether it is possible to develop fault tolerance schemes for channel-adapted codes~\cite{hui_prabha}. QEC codes adapted to specific noise models are often known to offer a similar level of protection as the general purpose codes, while using fewer qubits to encode the information~\cite{leung, fletcherthesis}. Furthermore, in today's era of noisy intermediate-scale quantum (NISQ) devices~\cite{preskill_nisq}, there has been growing interest in identifying shorter QEC codes that are amenable to fault-tolerant circuit constructions. For example, it has been demonstrated~\cite{gottesman_2016, flammia} that a family of fault-tolerant encoding and error-detection circuits can be constructed using five qubits for the $[[4,2,2]]$ erasure code, using which a criterion to test experimental fault tolerance may also be derived.

In our work~\cite{ak_ft}, we develop fault-tolerant gadgets for qubits that are susceptible to amplitude-damping noise. Physically, this corresponds to a $T_{1}$ relaxation process and is known to be a dominant source of noise in superconducting qubit systems, for example. Moving away from general purpose codes, we use the well known $4$-qubit code~\cite{leung} to develop gadgets that are fault-tolerant against \emph{local}, stochastic, amplitude-damping noise. We analytically estimate the error threshold for the logical \textsc{cphase}, as well as the memory threshold, using our prescription. Our work thus provides the first rigorous level-$1$ threshold estimates for fault tolerance against amplitude-damping noise. Given the \emph{approximate} nature of the $[[4,1]]$ code used here, it is not immediately clear as to how our our level-$1$ gadgets maybe concatenated to obtain fault tolerance at higher levels. 
\section{Preliminaries}\label{sec:prelim} 

We follow the basic framework of quantum fault tolerance developed by Aliferis \emph{et al.}~\cite{aliferis}, which we briefly review this formalism here for completeness. In a fault-tolerant quantum computation, ideal operations are simulated by performing encoded operations on logical qubits. Encoded operations, in turn, are implemented by composite objects called \emph{gadget}s which are made of elementary physical operations such as single- and two-qubit gates (including identity gates for wait times), state preparation, and measurements. A \emph{location} refers to any one of these elementary physical operations. A location in a gadget is said to be \emph{faulty} whenever it deviates from the ideal operation, and can result in errors in the qubits storing the computational data. The key challenge of is to design the gadgets in such a way as to minimize the propagation of errors due to the faults within the same encoded block. 

In what follows we refer to the elementary physical operations are called \emph{unencoded} gadgets. Our goal is to construct the \emph{encoded} or logical gadgets corresponding to the {$4$-qubit code}, which are resilient to a specific local stochastic noise model, namely, the amplitude-damping channel defined in Eq.~\ref{eq:ampdamp} below. In our scheme, we use the following unencoded operations to build the fault-tolerant encoded gadgets.  
\begin{equation} \label{eq:fault-tol}
 \{ \cP_{|+\rangle}, \cP_{|0\rangle}\} \cup \{ \cM_{X},\cM_{Z},\textsc{cnot},\textsc{cz},X, Z, S, T, H\} .
\end{equation}
Here,  $\cP_{|+\rangle}$ and $\cP_{|0\rangle}$  refer to the  eigenstates of single-qubit $X$ and $Z$ Pauli operators respectively, $\cM_{X}$ and  $\cM_{Z}$ refer to measurements in the $X$ and $Z$ basis respectively. \textsc{cnot} refers to the two-qubit controlled-\textsc{not} gate, \textsc{cz} refers to the two-qubit controlled-\textsc{phase} gate and $X,Z, H S$ and $T$ are the standard single-qubit gates. Note that $|0\rangle$ is the fixed state of the amplitude-damping channel defined in Eq.~\ref{eq:ampdamp} and is therefore inherently noiseless. We assume that rest of the gates and measurements in Eq.~\ref{eq:fault-tol} are susceptible to noise, as described below.

\subsection{Noise model and the $4$-qubit code}\label{sec:noisemodel} 

Our fault-tolerant construction is based on the assumption that the dominant noise process affecting the quantum device is amplitude-damping noise on each physical qubit. This is described by the single-qubit completely-positive and trace-preserving channel, $\cE_\mathrm{AD}(\,\cdot\,)=E_0(\,\cdot\,)E_0^\dagger + E_1(\,\cdot\,)E_1^\dagger$, with $E_0$ and $E_1$, the Kraus operators, defined as
 \begin{align}
E_0&\equiv\frac{1}{2}{\left[(1+\!\sqrt{1-p})I + (1-\!\sqrt{1-p}) Z\right]}\nonumber\\
&=|0\rangle\langle 0|+\sqrt{1-p}|1\rangle\langle 1|, \nonumber \\
\textrm{and}\quad E_1 &\equiv \frac{1}{2}\sqrt{p}(X+ \mi Y)=\sqrt p |0\rangle\langle 1|,  
\label{eq:ampdamp}
 \end{align}
 where $I$ is the qubit identity, $X, Y$, and $Z$ are the usual Pauli operators, and $|0\rangle$ and $|1\rangle$ are the eigenbasis of $Z$.
In our setting, we assume that storage errors, gate errors, as well as the measurement errors are all due $\cE_\mathrm{AD}$. Specifically, a noisy physical gate $\cG$ is modeled by the ideal gate followed by the noise $\cE_{AD}$ on each qubit. In the case of two-qubit gates such as the \textsc{cnot} and \textsc{cz}, we assume that a faulty gate implies an ideal gate followed by amplitude-damping noise acting on \emph{either} or \emph{both} of the control and the data qubits. All three possibilities are assumed to occur with equal probability and we build a two-qubit gate gadget that is tolerant to all three faults. A noisy measurement is modeled as an ideal measurement preceded by the noise $\cE_\mathrm{AD}$, while a noisy preparation is an ideal preparation followed by $\cE_\mathrm{AD}$. {Note that the noise acts on each physical qubit individually, and is assumed to be time- and gate-independent. One could more generally regard the parameter $p$ as an upper bound on the level of amplitude damping over time and gate variations. }
As the basis of our fault tolerance scheme, we make use of the well-known {$4$-qubit code}, originally introduced in \cite{leung} tailored to deal with amplitude-damping noise using four physical qubits.
The $[[4,1]]$ code space $\cC$ is the span of
\begin{align}\label{eq:4qubit}
 |0_{L}\rangle &\equiv \tfrac{1}{\sqrt 2}(|0000\rangle +|1111\rangle)\nonumber\\
\textrm{and}\quad |1_{L}\rangle &\equiv \tfrac{1}{\sqrt 2}(|1100\rangle +|0011\rangle),
 \end{align}
giving an encoded/logical qubit of information.
Recall that the $4$-qubit code given in Eq.~\ref{eq:4qubit} is an \emph{approximate} code, which does not perfectly satisfy the perfect QEC conditions~\cite{knill}, but can nonetheless correct for amplitude-damping noise with high fidelity. It has been noted that the set of $4$-qubit operators generated by the set $\left\langle XXXX,ZZII,IIZZ \right\rangle$ stabilize the code~\cite{fletcher}. We further observe that the {$4$-qubit code} can also be constructed by considering a pair of vectors in the $[[4,2,2]]$ code whose stabilizer generator set is given by $\left\langle XXXX,ZZZZ \right\rangle$~\cite{gottesman_2016}. The logical $\overline{X}$ and $\overline{Z}$ operators for the {$4$-qubit code} are identified as 
\begin{equation}\label{eq:logical} 
\overline{X} \equiv XXII; \overline{Z}\equiv ZIIZ,
\end{equation}
{up to multiplication by the stabilizer operators, of course.}

The action of the $4$-qubit noise channel $\cE_\mathrm{AD}^{\otimes4}$ can be discretized in terms of the single-qubit damping error $E_1$ $\propto$ $(X+ \mi Y)$ in Eq.~\ref{eq:ampdamp} and the single-qubit dephasing error $Z$ contained in $E_0$; there is also the no-error case of $I$ in $E_0$. We list here the complete set of errors that the $4$-qubit codespace $\cC$ is subject to, under the action of $\cE_\mathrm{AD}^{\otimes4}$. We first note the action of the single-qubit damping errors on the codespace $\cC$, as,
\begin{eqnarray}\label{eq:damping1}
(X + \mi Y)\otimes I^{\otimes 3}: \cC \rightarrow \cS\left(\{|0111\rangle,|0100\rangle \}\right) \nonumber \\ 
 I\otimes (X+ \mi Y)\otimes I^{\otimes 2} : \cC \rightarrow \cS\left(\{|1011\rangle,|1000\rangle \}\right) \nonumber \\ 
I^{\otimes 2}\otimes (X+\mi Y)\otimes I: \cC \rightarrow \cS\left(\{|1101\rangle,|0001\rangle \}\right)  \nonumber \\ 
I^{\otimes 3}  \otimes (X+\mi Y) : \cC \rightarrow \cS\left(\{|1110\rangle,|0010\rangle \}\right), 
\end{eqnarray} 
{where $\cS(.)$ denotes the span of the corresponding states. The two-qubit damping errors act as follows.}
\begin{eqnarray}\label{eq:damping2}
(X + \mi Y)^{\otimes 2} I^{\otimes 2}: \cC \rightarrow \cS\left(\{|0011\rangle,|0000\rangle \}\right) \nonumber \\ 
I^{\otimes 2} (X + \mi Y)^{\otimes 2}: \cC \rightarrow \cS\left(\{|1100\rangle,|1111\rangle \}\right) \nonumber \\ \nonumber
(X+\mi Y) \otimes I \otimes (X+\mi Y) \otimes I: \cC \rightarrow \cS\left(\{|0101\rangle\}\right) \nonumber \\ 
I \otimes (X+\mi Y)^{\otimes 2}\otimes I: \cC \rightarrow \cS\left(\{|1001\rangle \}\right) \nonumber \\ 
I \otimes (X+\mi Y) \otimes I\otimes (X+\mi Y): \cC \rightarrow  \cS \{|1010\rangle \} \nonumber \\ 
( X+\mi Y) \otimes I ^{\otimes 2} \otimes (X+\mi Y): \cC \rightarrow \cS\left(\{|0110\rangle\}\right).
\end{eqnarray}
The three-qubit damping errors $(X+\mi Y)^{\otimes 3} \otimes I$ map the codespace to subspaces that overlap with the single-damping error subspaces. The $4$-qubit damping error  $(X+\mi Y)^{\otimes 4}$ maps the codespace to the state $|0000\rangle$. The error set also comprises single-qubit dephasing errors of the form, 
\begin{equation}\label{eq:dephasing}
 Z\otimes I ^{\otimes 3}: \cC \rightarrow \cS\left(\{{|0000\rangle-|1111\rangle}/{\sqrt{2}},{|0011\rangle -|1100\rangle}/{\sqrt{2}}\} \right).
\end{equation}
For completeness sake, we note the no-error operator $I^{\otimes 4} : \cC \rightarrow \cC$ that leaves the codepsace invariant. It is easy to check that the $4$-qubit code \emph{perfectly} corrects for single-qubit damping errors and a partial set of two-qubit {damping }errors~\cite{leung}.  The set of syndrome operators that allow us to correct for the no-error case and the single-qubit damping errors include $(ZZII, IIZZ, ZIII, IIZI)$. 

We note that the operator $E_0$ in the expansion of amplitude-damping channel in Eq.~\ref{eq:ampdamp} is more than Identity with an additional back-action term of $O(p)$, which is a non-CP map of the form given below
\begin{equation}
\cF_{Z}(\cdot) = Z (\cdot) I+ I (\cdot) Z
\end{equation} 
Such effects of back-action have been studied in the past using bosing codes both theoritically and experimentally~\cite{liang, backaction}. In this chapter, we construct gadgets made of noisy elementary physical gates, tolerant against faults due to a single damping error. However, Ref.~\cite{ak_ft} provides a more extensive treatment by constructing gadgets fault-tolerant against both back-action errors and single-qubit damping errors arising due to amplitude damping channel.

\subsection{Principles of fault tolerance}\label{sec:faulttolerance}
We conclude this section by formalizing the notion of fault-tolerant {circuits and gadgets}, in the context of amplitude-damping noise. Specifically, we list the properties that the {error correction unit and the encoded gadgets} must satisfy, in order to lead to logical operations and circuits that are fault-tolerant against amplitude-damping noise. 
\begin{itemize}
\item[(Prop 1)] If an error correction unit has no fault, it takes an input with at most one damping error to an output with no errors.
\item[(Prop 2)] If an error correction unit contains at most one fault, it takes an input with no errors to an output with at most one damping error.
\item[(Prop 3)]{ A preparation unit without any fault propagates an input with upto one single-qubit damping error to an output with at most one single-qubit damping error. A preparation unit with at most one fault propagates an incoming state with no errors to an output with at most one single-qubit damping error. }

\item[(Prop 4)]{A measurement unit with no faults leads to a correctable classical outcome for an input with at most one damping error. A measurement unit with at most one fault anywhere leads to a correctable classical outcome for an input state with no errors. }

\item[(Prop 5)] An encoded gadget without any fault takes an input with upto a single-qubit damping error to an output state in each output block with at most one single-qubit damping error. An encoded gadget with a single damping fault takes an input with no error and leads to an output with a single-qubit damping error.
\end{itemize}

{
In what follows, we first develop fault-tolerant gadgets, which are resilient to single-qubit damping errors that occur with probability $O(p)$, neglecting the higher order dephasing and multi-qubit damping errors. Using the $4$-qubit code, we build the error correction unit, preparation and measurement units as well as a universal set of enoded gate gadgets. We prove that our gadgets satisfy the principles of fault tolerance listed above, and finally establish a pseudothreshold for the memory unit and controlled-\textsc{phase} gadget.}


\section{Basic units and encoded gadgets}\label{sec:enc_gadget}

In this section, we introduce the basic units which constitute the building blocks of our fault tolerance scheme. We begin by constructing the error correction unit, comprising the syndrome measurement and recovery circuits, which {corrects single-qubit damping errors} in a fault-tolerant manner. We then obtain a fault-tolerant preparation {unit} for a {two-qubit Bell state}, which eventually helps us to prepare the encoded states $|0\rangle_{L}$ and $|+\rangle_{L}$.  We also construct fault-tolerant measurement units corresponding to the logical $X$ and {$Z$ measurements}. We finally construct a set of basic encoded gates, namely, an encoded $X$ gadget and an encoded \textsc{cz} gadget. In every case, the physical gates come from the elementary set given in Eq.~\eqref{eq:fault-tol}.


\subsection{Error correction unit}\label{sec:ec_unit}

The error correction {unit}, henceforth referred to as the \textsc{ec} unit, is shown in Fig.~\ref{fig:ec}. It comprises two separate sub-units, one a fault-tolerant syndrome-extraction unit for extracting the syndrome bits corresponding to the set of correctable errors, the other a recovery unit. {We note here that we need a nontrivial recovery unit to correct for the single damping errors, rather than a classical Pauli-frame change. This is indeed due to the fact that the single damping errors of the form $X+\mi Y$ do not conjugate past all the elementary gate operations in Eq.~\eqref{eq:fault-tol}.} 
\begin{figure}[H]
\hspace{2cm}\includegraphics[width=.7\columnwidth]{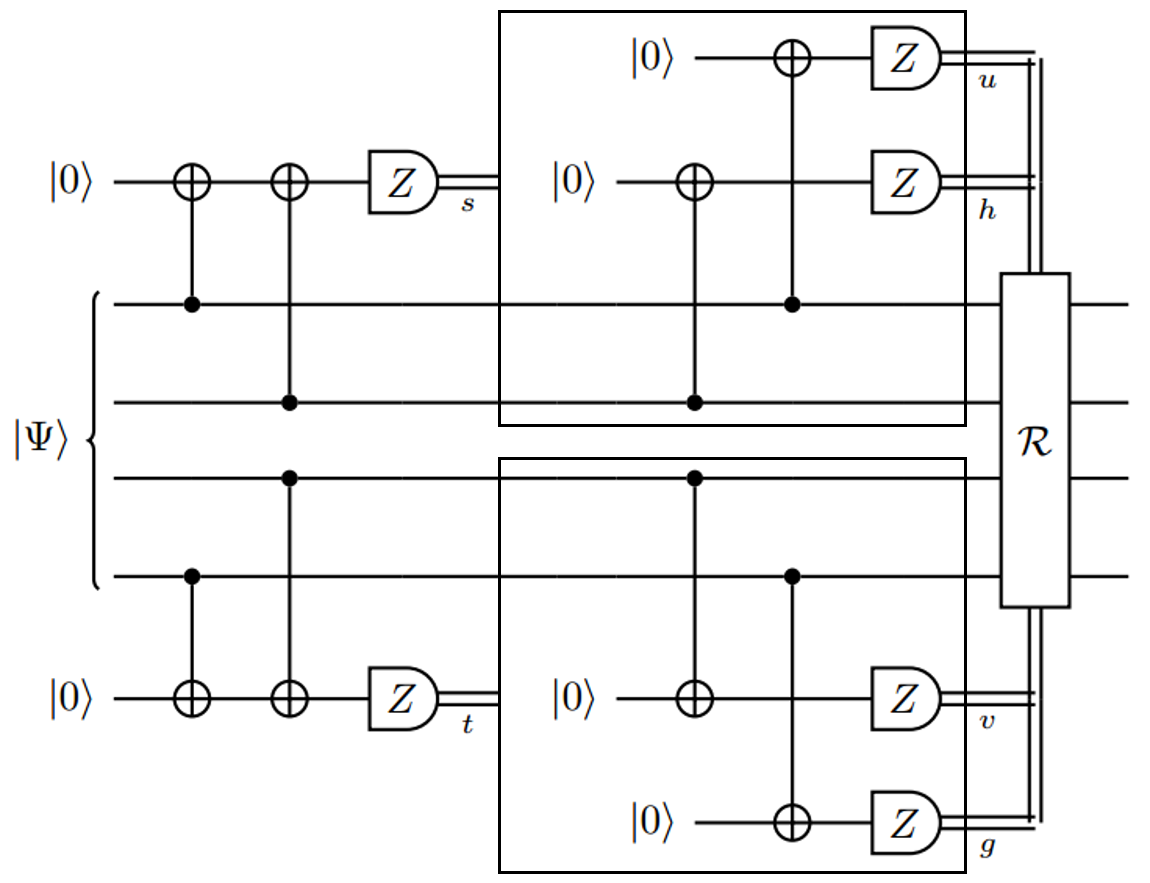}
\caption{Fault-tolerant error correction unit. $\ket{\Psi}$ denotes the four input data qubits, others are ancilla qubits. The circuits in the box are performed only if the corresponding syndrome bits are triggered (set to $1$). 
The recovery unit $\mathcal{R}$ is always performed if either of the two bits $s, t$ is triggered (see section \ref{sec:ec_unit}).}
\label{fig:ec}
 \end{figure}
The syndrome extraction procedure is a two-step process, {as shown in Fig.~\ref{fig:ec}. $|\Psi\rangle$ denotes the state of the encoded or \emph{data} qubits. The syndrome extraction} first involves parity measurements, using two ancilla qubits initialized to $|0\rangle$, on the two pairs of data qubits  separately. This results in two classical bits, denoted as $s$ and $t$ in Fig.~\ref{fig:ec}. The parity measurements correspond to measuring the stabilizers $ZZII$ and $IIZZ$ on the data qubits, with even(odd) parity registering as the outcome $0(1)$. 

{If both syndrome bits $s$ and $t$ take the value $0$, we terminate the error correction circuit. If either of $s$ or $t$ takes the value $1$, we perform another pair of non-stabilizer-group syndrome measurements, corresponding to measuring the operators $\{ZIII, IZII\}$ or $\{IIZI, IIIZ\}$. Specifically, if $s=1$ and $t=0$, we perform two measurements corresponding to $\{ZIII, IZII\}$, with outcomes denoted as $u$ and $h$. If $s=0$ and $t=1$, we perform two measurements corresponding to $\{IIZI, IIIZ\}$, with outcomes denoted as $v$ and $g$. Based on these four syndrome bits $s,t,u,h$ (or $s,t,v,g$), which are collectively denoted as $\textbf{P}$, we are able to diagnose all single fault events. Table \ref{tab:syndrome_bits} summarizes the mapping between a single fault/error and the syndrome bits. Note that there are higher-order fault events that can be detected -- for example, $s=1$, $t=1$ corresponds to damping errors in two data qubits -- however, these errors are uncorrectable by the $4$-qubit code. We decide to end the computation in such cases.}

\begin{table}[h] 
\begin{center}
\begin{tabular}{c | c | c | c | c | c || l}
	$s$ \ & \ $t$ \ & \ $u$ \ & \ $h$ \ & \ $v$ \ & \ $g$ \ & \textbf{Diagnosis}\\ 
			\hline\hline
			0&0&$\times$&$\times$&$\times$&$\times$& no error or undetected fault\\
			\hline
			1&0&0&1&$\times$&$\times$& Qubit 1 is damped \\
			\hline
			1&0&1&0&$\times$&$\times$& Qubit 2 is damped \\ 
			\hline
			1&0&1&1&$\times$&$\times$& Fault in one CNOT \\ 
			\hline
			0&1&$\times$&$\times$&0&1& Qubit 3 is damped \\ 
			\hline
			0&1&$\times$&$\times$&1&0& Qubit 4 is damped \\ 
			\hline
			0&1&$\times$&$\times$&1&1& Fault in one CNOT \\ 
		\end{tabular}
	\end{center}
	\caption{\label{tab:syndrome_bits} Mapping between syndrome bits and single fault/error events. Other combinations that do not appear in the table correspond to higher-order fault events.}
\end{table}
The recovery unit, denoted $\cR$ in Fig.~\ref{fig:ec}, is shown in detail in Fig.~\ref{fig:r_ec}. We note that the damping errors we wish to correct for, are of the form $X+\mi Y = (I+Z)X$. 
The recovery unit $\cR$ consists of two sub-units, namely, the $\cR_X$ unit for correcting the $X$ part, and the $\cR_Z$ unit for correcting the $(I+Z)$ part. The unit $\cR_X$ applies a local  $X$ gate to the damped qubits, as diagnosed by the syndrome bits. It may not apply any $X$ if the syndrome bits diagnose no damped qubits, for example, cases $4$ and $7$ in Table \ref{tab:syndrome_bits} . 

$\cR_Z$ corrects for the $(I+Z)$ part by measuring the stabilizer $XXXX$ using an ancilla qubit initialized to $\ket{+}$. The measurement, with outcome denoted as $c$, projects the encoded state into the code space, corresponding to the subspace with eigenvalue $+1$ ($c=0$), or to the subspace with eigenvalue $-1$ ($c=1$). The latter case corresponds to a $Z$ error and $\cR_Z$ applies a suitable local $Z$ gate from the set $\{ZIII,IZII,IIZI,IIIZ\}$, to correct for it. For example, if the first data qubit is damped and $c=1$, $\cR_Z$ can apply $Z$ to the first data qubit or to the second data qubit (since $ZZII$ is a stabilizer).

We note that the $\cR_Z$ unit will be reused in other gadgets with a slight variations on how the local $Z$ gates are applied at the end. Therefore, we represent this set of local, single-qubit $Z$ operations by a generic function $Z(\mathbf{P}, c)$ (see Fig.~\ref{fig:r_ec}) that depends on all the syndrome bits. We will explain the details further in each specific case.

\begin{figure}[H]
\begin{center}
		\includegraphics[width=0.5\columnwidth]{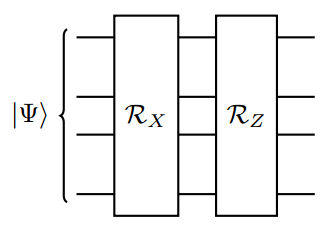}
		\caption{ Recovery unit $\mathcal{R}$ consists of two parts, $\cR_X$ and $\cR_Z$.}
\end{center}
\end{figure}
	\begin{figure}[H]
\begin{center}
	\includegraphics[width=.7\columnwidth]{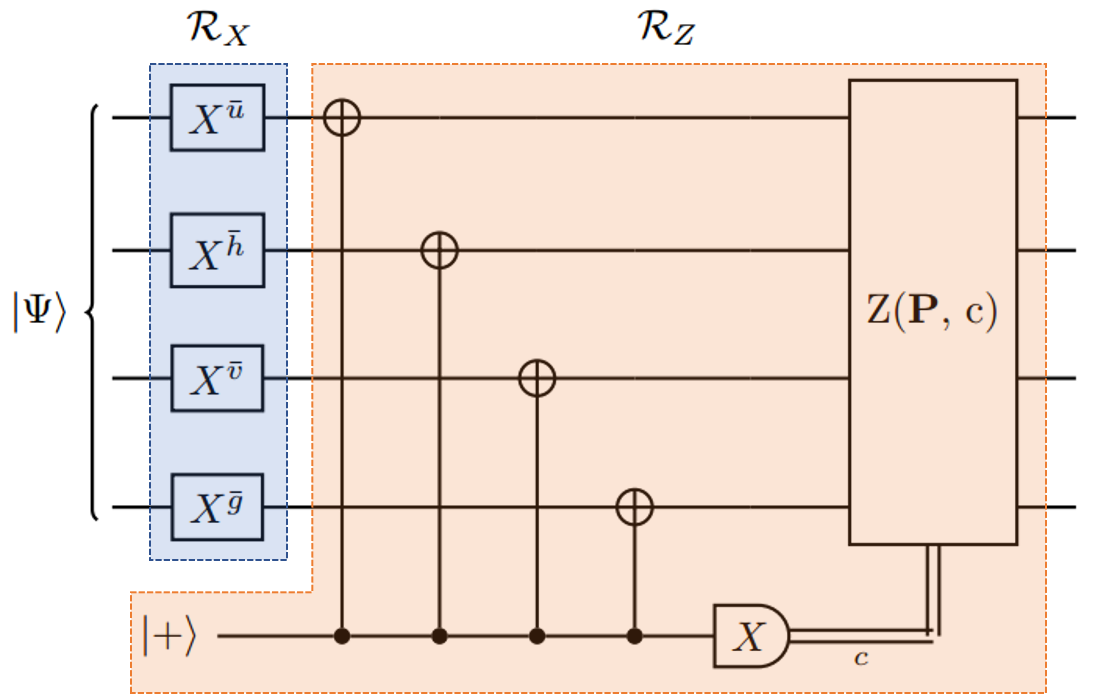}
		\caption{ ( Details of each sub-units. The bar notation denotes the complement of a syndrome bit. $\mathbf{P}$ denotes the set of syndrome bits $\{s,t,u,h,v,g\}$. $\cR_X$ applies an identity if the syndrome bit is unavailable. $\mathcal{R}_Z$ measures the stabilizer $XXXX$ and applies a corresponding $Z$ if $c=1$.}
\label{fig:r_ec}
\end{center}
\end{figure}
{The fault-tolerant property of the \textsc{ec} unit in Fig.~\ref{fig:ec} can be intuitively understood as follows. If there is one error in the input state and no fault in the \textsc{ec} unit, the syndrome extraction unit will detect the error and the recovery unit will correct for it. On the other hand, if we allow for one fault location in the \textsc{ec} unit, it can only occur in the first two parity measurements. If the fault is detected, it will be corrected by the recovery unit (fault-free). Otherwise, if the fault is undetected, the recovery unit will not be triggered and the fault will cause at most one error to the output, due to the way we construct the parity measurements.} A detailed proof that \textsc{ec} unit is fault tolerant is presented in Appendix \ref{app:FT}. 


\subsection{Bell-state preparation}\label{sec:prep}

We next demonstrate a fault-tolerant preparation of the two-qubit Bell state $|\beta_{00}\rangle$,
\begin{equation}\label{eq:bell}
 |\beta_{00}\rangle = \frac{1}{\sqrt 2}(|00\rangle+|11\rangle),
\end{equation}
serves as the input state to multiple fault-tolerant gadgets constructed in Sec.~\ref{sec:logical_prep}.
 
\begin{figure}[h]
		\centering
		\includegraphics[width=0.5\textwidth]{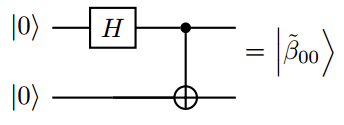}
		\caption{Preparation of the Bell state: (top) non-fault-tolerant preparation of the two-qubit Bell state $|\beta_{00}\rangle$}
		\label{fig:b_00}
	\end{figure}
\begin{figure}[h]
		\centering
		\includegraphics[width=0.7\textwidth]{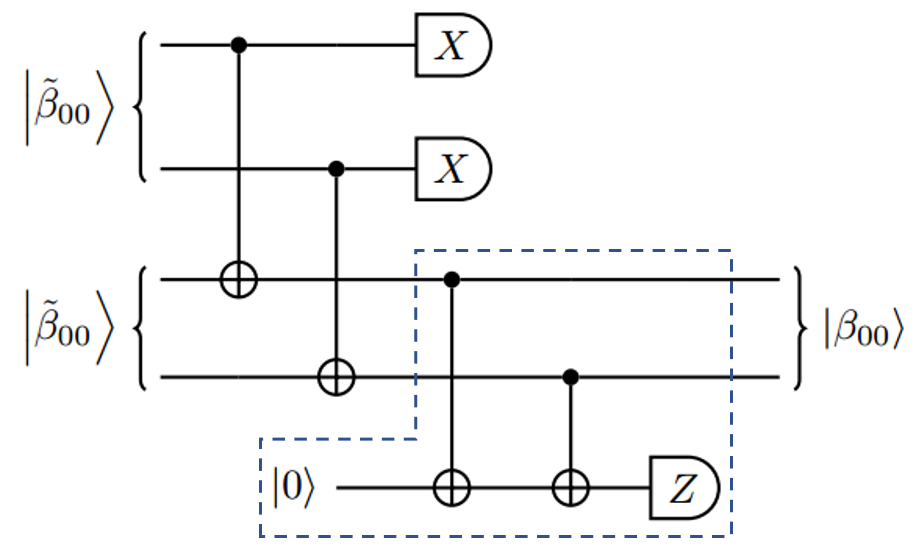}
		\caption{Preparation of the Bell state:a fault-tolerant version built from two non-fault-tolerant copies of $|\beta_{00}\rangle$ from the top circuit.}
		\label{fig:plus}
	\end{figure} 
The state $|\beta_{00}\rangle$ is first prepared in a non-fault-tolerant manner, as shown in Fig.~\ref{fig:b_00}. We then verify this Bell state using another copy of $|\beta_{00}\rangle$ prepared in a similar fashion, as shown in Fig.~\ref{fig:plus}. {When the circuits have no faults, both the $X$ as well as $Z$ measurements outcomes have even parity. However, when there is a single damping error any where in the circuit, one or both outcomes have odd parity, as explained in Appendix \ref{app:BellFT}. In the latter case, we reject the final state and start again, thus ensuring a fault-tolerant preparation of the Bell state $|\beta_{00}\rangle$.}


\subsection{Fault-tolerant $\overline Z$ and $\overline X$ measurements}\label{sec:X_meas}

\begin{figure}[H]
\begin{center}
\includegraphics[scale=0.6]{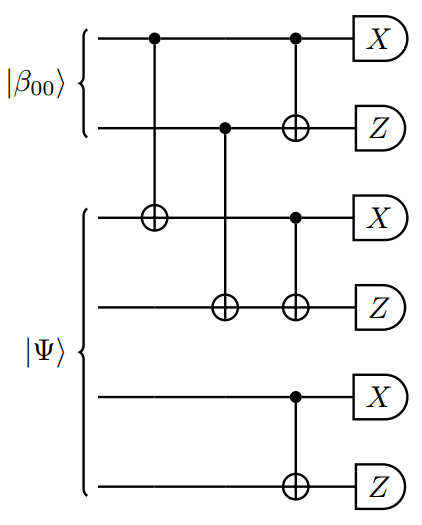}
\caption{\label{fig:xL}Measurement of the logical operator $\overline{X}$=$XXII$ on the encoded information (carried by the bottom $4$ qubits) with the use of two additional ancilla qubits (top $2$ qubits) fault-tolerantly prepared in $|\beta_{00}\rangle$.}
\end{center}
\end{figure}

{In this section, we demonstrate fault-tolerant circuits that perform logical $\overline X$ and $\overline Z$ measurements corresponding to the $4$-qubit code. Recall that a measurement unit is said to be fault-tolerant, if the presence of a single fault in the measurement circuit always leads to a correctable error in the classical outcome. In other words, distinct faults lead to distinct classical outcomes, so as to ensure that the faults can be diagnosed and corrected for, classically.}

{We first argue that the logical $Z$ -- $\overline{Z} = ZIIZ$ -- can be realised simply by performing four local $Z$ measurements on the encoded or data qubits. In the ideal, no-error scenario, the measurement outcomes (which are four-bit binary strings) have even parity. Specifically, the strings $0000$ and $1111$ project the data qubits onto the $|0\rangle_{L}$ state, whereas the strings $0011$ and $1100$ project the data qubits onto $|1\rangle_{L}$. A single fault in one of the local $Z$ measurements, or a single-damping error in the data qubits leads to outcomes strings with odd parity. Furthermore, faults in distinct locations lead to distinct four-bit classical strings, thus ensuring that the faults can diagnosed and corrected for. In particular, if the outcome string has more $1$'s, the \emph{correct} outcome corresponds to a projection onto the $|0\rangle_{L}$ state, whereas if the outcome string has more $0$'s the correct outcome corresponds to a projection onto the $|1\rangle_{L}$ state.}

Next, we give a fault-tolerant construction of the logical $X$ --- $\overline{X}=XXII$ --- measurement unit, by introducing two additional ancilla qubits initialized to the Bell pair $|\beta_{00}\rangle$, as shown in Fig.~\ref{fig:xL}.{Note that in Fig.~\ref{fig:xL} the first two qubits are the ancilla qubits and last four are the data qubits. We then perform three Bell measurements, one on the two ancilla qubits, and one each on each pair of data qubits. The Bell measurements on an \emph{ideal} encoded state (an encoded state with no error) lead to outcomes $(0,0)$ or $(1,0)$. Outcome $(0,0)$ indicates a projection onto $|+\rangle_{L}$, and the outcome $(1,0)$ indicates a projection onto $|-\rangle_{L}$, thus realising the logical $X$ measurement. 

We show that this measurement unit is fault-tolerant, that is, it leads to a correctable classical outcome even when there is a single fault anywhere in the circuit. We refer to Appendix \ref{app:X_meas} for the detailed argument. We merely note here that a single fault in the circuit would lead to a \emph{faulty} outcome ($(1,1)$ or $(0,1)$) for one of the Bell measurements. We simply ignore the faulty outcome and majority voting among the other outcome pairs ($(0,0)$ and $(1,0)$) decide if the projection is onto $|+\rangle_{L}$ or $|-\rangle_{L}$.}

\subsection{Logical $X$ operation}\label{sec:logicalX}

To obtain a universal set of logical gates, as explained in the Sec.~\ref{sec:universal}, we also need a fault-tolerant implementation of the logical $X$ operation. In standard fault tolerance schemes making use of Pauli-based codes, an operator like $\overline X=XXII$ (or alternatively, $IIXX$, with the two differing by a stabilizer operator) can be applied simply by performing $X$ on the first two (or last two) physical qubits, the fault tolerance guaranteed by the transversal nature of the operation. {In the case of the amplitude-damping code, however, the transversal operation is actually not fault-tolerant, due to the fact that the damping errors of the form $X+\mi Y$ do not conjugate past the $X$ operator.} 

Instead, a more elaborate scheme to ensure a fault-tolerant implementation of $\overline X$ is needed. Figure~\ref{fig:xgadget} gives a fault-tolerant gadget that accomplishes the $\overline X$ operation. {As before, $\ket{\Psi}$ denotes the state of the encoded or data qubits. The logical $X$ gadget requires six additional ancilla qubits, two of which are used in the implementation of the logical $X$ operation. The remaining four ancillas are used for parity checks, similar to the ones used in our \textsc{ec} unit in Fig.~\ref{fig:ec}. Depending on the outcomes of these parity checks, the logical $X$ gadget is made fault-tolerant by applying a recovery unit $\cR_{Z}$ (similar to the one already shown in Fig.~\ref{fig:r_ec}), or a modified error correction unit \textsc{ec}', shown in Fig.~\ref{fig:ecprime}.}

\begin{figure}[H]
	\centering
	\includegraphics[scale=0.45]{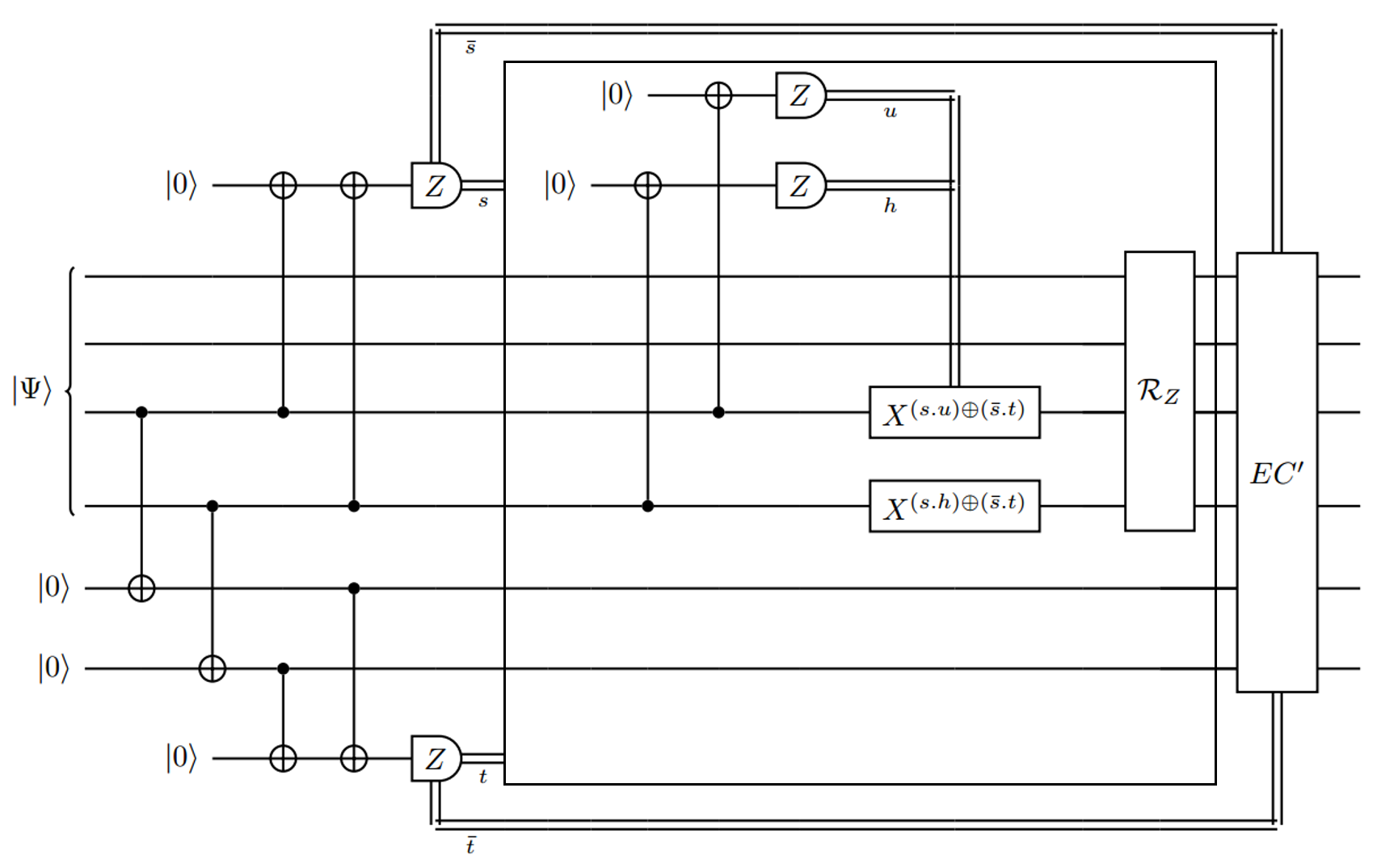}
	\caption{Fault-tolerant implementation of the logical $X$ operation}
	\label{fig:xgadget}
\end{figure} 

{Our logical $X$ operation proceeds as follows. If the parity check bits $s$ and $t$ are both trivial, we proceed to the \textsc{ec}' unit in Fig.~\ref{fig:ecprime}. This unit first applies  a pair of local $X$ gates on the data qubits, thus realising the logical operation $\overline{X} = IIXX$. In order to take into account faults that might have occurred in the logical $X$ implementation, the \textsc{ec}' unit includes additional parity checks on the last two data qubits, leading to an application of the \textsc{ec} unit described in Fig.~\ref{fig:ec}, conditioned on the syndrome bits $x,y$. }

\begin{figure}[h]
\centering
\includegraphics[scale=0.6]{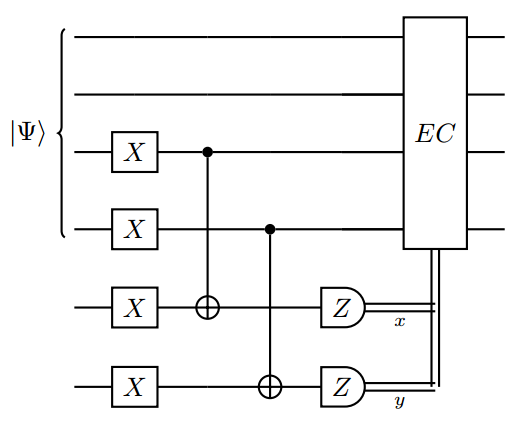}
\caption{The modified \textsc{ec} unit $EC'$ used in the logical $X$ circuit}
\label{fig:ecprime}
\end{figure}

{If either of the bits $s$ or $t$ takes the value $1$, we extract another pair of syndrome bits $u,h$ and proceed to apply one or both local $X$ gates on the encoded state $|\Psi\rangle$, as indicated in Fig.~\ref{fig:xgadget}. The presence of any faults in the gadget is flagged by the syndrome bits $s,t,u,h$, and the recovery unit $\cR_{Z}$ then fixes the state upto a single damping error, as explained in Sec.~\ref{sec:ec_unit} earlier.}


\subsection{Logical \textsc{cphase} operation}\label{sec:cphase}

We next demonstrate a fault-tolerant two-qubit logical operation, which is an essential ingredient for realising a universal set of logical gates. We first note that the logical \textsc{cnot} and the \textsc{cphase} gadgets for the $4$-qubit code both admit transversal constructions 
However, as noted earlier, transversality does not automatically translate into fault tolerance in the case of amplitude-damping errors and the $4$-qubit code. In fact, the transversal \textsc{cnot} is not fault-tolerant to amplitude-damping noise: a single error caused by the amplitude-damping noise can propagate through the transversal circuit into an error that is not correctable by the $4$-qubit code. 
For example, observe that, for two physical qubits connected by a physical \textsc{cnot} operation, a damping error $E_1$ [see Eq.~\eqref{eq:ampdamp}] on the control qubit propagates into  an $X$ error on the target, while the $E_1$ error on the control remains unchanged (see Fig.~\ref{fig:cnot_analysis}). The $4$-qubit code cannot deal with the propagated $X$ error. If such a situation occurs in one of the pairs of physical qubits in the transversal logical \textsc{cnot} gadget, the output of the gadget will not be correctable. A single fault on one of the qubits can thus result in an uncorrectable error, violating the requirements of fault tolerance, despite the transversal structure. 
\begin{figure}[H]
\centering
\hspace{2cm}\includegraphics[scale=.65]{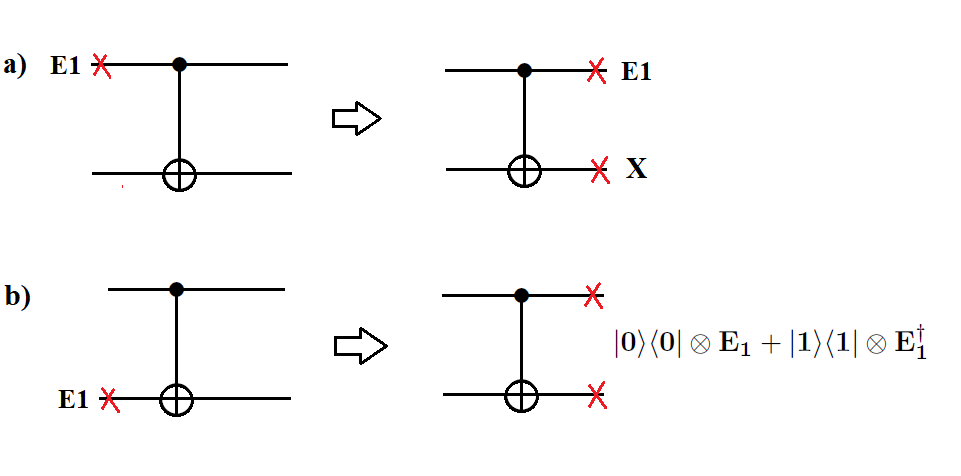}
\caption{a) Propagation of a damping error $E_1$ before the control (top) and b) target (bottom) of a \textsc{cnot} gate.}
\label{fig:cnot_analysis}
\end{figure}

Unlike the \textsc{cnot}, it turns out that the transversal \textsc{cphase} gadget shown in Fig.~\ref{fig:cphase} is fault tolerant against amplitude-damping noise. This is explained in detail in Appendix \ref{app:cz}. {The basic idea, however, is easy to understand by contrasting with the \textsc{cnot} gadget: a damping error at the control (target), after propagating through a physical \textsc{cphase} gate, propagates as a damping error at the control (target), as seen in Fig.~\ref{fig:cphase_analysis}. However, the damping error at the control (target) of the \textsc{cphase}, does lead to an additional phase ($Z$) error in the target (control). This explains the need for the trailing \textsc{ec} and recovery units in the \textsc{cphase} gadget.}

Note that the trailing \textsc{ec} unit in each block in the \textsc{cphase} gadget in Fig.~\ref{fig:cphase} conditions the phase recovery unit $\cR_{Z}$ in the other block. These recovery units are identical to the $\cR_{Z}$ unit in Fig.~\ref{fig:r_ec}, and are introduced to correct for the phase errors which may occur in the event of a single fault before the \textsc{cphase} gadget either in the control or target as shown in Fig.~\ref{fig:cphase_analysis}. The $\cR_{Z}$ unit in each block is  triggered when a fault occurs anywhere before the target (control) and applies a local $Z$ recovery operator depending on the outcomes of the syndrome extraction units of the trailing \textsc{ec} unit in the control (target) block. 


\section{Universal logical gate set}\label{sec:universal}

In the previous section, we described fault-tolerant constructions of a set of basic encoded gadgets: an error correction unit that implements correction with the $4$-qubit code, state preparation of the two-qubit Bell state $|\beta_{00}\rangle$, 
the $\overline X$ and $\overline Z$ measurements, the logical $X$ operation, as well as the logical \textsc{cphase} operation. We show in this section how to construct a fault-tolerant universal (encoded) gate set, tailored for amplitude-damping noise. In particular, we demonstrate the fault-tolerant implementation of a standard universal gate set comprising the logical \textsc{cphase}, Hadamard $H$, $S$ (or phase), and $T$ (or $\pi/8$) gates, where $H$ is the operation $H\equiv |+\rangle_L\langle +|-|-\rangle_L\langle -|$, $S\equiv |0\rangle_L\langle 0|+\mi |1\rangle_\langle 1|$, and $T\equiv |0\rangle_L\langle 0| + \mathrm{e}^{\mi\pi/4}|1\rangle_L\langle 1|$. The fault-tolerant logical \textsc{cphase} gate is already described in Sec.~\ref{sec:cphase}; here we complete the discussion with the logical Hadamard, $S$ and $T$ gates.

\begin{figure}[H]
\centering
\includegraphics[width=0.65\columnwidth]{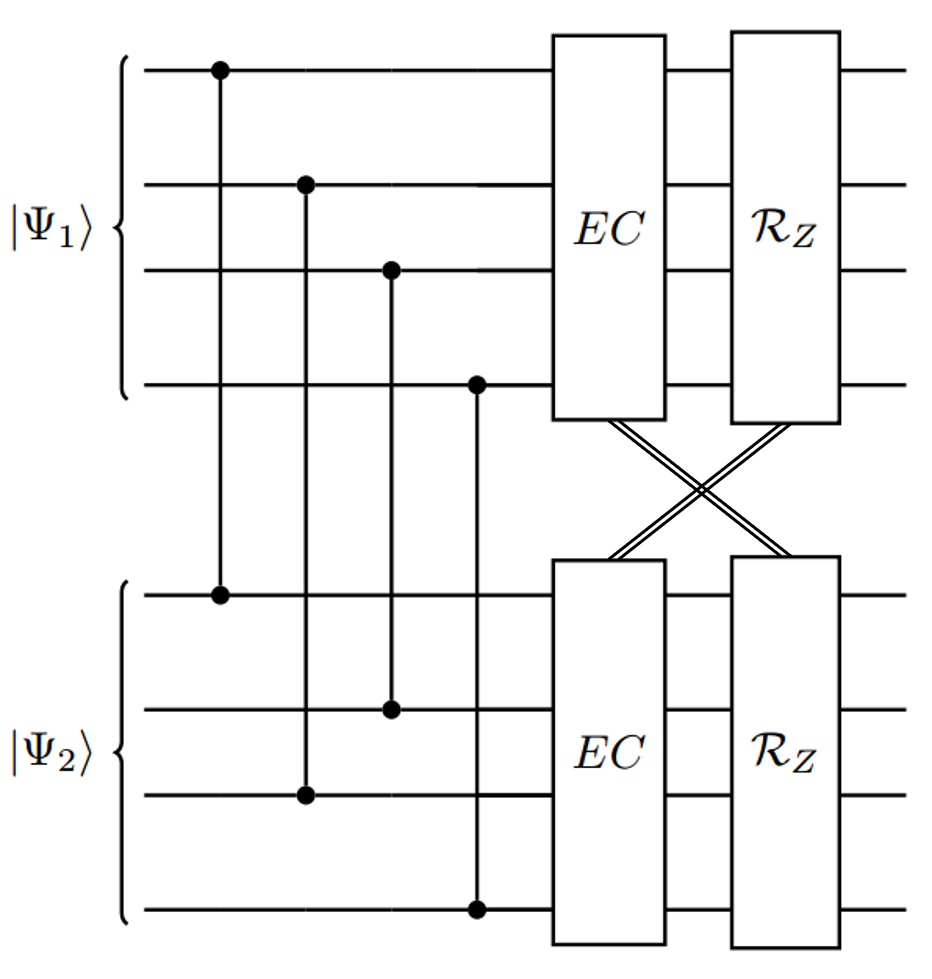}
\caption{The logical \textsc{cz} gadget.}
\label{fig:cphase}
\end{figure}  

As was the problem with the \textsc{cnot} gate, the physical $H$, $S$ and $T$ gates are not noise-structure-preserving: They change an input damping error into an error not correctable by the $4$-qubit code. We thus {do not have tranvsersal implementations of these logical gates}; rather, we need a different approach for getting fault-tolerant logical gate operations. Below, we make use of the well-known technique of gate teleportation \cite{knill_FT,nielsen} to construct our fault-tolerant logical gadgets. The resulting logical gadgets are manifestly fault-tolerant against amplitude-damping noise as we build the teleportation circuits using the basic encoded gadgets shown to be fault-tolerant in Sec.~\ref{sec:enc_gadget}.


\subsection{Preparation units}\label{sec:logical_prep}

Armed with the ability to prepare the Bell state fault-tolerantly, as shown in Sec.~\ref{sec:prep}, we can obtain fault-tolerant preparations of the logical states $|0\rangle_L$ and $|+\rangle_L$. The $|+\rangle_L$ state is easy: $|+\rangle_L$ is simply two copies of $|\beta_{00}\rangle$, that is, 
\[ |+\rangle_L=|\beta_{00}\rangle\otimes|\beta_{00}\rangle.\] 

To get the state $|0\rangle_L$, we start with a single copy of a fault-tolerantly prepared $|\beta_{00}\rangle$ and make use of the circuit shown in Fig.~\ref{fig:zero}, with two additional ancillas initialized to $|0\rangle$. Note that a single fault anywhere in the circuit of Fig.~\ref{fig:zero} is detected by performing parity measurements, identical to the measurements giving syndromes $s$ and $t$ shown in Fig.~\ref{fig:ec}. The prepared state is accepted only when both the parity measurements are even.
\begin{figure}[H]
\begin{center}
\includegraphics[width=0.65\columnwidth]{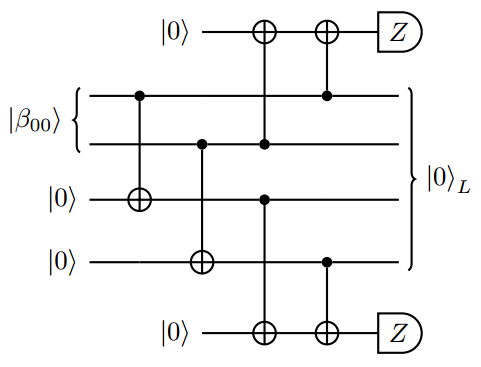}
\caption{Fault-tolerant preparation of $|0\rangle_L$. The input Bell state $|\beta_{00}\rangle$ is assumed to have been prepared fault-tolerantly by the preparation circuit of Fig.~\ref{fig:plus}.}
\label{fig:zero}
\end{center}
\end{figure} 

Finally, we also demonstrate fault-tolerant preparation units for the two-qubit states $\ket{\Phi_{S}}$ and  $\ket{\Phi_{S}}$, which act as resource states for constructing the logical $S$ and $T$ gates, respectively.  
\begin{align}\label{eq:resource}
  \ket{\Phi_S} &= \frac{1}{\sqrt{2}}{\left(\ket{0}_L + \mi\ket{1}_L\right)} \\
\textrm{and }\quad    \ket{\Phi_T} &= \frac{1}{\sqrt{2}}{\left(\ket{0}_L + \mathrm{e}^{\mi\pi/4}\ket{1}_L\right)}.\nonumber
\end{align}
The resource states $\ket{\Phi_S}$ and $\ket{\Phi_T}$ can be prepared and verified as shown in Fig.~\ref{fig:resource} starting with a fault-tolerant preparation of the states $|\beta_{S}\rangle$,$|\beta_{T}\rangle$, which are local-unitary equivalents of the Bell state $|\beta_{00}\rangle$.
\begin{align}
|\beta_{S}\rangle &\equiv \frac{1}{\sqrt 2}{\left(|00\rangle+ \mi |11\rangle\right)} \nonumber \\
\textrm{and}\quad |\beta_{T}\rangle &\equiv \frac{1}{\sqrt 2}{\left(|00\rangle+ \mathrm{e}^{\mi \pi/4} |11\rangle\right)}.
\end{align} 

We construct fault-tolerant preparation units for the states $|\beta_{S/T}\rangle$ as shown in Fig.~\ref{fig:b_ST}, using a circuit which is similar to the preparation of $|\beta_{00}\rangle$ in Fig.~\ref{fig:plus}. In each case, we accept the output state only when $X$-measurement outcomes are of even parity and the $Z$ measurement provides a trivial outcome. Using $|\beta_{S/T}\rangle$ states we then prepare and verify the resource states $|\phi_{S/T}\rangle$ in Eq.~\ref{eq:resource} as shown in Fig.~\ref{fig:resource}. We accept the prepared state only when both the $Z$ measurements show trivial outcomes.

\begin{figure}[H]
		\centering
		\includegraphics[width=0.65\textwidth]{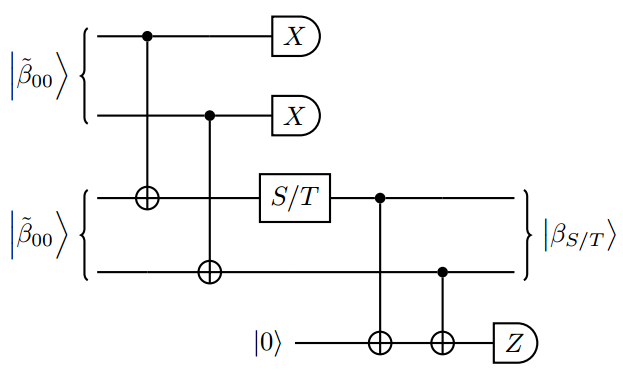}
		\caption{Fault-tolerant preparation of $|\beta_{S/T}\rangle$ (top)}
		\label{fig:b_ST}
	\end{figure}

	\begin{figure}
		\centering
		\includegraphics[width=0.65\textwidth]{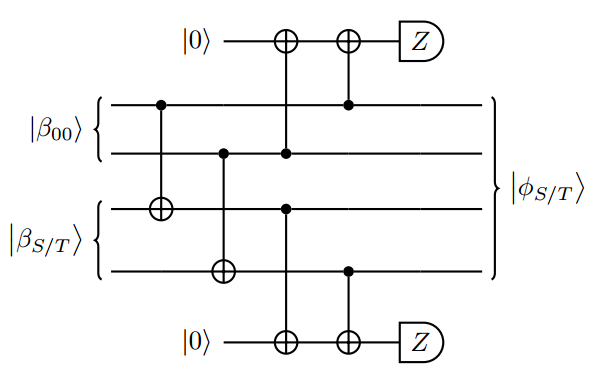}
		\caption{Fault-tolerant preparation of $|\phi_{S/T}\rangle$ (bottom)}
	\label{fig:resource}
\end{figure}

\subsection{Logical Hadamard}\label{sec:Hadamard}

Figure~\ref{fig:Hadamard} gives the construction of our fault-tolerant logical Hadamard. The teleportation scheme requires the ancillary input of $|+\rangle_L$, as well as the use of the logical \textsc{cphase} gadget , whose fault-tolerant preparation was described in Sec.~\ref{sec:prep}.

\begin{figure}[H]
\begin{center}
\includegraphics[width=0.55\columnwidth]{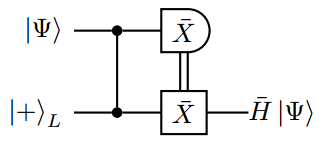}
\caption{Fault-tolerant logical Hadamard gadget, implemented using a teleportation scheme. Here, each horizontal qubit line represents 4 physical qubits, or, equivalently, one logical qubit. $|\Psi\rangle_L$ is the state of the incoming logical qubit on which we want to implement $\overline H$. }
\label{fig:Hadamard}
\end{center}
\end{figure}

In Fig.~\ref{fig:Hadamard}, a logical \textsc{cphase} gadget is applied to the incoming logical (data) state $|\Psi\rangle_L$ and the ancillary input state $|+\rangle_L$. The data qubits are then measured in the $\overline X$ basis, using our $\overline X$ measurement gadget. Depending on the outcome of the $\bar{X}$ measurement, the remaining four qubits (originally ancillas, but now the data qubits after the teleportation) end up in one of two states: $\overline H |\Psi\rangle_L$ or $\overline{X} \overline{H}|\Psi\rangle_L$. In the latter case, we apply a corrective logical $\overline{X}$ gadget to the outgoing state.

This logical Hadamard gate is fault-tolerant in the following sense: The gadget in Fig.~\ref{fig:Hadamard} with at most single damping fault teleports an incoming state with no-error to an output with upto one single-qubit damping error. The teleportation gadget without any faults, teleports an input with upto one single-qubit damping error to an output with at most one single-qubit damping error. This holds simply due to the fact the Hadamard gadget is constructed from the fault-tolerant basic encoded gadgets described in Sec.~\ref{sec:enc_gadget}, each tolerant to a single damping error.

\subsection{Logical $S$ and $T$ gates}\label{sec:ST}

We provide the teleportation circuits for the logical $S$ and $T$ gadgets in Fig.~\ref{fig:sgate} and Fig.~\ref{fig:tgate} respectively. Like the logical $H$, these circuits are composed of elementary encoded gadgets described in Sec.~\ref{sec:enc_gadget}, which can be realized fault-tolerantly.
\begin{figure}[H]
	\centering
    \includegraphics[width=0.55\columnwidth]{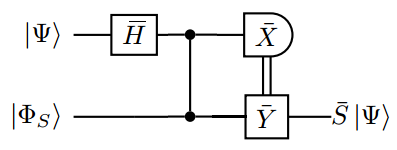}
    \caption{Logical S gate}
    \label{fig:sgate}
\end{figure}

\begin{figure}[H]
	\centering
    \includegraphics[width=0.55\columnwidth]{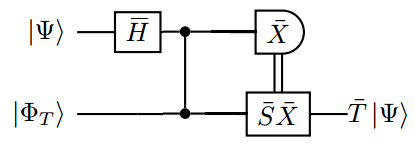}
    \caption{Logical T gate }
    \label{fig:tgate}
\end{figure}
In both cases, the state $|\Psi\rangle$ denotes the state of the incoming logical (data) qubits. In both cases, this state is acted upon by the logical $\overline H$ gadget, followed by a logical \textsc{cphase} gate applied to $\ket{\Psi}_{L}$ and the appropriate ancillary input state $\ket{\Phi_{S/T}}$. In both cases, the data qubits are measured in the $\overline X$ basis, which teleports the desired state to the ancilla block, upto a logical operation. 

The final step involves a conditional logical operation on the remaining qubits. In the case of the $\overline S$ gadget, this conditional operation is $\overline{Y} = \overline{Z}\overline{X} =\overline {H}\overline{X}\overline{H}\overline {X} $, whereas in the case of the $\overline T$ gadget, the conditional logical operation to be implemented is $\overline{S}\overline{X}$. Once again, the fault tolerance of the logical $\overline S$ and $\overline T$ gadgets follows from the fact that they are composed of basic encoded gadgets which are all known to be tolerant to single damping errors.

\section{Pseudothreshold calculation}\label{sec:threshold}

The reliability of a fault tolerance scheme is quantified by its error threshold, which refers to that critical value of the  error rate, $p_{th}$, at which the fault-tolerant protocol fails to perform any better than its unencoded version.  If $p$ refers to the unencoded error rate, and if the fault-tolerant procedure can correct upto $t$ errors, then the failure rate of a fault-tolerant circuit is given by $C p^{t+1}$, where $C$ refers to the number of ways in which a unit can have two faults and thus propagate an output which is not correctable.

The critical error threshold $p_{th}$ is then obtained as $C^{-1}$, for the $C$ that solves,  
 \begin{equation}\label{eq:threshold}
 p_{\rm th}= C p_{\rm th}^{t+1}.
\end{equation}
Since this value in Eq.~\ref{eq:threshold} describes the reliability of the scheme at just one level of encoding, we may think of it as a \emph{pseudo-threshold}. This is a key departure from the standard notion of a fault tolerance threshold which is defined for an entire class of codes, by concatenating a given code multiple times~\cite{aliferis}. 

\subsection{Memory Pseudothreshold}\label{sec:mem_threshold}

 \begin{figure}[H]
\centering
\includegraphics[scale=.65]{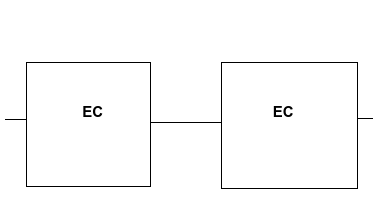}
\caption{Memory unit}
\label{fig:memory}
\end{figure} 

We first calculate such a pseudothreshold for a memory unit comprising a pair of \textsc{ec} gadgets interlaced with identity gates or rests, as shown in the Fig.~\ref{fig:memory}. By explicit enumeration of the pairs of faults, referred to as \emph{malignant} pairs, that lead to errors of $O(p^2)$ in the output state, we obtain,
\begin{equation}
p_{\rm th} \approx 2.8 \times 10^{-3}.
\end{equation}
We refer to Appendix~\ref{app:memory} for the details of the calculation. 

\subsection{Computational Pseudothreshold}\label{sec:comp_threshold}

In order to obtain a computational pseudothreshold, we go beyond our encoded gadgets to construct \emph{extended} units, which take into account both incoming errors as well as all faults occurring within a given gadget. The extended unit with the maximum number of malignant pairs determines the computational pseudothreshold. In our case, we have determined by exhaustive check that an extended \textsc{cphase} unit, shown in Fig.~\ref{fig:exrec}, is the unit that determines the error threshold. The extended \textsc{cphase} unit comprises of the encoded \textsc{cphase}gadget in Fig.~\ref{fig:cphase}, and \textsc{ec} units appended before each block of the encoded gadget. 

\begin{figure}[H]
\centering
\includegraphics[scale=0.7]{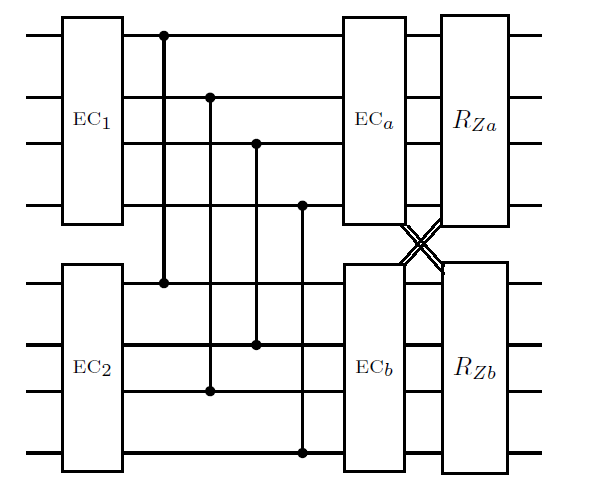}
\caption{Extended \textsc{cphase} unit for computing the pseudothreshold}
\label{fig:exrec}
\end{figure}

In order to compute the pseudothreshold of this extended unit, we first enumerate the number of malignant pairs of damping faults leading to an output that is not correctable in the extended \textsc{cz} unit, assuming that the incoming state has no errors. We first label the different blocks that constitute the extended unit, as follows.
\begin{enumerate}
	\item \textsc{ec$_1$} - leading \textsc{ec} in block 1
	\item \textsc{ec$_2$} - leading \textsc{ec} in block 2
	\item \textsc{cz} gadget
	\item \textsc{ec$_a$} + $\cR_{Za}$
	\item \textsc{ec$_b$} + $\cR_{Zb}$
\end{enumerate}
We can then represent the number of malignant pairs via a matrix whose rows and columns correspond to each block in Fig.~\ref{fig:exrec}, with the entries of the matrix denoting the total malignant pair contributions from the respective blocks. We merely note the final numbers here and refer to Appendix~\ref{app:cphase_threshold} for the detailed enumeration of malignant pairs for every pair of blocks.
\[
\begin{blockarray}{cccccccc}
	 & 1 & 2 & 3 & 4 & 5 \\
	\begin{block}{c (ccccccc)}
		1\  & 0 \\
		2\ & 74 & 0 \\
		3\  & 30 & 30 &  6  \\
		4\  & 274 & 240 & 110 & 134 \\
		5\  & 240 & 274 & 110 & 224 & 134 \\
	\end{block}	
\end{blockarray}
\]

The total number of malignant pairs of damping errors is simply the sum of the entries of the matrix. Apart from the pairs of damping errors, we also need to keep track of the $O(p^{2})$ phase errors, and we argue in Appendix~\ref{app:cphase_threshold} that there are $52$ malignant fault locations for the phase errors.  

The total number of malignant pairs due to damping errors and malignant fault locations due to $Z$ errors is  given by, $A = 1880$, leading to a computational pseudothreshold of 
\[ p_{\rm th} = 1/ A \approx 5.31 \times 10^{-4}.\]

\section{Conclusion}\label{sec:concl}
We develop a fault-tolerant scheme specific to the amplitude-damping channel, using the {$4$-qubit} code. We construct gadgets, namely, the error correction unit, a universal gate set, and a two-qubit gate which are tolerant against single-qubit damping errors. Our constructions show that achieving fault tolerance using channel-adapted codes for nontrivial, non-Pauli noise models like amplitude-damping channel poses interesting challenges, and often leads to counter-intuitive results when viewed from the standpoint of the well established principles of quantum fault tolerance. For instance, although the logical two-qubit \textsc{cnot} gadget is transversal, it is not tolerant to single damping errors, thus violating the wisdom gained from standard fault tolerance schemes for Pauli channels based on stabilizer codes. 

The structure of the non-Pauli noise dictates our choice of fault-tolerant gate gadgets. Thus, the transversal two-qubit \textsc{cphase} gadget and three-qubit \textsc{ccz} gadget turn out to be a naturally favoured set since they can be made tolerant against single damping errors. When it comes to single-qubit damping errors, we do not obtain transversal constructions for the $[[4,1]]$ code; rather, we have to rely on a teleportation-based scheme to implement the Hadamard gadget, $S$ and $T$ gadgets that are fault-tolerant.

Finally, we obtain a rigorous estimate on the level-$1$ \emph{pseudo threshold} for the memory unit as $2.8\times 10^{-3}$ and \textsc{cphase}-\textsc{exrec} unit of $5.31 \times 10^{-4}$, 
by explicitly enumerating the malignant pairs of faults. 
We present an updated, extensive fault-tolerant scheme in Ref.~\cite{ak_ft} which takes into account the back-action errors besides single-qubit damping errors.
While our work presents a first step towards achieving fault tolerance against a specific non-Pauli noise model, the ultimate goal would be to obtain a true fault tolerance threshold via concatenation or using surface codes. Whether such a concatenated scheme can be developed entirely using channel-adapted codes such as the $[[4,1]]$ code, or whether this would require us to concatenate the $[[4,1]]$ code with a suitable general purpose code remains an interesting open question. Another immediate question is whether gadgets that are tolerant to amplitude-damping noise can be constructed using $2$-d Bacon-shor code~\cite{renes}. Such a code might lead to a genuine fault tolerance threshold for amplitude-damping noise.


%% file: Chapter6.tex

\chapter{Summary and Future Scope} 

\label{Chapter6} 

\lhead{Chapter 6. \emph{Chap:6}} 

Our thesis work entails a study of three interesting aspects centered around channel-adapted quantum codes, namely, 
\begin{itemize}
\item[(a)] Construction: we have proposed an efficient numerical way to search for channel-adapted codes and demonstrated its usefulness in the case of amplitude-damping noise~\cite{ak_cartan}.
\item[(b)] Application: we have shown how channel-adapted codes may be used in the context of the quantum state transfer protocol to obtain better fidelities over longer spin chains~\cite{ak_statetransfer}.
\item[(c)] Fault tolerance: we have developed a fault-tolerant scheme comprising a universal gate set and error-correction units, based on the $4$-qubit code adapted to amplitude-damping noise~\cite{ak_ft}.
\end{itemize}

The work presented in this thesis has several novel aspects and original results. For instance, the use of the Cartan form for finding good codespaces (Chapter~\ref{Chapter3}) to protect quantum information is novel and leads to an efficient numerical procedure to search for quantum codes. Similarly, in Chapter~\ref{Chapter5}, a new paradigm has been developed for quantum fault tolerance using channel-adapted quantum error correction, including several novel circuit constructions.

Our work opens up several directions of research. It would be interesting to study how the numerical procedure elaborated in Chapter~\ref{Chapter3} extends to the case of passive error suppression techniques such as decoherence-free subspaces (DFS) and noiseless subsystems(NS). Recently, \cite{wang} proposed a numerical procedure to search for the existence of DFS/NS in noise models with and without perturbations. We note here that it is indeed straightforward to extend our search procedure to check for the existence of DFS for a given noise model. For example, motivated by existence of a DFS for correlated amplitude-damping noise~\cite{correlated_AD} on two qubits, we examined the case of the $3$-qubit and $4$-qubit correlated amplitude-damping noise channels. Preliminary results suggest that our unstructured search procedure, where we make use of a full, unstructured parameterization of the encoding unitary, may work well and can identify the subspaces corresponding to a DFS for both cases. 

Channel-adapted codes are expected to have a wide variety of applications in different quantum information processing tasks. We have demonstrated one such application in the context of quantum state transfer over $1$-d  spin chains in Chapter~\ref{Chapter4}. While our work has focussed on spin-preserving Hamiltonians, our protocol can still be extended to study more general Hamiltonians based on realistic implementations. For example, the Josephson junction arrays with a more general non-spin preserving Hamiltonian~\cite{lyakhov} gives rise to an arbitrary noise channel. In such a case our state transfer protocol with channel-adapted codes could be more useful over the perfect QEC protocols.

The important open problem of achieving fault tolerance using a channel-adapted code has been studied in Chapter~\ref{Chapter5}. The problem of constructing fault-tolerant gadget set for a non-Pauli noise model is a counter-intuitive, challenging task, deviating from the standard fault-tolerant constructions that deal with Pauli errors and use \emph{perfect} codes. While we study fault tolerance using a specific channel-adapted code, namely,the $[[4,1]]$ code~\cite{leung}, it would be interesting to see if there are other codes tailormade to specific noise models which maybe used in this regard. For example, an immediate question is whether gadgets that are tolerant to amplitude-damping noise can be constructed using the $2$-d Bacon-shor code~\cite{renes}. Such a code might allow us to build a universal fault-tolerant scheme, leading to a genuine fault tolerance threshold for amplitude-damping noise.

%% file: appendixA_v2.tex

\chapter{Channel-adapted codes for the amplitude-damping channel} 
\label{AppendixA} 

\lhead{Appendix A. \emph{Title of the Appendix in short}} 
\section{Nelder-Mead search algorithm}\label{sec:nmsearch}
The Nelder-Mead algorithm~\cite{nelderpaper, NumericalRecipes} to search for {\it optimal} encoding proceeds through the following steps. 
\begin{itemize}
\item[1.] We first create a $m+1$-fold simplex for a $m$ parameter search space, with each of the vertex initialized to a unitary $U_{i}$ $\in$ $SU(2^{n})$ by sampling the parameters of the unitary $U_{i}$ from a uniform distribution over $[0,2\pi]$. We also pick an initial code space $\cC_{\rm initial}$ spanned by $\{|0_{\rm initial}\rangle,|1_{\rm   initial}\rangle \}$. At each vertex, we then obtain,
\begin{eqnarray}\label{eq:nmiterate}
 |0^{i}_{\rm new}\rangle &=& U_{i} |0_{\rm initial}\rangle \nonumber \\ \nonumber
 |1^{i}_{\rm new}\rangle &=& U_{i} |1_{\rm initial}\rangle
\end{eqnarray}
where $\{|0^{i}_{\rm new}\rangle,|1^{i}_{\rm new}\rangle\}$ span the code space $\cC_{i}$ at each vertex $i$.
\item[2.] Each iteration of the Nelder-mead algorithm proceeds via the following steps.
\begin{itemize}
\item[(i)] Sort : Each vertex $i$ is given by a set of real parameters $\mathbf{x_{i}} = (a_{i},b_{i},c_{i},\ldots)$, describing the unitary $U_i$ $\in$ $SU(2^{n})$ according to the Cartan decomposition given in Sec.~\ref{sec:cartan}. The objective function, which is the fidelity-loss function given in Eq.~\ref{eq:fidelityloss} is then evaluated at each vertex $i$ for $\mathbf{x_{i}}$ as $\eta(\mathbf{x_i})$, and then the set $\{\eta(\mathbf{x_{i}})\}$ is sorted such that, 
\begin{equation}
\eta(\mathbf{x_1}) <\eta(\mathbf{x_2}) <\eta(\mathbf{x_3})...<\eta(\mathbf{x_{m+1}})
\end{equation}
where $\eta(\mathbf{x_1})$ is the minimum value and $\eta(\mathbf{x_{(m+1)}})$ is the maximum value of the fidelity-loss. The centroid corresponding to this set of $(m+1)$ vertices after sorting is obtained as given below,
\begin{equation}
\bar{\mathbf{x}}= \frac{\sum_{i=1}^{m} \mathbf{x_{i}}}{m}
\end{equation}
\item[(ii)] Reflect : We now obtain $\mathbf{x_{r}}$ as follows,
\begin{equation}
\mathbf{x_{r}} =\bar{\mathbf{x}} +\alpha(\bar{\mathbf{x}}-\mathbf{x_{m+1}})
\end{equation}
where the constant, $\alpha$ $\in$ $\mathbb{R}$. If $\eta(\mathbf{x_{1}}) \leq \eta(\mathbf{x_{r}})< \eta(\mathbf{x_{m+1}})$, replace $\mathbf{x_{m+1}} $ with $\mathbf{x_{r}} $.
\item[(iii)] Expand: If $\eta(\mathbf{x}_{r})$ $<$ $\eta(\mathbf{x}_{1})$, then obtain,
\begin{equation}
\mathbf{x_{e}} =\bar{\mathbf{x}} +\beta(\mathbf{x_{r}}-\bar{\mathbf{x}})
\end{equation}
where the constant, $\beta$ $\in$ $\mathbb{R}$. If $\eta(\mathbf{x_{e}}) \leq \eta(\mathbf{x_{r}})$, replace $\mathbf{x_{m+1}} $ with $\mathbf{x_{e}} $, otherwise replace $\mathbf{x_{m+1}} $  with $\mathbf{x_{r}}$.
\item[(iv)] Contract: If  $\eta(\mathbf{x}_{m})$ $\leq$ $\eta(\mathbf{x}_{r})$ $<$ $\eta(\mathbf{x}_{m+1})$, then obtain,
\begin{equation}
\mathbf{x_{c}} =\bar{\mathbf{x}} +\Delta(\mathbf{x_{m+1}}-\bar{\mathbf{x}})
\end{equation}
where the constant, $\Delta$ $\in$ $\mathbb{R}$. If $\eta(\mathbf{x_{c}}) \leq \eta(\mathbf{x_{r}})$, replace $\mathbf{x_{m+1}} $ with $\mathbf{x_{c}} $, otherwise go to step (v).
\item[(v)] Shrink: for $2<i<m+1$, replace
\begin{equation}
\mathbf{x_{i}}= \mathbf{x_{1}} +\sigma(\mathbf{x_{i}}-\mathbf{x_{1}})
\end{equation}
where the constant, $\sigma$ $\in$ $\mathbb{R}$. We again go to step (i) .
\end{itemize}
\end{itemize}
Thus the algorithm proceeds by generating newer simplices at every step and flows towards the optimum. Also, note that at each step a new $U'_{i}$ is generated at $i^{th}$ vertex, giving a new code space $\cC'^{i}$ spanned by
\begin{eqnarray}\label{eq:nmiterate1}
 |0^{'i}_{\rm new}\rangle &=& U'_{i} |0_{\rm initial}\rangle \nonumber \\ \nonumber
 |1^{'i}_{\rm new}\rangle &=& U'_{i} |1_{\rm initial}\rangle
\end{eqnarray}
The search is terminated once the optimum is attained for the desired accuracy level.

\section{Structured Codes and Encoding Circuits}\label{sec:circuit}
Here, we describe simple circuits by means of which a structured encoding of the form given in Eq.~\eqref{eq:U_8} can be implemented. Recall that the Cartan decomposition for such a unitary is given by,
\begin{equation}\label{eq:nonlocalunitary}
U=F^{(1)}J F^{(2)},
\end{equation}
where the unitary operators $\{F^{(1)},{J},F^{(2)}\}$ are constructed from elements of an Abelian subgroup of the $n$-fold Pauli group. For the case of $SU(2^{3})$, these operators are given by{ [see Eq.~\eqref{eq:su8}]},
\begin{eqnarray}
 F^{(1)} &=& \upe^{-\upi(c_{1}XXZ+c_{2}YYZ+c_{3}ZZZ)} \nonumber \\ 
  &=& \upe^{-\upi c_{1}XXZ} \upe^{-\upi c_{2}YYZ} \upe^{-\upi c_{3}ZZZ} \nonumber \\ 
F^{(2)}&=& \upe^{-\upi c_{4}XXZ} \upe^{-\upi c_{5}YYZ} \upe^{-\upi c_{6}ZZZ} \nonumber \\ 
J &=& \upe^{-\upi(a_{1}XXX+a_{2}YYX+a_{3}ZZX+a_{4}IIX)} \nonumber \\
 &=&\upe^{-\upi a_{1}XXX} \upe^{-\upi a_{2}YYX} \upe^{-\upi a_{3}ZZX}\upe^{-\upi a_{4}IIX} . \label{eq:abeliangroup2}
\end{eqnarray}

Following a simple prescription in~\cite{swaddle}, we can construct quantum circuits to implement $F^{(1)}, F^{(2)}$  and $J$ by suitably combining the simple circuits given in Fig.~\ref{fig:circuits}. The full encoding unitary $U$ given in Eq.~\eqref{eq:nonlocalunitary} is then composed from the circuits for $F^{(1)}$, $F^{(2)}$ and $J$ as given in Fig.~\ref{fig:circuits2}.

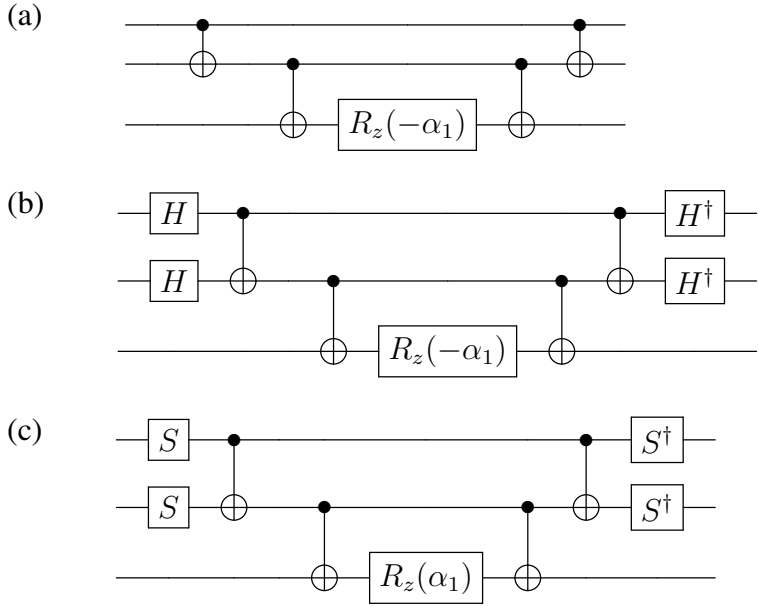
\begin{figure}[H]
\flushleft
(a) \hspace*{-0.7cm}\Qcircuit @C=1em @R=.7em {
&&&&& \qw & \ctrl{1} & \qw & \qw  &\qw & \qw  & \ctrl{1}    &\qw \\
&&&&& \qw & \targ  &\qw  & \ctrl{1} & \qw & \ctrl{1} & \targ  &\qw \\
&&&&& \qw &\qw  &\qw &\targ  &\gate{R_{z}(-\alpha_{1})}  &\targ &\qw  &\qw 
}
\vspace*{0.5cm}

(b) \hspace*{-0.4cm}\Qcircuit @C=1em @R=.7em {
&&&& \gate{H} & \ctrl{1} & \qw & \qw  &\qw & \qw  & \ctrl{1}   & \gate{H^{\dagger}} &\qw \\
&&&& \gate{H} & \targ  &\qw  & \ctrl{1} & \qw & \ctrl{1} & \targ & \gate{H^{\dagger}} &\qw \\
&&&& \qw &\qw  &\qw &\targ  &\gate{R_z(-\alpha_{1})}  &\targ &\qw & \qw &\qw
}
\vspace*{0.5cm}

(c) \hspace*{-0.4cm}\Qcircuit @C=1em @R=.7em {
&&&& \gate{S} & \ctrl{1} & \qw & \qw  &\qw & \qw  & \ctrl{1}   & \gate{S^{\dagger}} &\qw \\
&&&& \gate{S} & \targ  &\qw  & \ctrl{1} & \qw & \ctrl{1} & \targ & \gate{S^{\dagger}} &\qw \\
&&&& \qw &\qw  &\qw &\targ  &\gate{R_z(\alpha_{1})}  &\targ &\qw & \qw &\qw
}
\caption{Quantum circuit implementing {(for $\alpha_{1} \in \mathbb{R}$) (a) the gate $e^{(-i \alpha_{1} Z \otimes Z \otimes Z)}$; (b) the gate $e^{(-i \alpha_{1} X \otimes X \otimes Z)}$; and (c) the gate $e^{(-i \alpha_{1} Y \otimes Y \otimes Z)}$. Here, $H\equiv|+\rangle\langle 0|+|-\rangle\langle 1|$ is the Hadamard gate, and $S\equiv \frac{1}{\sqrt 2}(\id+\upi \sigma_x)$.}}
\label{fig:circuits}
\end{figure}

\begin{figure}[H]
\includegraphics[scale=.65]{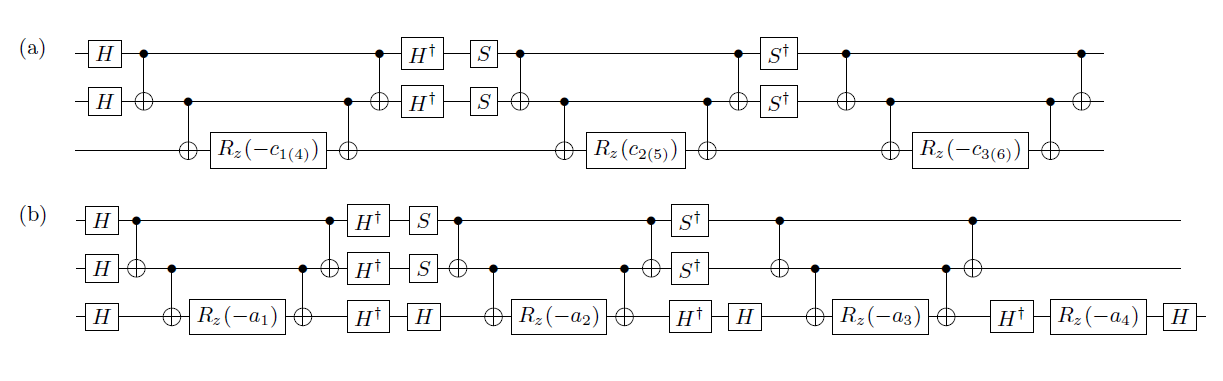}
\caption{\label{fig:circuits2} Circuits for (a) $F^{(1(2))}$ and (b) $J$ for $SU(2^3)$.} 
\end{figure}

Fig.~\ref{fig:circuits2} indicates that the encoding circuits are made up of \textsc{cnot} gates and single qubit unitaries which are rotations about $z$ axis on the Bloch sphere by angles $\alpha_{i}$ determined by the search parameters.  In other words, once we obtain the optimal code, we can easily encode into the desired subspace by only changing the rotation angle about the $z$ axis, while keeping the rest of the components in the encoding circuit fixed.

\section{{Optimal codes for the amplitude-damping channel}}\label{sec:numericalcodes}

We list {in Table \ref{tab1},~\ref{tab2} }the optimal codes obtained using our numerical search, for the standard amplitude-damping channel, whose performance has been described in the plots in Fig.~\ref{fig:3-4qubit}.

\begin{table*}
\centering
\begin{tabular}{||c|c|c||}
Encoding&$|0_L\rangle$&$|1_L\rangle$\\
\hline\hline
&&\\
3-qubit unstructured & 
~$\begin{pmatrix}
  -0.426 + 0.235\upi\\
   0.040 - 0.415\upi\\
   0.014 + 0.084\upi\\
  -0.312 + 0.323\upi\\
   0.021 + 0.278\upi\\
   0.089 + 0.167\upi\\
  -0.303 + 0.038\upi\\
  -0.403 + 0.108\upi\\
  \end{pmatrix}$~& 
  ~$\begin{pmatrix}
   0.275 + 0.103\upi\\
   0.248 + 0.191\upi\\
   0.116 - 0.116\upi\\
   0.008 - 0.194\upi\\
   0.429 + 0.266\upi\\
  -0.066 - 0.269\upi\\
  -0.086 + 0.305\upi\\
  -0.488 - 0.285\upi\\
  \end{pmatrix}$~\\
&&\\
\hline
&&\\
4-qubit unstructured & 
$\begin{pmatrix}
  0.448 + 0.236\upi\\
  -0.066 + 0.134\upi\\
  -0.052 + 0.003\upi\\
  -0.044 - 0.027\upi\\
  -0.037 + 0.058\upi\\
   0.313 + 0.048\upi\\
  -0.338 - 0.057\upi\\
   0.001 - 0.060\\
   0.006 - 0.114\upi\\
  -0.310 - 0.088\upi\\
  -0.356 - 0.073\upi\\
   0.020 - 0.004\upi\\
   0.059 - 0.004\upi\\
   0.041 - 0.002\upi\\
  -0.038 + 0.040\\
   0.412 + 0.250
\end{pmatrix}$&
$\begin{pmatrix}
   0.379 - 0.350\\
  -0.012 + 0.001\upi\\
  -0.040 + 0.042\upi\\
   0.027 - 0.038\upi\\
   0.038 - 0.024\upi\\
  -0.170 + 0.292\upi\\
   0.191 - 0.321\upi\\
   0.027 - 0.053\upi\\
   0.031 - 0.041\upi\\
   0.200 - 0.276\upi\\
   0.210 - 0.290\\
   0.027 + 0.077\upi\\
  -0.014 + 0.033\upi\\
   0.098 + 0.013\upi\\
   0.022 + 0.026\upi\\
   0.356 - 0.276\upi
\end{pmatrix}$~\\
&&
\end{tabular}
\caption{\label{tab1} Optimal unstructured codes for amplitude-damping channel plotted in Fig.~\ref{fig:3-4qubit}.}
\end{table*}

\begin{table*}
\centering
\begin{tabular}{||c||c|c||}
Encoding&$|0_L\rangle$&$|1_L\rangle$\\
\hline\hline
&&\\
3-qubit structured & 
~$\begin{pmatrix}
  -0.013 + 0.076\upi\\
  -0.587 + 0.370\\
   0\\
   0\\
   0\\
   0\\
   0.026 + 0.052\upi\\
   0.385 + 0.601\upi\\
  \end{pmatrix}$~&
  ~$\begin{pmatrix}
   0\\
   0\\
  -0.152 + 0.056\upi\\
  -0.330 - 0.177\upi\\
   0.491 + 0.763\upi\\
  -0.044 - 0.095\upi\\
   0\\
   0\\
  \end{pmatrix}$~\\
&&\\
\hline
&&\\
4-qubit structured & 
$\begin{pmatrix}
   0.580 - 0.352\upi\\
   0.026 - 0.210\\
   0.027 + 0.040\\
  -0.001 + 0.042\upi\\
   0\\
   0\\
   0\\
   0\\
   0\\
   0\\
   0\\
   0\\
  -0.014 + 0.030\\
  -0.056 + 0.025\upi\\
   0.134 - 0.166\upi\\
   0.048 + 0.662\upi
\end{pmatrix}$ &
$\begin{pmatrix}
 0\\
   0\\
   0\\
   0\\
   0.186 + 0.028\upi\\
  -0.353 + 0.178\upi\\
  -0.434 - 0.017\upi\\
  -0.099 + 0.059\upi\\
  -0.191 + 0.123\upi\\
   0.071 - 0.511\upi\\
  -0.346 + 0.379\upi\\
   0.051 + 0.157\upi\\
   0\\
   0\\
   0\\
   0
\end{pmatrix}$\\
&&
\end{tabular}
\caption{\label{tab2} Optimal structured codes for amplitude-damping channel plotted in Fig.~\ref{fig:3-4qubit}.}
\end{table*}

%% file: appendixB_v2.tex

\chapter{Channel-adapted recovery for state transfer over $1$-d spin chains} 
\label{AppendixB} 

\lhead{Appendix B \emph{Title of the Appendix in short}} 

\section{Effect of the noise channel $\widetilde{\cE}_{AD}$ on the $4$-qubit code}\label{sec:E(P)}

We note the following structure for the Kraus operators of the $4$-qubit channel, by expanding them in the $4$-qubit computational basis. First, we note that the only Kraus operator diagonal in the computational basis is $E_{0}^{\otimes 4}$, with diagonal entry $e^{ij\Theta}|f_{r,s}^{N}(t)|^{j}$, corresponding to those basis vectors with $j$ $1$'s in them. 
All the other operators are off-diagonal matrices with support on some subset of computational basis states. For example, a three-qubit error operator (involving $E_{1}$ in three of the four  qubits) is of the form,
\begin{equation}
E_{0}\otimes E_{1}^{\otimes 3} = (1-|f^{N}_{r,s}(t)|^{2})^{3/2}|0000\rangle\langle 0111|+ e^{i\Theta}|f^{N}_{r,s}(t)|(1-|f^{N}_{r,s}(t)|^{2})^{3/2}|1000\rangle\langle 1111| . \label{eq:E_03}
\end{equation}
The remaining three-qubit errors are of the same form, with the strings $\{0111, 1000\}$ replaced by their permutations. Similarly, an operator which has $E_{1}$ errors on two of the qubits is a linear combination of the form,
\begin{eqnarray}
E_{0}^{\otimes 2}\otimes E_{1}^{\otimes 2} &=& (1-|f^{N}_{r,s}(t)|^{2})|0000\rangle\langle 0011| + e^{2i\Theta}|f^{N}_{r,s}(t)|^{2}(1-|f^{N}_{r,s}(t)|^{2})|1100\rangle\langle 1111| \nonumber \\&& +  e^{i\Theta}|f^{N}_{r,s}(t)|(1-|f^{N}_{r,s}(t)|^{2})\left(|0100\rangle\langle 0111| + |1000\rangle\langle 1011| \right). \label{eq:E_02}
\end{eqnarray}
Other two-qubit error operators are realized by replacing the strings $\{0011,1100,0100,1000\}$ with permutations thereof. A single-qubit error operator, with $E_{1}$ error on only one of the qubits has the form,
\begin{eqnarray}
 E_{0}^{\otimes 3}\otimes E_{1} &=& \sqrt{1-|f^{N}_{r,s}(t)|^{2}}|0000\rangle \langle 0001| +\nonumber \\&& e^{i\Theta}|f^{N}_{r,s}(t)|\sqrt{1-|f^{N}_{r,s}(t)|^{2}}(|0010\rangle\langle 0011|+ |0100\rangle\langle 0101| + |1000\rangle\langle1001|) \nonumber \\ 
&+& e^{2i\Theta} |f^{N}_{r,s}(t)|^{2}\sqrt{1-|f^{N}_{r,s}(t)|^{2}}\left(|1100\rangle\langle 1101| + |0110\rangle\langle 0111| + |1010\rangle\langle 1011| \right)  \nonumber \\ && + e^{3i\Theta} |f^{N}_{r,s}(t)|^{3}\sqrt{1-|f^{N}_{r,s}(t)|^{2}}|1110\rangle\langle 1111| .  \label{eq:E_01}
\end{eqnarray}
Finally, the four-qubit error operator $E_{1}^{\otimes 4}$ is of the form,
\begin{equation}
E_{1}^{\otimes 4} = (1-|f^{N}_{r,s}(t)|^{2})^{2} |0000\rangle \langle 1111| . \label{eq:E_14}
\end{equation}

We next explicitly write out the operator $\small{\widetilde{\cE}_{AD}^{\otimes 4}(P)}$ in the computational basis of the $4$-qubit space.
\tiny{
\begin{eqnarray}
&& \widetilde{\cE}_{AD}^{\otimes 4}(P) = \nonumber \\
&& \left[
\begin{array}{*{16}c}
 \cQ_{1} & 0 & 0 & 0 & 0 & 0 & 0 & 0 & 0 & 0 & 0 & 0 & 0 & 0 & 0 &  e^{-4 i \Theta }\, \cQ_{17} \\
 0 & \cQ_{2} & 0 & 0 & 0 & 0 & 0 & 0 & 0 & 0 & 0 & 0 & 0 & 0 & 0 & 0 \\
 0 & 0 &\cQ_{3} & 0 & 0 & 0 & 0 & 0 & 0 & 0 & 0 & 0 & 0 & 0 & 0 & 0 \\
 0 & 0 & 0 & \cQ_{4}& 0 & 0 & 0 & 0 & 0 & 0 & 0 & 0 & \cQ_{18}& 0 & 0 & 0 \\
 0 & 0 & 0 & 0 & \cQ_{5} & 0 & 0 & 0 & 0 & 0 & 0 & 0 & 0 & 0 & 0 & 0 \\
 0 & 0 & 0 & 0 & 0 & \cQ_{6} & 0 & 0 & 0 & 0 & 0 & 0 & 0 & 0 & 0 & 0 \\
 0 & 0 & 0 & 0 & 0 & 0 & \cQ_{7}& 0 & 0 & 0 & 0 & 0 & 0 & 0 & 0 & 0 \\
 0 & 0 & 0 & 0 & 0 & 0 & 0 & \cQ_{8} & 0 & 0 & 0 & 0 & 0 & 0 & 0 & 0 \\
 0 & 0 & 0 & 0 & 0 & 0 & 0 & 0 & \cQ_{9} & 0 & 0 & 0 & 0 & 0 & 0 & 0 \\
 0 & 0 & 0 & 0 & 0 & 0 & 0 & 0 & 0 &\cQ_{10}& 0 & 0 & 0 & 0 & 0 & 0 \\
 0 & 0 & 0 & 0 & 0 & 0 & 0 & 0 & 0 & 0 &\cQ_{11} & 0 & 0 & 0 & 0 & 0 \\
 0 & 0 & 0 & 0 & 0 & 0 & 0 & 0 & 0 & 0 & 0 & \cQ_{12} & 0 & 0 & 0 & 0 \\
 0 & 0 & 0 & \cQ_{18} & 0 & 0 & 0 & 0 & 0 & 0 & 0 & 0 & \cQ_{13} & 0 & 0 & 0 \\
 0 & 0 & 0 & 0 & 0 & 0 & 0 & 0 & 0 & 0 & 0 & 0 & 0 & \cQ_{14}& 0 & 0 \\
 0 & 0 & 0 & 0 & 0 & 0 & 0 & 0 & 0 & 0 & 0 & 0 & 0 & 0 & \cQ_{15} & 0 \\
 e^{4 i \Theta }\,\cQ_{17}& 0 & 0 & 0 & 0 & 0 & 0 & 0 & 0 & 0 & 0 & 0 & 0 & 0 & 0 & \cQ_{16} \\
\end{array}
\right], \nonumber
\end{eqnarray}
}
\normalsize
with $\{\cQ_{i}\}$ denoting polynomial functions of the transition amplitude $|f^{N}_{r,s}(t)|$. In terms of the rank-$1$ projectors onto the computational basis states, we may write $\small{\widetilde{\cE}_{AD}^{\otimes 4}(P)}$ as,
\begin{eqnarray}
 && \small{\widetilde{\cE}_{AD}^{\otimes 4}(P)} = \sum_{i=1}^{16}\cQ_{i}|i\rangle\langle i| + e^{-4 i \Theta } \cQ_{17}|0000\rangle\langle 1111| \nonumber \\
 && + e^{i 4\Theta} \cQ_{17}|1111\rangle\langle0000| + \cQ_{18}(|1100\rangle \langle0011| +|0011\rangle \langle1100|),
\end{eqnarray}
wherein $|i\rangle \in \{|0000\rangle,\ldots,|0100\rangle, \ldots, |1111\rangle\}$ denote the computational basis states of the $4$-qubit space.

Similarly, we can also express the pseudo-inverse $\small{\widetilde{\cE}_{AD}^{\otimes 4}(P)^{-1/2}}$ in the $4$-qubit computational basis, as follows.
\hspace*{-2.5cm}\begin{eqnarray}
&& \widetilde{\cE}_{AD}^{\otimes 4}(P)^{-1/2} = \nonumber \\
&& {\scriptstyle \left[
\begin{array}{*{16}c}
 \cG_{1} & 0 & 0 & 0 & 0 & 0 & 0 & 0 & 0 & 0 & 0 & 0 & 0 & 0 & 0 &  e^{-4 i \Theta }\, \cG_{17} \\
 0 & \cG_{2} & 0 & 0 & 0 & 0 & 0 & 0 & 0 & 0 & 0 & 0 & 0 & 0 & 0 & 0 \\
 0 & 0 &\cG_{3} & 0 & 0 & 0 & 0 & 0 & 0 & 0 & 0 & 0 & 0 & 0 & 0 & 0 \\
 0 & 0 & 0 & \cG_{4}& 0 & 0 & 0 & 0 & 0 & 0 & 0 & 0 & \cG_{18}& 0 & 0 & 0 \\
 0 & 0 & 0 & 0 & \cG_{5} & 0 & 0 & 0 & 0 & 0 & 0 & 0 & 0 & 0 & 0 & 0 \\
 0 & 0 & 0 & 0 & 0 & \cG_{6} & 0 & 0 & 0 & 0 & 0 & 0 & 0 & 0 & 0 & 0 \\
 0 & 0 & 0 & 0 & 0 & 0 & \cG_{7}& 0 & 0 & 0 & 0 & 0 & 0 & 0 & 0 & 0 \\
 0 & 0 & 0 & 0 & 0 & 0 & 0 & \cG_{8} & 0 & 0 & 0 & 0 & 0 & 0 & 0 & 0 \\
 0 & 0 & 0 & 0 & 0 & 0 & 0 & 0 & \cG_{9} & 0 & 0 & 0 & 0 & 0 & 0 & 0 \\
 0 & 0 & 0 & 0 & 0 & 0 & 0 & 0 & 0 &\cG_{10}& 0 & 0 & 0 & 0 & 0 & 0 \\
 0 & 0 & 0 & 0 & 0 & 0 & 0 & 0 & 0 & 0 &\cG_{11} & 0 & 0 & 0 & 0 & 0 \\
 0 & 0 & 0 & 0 & 0 & 0 & 0 & 0 & 0 & 0 & 0 & \cG_{12} & 0 & 0 & 0 & 0 \\
 0 & 0 & 0 & \cG_{18} & 0 & 0 & 0 & 0 & 0 & 0 & 0 & 0 & \cG_{13} & 0 & 0 & 0 \\
 0 & 0 & 0 & 0 & 0 & 0 & 0 & 0 & 0 & 0 & 0 & 0 & 0 & \cG_{14}& 0 & 0 \\
 0 & 0 & 0 & 0 & 0 & 0 & 0 & 0 & 0 & 0 & 0 & 0 & 0 & 0 & \cG_{15} & 0 \\
 e^{4 i \Theta }\,\cG_{17}& 0 & 0 & 0 & 0 & 0 & 0 & 0 & 0 & 0 & 0 & 0 & 0 & 0 & 0 & \cG_{16} \\
 \end{array}
\right], }\nonumber
\end{eqnarray}
with $\{\cG_{i}\}$ denoting a set of polynomials in $|f^{N}_{r,s}(t)|$. In terms of the rank-$1$ projectors onto the computational basis states, we have,
\begin{eqnarray}
 && \small{\widetilde{\cE}_{AD}^{\otimes 4}(P)^{-1/2}}= \sum_{i=1}^{16}\cG_{i}|i\rangle\langle i| + e^{-4 i \Theta } \cG_{17}|0000\rangle\langle1111| \nonumber \\
 && + e^{i 4\Theta} \cG_{17}|1111\rangle\langle0000| + \cG_{18}(|1100\rangle \langle0011| +|0011\rangle \langle1100|). \label{eq:EP_inv}
\end{eqnarray}
Upon sandwiching the operator in Eq.~\eqref{eq:EP_inv} between the different error operators of the four-qubit noise channel (as described in Eqs.~\eqref{eq:E_03},~\eqref{eq:E_02},~\eqref{eq:E_01},~\eqref{eq:E_14}) and their adjoints, it is easy to see that the phases cancel out everywhere. In other words, the Kraus operators of the composite channel comprising noise and recovery are all independent of the phase $\Theta$ of the transition amplitude.

\section{$\Theta$-cancellation in $\left(\cR_{P}^{(4)}\circ\widetilde{\cE}_{AD}^{\otimes 4}\right)$}\label{sec:theta_cancellation}
\begin{theorem}\label{theorem:thetacancellation}
The composite map $\left(\cR_{P}^{(4)}\circ\widetilde{\cE}_{AD}^{\otimes 4}\right)$ with the Kraus operators given in Eq.~\ref{eq:kraus_composite} is independent of $\Theta$, for any choice of the code space $\cC$ in a $4$-qubit Hilbert space. 
\end{theorem}
\begin{proof}
In order to prove this theorem we need to show that the Kraus operators of the composite map $\left(\cR_{P}^{(4)}\circ\widetilde{\cE}_{AD}^{\otimes 4}\right)$ =$\{ \small{P\left(E^{(4)}_{j}\right)^{\dagger}}\widetilde{\cE}_{AD}^{\otimes 4}(P)^{-1/2}E^{(4)}_{i}P\}$ given in Eq.~\ref{eq:kraus_composite} are independent of $\Theta$, where $\{E_{i}^{(4)}\}$ are the Kraus operators of the map $\widetilde{\cE}_{AD}^{\otimes 4}$ and are realized as four-fold tensor products of $E_0$, $E_1$ given in Eq.~\ref{eq:Kraus_ideal}.  In other words, it is enough to show that $ \small{\left(E^{(4)}_{j}\right)^{\dagger}}\widetilde{\cE}_{AD}^{\otimes 4}(P)^{-1/2}E^{(4)}_{i}$\ $\forall$  $\{i,j\}$ is independent of $\Theta$, since the projection $P$ on the code space has no dependence on $\Theta$.

Recall from Eq.~\ref{eq:Kraus_ideal} that the Kraus operators $\{E_0,E_1\}$ of the single-qubit noisy channel $\widetilde{\cE}_{AD}$ are of the form,
\begin{eqnarray} \label {eq:error}
E_{0}&= &|0\rangle \langle 0| + f _{r,s}^{N}(t) |1\rangle \langle 1|  \nonumber \\
 E_{1}&=& |0\rangle \langle 1| \sqrt{1-|f _{r,s}^{N}(t)|^{2}} 
\end{eqnarray}
where $f_{r,s}^{N}(t)$=$|f_{r,s}^{N}(t)| e^{i \Theta}$.
The action of these Kraus operators in Eq.~\ref{eq:Kraus_ideal} and their adjoint on the standard basis, i.e the eigenbasis of Pauli Z operator, can be given as,
\begin{eqnarray}\label{eq:basis1}
E_{0}^{\dagger}|0\rangle &=& |0\rangle \nonumber \\ \nonumber
E_{0}^{\dagger} |1\rangle &=&|f_{r,s}^{N}(t)| e^{-i \Theta}|1\rangle \nonumber \\ \nonumber
E_{1}^{\dagger}|0\rangle&=&  \sqrt{1-|f _{r,s}^{N}(t)|^{2}}  |1\rangle \nonumber \\ 
E_{1}^{\dagger}|1\rangle&=0 \end{eqnarray}
Similarly,
\begin{eqnarray}\label{eq:basis2}
\langle 0 |E_{0}&=&\langle 0| \nonumber \\ \nonumber
\langle 1 |E_{0} &=&|f_{r,s}^{N}(t)| e^{i \Theta} \langle 1| \nonumber \\ \nonumber
\langle 0|E_{1}&=& \sqrt{1-|f _{r,s}^{N}(t)|^{2}} \langle 1| \nonumber \\ 
\langle 1| E_{1}&=0 \end{eqnarray}

Observe that operators of the form $\{ E_{i}^{(4)\dagger}PE_{j}^{(4)}\}$, can always be expanded as a matrix with real entries, in the basis,
\begin{eqnarray}\label{eq:operatorbasis}
&&\{ e^{i (a+b+c+d)\Theta} |a\ b\ c\ d \rangle \langle m\ n\ q\ r| e^{-i (m +n+q+r)\Theta} \}\nonumber \\ \nonumber &&=\{ |0000\rangle \langle 0000|, |0000\rangle\langle 0001|e^{-i \Theta}, |0000\rangle\langle0010|e^{-i \Theta}, \ldots , |1111\rangle\langle 1111| \}
\end{eqnarray}
where $\{a,b,c,d,m,n,q,r\} \in \{0,1\}$, irrespective of the choice for the  codespace $\cC$. This choice of the operator basis is particularly due to the way the $\{E_{i}^{(4)}\}$ act on the standard basis $\{|0000\rangle,|0001\rangle ,\ldots, |1111\rangle\}$. Hence, we can always represent $(\widetilde{\cE}_{AD}^{\otimes 4}(P))^{-1/2}$ in the same operator basis. Note that $E_{i}^{(4)\dagger}(\widetilde{\cE}_{AD}^{\otimes 4}(P))^{-1/2} E_{j}^{(4)}$ can then be expanded in terms of the following elements,
\begin{equation}\label{eq:expand}
E_{i}^{(4)\dagger} e^{i (a+b+c+d)\Theta} |a\ b\ c\ d \rangle \langle m\ n\ q\ r| e^{-i (m +n+q+r)\Theta} E_{j}^{(4)}
\end{equation}
Based on the set of Eqs.~\ref{eq:basis1},~\ref{eq:basis2}, carefully observe that the Kraus operators $\{E_{i}^{(4)}\}$  that contribute to Eq.~\ref{eq:expand} are those that tend to cancel the appropriate $\Theta$ contribution coming from the basis elements in Eq.~\ref{eq:operatorbasis}, thereby leaving the term $E_{i}^{(4)\dagger} (\widetilde{\cE}_{AD}^{\otimes 4}(P))^{-1/2} E_{j}^{(4)}$ independent of $\Theta$.

\noindent Here is an example which illustrates the $\Theta$ cancellation. Consider the case where $a=0, b=0,c=1,d=1,m=0,n=1,q=1,r=1$ in Eq.~\ref{eq:expand} as given below.
\begin{equation}\label{eq:ex1}
E_{i}^{(4)\dagger} e^{i 2 \Theta} |0011 \rangle \langle 0111| e^{-i 3\Theta} E_{j}^{(4)}
\end{equation}
We should necessarily have the combinations of the following form, in order to contribute to Eq.~\ref{eq:ex1}.
\begin{equation}\label{eq:3}
 \left( \_ \ \ \otimes \_ \ \ \otimes E_{0}^{\dagger} \otimes E_{0}^{\dagger} \right) (e^{i 2 \Theta} 0011 \rangle \langle 0111| e^{-i 3\Theta})  \left( \_ \ \ \otimes E_{0} \otimes E_{0}\otimes E_{0}\right) . 
\end{equation}
The blank spaces in Eq.~\ref{eq:3} can be filled with either $E_{0}$ or $E_{1}$, which would act on $|0\rangle$ and hence will not lead to a new $\Theta$-dependant factor, as can be seen in Eq.~\ref{eq:basis1},~\ref{eq:basis2}. Therefore Eq.~\ref{eq:3} is independent of $\Theta$. The same argument holds for each operator element in Eq.~\ref{eq:expand}, thereby leaving  $\{ \small{\left(E^{(4)}_{j}\right)^{\dagger}}\widetilde{\cE}_{AD}^{\otimes 4}(P)^{-1/2}E^{(4)}_{i}\}$ and hence $\left(\cR_{P}^{(4)}\circ\widetilde{\cE}_{AD}^{\otimes 4}\right)$ independent of $\Theta$. 
\end{proof}
\section{$4$-qubit numerical structured code}\label{sec:numerical_code}
We present here the codewords of the $4$-qubit code obtained via our numerical search procedure described in Chapter~\ref{Chapter3}, for the amplitude-damping channel that arises in the context of state transfer over the $1$-d Heisenberg chain.  The performance of the corresponding state transfer protocol is shown in Fig.~\ref{fig:fmin_QEC}.

$|0\rangle_L$=$\begin{pmatrix}
0.2618 - 0.6116\upi \\
0.0000 - 0.0740\upi \\
-0.3098 + 0.0389 \upi \\ 
0.0000 + 0.0786 \upi \\
0  \\
0 \\
0 \\
0 \\
0 \\
0 \\
0  \\
0 \\
0.0960 - 0.0112\upi \\
-0.0000 + 0.1473 \upi\\ 
0.2209 - 0.0093\upi\\
0.0000 + 0.6069\upi \\
\end{pmatrix}$
$|1\rangle_L$=$\begin{pmatrix}
0 \\
0 \\
0 \\
0  \\
0.0000 - 0.1827 \upi \\
0.5116 - 0.2553 \upi\\
-0.0000 - 0.3948 \upi \\
-0.2770 - 0.1778\upi \\
-0.0000 + 0.1036 \upi \\
0.1211 - 0.4826 \upi \\
0.0000 + 0.3189 \upi \\
0.1153 + 0.0473 \upi \\
0 \\
0 \\
0 \\
0 \\
\end{pmatrix}$
\section{Distribution of the transition amplitude for a disordered $XXX$ chain}\label{sec:transAmp_dist}

Here we derive the distribution of the transition amplitude $f^{N}_{r,s}(t,\{\Delta_{k}\})$ for the disordered $XXX$ chain described in Eq.~\eqref{eq:H_dis}, as a function of time $t$ and disorder strength $\delta$. Recall that the transition amplitude between the $r^{\rm th}$ and $s^{\rm th}$ site for the disordered Hamiltonian $\mathcal{H}_{\rm dis}$ is given by,
\begin{equation}
f^{N}_{r,s}(t,\{\Delta_{k}\}) = \langle \textbf{r} | e^{- i (\mathcal{H}_{o}+ \mathcal{H}_{\delta})t} |\textbf{s}\rangle = \langle \textbf{r}| e^{-i \mathcal{H}_{o} t}\mathcal{ T}\left[\exp{\left(-i\int_{0}^{t}e^{i\mathcal{ H}_{o} t'}\,\mathcal{H_{\delta}}\,e^{-i \mathcal{H}_{o} t'}dt'\right)}\right]| \textbf{s}\rangle, \label{eq:transAmp_delta}
\end{equation}
where $\cT$ denotes the time-ordering operator. We first expand the time-ordered perturbation series in Eq.~\eqref{eq:transAmp_delta} as follows, 
\begin{eqnarray}
f^{N}_{r,s}(t,\{\Delta_{k}\}) &=&\sum_{k=1}^{N}\langle \textbf{r}| e^{-i \mathcal{H}_{o} t} |\textbf{k}\rangle \langle \textbf{k}|\mathcal{T}\left[e^{(-i\int_{0}^{t}e^{i\mathcal{ H}_{o} t'}\mathcal{H_{\delta}}e^{-i \mathcal{H}_{o} t'}dt')}\right] |\textbf{s}\rangle  \nonumber \\
&=& \sum_{k=1}^{N}f^{N}_{r,k}(t) \langle \textbf{k} | \, I - i O(\mathcal{H}_{\delta}) + \frac{i^{2}}{2!} O(\mathcal{H}_{\delta}^{2}) + \ldots \, | \textbf{s}\rangle \label{eq:transAmp_disorder}
\end{eqnarray}
where, $f^{N}_{r,k}(t) = \langle \textbf{r}| e^{-i\mathcal{ H}_{o} t} |\textbf{k}\rangle$ is the transition amplitude in the absence of disorder. Expanding the first order term ($O(\mathcal{H}_{\delta})$) as a time-ordered form, we have,
\begin{eqnarray}
\langle \textbf {k} |O(\mathcal{H}_{\delta}) |\textbf{s} \rangle &=& \int_{0}^{t}\langle \textbf{k}|e^{i \mathcal{H}_{o} t'} \mathcal{H}_{\delta}e^{-i\mathcal{ H}_{o} t'}|\textbf{s}\rangle dt' \nonumber \\
&=& \sum_{l,m=1}^{N} \int_{0}^{t}\langle \textbf{k}|e^{i \mathcal{H}_{o} t'} |\textbf{l}\rangle \langle \textbf{l}|\mathcal{H}_{\delta}|\textbf{m}\rangle\langle \textbf{m}|e^{-i\mathcal{ H}_{o} t'}|\textbf{s}\rangle dt' \label{eq:first_order1}
\end{eqnarray}
where,
\begin{equation}
\langle \textbf{l} | \mathcal{H}_{\delta} |\textbf{m}\rangle = \frac{\overline{J}}{2}\left(\sum_{i=1}^{N-1}( u_{i}^l \Delta_{i})\delta_{lm} - 2\Delta_{l}\delta_{m(l+1)}  - 2\Delta_{l-1}\delta_{m(l-1)} \right), \label{eq:first_order2}
\end{equation}
with the coefficients $u^{l}_{i} \in \{\pm 1\}$. For example, $\cH_{\delta}$ for a $4$-qubit spin chain is a tridiagonal matrix of the form,

\begin{equation}
\cH_{\delta} =
\frac{\overline{J}}{2}\left(
\begin{array}{cccc}
 -\Delta_{1}-\Delta_{2}+\Delta_{3} & -2 \Delta_{3} & 0 & 0 \\
 -2 \Delta_{3} & -\Delta_{1}+\Delta_{2}+\Delta_{3} & -2 \Delta_{2} & 0 \\
 0 & -2 \Delta_{2} & \Delta_{1}+\Delta_{2}-\Delta_{3} & -2 \Delta_{1} \\
 0 & 0 & -2 \Delta_{1} &  \Delta_{1}-\Delta_{2}-\Delta_{3} \\
\end{array}
\right) . \nonumber
\end{equation}

Substituting the form of $\cH_{\delta}$ in Eq.~\eqref{eq:first_order2} to the first order term in Eq.~\eqref{eq:transAmp_disorder}, and setting  $\overline{J}=1$ throughout, we get, 
\begin{eqnarray}
f^{N}_{r,s}(t,\{\Delta_{k}\})= f_{r,s}^{N}(t) &-& \frac{i}{2} \int_{0}^{t} \sum_{l,k=1}^{N} f^{N}_{r,k}(t)(f^{N}_{k,l}(t'))^{*}f^{N}_{l,s}(t')(\sum_{i=1}^{N-1}u_{i}^{l}\Delta_{i})dt'\nonumber \\&-&\frac{i}{2} \int_{0}^{t} \sum_{l=1}^{N-1}\sum_{k=1}^{N} f^{N}_{r,k}(t)(f^{N}_{k,l}(t'))^{*}f^{N}_{l+1, s}(t')(-2\Delta_{l})dt' \nonumber \\ 
&-& \frac{i}{2} \int_{0}^{t} \sum_{l=1}^{N-1}\sum_{k=1}^{N} f^{N}_{r,k}(t)(f^{N}_{k,l+1}(t'))^{*}f^{N}_{l,s}(t')(-2\Delta_{l})dt' . \nonumber
 \end{eqnarray}
Thus, up to first order in perturbation, $f^{N}_{r,s}( t,\{\Delta_{k}\})$ is simply a linear combination of the random variables $\{\Delta_{k}\}$, of the form,
\begin{equation}\hspace{-1cm}
 f^{N}_{r,s}(t,\{\Delta_{k}\}) = f^{N}_{r,s}(t)+ \sum_{i=1}^{N-1} c^{N}_{i}(t) \Delta_{i}, \label{eq:transAmp_final2}
 \end{equation}
where $\{c^{N}_{i}(t)\}$ are complex coefficients given by,
\begin{eqnarray}
c^{N}_{i}(t) &=& -\frac{i}{2}\sum_{k=1}^{N}f^{N}_{r,k}(t) \left[\int_{0}^{t}  \sum_{l=1}^{N} u_{i}^{l} (f^{N}_{k,l}(t'))^{*}f^{N}_{l,s}(t')dt'\right]\nonumber \\ &-&\frac{i}{2}\sum_{k=1}^{N}f^{N}_{r,k}(t)\left[2 \int_{0}^{t} ( f^{N}_{k,i}(t'))^{*}f^{N}_{i+1, s}(t')dt' -2 \int_{0}^{t}(f^{N}_{k,i+1}(t'))^{*}f^{N}_{i,s}(t')dt' \right]. \nonumber \\ \label{eq:c-coeff}
\end{eqnarray}

We first note that in the limit of large $N$, the distribution of $f^{N}_{r,s}(t)$ tends towards a normal distribution. This is a direct consequence of the central limit theorem, since $\{\Delta_{i}\}$ are i.i.d random variables. In what follows, we will obtain the exact form of the distribution of $f^{N}_{r,s}(t,\{ \Delta_{k} \})$, specifically, the real and imaginary parts of $f^{N}_{r,s}(t,\{ \Delta_{k} \})$ in terms of $N, t$ and $\delta$. 

Since the $\{\Delta_{i}\}$ are randomly drawn from a uniform distribution between $\left [ -\delta,\delta \right ]$, the joint probability density $P\left(\Delta_{1},\Delta_{2}, \ldots, \Delta_{N} \right)$ is given by,
\begin{equation}
 P \left(\, \Delta_{1},\Delta_{2}, \ldots, \Delta_{N-1} \,\right) = \left\lbrace \begin{array}{cc}
 \frac{1}{(2\delta)^{N-1}}, &  -\delta \leq \Delta_{i} \leq \delta, \; \forall i=1,2,\ldots, N-1. \\
 0, & {\rm otherwise}. 
  \end{array} \right.
\end{equation}
Let $x \equiv \texttt{Re}[f^{N}_{r,s}(t,\{\Delta_{k}\})]$ and $y \equiv \texttt{Im}[f^{N}_{r,s}(t,\{\Delta_{k}\})]$ denote the real and imaginary parts of the transition amplitude in Eq.~\eqref{eq:transAmp_final2}. Then, we may obtain the distribution of $x$ and $y$ as follows:
\begin{eqnarray}
\cP^{\delta,t,N}(x) &=& \int_{\Delta_{1}=-\delta}^{\delta}\ldots\int_{\Delta_{N-1}=-\delta}^{\delta}  \left(\prod_{i=1}^{N-1}d\Delta_{i}\right) P(\Delta_{1},\Delta_{2},\ldots,\Delta_{N-1}) \, \nonumber \\ &&\delta\left( x - ( \, \texttt{Re}[f^{N}_{r,s}(t)] + \sum_{i=1}^{N-1}\texttt{Re}[c^{N}_{i}(t)]\Delta_{i})\right), \nonumber \\
\cP^{\delta,t,N}(y) &=& \int_{\Delta_{1}=-\delta}^{\delta}\ldots\int_{\Delta_{N-1}=-\delta}^{\delta} \left(\prod_{i=1}^{N-1}d\Delta_{i}\right) P(\Delta_{1},\Delta_{2},\ldots,\Delta_{N-1}) \,\nonumber \\ && \delta\left( y - (\, \texttt{Im}[f^{N}_{r,s}(t)] + \sum_{i=1}^{N-1}\texttt{Im}[c^{N}_{i}(t)] \Delta_{i}) \right). \nonumber 
\end{eqnarray}
Replacing the Dirac delta functions with their Fourier transforms, and then integrating out the $\{\Delta_{k}\}$ variables, we get,
\begin{eqnarray}
\cP^{\delta,t,N}(x) &=& \frac{1}{\sqrt{2\pi}(2\delta)^{N-1}}\int_{\Delta_{1}=-\delta}^{\delta}\ldots\int_{\Delta_{N-1}= -\delta}^{\delta} \int_{k=-\infty}^{\infty} \prod_{i=1}^{N-1}d\Delta_{i} dk \nonumber \\ && \exp{\left(-i k \left( x - \left[ \, \texttt{Re}[f^{N}_{r,s}(t)]+\sum_{i=1}^{N-1}\texttt{Re}[c^{N}_{i}(t)]\Delta_{i}\right]\right)\right)}  \nonumber \\
 &=& \frac{1}{\sqrt{2\pi}(2\delta)^{N-1}} \int_{k=-\infty}^{\infty} dk \exp{\left(-i k ( \, x - \texttt{Re}[f^{N}_{r,s}(t)] \,) \right)} \nonumber \\ \nonumber &&\prod_{i=1}^{N-1} \frac{2\sin \left( k\delta\,\texttt{Re}[c^{N}_{i}(t)] \right)}{k\,\texttt{Re}[c^{N}_{i}(t)]}    \label{eq:dist_real1} \\
&=& \frac{1}{\sqrt{2\pi}(2\delta)^{N-1}} \int_{k=-\infty}^{\infty} dk \exp{\left(-i k ( \, x - \texttt{Re}[f^{N}_{r,s}(t)] \,) \right)} \nonumber \\ && \prod_{i=1}^{N-1} \frac{e^{i \left( k\delta\,\texttt{Re}[c^{N}_{i}(t)] \right)}-e^{-i \left( k\delta\,\texttt{Re}[c^{N}_{i}(t)] \right)}}{i k\,\texttt{Re}[c^{N}_{i}(t)]} \nonumber  \\ \nonumber
&=& \frac{1}{\sqrt{2\pi}(2\delta)^{N-1}} \int_{k=-\infty}^{\infty} dk \exp{\left(-i k ( \, x - \texttt{Re}[f^{N}_{r,s}(t)] \,) \right)}\nonumber \\ &&  \frac{\sum_{j=1}^{2^{N-1}} (-1)^{\alpha_{j}}e^{i (k\,\delta\sum_{i=1}^{N-1} (-1)^{r_{i}^{j}}\texttt{Re}[c^{N}_{i}(t)])}}{(i k)^{N-1}\,\prod_{i=1}^{N-1}\texttt{Re}[c^{N}_{i}(t)]}, \nonumber  
\end{eqnarray}
where, $\alpha_{j}, r_{i}^{j} \in [0,1], \; \forall i,j$. Simplifying further, we get,
\begin{eqnarray}
\cP^{\delta,t,N}(x) &=&
\frac{1}{\sqrt{2\pi}(2\delta)^{N-1}\prod_{i=1}^{N-1}\texttt{Re}[c^{N}_{i}(t)]}\nonumber \\ && \int_{k=-\infty}^{\infty} dk \frac{\sum_{j=1}^{2^{N-1}} (-1)^{\alpha_{j}}\exp{\left(-i k ( \, x - \texttt{Re}[f^{N}_{r,s}(t)] + \delta \sum_{i=1}^{N-1}(-1)^{r^{j}_{i}}\texttt{Re}[c^{N}_{i}(t)] ) \right)}}{{(ik)}^{N-1}} \nonumber \\ 
& =& \left(\frac{1}{(2\delta)^{N-1}}\right)\left(\frac{1}{\prod_{i=1}^{N-1}\texttt{Re}[c^{N}_{i}(t)]}\right) \sum_{j=1}^{2^{N-1}} (-1)^{u_{j}}(q_{j})^{N-2}\,{\rm Sign}[q_{j}], \label{eq:dist_real2}
\end{eqnarray}
where $u_{j} \in [0,1]$, and $q_{j} (x,\texttt{Re}[f^{N}_{r,s}(t)], \{\texttt{Re}[c^{N}_{i}(t)]\})$ are linear combinations of the form,
\begin{equation}
q_{j} \equiv x - \texttt{Re}[f^N_{r,s}(t)] + \delta\sum_{i=1}^{N-1} (-1)^{r_{i}^{j}}\texttt{Re}[c^{N}_{i}(t)]  , \; r_{i}^{j}\in [0,1], \; \forall i=1,\ldots, N-1. 
\end{equation} 
We may evaluate the distribution of the imaginary part of the transition amplitude in a similar fashion, to get,
\begin{equation}
\cP^{\delta,t,N}(y) = \left(\frac{1}{(2\delta)^{N-1}}\right)\left(\frac{1}{\prod_{i=1}^{N-1}\texttt{Im}[c^{N}_{i}(t)]}\right) \sum_{i=1}^{2^{N-1}} (-1)^{u_{j}}(\tilde{q}_{j})^{N-2}\,{\rm Sign}[\tilde{q}_{j}], \label{eq:dist_im2}
\end{equation}
where the $\tilde{q}_{j} (x,\texttt{Im}[f^{N}_{r,s}(t)], \{\texttt{Im}[c^{N}_{i}(t)]\})$ are linear combinations of the form,
\begin{equation}
\tilde{q}_{j} \equiv y - \texttt{Im}[f^{N}_{r,s}(t)] + \delta\sum_{i=1}^{N-1} (-1)^{r_{i}^{j}}\texttt{Im}[c^{N}_{i}(t)]  , \; r_{i}^{j}\in [0,1], \; \forall i=1,\ldots, N-1. 
\end{equation}
We see from Eq.~\eqref{eq:dist_real1} that the limiting distribution in the case of no disorder ($\delta \rightarrow 0$), is indeed a delta distribution peaked around $\texttt{Re}[f^{N}_{r,s}(t)]$:
\begin{equation}
\lim_{\delta\rightarrow 0}\cP^{\delta,t,N}(x) = \frac{1}{\sqrt{2\pi}}\int_{k=-\infty}^{\infty} dk \exp{\left(-i k ( \, x - \texttt{Re}[f^{N}_{r,s}(t)] \,) \right)} = \delta\left( x - \texttt{Re}[f^{N}_{r,s}(t)]\right).
\end{equation}
\begin{figure}  [H]
\centering
\begin{subfigure}{
\includegraphics[width=.75\textwidth]{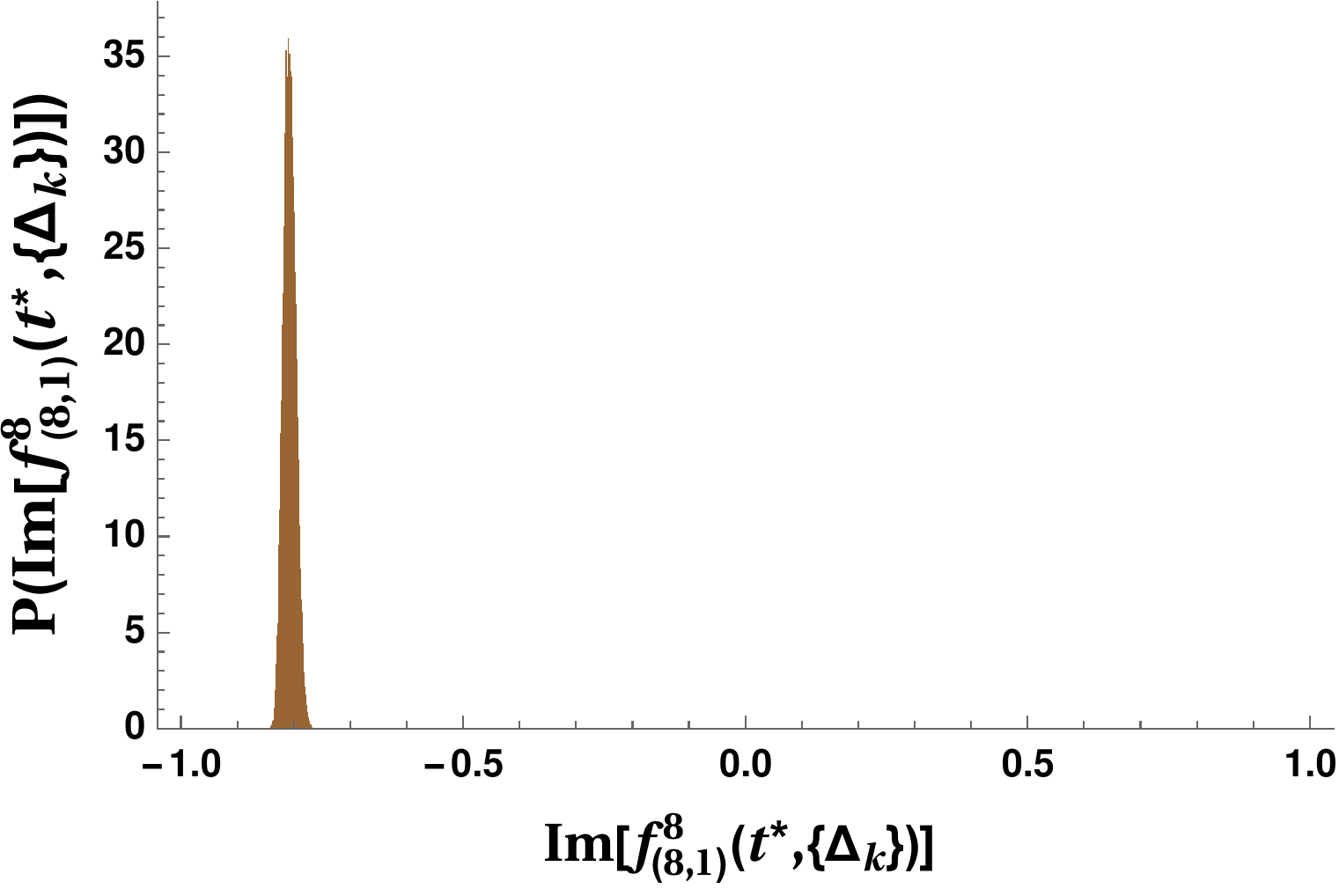}
}
\end{subfigure}
\begin{subfigure}{
\includegraphics[width=.75\textwidth]{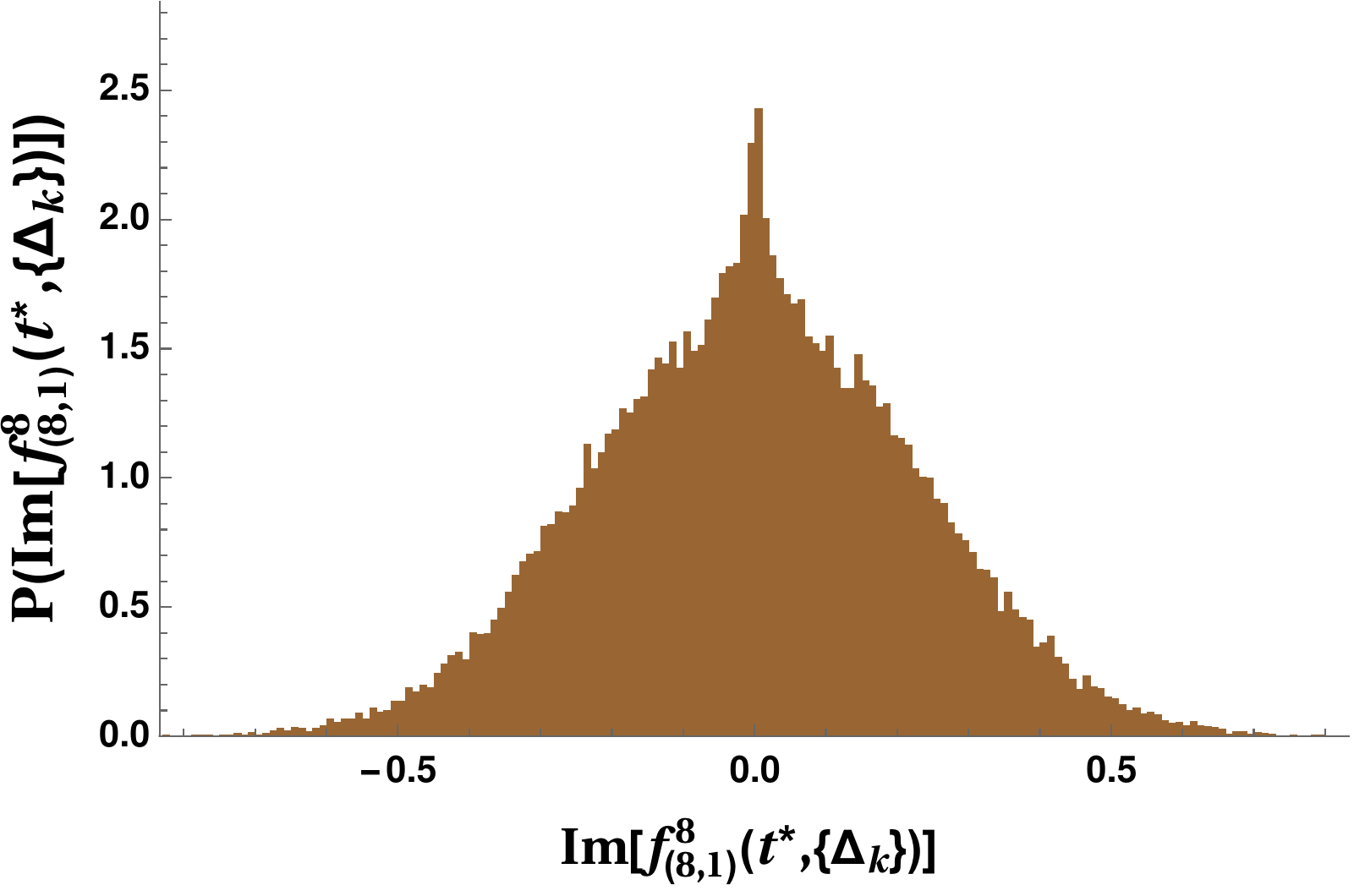}
}
\end{subfigure}
\caption{Distribution of $\texttt{Im}[f_{8,1}^{8}(t^{*}, \{\Delta_{k}\})]$ 
over different disorder realizations drawn from a uniform distribution with disorder strengths $\delta = 0.001$ and $\delta=1$, respectively.}
\label{fig:f_dist_Im}
\end{figure}
Finally, we compute the disorder-averaged value of the transition amplitude 
upto $O(\cH_{\delta}^{2})$. We first modify the expression in Eq.~\eqref{eq:transAmp_final2} to include the second-order perturbation terms:
\begin{equation}
 f^{N}_{r,s}(t,\{\Delta_{k}\}) = f_{r,s}^{N}(t)+ \sum_{i=1}^{N-1}c^{N}_{i}(t) \Delta_{i} + \sum_{i,j=1}^{N-1}d^{N}_{ij} \Delta_{i}\Delta_{j} + \ldots , \label{eq:transAmp_final3}
 \end{equation}
where $\{d^{N}_{ij}\}$ are complex coefficients which are convolutions of the zero-disorder transition amplitude, similar to $\{c^{N}_{i}(t)\}$. Next, using the fact that the random couplings $\{\Delta_{i}\}$ are drawn from a uniform distribution, we obtain,
\begin{eqnarray}\label{eq:trans_ampAvg}
\langle f^{N}_{r,s}(t,\{\Delta_{k}\}) \rangle_{\delta}&=&\frac{1}{(2\delta)^{N-1}}\int_{-\delta}^{\delta} \left( f^{N}_{r,s}(t) + \sum_{i=1}^{N-1}c^{N}_{i}(t)\Delta_{i} +\sum_{l,m=1}^{N-1} d^{N}_{lm}(t) \Delta_{l}\Delta_{m}+\ldots \right) \prod_{i=1}^{N-1}d\Delta_{i} \nonumber \\ 
&=& f^{N}_{r,s}(t)+ \frac{\delta^{2}}{3}\sum_{i}d^{N}_{ii}(t) + O(\delta^{4}) .
\end{eqnarray}
The second moment of $f^{N}_{r,s}(t,\{\Delta_{k}\})$,
\begin{eqnarray}\label{eq:2ndmoment}
\left\langle \, (f^{N}_{r,s}(t,\{\Delta_{k}\}))^{2} \, \right\rangle_{\delta}&=&\frac{1}{(2\delta)^{N-1}} \int_{-\delta}^{\delta} \left( f^{N}_{r,s}(t) + \sum_{i=1}^{N-1}c^{N}_{i}(t)\Delta_{i} +\sum_{l,m=1}^{N-1} d^{N}_{lm}(t) \Delta_{l}\Delta_{m})+\ldots \right)^{2}  \nonumber \\ &&\prod_{i=1}^{N-1}d\Delta_{i} \nonumber \\ 
&=& (f^{N}_{r,s}(t))^{2} + \frac{\delta^{2}}{3}\left( 2 f^{N}_{r,s}(t) \sum_{l=1}^{N-1}d^{N}_{ll}(t)+\sum_{j=1}^{N-1}(c^{N}_{j}(t))^{2}\right) \nonumber \\
&&  + \frac{\delta^{4}}{5}\sum_{l=1}^{N-1}(d^{N}_{ll}(t))^{2}+ \frac{\delta^{4}}{9}\sum_{l \neq m=1}^{N-1}(d^{N}_{lm}(t))^{2} + O(\delta^{6}).
\end{eqnarray}
We can now calculate the variance from Eq~\ref{eq:trans_ampAvg} and Eq~\ref{eq:2ndmoment} as follows:
 \begin{eqnarray}\label{eq:variance}
{\rm Var}[f^{N}_{r,s}(t,\{\Delta_{k}\})] &=& \langle \, (f^{N}_{r,s}(t,\{\Delta_{k}\}))^{2} \, \rangle_{\delta} - \langle f^{N}_{r,s}(t,\{\Delta_{k}\}) \rangle^{2}_{\delta} \nonumber \\
&=& \frac{\delta^{2}}{3}\sum_{j=1}^{N-1}(c^{N}_{j}(t))^{2}+\nonumber \\ &&\delta^{4}\left(\frac{1}{5}\sum_{l=1}^{N-1}(d^{N}_{ll}(t))^{2}+ \frac{1}{9}\sum_{l \neq m=1}^{N-1}(d^{N}_{lm}(t))^{2}-\frac{1}{9}\left(\sum_{l=1}^{N-1}d^{N}_{ll}(t)\right)^{2} \right)  \nonumber \\ &&+ O(\delta^{6}). 
 \end{eqnarray}
To summarize, from Eq.~\eqref{eq:trans_ampAvg} we see that as  $\delta \rightarrow 0$, $\langle f^{N}_{r,s}(t,\{\Delta_{k}\}) \rangle_{\delta}$  approaches the zero-disorder value $f^{N}_{r,s}(t)$. As expected, the variance given in Eq.~\eqref{eq:variance} vanishes in this limit. However as the disorder strength $\delta$ increases, $\langle f^{N}_{r,s}(t,\{\Delta_{k}\}) \rangle_{\delta}$ deviates from the no-disorder case, and the variance also starts growing since terms of $O(\delta^{2})$ become increasingly significant now. 
 

%% file: appendixC_v3.tex
\chapter{Threshold calculation using $[[4,1]]$ code} 
\label{AppendixC} 

\lhead{Appendix sA. \emph{Title of the Appendix in short}} 
\section{Fault tolerance proofs}\label{app:FT}
Here, we present the detailed proofs of the fault tolerance of the various basic encoded gadgets presented in Sec.~\ref{sec:enc_gadget} of the main text.
\subsection{\textsc{ec} unit}\label{app:ECFT}

\begin{lemma}[Properties of the ideal \textsc{ec} unit]\label{lem:ec_prop1}
If the \textsc{ec} unit described in Fig.~\ref{fig:ec} has no fault, it takes an input with at most one damping error to an output with no errors.
 \end{lemma}
\begin{proof}
Our error-correction gadget in Fig.~\ref{fig:ec} has been designed such that the no-damping error and the single-qubit damping errors point to a unique syndrome bit string $\{s, t, u, h\}$ or $\{s,t,v,g\}$, as tabulated in Tab.~\ref{tab:syndrome_bits}. 

Finally, the recovery unit $\cR$ in Fig.~\ref{fig:r_ec} maps the input state back to a state in the code subspace. Thus our error correction unit without any fault maps an input with upto single-qubit damping errors into the codespace perfectly. 
\end{proof}

We next show that a faulty \textsc{ec} unit does not lead to any errors beyond the single-qubit damping errors. 

\begin{lemma}[Properties of a faulty \textsc{ec} unit]\label{lem:prop2}
If the \textsc{ec} unit in Fig.~\ref{fig:ec} contains at most one fault, it takes an input state with no errors to a state with at most a single damping error.
\end{lemma}
\begin{figure}[t!]
\centering
\includegraphics[scale=.4]{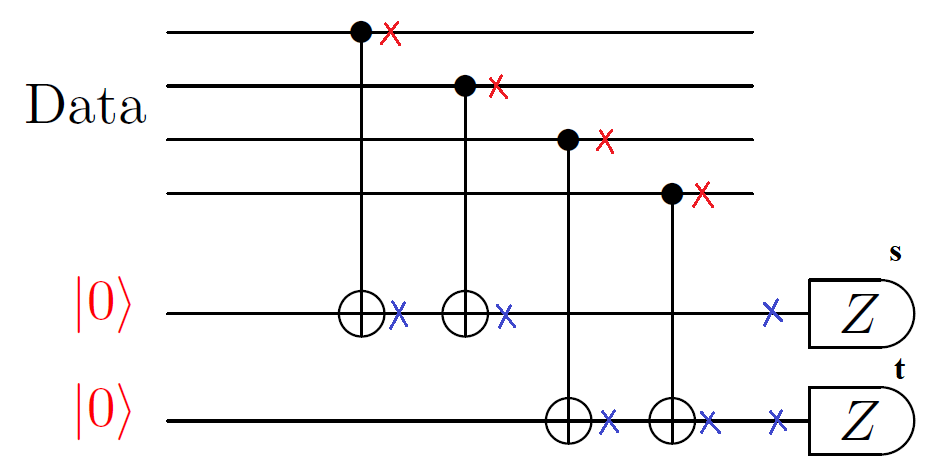}
 \caption{A faulty parity measurement unit, where the faults are marked in red for the data qubits and blue for the ancilla qubits}
 \label{fig:faulty_syndrome}
 \end{figure}
\begin{proof}
Our \textsc{ec} unit consists of a syndrome measurement unit followed by a recovery unit as given in Fig.~\ref{fig:ec}. Unless atleast one of the syndrome bits $s$ or $t$ record an odd parity, the subsequent units are not triggered. Consider a single fault in the parity measurement part of the syndrome extraction unit given in Fig.~\ref{fig:faulty_syndrome}.
\begin{itemize}
\item A faulty \textsc{cnot} in the control at location $1$ in Fig.~\ref{fig:faulty_syndrome} will take an incoming state without any error to an output with single-qubit damping error as given below.
\begin{equation}
a|0\rangle_{L} +b|1\rangle_{L} \rightarrow a| 0111\rangle + b|0100\rangle
\end{equation}
where $\{|0\rangle_L, |1\rangle_L\} \in \cC$,  and $|a|^2 +|b|^2 = 1$. However, one may not correct for this error because syndrome bit $s=0$ in this case.
\item  A faulty \textsc{cnot} in the target at location $2$ in Fig.~\ref{fig:faulty_syndrome} will take an incoming state without any error to an output with partial superposition as given below
\begin{equation}
a|0\rangle_{L} +b|1\rangle_{L} \rightarrow a| 1111\rangle + b|1100\rangle
\end{equation}
In this case one obtains an odd parity outcome where $s=1$, following which one obtains $u=1$, $h=1$. Finally, one passes the state through the recovery unit $\cR$ in Fig.~\ref{fig:r_ec} to rebuild the superposition. 
\item A faulty \textsc{cnot} both in the control (location $1$) and target (location $2$) will take an incoming state without any error to an output with single-qubit damping error as given below.
\begin{equation}
a|0\rangle_{L} +b|1\rangle_{L} \rightarrow a| 0111\rangle + b|0100\rangle .
\end{equation}
In this case we have $s=1$, therefore subsequent units are triggered recording an outcome $u=0$, $h=1$. One then passes through the recovery to correct for the error. Note that we extract additional syndrome bits $h, g$ in order to distinguish between the two instances of a faulty \textsc{cnot} in Fig.~\ref{fig:faulty_syndrome}: (a) one with the fault both on control (location $1$) and target (location $2$), leading to a damping error at the output and (b) a faulty \textsc{cnot} at the target (location $2$) in Fig.~\ref{fig:faulty_syndrome}, leading to a state with partial superposition on the codespace.

\item A faulty $Z$ measurement on the ancilla will not introduce any error to the already error-free incoming state. This is because the ancilla qubit would already be in $|0\rangle$ just before the $Z$ measurement for an error-free incoming state. 
\end{itemize}

We note that similarly, the other faulty locations in Fig.~\ref{fig:faulty_syndrome} will lead to at most one single damping error at the output. For an incoming state without any error, the parity check units in Fig.~\ref{fig:ec} and the recovery unit in Fig.~\ref{fig:r_ec} are triggered only when there is atleast one fault in the parity measurement unit in Fig.~\ref{fig:faulty_syndrome}. Therefore, the \textsc{ec} unit with a single fault propagates at most single damping error at the output for an incoming state without any error.

\end{proof}


\subsection{Bell-state preparation unit}\label{app:BellFT}

\begin{lemma} [Properties of the Bell-state preparation unit]\label{lem:prep}
If the preparation unit in Fig.~\ref{fig:plus} has no faults, it propagates an input with upto one single-qubit damping error to an output with at most one single-qubit damping error. A faulty unit in Fig.~\ref{fig:plus}, with at most single damping fault propagates an incoming state with no error to an output with at most one single-qubit damping error. 
\end{lemma}
\begin{proof}
We now explain how the circuits in Fig.~\ref{fig:b_00} and Fig.~\ref{fig:plus} lead to a fault-tolerant preparation of the Bell state $|\beta_{00}\rangle$.
\begin{itemize}
\item If any of the Hadamard gates turns out to be faulty while preparing two copies of the Bell state $|\beta_{00}\rangle$ in Fig.~\ref{fig:b_00}, we obtain an odd parity outcome in the $X$ measurements in the verifier, in which case we reject the prepared state.
\item A fault in any other location in the verifier in Fig.~\ref{fig:plus} will show up as an odd parity outcome during the $Z$ measurement, in which case we reject the prepared Bell state. 
\item When there are no faults at all, both $X$ measurements as well as $Z$ measurements show an even parity outcome, in which case we accept the prepared Bell state. 
\end{itemize}
\end{proof}


\subsection{Logical $X$ measurement}\label{app:X_meas}
We next prove that the $\overline{X}$ measurement unit satisfies the following fault tolerance properties.
\begin{lemma} [Properties of the $\overline{X}$ measurement unit]\label{lem:meas}
If the $\overline{X}$ measurement unit in Fig.~\ref{fig:xL} has no faults, it leads to a correctable classical outcome when it measures an incoming state with at most one single-qubit damping error. When a faulty $\overline{X}$ measurement unit, with at most a single damping fault measures an incoming state with no error, it leads to a classical outcome that is correctable. 
\end{lemma}
\begin{proof}
Note that after the action of the two \textsc{cnot}s between the ancilla qubits at the top and the four (data) qubits in Fig.~\ref{fig:xL}, the joint state of the system is given by,
\begin{eqnarray} \label{eq:xmeasurement3}
&& |\beta_{00}^{d}\rangle^{\otimes3}(a|\beta_{00}^{\rm anc}\rangle^{\otimes 2} + b |\beta_{10}^{\rm anc}\rangle^{\otimes 2} ) \nonumber \\
&& + |\beta_{10}^{d}\rangle^{\otimes 3}(a|\beta_{00}^{\rm anc}\rangle^{\otimes 2} - b |\beta_{10}^{\rm anc}\rangle^{\otimes 2} ). 
\end{eqnarray}
Amongst the three copies of Bell pairs, one copy refers to the state of the ancilla qubits at the top and the other two refer to the state of the remaining four data qubits. We now discuss all possible single fault scenarios below.
\begin{itemize}
\item When there is no fault anywhere, the ($X$, $Z$) measurements in Fig.~\ref{fig:xL} lead to (a) three copies of the  outcome $(0,0)$ while projecting the four data qubits onto the state $a|+\rangle_L +b |-\rangle_L$, or, (b) three copies of $(1,0)$, in which case the data qubits are projected onto the state $a|+\rangle_L - b |-\rangle_L$. Thus, in the ideal case, our circuit in Fig.~\ref{fig:xL} realises a logical $\overline {X}$ measurement. Note that the measurement outcomes $0$ and $1$ correspond to $+1$ and $-1$ eigenstates, respectively, of the measurement operators in Fig.~\ref{fig:xL}.
\item Observe that whenever there is a fault anywhere within the unit in Fig.~\ref{fig:xL}, or an error in the incoming state, the ($X$, $Z$) measurements at the end in Fig.~\ref{fig:xL} will lead to outcomes of the form $(0,1)$, $(1,0)$ or $(1,1)$. This is due to the realisation of a different Bell pair, namely, $|\beta_{01}\rangle$ or $|\beta_{11}\rangle$, as a result of the fault. However, we note that this can still be detected and corrected. 
\end{itemize}
Thus the measurement of logical operator $\overline{X}$ in Fig.~\ref{fig:xL} is tolerant upto single-qubit damping errors thereby satisfying the desired properties.
\end{proof}

\subsection{Logical $X$ operation}\label{app:xlogical}
\begin{lemma} [Properties of the $\overline{X}$ gadget]\label{lem:xbar1}
A logical $\overline{X}$ gadget without any faults, propagates an input with at most one single-qubit damping error to an output with at most one single-qubit damping error. A logical $\overline{X}$ gadget with at most a single damping fault propagates an incoming state with no error to an output with at most one single-qubit damping error. 
\end{lemma}
\begin{proof}
We demonstrate how the logical $X$ gadget described in Fig.~\ref{fig:xgadget} is tolerant against single-qubit damping errors. We need six ancillas initialized to $|0\rangle$ to establish the action of the logical $X$ operation on the incoming encoded state $|\Psi\rangle$. The joint state of the system, which is any arbitrary single qubit state in the codespace $|\psi\rangle$ and two ancillas initialized to $|0\rangle$ is given by
\begin{eqnarray}
|\psi\rangle |00\rangle=  (a|0\rangle_L+b|1\rangle_L)|00\rangle \nonumber
\end{eqnarray}
where $a,b$ define any arbitrary single qubit state, satisfying $|a|^2 +|b|^2 = 1$ and $\{|0\rangle_L,|1\rangle_L\}$ $\in$ $\cC$. Note that the first four qubits represent the data qubits and the last two qubits represent the ancillas. The joint state of the data  qubits $|\Psi\rangle$  and the ancillas after the action of two \textsc{cnot}s is given by
 \begin{eqnarray}
&& a(|0000\rangle{|00\rangle}+ |1111\rangle{|11\rangle}) \nonumber \\ 
&& + b(|0011\rangle{|11\rangle} +|1100\rangle{|00\rangle}) /\sqrt{2} \nonumber
 \end{eqnarray} 
A single fault at this stage can be detected by the parity check  bits $s,t$ and corrected by the recovery unit $\cR_{Z}$ in Fig.~\ref{fig:xgadget}. Specifically, we may consider the following faulty cases that may occur until this stage:
\begin{itemize}
\item A faulty \textsc{cnot} (at both control and target) leads to
\[a|1101\rangle|01\rangle + b|0001\rangle|01\rangle \]
\item A faulty \textsc{cnot} (at control) leads to
\[a|1101\rangle|11\rangle + b|0001\rangle|11\rangle \]
\item A faulty \textsc{cnot} (at target) leads to
\[a|1111\rangle|01\rangle + b|0011\rangle|01\rangle \]
\item An incoming state with an error
\[(a|1101\rangle+b|0001\rangle)|01\rangle\]
\end{itemize}
The parity measurements on the last two data qubits (leading to outcome $s$) and on the two ancilla qubits (leading to parity check bit $t$) trigger another pair of parity checks, denoted by the bits $u$ and $h$, as shown in Fig.~\ref{fig:xgadget}. Depending on the outcomes $s, t, u, h$, we apply local $X$ gates to correct the system state and then pass through the recovery unit $\cR_{Z}$ (described earlier in Sec.~\ref{sec:ec_unit}) to obtain the state $a|1\rangle_L+ b|0\rangle_L$.

In case there is no fault detected by the $s,t$ parity checks, we proceed to the \textsc{ec}' unit in Fig.~\ref{fig:ecprime}. This unit applies four local $X$ gates (two each on the data qubits and the ancilla qubits) to implement the logical $\overline X$ operation. This is followed by another parity check unit, whose outcomes $x,y$ trigger an \textsc{ec} unit, thus ensuring that that the data qubits are finally in the state $a|1\rangle_L+ b|0\rangle_L$. 

Thus the logical $\overline{X}$ gadget is tolerant upto single-qubit damping errors, thereby satisfying the desired fault tolerance properties.
\end{proof}

\subsection{Logical \textsc{cphase} operation}\label{app:cz}
\begin{lemma} [Properties of the \textsc{cphase} gadget]\label{lem:cphase}
If the \textsc{cphase} gadget in Fig.~\ref{fig:cphase} does not have any faults, it propagates an input with upto one single-qubit damping error to an output with at most one single-qubit damping error. A faulty \textsc{cphase} gadget in Fig.~\ref{fig:cphase}, with at most one single damping fault, propagates an incoming state with no errors to an output with upto one single-qubit damping error.  
\end{lemma}
\begin{proof} 
Consider the case of a damping error propagating through a two-qubit \textsc{cphase} gate, as shown in Fig.~\ref{fig:cphase_analysis}.

\begin{figure}[H]
\centering
\includegraphics[scale=.65]{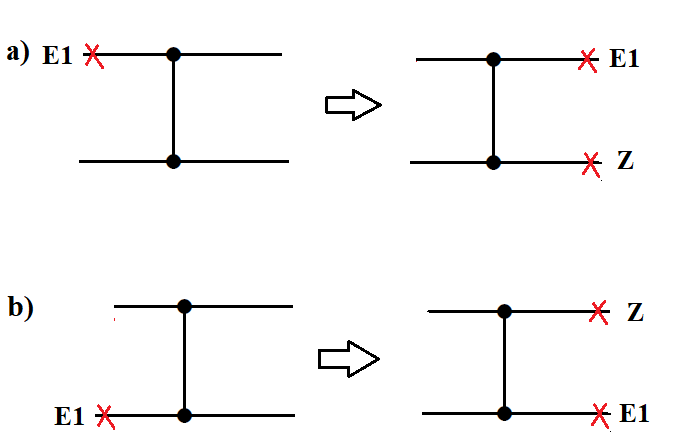}
\caption{a) Propagation of a damping error $E_1$ before the control (top) and b) target (bottom) of a \textsc{cphase} gate. }
\label{fig:cphase_analysis}
\end{figure}

Note that in Fig.~\ref{fig:cphase_analysis}, an incoming damping error right before the control as well as target propagates as a damping error $E_1$ and a phase error $Z$, both of which are correctable by the \textsc{ec} unit, as explained in Sec.~\ref{sec:ec_unit}. Furthermore, the additional $Z$ error in one of the blocks due to a damping error in the other block can be corrected by the $\cR_{Z}$ units at the of the \textsc{cphase} gadget. Hence, our transversal \textsc{cphase} gadget satisfies the desired fault tolerance properties. 
\end{proof}

\section{Pseudothreshold Calculation}\label{app:thresh_calc}

We describe here the details of our pseudothreshold calculation, for the memory unit and the extended \textsc{cz} unit. We assume that the inputs to the units do not have any errors and explicitly count the total number of malignant faults of $O(p^2)$ which will cause a given unit to fail. These malignant faults include phase faults $Z$ and two-qubit damping errors.

\subsection{Total number of locations within a \textsc{EC} unit}
We first count the total number of locations within our \textsc{ec}-unit, which is redrawn here for clarity.

\begin{figure}[H]	
\centering
	\includegraphics[scale=0.55]{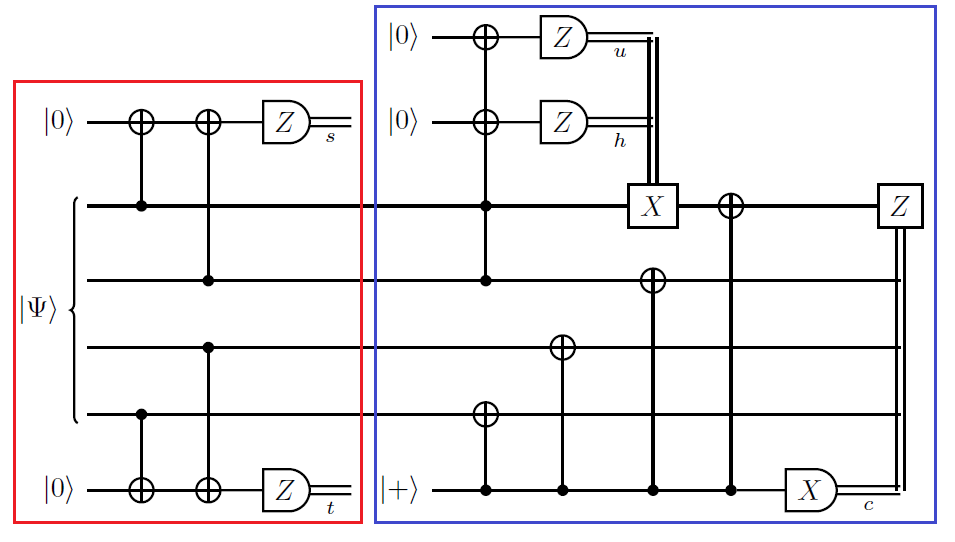}
		\caption{\textsc{ec} unit with parity measurements (red box to the left, we call this part $1$), the rest of the unit involving remaining syndrome extraction and recovery unit $\mathcal{R}$ (blue box to the right, call this part $2$)}
\label{fig:ec_threshold}		
\end{figure}
There are $10$ \textsc{cnot}s, $4$ $Z$ measurements, $1$ $X$ measurement, $1$ preparation of $\ket{+}$, $1$ $X$, $1$ $Z$, and $24$ rest locations adding to a total of $42$ locations. There are $14$ locations in part $1$ and $28$ locations in part $2$. 

Note that the syndrome extraction unit in part $2$ could apply to any pair of qubits depending on the outcome $s$ and $t$. Similarly the $Z$ gate in part $2$ could be applied to any of the four data qubits depending on the syndrome bits and the outcome $c$, although we present a specific case of damping and $Z$ application in Fig.~\ref{fig:ec_threshold}. We explain our calculation with one specific case of damping as depicted in Fig.~\ref{fig:ec_threshold} since the threshold calculation does not change with the other cases of damping. Furthermore, counting of fault locations is easier to explain with a specific case of damping.

\subsection{Memory Pseudothreshold}\label{app:memory}

We now calculate the pseudothreshold for the memory unit in Fig.~\ref{fig:memory}. We first count the malignant fault pairs due to two damping faults (assuming a no-error input to the memory unit) leading to an output that is uncorrectable. This could happen in one of three possible ways: (1) two damping faults could occur within the \textsc{ec} units, (2) both faults could occur within the resting qubits, or, (3) one fault in the \textsc{ec} unit and one fault in the resting qubits. We now enumerate the number of malignant pairs in each of the three cases. 

\begin{enumerate}
\item {\bf Malignant pairs within  a \textsc{ec} unit:} We count the total number of malignant pairs within a \textsc{ec} block. There are totally $14$ locations in the parity measurements part of a \textsc{ec} unit in Fig.~\ref{fig:ec_threshold}. Fault pairs in the same data qubit lines are not malignant because this will not lead to more than a single damping error at the output. Therefore we have $3 \times 4=12$ such benign pairs, where there are $3$ locations in each data qubit line in part $1$ of Fig.~\ref{fig:ec_threshold}. A fault anywhere in the data qubit block and a faulty $Z$ measurement in the ancilla block is not malignant, leading to $12\times 2=24$ benign pairs. Furthermore, a pair of faulty $Z$ measurements in the ancilla block is not malignant. Thus we have totally $C(14,2)- 37=54$ malignant fault pairs. 

Consider the case where there is a single fault in the parity measurements in the \textsc{ec} unit and another fault in the rest of the \textsc{ec} unit (part $2$) in Fig.~\ref{fig:ec_threshold}. To trigger part 2 in Fig.~\ref{fig:ec_threshold}, at least one of the parity outcomes must be $1$, this includes $4$ locations in the first time step on the data line in the parity measurements in part $1$ in Fig.~\ref{fig:ec_threshold}. Then, a second fault in part $2$ in Fig.~\ref{fig:ec_threshold} which will combine with the first fault in part $1$ forming a malignant pair includes all except the $7$ rest locations  during the measurement of $X$  and application of $Z$ gate, and one potential application of $Z$ gate in part $2$ in Fig.~\ref{fig:ec_threshold}, leading to a total of $4 \times (28-8) = 80$ malignant pairs. 

Denoting the number of malignant pairs within each \textsc{ec} block as $N_{\textsc{ec}}$, we thus have, 
\begin{equation}
N_{\textsc{ec}} = 54+80 = 134 . \label{eq:ec_pairs}
\end{equation} 
However, since a memory unit has two \textsc{ec} units, the total number of malignant pair contributions from both the units leads to $134\times 2= 268$ possibilities.

\item {\bf Malignant pairs in the resting qubits:} There are $4$ locations when the qubits rest or there are $4$ applications of Identity gates leading to $C(4,2)=6$ malignant pairs.

\item { \bf $1$ fault in an \textsc{ec}-unit and $1$ fault in the resting qubits:} 
There are $10$ bad locations that cause an undetected error in the output of part $1$ of a \textsc{ec} unit in Fig.~\ref{fig:ec_threshold}. This includes $3$ locations ($1$\textsc{cnot}, $2$ rests) along the $1$st and $4$th data qubit lines each, $2$ ($1$ \textsc{cnot}, $1$ rest) along the 2nd and 3rd data qubit lines each, $2$ for 3rd qubit and $3$ for 4th qubit. These fault locations in combination with a fault during rest could propagate two errors at the output of the memory unit. A rest fault along the same data qubit line as a fault in the \textsc{ec} unit will not lead to any further damping, thereby propagating an output with at most one single damping error. Therefore, each fault location in a \textsc{ec} unit combines with $3$ bad rest locations between the \textsc{ec} units to give $10 \times 3=30$ malignant pairs. There are two such cases: (a) a single fault in the preceding \textsc{ec} unit and  a single fault in the intermediate rest locations, and, (b) a single fault in the intermediate rest locations and a single fault in the succeeding \textsc{ec} unit leading to a total of $30 \times 2=60$ locations.
\end{enumerate}

Finally, we count the malignant faults leading to $Z$ errors in the memory unit in Fig.~\ref{fig:memory}. This includes only the locations from part $1$ of the two \textsc{ec} units except the two $Z$ measurements on the ancilla leading to $14-2=12$ bad locations and $4$ intermediate rest locations, adding up to $(12 \times 2) +4= 28$ locations.

Therefore, the total number of malignant pairs due to malignant pairs of damping errors and malignant fault locations due to $Z$ errors is given by, $A = 356$, leading to a threshold of $p_{\rm th} = 1/ A \approx 2.8 \times 10^{-3}$.

\subsection{Pseudothreshold for the extended \textsc{cphase} unit}\label{app:cphase_threshold}

We first note that the total number of locations in the extended \textsc{cz} unit in Fig.~\ref{fig:exrec} is $42 (\textsc{ec}) \times 4 + 4 (\textsc{cz} \  \rm gadget) = 172$. 

\subsubsection{Malignant fault pairs due to damping errors}
We first count the malignant pairs of damping faults leading to an output that is not correctable in the extended \textsc{cz} unit given in Fig.~\ref{fig:exrec}, assuming a no-error input.
We present a matrix whose rows and columns correspond to each block in Fig.~\ref{fig:exrec}. The entries of the matrix are populated with the total malignant pair contributions from the respective blocks which are labelled as follows.
\begin{enumerate}
	\item \textsc{ec$_1$} - leading \textsc{ec} in block 1
	\item \textsc{ec$_2$} - leading \textsc{ec} in block 2
	\item \textsc{cz} gadget
	\item \textsc{ec$_a$} + $\cR_{Za}$
	\item \textsc{ec$_b$} + $\cR_{Zb}$
\end{enumerate}
\[
\begin{blockarray}{cccccccc}
	 & 1 & 2 & 3 & 4 & 5 \\
	\begin{block}{c(ccccccc)}
		1 & 0 \\
		2 & 74 & 0 \\
		3 & 30 & 30 &  6  \\
		4 & 274 & 240 & 110 & 134 \\
		5 & 240 & 274 & 110 & 224 & 134 \\
	\end{block}	
\end{blockarray}
\]
We now explain in detail, as to how we arrive at the entries of the matrix by listing the malignant pairs for each of the above combinations.

\begin{enumerate}
\item {\bf Malignant pairs within an \textsc{ec} block: } 
We recall from Eq.~\eqref{eq:ec_pairs} that there are $134$ malignant pairs within an \textsc{ec} unit. There are two such \textsc{ec} units, so in total we have $134\times 2=268$ pairs.

\item {$1$ fault in \textsc{ec$_1$} and $1$ fault in \textsc{ec$_2$}:} Any fault that triggers a non-trivial parity in part $1$ of Fig.~\ref{fig:ec_threshold} is corrected by the same \textsc{ec} unit. Recall from Sec.~\ref{app:memory} that there are $10$ bad locations that can cause undetected errors in the output of part $1$ in Fig.~\ref{fig:ec_threshold}. The same set of faults in part $1$ of the other \textsc{ec} unit together lead to a malignant pair. However, this excludes a fault in the same data qubit line in the other \textsc{ec}, since such faults after the \textsc{cz} gadget propagate a global phase. We illustrate this with two contrasting cases as shown below
\begin{itemize}
\item Consider a damping fault $E_1$ in \textsc{ec}$_1$ and \textsc{ec}$_2$ on the first qubit. These faults after the action of the \textsc{cz} gadget propagate a global phase which will not affect our output.
\item Consider a damping fault $E_1$ in \textsc{ec}$_1$ and \textsc{ec}$_2$ on the first qubit and second qubit respectively. These faults after the action of the \textsc{cz} gadget propagate a phase error in addition to damping errors, which cannot be corrected.

\end{itemize}

 Therefore, we have a total of $(3\times7) + (2\times8) + (2\times8) + (3\times7) = 74$ malignant pairs. A fault in part $1$ of of one \textsc{ec} unit and part $2$ of the other \textsc{ec} unit do not contribute to malignant pairs.

\item {\bf $1$ fault in \textsc{EC$_1$} (or \textsc{EC$_2$}) and $1$ fault in \textsc{CZ} gadget: } There are $10$ bad locations in a \textsc{ec} unit causing an error in the output without being detected in Fig.~\ref{fig:ec_threshold} as illustrated in Sec.~\ref{app:memory}. A further fault in the \textsc{cz} gadget in a different data qubit line can lead to a malignant pair. This includes $10 \times 3 = 30$ possibilities. A fault pair in the same data qubit line in a \textsc{ec} unit and \textsc{cz} gadget is benign since it propagates the same damping error at the output of one block and a phase error to the other block which is correctable. 

\item {\bf Malignant fault pairs in a \textsc{cz} gadget:} There are totally $4$ locations in a \textsc{cz} gadget leading to a total of $C(4,2)=6$ malignant pairs.

\item {\bf $1$ fault in the \textsc{ec$_1$} unit and 1 fault in \textsc{ec$_a$} + $\cR_{Za}$ (or $1$ fault in \textsc{ec$_2$} and $1$ fault in \textsc{ec$_b$}+ $\cR_{Zb}$): }
There are $10$ bad locations in a \textsc{ec} unit causing an error in the output without being detected in Fig.~\ref{fig:ec_threshold} as described in Sec.~\ref{app:memory}. 

Consider a further fault in part $1$ of \textsc{ec$_a$} in Fig.~\ref{fig:ec_threshold}. If the first fault is in \textsc{ec$_1$} is in the 1st qubit, there are $10$ bad locations in part $1$ of \textsc{ec$_a$} in Fig.~\ref{fig:ec_threshold} contributing to malignant pairs. Recall that part $1$ has totally $14$ locations, hence the remaining $4$ locations in part $1$ of \textsc{ec}$_a$ that are benign include $1$ \textsc{cnot}, $2$ rests along the data line of the first qubit and $1$ $Z$ measurement in the second ancilla line giving the parity information of the $3^{\rm rd}$ and $4^{\rm th}$ data qubits. 

If the first fault is in \textsc{ec$_1$} in part 1 of in Fig.~\ref{fig:ec_threshold} in the $2^{\rm nd}$ qubit, then there are $11$ bad locations in part $1$ of \textsc{ec$_a$} ($3$ benign locations include $2$ rests along the second data qubit line and $1$ $Z$ measurement in the first ancilla line giving the parity information of $1^{\rm st}$ and $2^{\rm nd}$ data qubits in Fig.~\ref{fig:ec_threshold}) contributing to malignant pairs. Due to symmetry, there are $10$ and $11$ bad locations in part $1$ of \textsc{ec$_a$} for a fault in the $4^{\rm th}$ and $3^{\rm rd}$ data qubit line in part $1$ of \textsc{ec$_1$} respectively. So, there are $(3 \times 10) + (2 \times 11)+ (3 \times 10) + (2 \times 11)=104$ malignant pairs in this case. 

Consider a second fault in part $2$ of \textsc{ec$_a$} in Fig.~\ref{fig:ec_threshold}. Recall that part $2$ of the \textsc{ec} unit has $28$ locations totally. The benign locations include the rests during $X$ measurement, $Z$ application and a fault during $Z$ application, a damping in the same data qubit line as the first fault ($1$ \textsc{cnot}, $1$ rest, $Z$ measurement on the ancilla line) in part $2$ in Fig.~\ref{fig:ec_threshold}. This includes $10 \times (28-11)=170$ malignant pairs.   

In total, there are $104+170=274$ malignant pairs in this case.

\item { \bf $1$ fault in \textsc{ec$_1$} and 1 fault in \textsc{ec$_b$}+ $\cR_{Zb}$ ($1$ fault in \textsc{ec$_2$} and 1 fault in \textsc{ec$_a$} + $\cR_{Za}$)}
Recall from Sec.~\ref{app:memory} that there are $10$ bad locations in part $1$ of the \textsc{ec} unit in Fig.~\ref{fig:ec_threshold} causing an error in the output without being detected . 

If the first fault in the \textsc{ec$_1$} affects the first or second (similarly third or fourth) data qubit in part $1$ of \textsc{ec$_1$}, a further fault in the third and fourth (first and second) data qubit lines in part $1$ of \textsc{ec$_b$} in Fig.~\ref{fig:ec_threshold} lead to an uncorrectable output contributing to the malignant pair. This includes $6$ bad locations ($4$ rests and $2$ \textsc{cnot}s) in part $1$ of \textsc{ec$_b$} in Fig.~\ref{fig:ec_threshold} in each case.  Hence, there are $10 \times 6 = 60$ malignant pairs. 

Consider a second fault in part $2$ of \textsc{ec$_b$} in Fig.~\ref{fig:ec_threshold} where $\cR_{Zb}$  alone is executed. Note that  $\cR_{Za}$ and $\cR_{Zb}$ have the same construction as $\mathcal{R}$ in Fig.~\ref{fig:ec_threshold}.There are totally $18$ bad locations for $\cR_{Zb}$, leading to a total of $10\times 18=180$ malignant pairs. 

In total, there are $240$ malignant pairs in this case. 

\item {\bf $1$ fault in \textsc{cz} gadget and $1$ fault in \textsc{ec$_a$}+$\cR_{Za}$ (or \textsc{ec$_b$}+$\cR_{Zb}$):} This case is similar to the case (5) above, except that there are only $4$ bad locations in a\textsc{cz} gadget. Each bad location in a \textsc{cz} gadget combines with part $1$ of \textsc{ec} to give $10 + 11 + 11 + 10=42$. Each bad location in a\textsc{cz} gadget combines with part $2$ of the \textsc{ec} unit leading to  $17\times4 = 68$ malignant pairs. We refer to item (5) above to recall how only $17$ out of $28$ locations in part $2$ of the \textsc{ec} unit contribute. In total, we have $110$ malignant pairs.

\item {\bf $1$ fault in \textsc{ec$_a$}+ $\cR_{Za}$ and $1$ fault in \textsc{ec$_b$}+ $\cR_{Zb}$:} The \textsc{ec$_a$} unit triggers the R$_{Zb}$ only when there is a nontrivial parity detected in \textsc{EC$_a$} and trivial parity detected in \textsc{EC$_b$}. In order to have non-trivial parity outcomes, the fault must be in one of $4$ locations in the first time step in \textsc{EC$_a$}. Although these faults are correctable within \textsc{EC$_a$}. A further fault after detecting a trivial parity in \textsc{EC$_b$} could introduce a logical $Z$ error. This includes $10$ bad locations in part $1$ in Fig.~\ref{fig:ec_threshold} or $18$ bad locations in R$_{Zb}$. So, there are totally $4 \times (10 + 18) \times 2 = 224$ malignant pairs. The factor of $2$ is due the possibility of the other case where \textsc{ec$_b$} could record nontrivial parity and \textsc{ec$_a$} records trivial parity outcome .  

\item {\bf $2$ faults in \textsc{ec$_a$}+$\cR_{Za}$ (or \textsc{ec$_b$}+ $\cR_{Zb}$): }
We do not count for malignant pairs within a \textsc{ec$_a$}+ $\cR_{Za}$ (or \textsc{ec$_b$}+$\cR_{Zb}$) since that could go into the threshold calculation of a subsequent extended unit. This avoids double counting of the malignant pairs.
\end{enumerate}

\subsubsection{Malignant faults due to second order phase errors}
For a second order $Z$ error, the total number of malignant locations are contributed by locations from part $1$ of a \textsc{ec} unit and \textsc{cz} gadget. The $Z$ measurements on the ancilla of part $1$ in Fig.~\ref{fig:ec_threshold} are benign since they are not affected by phase errors. Therefore we have totally $(14-2) \times 4 + 4 = 52$ malignant fault locations.  

Therefore, the total number of malignant pairs due to damping errors and malignant fault locations due to $Z$ errors is  given by, $A = 1880$, leading to a threshold of $p_{\rm th} = 1/ A \approx 5.31 \times 10^{-4}$.


\section{The \textsc{ccz} gadget}\label{sec:ccz}
Recall that the three-qubit controlled-controlled-\textsc{phase} (\textsc{ccz}) gate, which is a non-Clifford gate, forms a universal gate set along with the Hadamard gate ($H$), a single-qubit Clifford gate. From the analysis shown in Sec.~\ref{app:cz} it is clear that a gate construction with a phase gate is naturally more immune to the amplitude-damping noise. Hence it emerges that an alternate set of universal gates for the $[[4,1]]$ code, would be the set comprising the \textsc{ccz} and $H$ gadgets. We demonstrate a fault-tolerant construction of the \textsc{ccz} gate in this section.

We obtain a construction for the \textsc{ccz} gadget as two successive transversal applications of individual \textsc{ccz} gates, as shown below in Fig.~\ref{fig:ccz_2}. The transversal \textsc{ccz} gadget satisfies the properties described in Lemma~\ref{lem:ccz}.
 \begin{figure}[H]
\centering
\includegraphics[scale=.7]{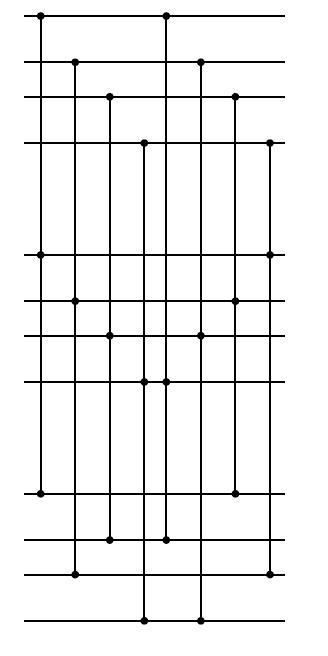}
\caption{\textsc{ccz} gadget}
\label{fig:ccz_2}
\end{figure}

\begin{lemma} [Properties of the \textsc{ccz} gadget]\label{lem:ccz}
The \textsc{ccz} gadget with at most a single damping fault in a block propagates an incoming state with no-error to an output with upto one single-qubit damping error in the same block with or without an induced \text{cphase} error in the other two blocks. Furthermore, a \textsc{ccz} gadget without any faults, propagates an input with upto one single-qubit damping error in a block to an output with at most one single-qubit damping error in the same block with or without an induced \textsc{cphase} error in the other two blocks.
\end{lemma}
\begin{proof}
A \textsc{ccz} gate propagates a damping error occurring before the control or target as a damping error after control or target respectively and induces a \textsc{cphase} error on the other two qubits, as denoted in Fig.~\ref{fig:cczgate}. This is easily verifiable via the following steps. 
\begin{eqnarray}
\textsc{CCZ} (E_1 \otimes I \otimes I)\textsc{CCZ}&=&E_1 \otimes (|0\rangle \langle0| \otimes I +|1\rangle \langle1| \otimes Z) \nonumber \\ \nonumber
\textsc{CCZ} (I \otimes E_1 \otimes I)\textsc{CCZ}&=&(|0\rangle \langle0| \otimes E_1\otimes I +|1\rangle \langle1| \otimes E_1\otimes Z) \nonumber \\ \nonumber
\textsc{CCZ} (I \otimes I \otimes E_1)\textsc{CCZ}&=&(|0\rangle \langle0| \otimes I +|1\rangle \langle1| \otimes Z)\otimes E_1
\end{eqnarray}

\begin{figure}[H]
\centering
\includegraphics[scale=.55]{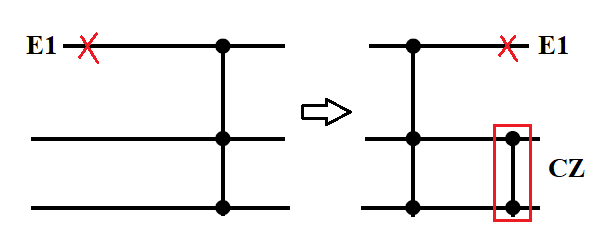}
\caption{ Propagation of a damping error $E_1$ through a faulty \textsc{ccz} gate.}
\label{fig:cczgate}
\end{figure}
Thus, this explains why a \textsc{ccz} gadget satisfies Lemma~\ref{lem:ccz}.
However, note that the \textsc{ccz} gadget propagates a damping error to one of the errors within the correctable set as discussed in Sec.~\ref{sec:noisemodel}.
\end{proof}
We now present the construction of a \textsc{ccz}-\textsc{exrec} in Fig.~\ref{fig:cczexrec} which is tolerant upto single-qubit damping errors. This unit has been constructed such that it can correct for a single damping fault or a single-qubit damping error occurring anywhere and a \textsc{cphase} error induced due to a single damping fault or single damping error. We explain this construction below.

 \begin{figure}[H]
\centering
\includegraphics[scale=.7]{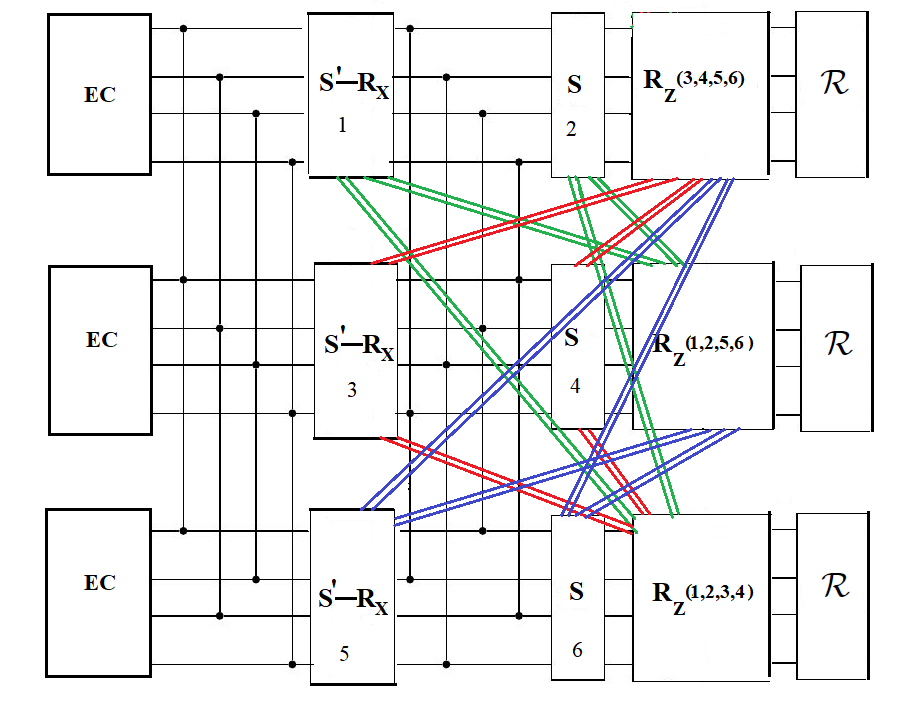}
\caption{\textsc{ccz}-\textsc{exrec}}
\label{fig:cczexrec}
\end{figure}

The \textsc{ccz}-\textsc{exrec} has a \textsc{ec} unit prefixed to each block, followed by two successive transversal applications of the level-$0$ \textsc{ccz} gates, interspersed with local syndrome and recovery units on each block. The units $S$ and $S'$ are the syndrome extraction units given in Fig.~\ref{fig:ec}. In addition to the single-damping recovery unit $\cR$, the \textsc{ccz} gadget also requires two other recovery units, namely, the $\cR_X$ unit described in the Fig.~\ref{fig:r_ec} and the $\cR_Z$ unit already described in Fig.~\ref{fig:r_ec}, whose application is conditioned on the syndrome bits $\{s,t,u,v,h,g\}$. 

The syndrome units $S'$ and $S$ indicate the location of the fault and accordingly one or more of the recovery units are applied. While the phase recovery unit $R_Z$ corrects for the \textsc{cphase} error that arises in the event of a fault occurring before any of the \textsc{ccz} gates in Fig.~\ref{fig:cczexrec}. The $R_{X}$ unit has been suitably introduced so as to allow us to distinguish between the faults occurring before the first application of \textsc{ccz} gates and the faults occurring after the first application of \textsc{ccz} gates . The $\cR$ units at the end correct for the single-damping errors after the second set of \textsc{ccz} gates, and bring the state back into the codespace. 

\begin{figure}[H]
\begin{center}
\caption{ Recovery units $R_X$ (left) and $R_Z$ (right) used in the \textsc{ccz} gadget.}
\label{fig:Rx}
\end{center}
\end{figure}

The phase recovery unit $R_Z$ in each block is conditioned on the $S,S'$ units of other blocks, and this is indicated by the notation $R_{Z}(m,n,q,r)$, where $\{m,n,q,r\}$ $\in$ $\{1,2,3,4,5,6\}$ denote the block numbers of the conditioning syndrome units $S,S'$. We indicate how each $R_Z$ is conditioned using coloured double lines in the Fig.~\ref{fig:cczexrec}. Note that in Fig.~\ref{fig:r_ec}, $Z^{(S,S')/c}$ applies one of the four single-qubit $Z$ operators $\{ZIII,IZII,IIZI,IIIZ\}$, depending on the outcomes of the units $S,S'$. As already explained in the proof of Lemma~\ref{lem:cphase}, the outcome $c$ in the $R_{Z}$ unit allows us to identify if a phase error has been induced due to a damping fault. In case $c=1$, then the phase error has been induced and one corrects for it by applying a local $Z$ gate on one of the four qubits depending on the conditioning syndrome bits from the other blocks $S, S'$. If $c=0$, no phase error has occurred and one need not apply any local $Z$ gate. The unit $\cR$ in Fig.~\ref{fig:cczexrec} refers to the original single-damping recovery unit described in Fig.~\ref{fig:r_ec}. We can also demonstrate the fault tolerance of the \textsc{ccz} gadget to amplitude-damping noise, based on Lemma~\ref{lem:ccz}.

%% file: Thesis_arxiv.bbl
\begin{thebibliography}{100}
\providecommand{\url}[1]{\texttt{#1}}
\providecommand{\urlprefix}{URL }

\bibitem{deutsch}
D.~Deutsch.
\newblock \emph{Quantum theory, the Church--Turing principle and the universal
  quantum computer}.
\newblock Proceedings of the Royal Society of London. A. Mathematical and
  Physical Sciences \textbf{400}, 97 (1985).

\bibitem{deutschjosza}
D.~Deutsch and R.~Jozsa.
\newblock \emph{Rapid solution of problems by quantum computation}.
\newblock Proceedings of the Royal Society of London. Series A: Mathematical
  and Physical Sciences \textbf{439}, 553 (1992).

\bibitem{shor_factor}
P.~W. Shor.
\newblock \emph{Polynomial-time algorithms for prime factorization and discrete
  logarithms on a quantum computer}.
\newblock SIAM review \textbf{41}, 303 (1999).

\bibitem{grover}
L.~K. Grover.
\newblock \emph{A fast quantum mechanical algorithm for database search}.
\newblock In \emph{Proceedings of the twenty-eighth annual ACM symposium on
  Theory of computing}, 212--219 (1996).

\bibitem{nielsen}
M.~A. Nielsen and I.~L. Chuang.
\newblock \emph{Quantum Computation and Quantum Information}.
\newblock Cambridge University Press (2000).

\bibitem{shor_qec}
P.~W. Shor.
\newblock \emph{Scheme for reducing decoherence in quantum computer memory}.
\newblock Physical review A \textbf{52}, R2493 (1995).

\bibitem{shor_ft}
P.~W. Shor.
\newblock \emph{Fault-tolerant quantum computation}.
\newblock In \emph{Proceedings of 37th conference on foundations of computer
  science}, 56--65. IEEE (1996).

\bibitem{terhal}
B.~M. Terhal.
\newblock \emph{Quantum error correction for quantum memories}.
\newblock Reviews of Modern Physics \textbf{87}, 307 (2015).

\bibitem{steane}
A.~M. Steane.
\newblock \emph{Error correcting codes in quantum theory}.
\newblock Physical Review Letters \textbf{77}, 793 (1996).

\bibitem{wootters}
W.~K. Wootters and W.~H. Zurek.
\newblock \emph{A single quantum cannot be cloned}.
\newblock Nature \textbf{299}, 802 (1982).

\bibitem{bennet}
C.~H. Bennett, D.~P. DiVincenzo, J.~A. Smolin, and W.~K. Wootters.
\newblock \emph{Mixed-state entanglement and quantum error correction}.
\newblock Physical Review A \textbf{54}, 3824 (1996).

\bibitem{knill}
E.~Knill, R.~Laflamme, and L.~Viola.
\newblock \emph{Theory of quantum error correction for general noise}.
\newblock Physical Review Letters \textbf{84}, 2525 (2000).

\bibitem{ekert}
A.~Ekert and C.~Macchiavello.
\newblock \emph{Error correction in quantum communication}.
\newblock arXiv preprint quant-ph/9602022  (1996).

\bibitem{laflamme}
R.~Laflamme, C.~Miquel, J.~P. Paz, and W.~H. Zurek.
\newblock \emph{Perfect quantum error correcting code}.
\newblock Physical Review Letters \textbf{77}, 198 (1996).

\bibitem{preskill_FT}
J.~Preskill.
\newblock \emph{Fault-tolerant quantum computation} 213--269 (1998).

\bibitem{preskill_nisq}
J.~Preskill.
\newblock \emph{Quantum Computing in the NISQ era and beyond}.
\newblock Quantum \textbf{2}, 79 (2018).

\bibitem{haroche}
J.-M. Raimond and S.~Haroche.
\newblock \emph{Exploring the quantum}.
\newblock Oxford University Press \textbf{82}, 86 (2006).

\bibitem{dephasing}
J.~J. Burnett, A.~Bengtsson, M.~Scigliuzzo, D.~Niepce, M.~Kudra, P.~Delsing,
  and J.~Bylander.
\newblock \emph{Decoherence benchmarking of superconducting qubits}.
\newblock npj Quantum Information \textbf{5}, 1 (2019).

\bibitem{leung}
D.~W. Leung, M.~A. Nielsen, I.~L. Chuang, and Y.~Yamamoto.
\newblock \emph{Approximate quantum error correction can lead to better codes}.
\newblock Physical Review A \textbf{56}, 2567 (1997).

\bibitem{fletcher_codes}
A.~S. Fletcher, P.~W. Shor, and M.~Z. Win.
\newblock \emph{Channel-adapted quantum error correction for the amplitude
  damping channel}.
\newblock IEEE Transactions on Information Theory \textbf{54}, 5705 (2008).

\bibitem{fletcher_rec}
A.~S. Fletcher, P.~W. Shor, and M.~Z. Win.
\newblock \emph{Optimum quantum error recovery using semidefinite programming}.
\newblock Physical Review A \textbf{75}, 012338 (2007).

\bibitem{hui_prabha}
H.~K. Ng and P.~Mandayam.
\newblock \emph{Simple approach to approximate quantum error correction based
  on the transpose channel}.
\newblock Physical Review A \textbf{81}, 062342 (2010).

\bibitem{gottesman_stabilizer}
D.~Gottesman.
\newblock \emph{Stabilizer codes and quantum error correction}.
\newblock arXiv preprint quant-ph/9705052  (1997).

\bibitem{Petz}
M.~Ohya and D.~Petz.
\newblock \emph{Quantum entropy and its use}.
\newblock Springer Science \& Business Media (2004).

\bibitem{fletcherthesis}
A.~S. Fletcher.
\newblock \emph{Channel-adapted quantum error correction}.
\newblock arXiv preprint arXiv:0706.3400  (2007).

\bibitem{schum}
B.~Schumacher.
\newblock \emph{Sending entanglement through noisy quantum channels}.
\newblock Phys. Rev. A \textbf{54}, 2614 (1996).
\newblock \urlprefix\url{https://link.aps.org/doi/10.1103/PhysRevA.54.2614}.

\bibitem{kosut}
R.~L. Kosut and D.~A. Lidar.
\newblock \emph{Quantum error correction via convex optimization}.
\newblock Quantum Information Processing \textbf{8}, 443 (2009).

\bibitem{yamamoto}
N.~Yamamoto, S.~Hara, and K.~Tsumura.
\newblock \emph{Suboptimal quantum-error-correcting procedure based on
  semidefinite programming}.
\newblock Phys. Rev. A \textbf{71}, 022322 (2005).
\newblock \urlprefix\url{https://link.aps.org/doi/10.1103/PhysRevA.71.022322}.

\bibitem{reimpell}
M.~Reimpell and R.~F. Werner.
\newblock \emph{Iterative optimization of quantum error correcting codes}.
\newblock Physical review letters \textbf{94}, 080501 (2005).

\bibitem{Barnum}
H.~Barnum and E.~Knill.
\newblock \emph{Reversing quantum dynamics with near-optimal quantum and
  classical fidelity}.
\newblock Journal of Mathematical Physics \textbf{43}, 2097 (2002).

\bibitem{beny}
C.~B\'eny and O.~Oreshkov.
\newblock \emph{General Conditions for Approximate Quantum Error Correction and
  Near-Optimal Recovery Channels}.
\newblock Phys. Rev. Lett. \textbf{104}, 120501 (2010).
\newblock
  \urlprefix\url{https://link.aps.org/doi/10.1103/PhysRevLett.104.120501}.

\bibitem{ak_cartan}
A.~Jayashankar, A.~M. Babu, H.~K. Ng, and P.~Mandayam.
\newblock \emph{Finding good quantum codes using the Cartan form}.
\newblock Physical Review A \textbf{101}, 042307 (2020).

\bibitem{khaneja_glaser}
N.~Khaneja and S.~J. Glaser.
\newblock \emph{Cartan decomposition of SU (2n) and control of spin systems}.
\newblock Chemical Physics \textbf{267}, 11 (2001).

\bibitem{bose}
S.~Bose.
\newblock \emph{Quantum communication through an unmodulated spin chain}.
\newblock Physical review letters \textbf{91}, 207901 (2003).

\bibitem{christandl}
M.~Christandl, N.~Datta, A.~Ekert, and A.~J. Landahl.
\newblock \emph{Perfect state transfer in quantum spin networks}.
\newblock Physical review letters \textbf{92}, 187902 (2004).

\bibitem{christandl2005perfect}
M.~Christandl, N.~Datta, T.~C. Dorlas, A.~Ekert, A.~Kay, and A.~J. Landahl.
\newblock \emph{Perfect transfer of arbitrary states in quantum spin networks}.
\newblock Physical Review A \textbf{71}, 032312 (2005).

\bibitem{albanesemirror}
C.~Albanese, M.~Christandl, N.~Datta, and A.~Ekert.
\newblock \emph{Mirror inversion of quantum states in linear registers}.
\newblock Physical review letters \textbf{93}, 230502 (2004).

\bibitem{karbach}
P.~Karbach and J.~Stolze.
\newblock \emph{Spin chains as perfect quantum state mirrors}.
\newblock Physical Review A \textbf{72}, 030301 (2005).

\bibitem{di}
C.~Di~Franco, M.~Paternostro, and M.~Kim.
\newblock \emph{Perfect state transfer on a spin chain without state
  initialization}.
\newblock Physical review letters \textbf{101}, 230502 (2008).

\bibitem{conclusive}
D.~Burgarth and S.~Bose.
\newblock \emph{Conclusive and arbitrarily perfect quantum-state transfer using
  parallel spin-chain channels}.
\newblock Physical Review A \textbf{71}, 052315 (2005).

\bibitem{perfect}
D.~Burgarth and S.~Bose.
\newblock \emph{Perfect quantum state transfer with randomly coupled quantum
  chains}.
\newblock New journal of physics \textbf{7}, 135 (2005).

\bibitem{efficient}
D.~Burgarth, V.~Giovannetti, and S.~Bose.
\newblock \emph{Efficient and perfect state transfer in quantum chains}.
\newblock Journal of Physics A: Mathematical and General \textbf{38}, 6793
  (2005).

\bibitem{bochkin}
G.~Bochkin, S.~Doronin, S.~Vasil'ev, A.~Fedorova, and E.~Fel'dman.
\newblock \emph{Experimental and theoretical investigations of quantum state
  transfer and decoherence processes in quasi-one-dimensional systems in
  multiple-quantum NMR experiments}.
\newblock In \emph{International Conference on Micro-and Nano-Electronics
  2016}, volume 10224, 102242E. International Society for Optics and Photonics
  (2016).

\bibitem{perez2013}
A.~Perez-Leija, R.~Keil, A.~Kay, H.~Moya-Cessa, S.~Nolte, L.-C. Kwek, B.~M.
  Rodr{\'\i}guez-Lara, A.~Szameit, and D.~N. Christodoulides.
\newblock \emph{Coherent quantum transport in photonic lattices}.
\newblock Physical Review A \textbf{87}, 012309 (2013).

\bibitem{chapman}
R.~J. Chapman, M.~Santandrea, Z.~Huang, G.~Corrielli, A.~Crespi, M.-H. Yung,
  R.~Osellame, and A.~Peruzzo.
\newblock \emph{Experimental perfect state transfer of an entangled photonic
  qubit}.
\newblock Nature communications \textbf{7}, 11339 (2016).

\bibitem{godsil2012}
C.~Godsil, S.~Kirkland, S.~Severini, and J.~Smith.
\newblock \emph{Number-theoretic nature of communication in quantum spin
  systems}.
\newblock Physical review letters \textbf{109}, 050502 (2012).

\bibitem{godsil}
L.~Banchi, G.~Coutinho, C.~Godsil, and S.~Severini.
\newblock \emph{Pretty good state transfer in qubit chains—the Heisenberg
  Hamiltonian}.
\newblock Journal of Mathematical Physics \textbf{58}, 032202 (2017).

\bibitem{osborne}
T.~J. Osborne and N.~Linden.
\newblock \emph{Propagation of quantum information through a spin system}.
\newblock Physical Review A \textbf{69}, 052315 (2004).

\bibitem{hasel}
H.~L. Haselgrove.
\newblock \emph{Optimal state encoding for quantum walks and quantum
  communication over spin systems}.
\newblock Physical Review A \textbf{72}, 062326 (2005).

\bibitem{ashhab}
S.~Ashhab.
\newblock \emph{Quantum state transfer in a disordered one-dimensional
  lattice}.
\newblock Physical Review A \textbf{92}, 062305 (2015).

\bibitem{chiara}
G.~De~Chiara, D.~Rossini, S.~Montangero, and R.~Fazio.
\newblock \emph{From perfect to fractal transmission in spin chains}.
\newblock Physical Review A \textbf{72}, 012323 (2005).

\bibitem{ak_statetransfer}
A.~Jayashankar and P.~Mandayam.
\newblock \emph{Pretty good state transfer via adaptive quantum error
  correction}.
\newblock Phys. Rev. A \textbf{98}, 052309 (2018).

\bibitem{kay}
A.~Kay.
\newblock \emph{Quantum error correction for state transfer in noisy spin
  chains}.
\newblock Physical Review A \textbf{93}, 042320 (2016).

\bibitem{kay2018perfect}
A.~Kay.
\newblock \emph{Perfect coding for dephased quantum state transfer}.
\newblock Physical Review A \textbf{97}, 032317 (2018).

\bibitem{allcock}
J.~Allcock and N.~Linden.
\newblock \emph{Quantum communication beyond the localization length in
  disordered spin chains}.
\newblock Physical review letters \textbf{102}, 110501 (2009).

\bibitem{anderson}
P.~W. Anderson.
\newblock \emph{Absence of diffusion in certain random lattices}.
\newblock Physical review \textbf{109}, 1492 (1958).

\bibitem{gottesman_intro}
D.~Gottesman.
\newblock \emph{An introduction to quantum error correction and fault-tolerant
  quantum computation}.
\newblock In \emph{Quantum information science and its contributions to
  mathematics, Proceedings of Symposia in Applied Mathematics}, volume~68,
  13--58 (2010).

\bibitem{gottesman_nature}
D.~Gottesman.
\newblock \emph{Quantum computing: Efficient fault tolerance}.
\newblock Nature \textbf{540}, 44 (2016).

\bibitem{aliferis}
P.~Aliferis, D.~Gottesman, and J.~Preskill.
\newblock \emph{Quantum accuracy threshold for concatenated distance-3 codes}.
\newblock Quantum Inf. Comput. \textbf{6}, 97 (2005).

\bibitem{aliferis_thesis}
P.~Aliferis.
\newblock \emph{Level reduction and the quantum threshold theorem}.
\newblock arXiv preprint quant-ph/0703230  (2007).

\bibitem{campbell}
E.~T. Campbell, B.~M. Terhal, and C.~Vuillot.
\newblock \emph{Roads towards fault-tolerant universal quantum computation}.
\newblock Nature \textbf{549}, 172 (2017).

\bibitem{lidar_brun}
D.~A. Lidar and T.~A. Brun.
\newblock \emph{Quantum error correction}.
\newblock Cambridge university press (2013).

\bibitem{bombin}
H.~Bomb{\'\i}n.
\newblock \emph{An introduction to topological quantum codes}.
\newblock arXiv preprint arXiv:1311.0277  (2013).

\bibitem{eastin_restrictions}
B.~Eastin and E.~Knill.
\newblock \emph{Restrictions on transversal encoded quantum gate sets}.
\newblock Physical review letters \textbf{102}, 110502 (2009).

\bibitem{chen_restrictions}
X.~Chen, H.~Chung, A.~W. Cross, B.~Zeng, and I.~L. Chuang.
\newblock \emph{Subsystem stabilizer codes cannot have a universal set of
  transversal gates for even one encoded qudit}.
\newblock Physical Review A \textbf{78}, 012353 (2008).

\bibitem{bravyi}
S.~Bravyi and A.~Kitaev.
\newblock \emph{Universal quantum computation with ideal Clifford gates and
  noisy ancillas}.
\newblock Physical Review A \textbf{71}, 022316 (2005).

\bibitem{aharonov}
D.~Aharonov, A.~Kitaev, and J.~Preskill.
\newblock \emph{Fault-tolerant quantum computation with long-range correlated
  noise}.
\newblock Physical review letters \textbf{96}, 050504 (2006).

\bibitem{knill_nature}
E.~Knill.
\newblock \emph{Quantum computing with realistically noisy devices}.
\newblock Nature \textbf{434}, 39 (2005).

\bibitem{ng}
H.~K. Ng and J.~Preskill.
\newblock \emph{Fault-tolerant quantum computation versus Gaussian noise}.
\newblock Physical Review A \textbf{79}, 032318 (2009).

\bibitem{aliferis_biasedExp}
P.~Aliferis, F.~Brito, D.~P. DiVincenzo, J.~Preskill, M.~Steffen, and B.~M.
  Terhal.
\newblock \emph{Fault-tolerant computing with biased-noise superconducting
  qubits: a case study}.
\newblock New Journal of Physics \textbf{11}, 013061 (2009).

\bibitem{gourlay}
I.~Gourlay and J.~F. Snowdon.
\newblock \emph{Concatenated coding in the presence of dephasing}.
\newblock Physical Review A \textbf{62}, 022308 (2000).

\bibitem{raus_FT}
R.~Raussendorf and J.~Harrington.
\newblock \emph{Fault-tolerant quantum computation with high threshold in two
  dimensions}.
\newblock Physical review letters \textbf{98}, 190504 (2007).

\bibitem{stephens}
A.~M. Stephens.
\newblock \emph{Fault-tolerant thresholds for quantum error correction with the
  surface code}.
\newblock Physical Review A \textbf{89}, 022321 (2014).

\bibitem{criger}
B.~Criger and B.~Terhal.
\newblock \emph{Noise thresholds for the [[4, 2, 2]]-concatenated toric code}.
\newblock arXiv preprint arXiv:1604.04062  (2016).

\bibitem{aliferis_biased}
P.~Aliferis and J.~Preskill.
\newblock \emph{Fault-tolerant quantum computation against biased noise}.
\newblock Physical Review A \textbf{78}, 052331 (2008).

\bibitem{FT_assymetry}
I.~Gourlay and J.~F. Snowdon.
\newblock \emph{Concatenated coding in the presence of dephasing}.
\newblock Phys. Rev. A \textbf{62}, 022308 (2000).
\newblock \urlprefix\url{https://link.aps.org/doi/10.1103/PhysRevA.62.022308}.

\bibitem{gottesman_2016}
D.~Gottesman.
\newblock \emph{Quantum fault tolerance in small experiments}.
\newblock arXiv preprint arXiv:1610.03507  (2016).

\bibitem{vaidman}
L.~Vaidman, L.~Goldenberg, and S.~Wiesner.
\newblock \emph{Error prevention scheme with four particles}.
\newblock Phys. Rev. A \textbf{54}, R1745 (1996).
\newblock \urlprefix\url{https://link.aps.org/doi/10.1103/PhysRevA.54.R1745}.

\bibitem{grassl}
M.~Grassl and T.~Beth.
\newblock \emph{A note on non-additive quantum codes}.
\newblock arXiv preprint quant-ph/9703016  (1997).

\bibitem{flammia}
R.~Harper and S.~T. Flammia.
\newblock \emph{Fault-tolerant logical gates in the ibm quantum experience}.
\newblock Physical review letters \textbf{122}, 080504 (2019).

\bibitem{ak_ft}
A.~Jayashankar, M.~D.~H. Long, H.~K. Ng, and P.~Mandayam.
\newblock \emph{Achieving fault tolerance against amplitude damping noise}.
\newblock  arXiv preprint quant-ph/2107.05485  (2021).
\bibitem{choi}
M.~D.~Choi.
\newblock \emph{Completely positive linear maps on complex matrices}.
\newblock Linear algebra and its applications \textbf{10}, 285 (1975).

\bibitem{kraus1971}
K.~Kraus.
\newblock \emph{General state changes in quantum theory}.
\newblock Annals of Physics \textbf{64}, 311 (1971).

\bibitem{kraus}
K.~Kraus, A.~B{\"o}hm, J.~D. Dollard, and W.~Wootters.
\newblock \emph{States, effects, and operations: fundamental notions of quantum
  theory. Lectures in mathematical physics at the University of Texas at
  Austin}.
\newblock Lecture notes in physics \textbf{190} (1983).

\bibitem{lidar}
D.~A. Lidar, I.~L. Chuang, and K.~B. Whaley.
\newblock \emph{Decoherence-free subspaces for quantum computation}.
\newblock Physical Review Letters \textbf{81}, 2594 (1998).

\bibitem{lidar_paper}
D.~A. Lidar, I.~L. Chuang, and K.~B. Whaley.
\newblock \emph{Decoherence-free subspaces for quantum computation}.
\newblock Physical Review Letters \textbf{81}, 2594 (1998).

\bibitem{zanardi}
P.~Zanardi.
\newblock \emph{Stabilizing quantum information}.
\newblock Phys. Rev. A \textbf{63}, 012301 (2000).
\newblock \urlprefix\url{https://link.aps.org/doi/10.1103/PhysRevA.63.012301}.

\bibitem{shabani}
A.~Shabani and D.~A. Lidar.
\newblock \emph{Theory of initialization-free decoherence-free subspaces and
  subsystems}.
\newblock Phys. Rev. A \textbf{72}, 042303 (2005).
\newblock \urlprefix\url{https://link.aps.org/doi/10.1103/PhysRevA.72.042303}.

\bibitem{holbrook}
J.~A. Holbrook, D.~W. Kribs, and R.~Laflamme.
\newblock \emph{Noiseless subsystems and the structure of the commutant in
  quantum error correction}.
\newblock Quantum Information Processing \textbf{2}, 381 (2003).

\bibitem{kribs}
D.~Kribs, R.~Laflamme, and D.~Poulin.
\newblock \emph{Unified and generalized approach to quantum error correction}.
\newblock Physical review letters \textbf{94}, 180501 (2005).

\bibitem{beny2009}
C.~B{\'e}ny.
\newblock \emph{Conditions for the approximate correction of algebras}.
\newblock In \emph{Workshop on Quantum Computation, Communication, and
  Cryptography}, 66--75. Springer (2009).

\bibitem{klesse}
R.~Klesse.
\newblock \emph{Approximate quantum error correction, random codes, and quantum
  channel capacity}.
\newblock Physical Review A \textbf{75}, 062315 (2007).

\bibitem{schumacher}
B.~Schumacher and M.~D. Westmoreland.
\newblock \emph{Approximate quantum error correction}.
\newblock Quantum Information Processing \textbf{1}, 5 (2002).

\bibitem{petz_entropy}
D.~Petz.
\newblock \emph{Monotonicity of quantum relative entropy revisited}.
\newblock Reviews in Mathematical Physics \textbf{15}, 79 (2003).

\bibitem{prabha_thesis}
P.~Mandayam~Doddamane.
\newblock \emph{Emerging Paradigms in Quantum Error Correction and Quantum
  Cryptography}.
\newblock Ph.D. thesis, California Institute of Technology (2011).

\bibitem{prabha_aqec}
P.~Mandayam and H.~K. Ng.
\newblock \emph{Towards a unified framework for approximate quantum error
  correction}.
\newblock Physical Review A \textbf{86}, 012335 (2012).

\bibitem{cafaro}
C.~Cafaro and P.~van Loock.
\newblock \emph{Approximate quantum error correction for generalized
  amplitude-damping errors}.
\newblock Physical Review A \textbf{89}, 022316 (2014).

\bibitem{elizabeth}
F.~G. Brandao, E.~Crosson, M.~B. {\c{S}}ahino{\u{g}}lu, and J.~Bowen.
\newblock \emph{Quantum error correcting codes in eigenstates of
  translation-invariant spin chains}.
\newblock Physical Review Letters \textbf{123}, 110502 (2019).

\bibitem{preskill_symmetry}
P.~Faist, S.~Nezami, V.~V. Albert, G.~Salton, F.~Pastawski, P.~Hayden, and
  J.~Preskill.
\newblock \emph{Continuous symmetries and approximate quantum error
  correction}.
\newblock arXiv preprint arXiv:1902.07714  (2019).

\bibitem{todd}
J.~R. Gonzalez~Alonso and T.~A. Brun.
\newblock \emph{Protecting orbital-angular-momentum photons from decoherence in
  a turbulent atmosphere}.
\newblock Phys. Rev. A \textbf{88}, 022326 (2013).
\newblock \urlprefix\url{https://link.aps.org/doi/10.1103/PhysRevA.88.022326}.

\bibitem{wilde}
L.~Lami, S.~Das, and M.~M. Wilde.
\newblock \emph{Approximate reversal of quantum Gaussian dynamics}.
\newblock Journal of Physics A: Mathematical and Theoretical \textbf{51},
  125301 (2018).

\bibitem{petz_implementation}
A.~Gily{\'e}n, S.~Lloyd, I.~Marvian, Y.~Quek, and M.~M. Wilde.
\newblock \emph{Quantum algorithm for Petz recovery channels and pretty good
  measurements}.
\newblock arXiv preprint arXiv:2006.16924  (2020).

\bibitem{preskill2018quantum}
J.~Preskill.
\newblock \emph{Quantum Computing in the NISQ era and beyond}.
\newblock Quantum \textbf{2}, 79 (2018).

\bibitem{shor}
P.~W. Shor.
\newblock \emph{Scheme for reducing decoherence in quantum computer memory}.
\newblock Physical review A \textbf{52}, R2493 (1995).

\bibitem{calderblank}
A.~R. Calderbank and P.~W. Shor.
\newblock \emph{Good quantum error-correcting codes exist}.
\newblock Physical Review A \textbf{54}, 1098 (1996).

\bibitem{gottesman}
D.~Gottesman.
\newblock \emph{Stabilizer codes and quantum error correction}.
\newblock arXiv preprint quant-ph/9705052  (1997).

\bibitem{tannu}
S.~S. Tannu and M.~K. Qureshi.
\newblock \emph{Not all qubits are created equal: a case for variability-aware
  policies for NISQ-era quantum computers}.
\newblock In \emph{Proceedings of the Twenty-Fourth International Conference on
  Architectural Support for Programming Languages and Operating Systems},
  987--999. ACM (2019).

\bibitem{yamam}
N.~Yamamoto, S.~Hara, and K.~Tsumura.
\newblock \emph{Suboptimal quantum-error-correcting procedure based on
  semidefinite programming}.
\newblock Physical Review A \textbf{71}, 022322 (2005).

\bibitem{wang}
X.~Wang, M.~Byrd, and K.~Jacobs.
\newblock \emph{Minimal noise subsystems}.
\newblock Physical Review Letters \textbf{116}, 090404 (2016).

\bibitem{nelderpaper}
J.~A. Nelder and R.~Mead.
\newblock \emph{A simplex method for function minimization}.
\newblock The Computer Journal \textbf{7}, 308 (1965).

\bibitem{NumericalRecipes}
W.~H. Press, S.~A. Teukolsky, W.~T. Vetterling, and B.~P. Flannery.
\newblock \emph{Numerical Recipes: The Art of Scientific Computing}.
\newblock Cambridge University Press, New York, 3 edition (2007).

\bibitem{cartan}
H.~N.~S. Earp and J.~K. Pachos.
\newblock \emph{A constructive algorithm for the Cartan decomposition of
  SU(2N)}.
\newblock Journal of Mathematical Physics \textbf{46}, 082108 (2005).

\bibitem{Zhang2003}
J.~Zhang, J.~Vala, S.~Sastry, and K.~B. Whaley.
\newblock \emph{Geometric theory of nonlocal two-qubit operations}.
\newblock Physical Review A \textbf{67}, 042313 (2003).

\bibitem{Rezakhani2004}
A.~T. Rezakhani.
\newblock \emph{Characterization of two-qubit perfect entanglers}.
\newblock Physical Review A \textbf{70}, 052313 (2004).

\bibitem{langshor}
R.~Lang and P.~W. Shor.
\newblock \emph{Nonadditive quantum error correcting codes adapted to the
  ampltitude damping channel}.
\newblock arXiv preprint arXiv:0712.2586  (2007).

\bibitem{gottesman_FT}
D.~Gottesman.
\newblock \emph{Theory of fault-tolerant quantum computation}.
\newblock Physical Review A \textbf{57}, 127 (1998).

\bibitem{kayreview}
A.~Kay.
\newblock \emph{Perfect, efficient, state transfer and its application as a
  constructive tool}.
\newblock International Journal of Quantum Information \textbf{8}, 641 (2010).

\bibitem{AD_reliable2017}
{\'A}.~Piedrafita and J.~M. Renes.
\newblock \emph{Reliable channel-adapted error correction: Bacon-Shor code
  recovery from amplitude damping}.
\newblock Physical review letters \textbf{119}, 250501 (2017).

\bibitem{dyson}
F.~J. Dyson.
\newblock \emph{The radiation theories of Tomonaga, Schwinger, and Feynman}.
\newblock Physical Review \textbf{75}, 486 (1949).

\bibitem{gottesman97}
D.~Gottesman.
\newblock \emph{Stabilizer codes and quantum error correction. Caltech Ph. D}.
\newblock Ph.D. thesis, Thesis, eprint: quant-ph/9705052 (1997).

\bibitem{magic_2005}
S.~Bravyi and A.~Kitaev.
\newblock \emph{Universal quantum computation with ideal Clifford gates and
  noisy ancillas}.
\newblock Physical Review A \textbf{71}, 022316 (2005).

\bibitem{knill_FT}
E.~Knill.
\newblock \emph{Scalable quantum computing in the presence of large
  detected-error rates}.
\newblock Phys. Rev. A \textbf{71}, 042322 (2005).
\newblock \urlprefix\url{https://link.aps.org/doi/10.1103/PhysRevA.71.042322}.

\bibitem{fletcher}
A.~S. Fletcher, P.~W. Shor, and M.~Z. Win.
\newblock \emph{Structured near-optimal channel-adapted quantum error
  correction}.
\newblock Physical Review A \textbf{77}, 012320 (2008).

\bibitem{liang}
M.~H. Michael, M.~Silveri, R.~T. Brierley, V.~V. Albert, J.~Salmilehto,
  L.~Jiang, and S.~M. Girvin.
\newblock \emph{New Class of Quantum Error-Correcting Codes for a Bosonic
  Mode}.
\newblock Phys. Rev. X \textbf{6}, 031006 (2016).
\newblock \urlprefix\url{https://link.aps.org/doi/10.1103/PhysRevX.6.031006}.

\bibitem{backaction}
L.~Hu, Y.~Ma, W.~Cai, X.~Mu, Y.~Xu, W.~Wang, Y.~Wu, H.~Wang, Y.~Song, C.-L.
  Zou, \emph{et~al.}
\newblock \emph{Quantum error correction and universal gate set operation on a
  binomial bosonic logical qubit}.
\newblock Nature Physics \textbf{15}, 503 (2019).

\bibitem{renes}
A.~Piedrafita and J.~M. Renes.
\newblock \emph{Reliable Channel-Adapted Error Correction: Bacon-Shor Code
  Recovery from Amplitude Damping}.
\newblock Phys. Rev. Lett. \textbf{119}, 250501 (2017).
\newblock
  \urlprefix\url{https://link.aps.org/doi/10.1103/PhysRevLett.119.250501}.

\bibitem{correlated_AD}
A.~D'Arrigo, G.~Benenti, G.~Falci, and C.~Macchiavello.
\newblock \emph{Classical and quantum capacities of a fully correlated
  amplitude damping channel}.
\newblock Physical Review A \textbf{88}, 042337 (2013).

\bibitem{lyakhov}
A.~Lyakhov and C.~Bruder.
\newblock \emph{Quantum state transfer in arrays of flux qubits}.
\newblock New journal of physics \textbf{7}, 181 (2005).

\bibitem{swaddle}
M.~Swaddle, L.~Noakes, H.~Smallbone, L.~Salter, and J.~Wang.
\newblock \emph{Generating three-qubit quantum circuits with neural networks}.
\newblock Physics Letters A \textbf{381}, 3391 (2017).

\end{thebibliography}
